\documentclass[11pt,leqno]{amsart}

\usepackage[english]{babel}
\usepackage[linktocpage=true,colorlinks=true, linkcolor=blue, citecolor=red, urlcolor=green]{hyperref}

\input xy
\xyoption{all} \CompileMatrices

\usepackage{amssymb,amsfonts}
\usepackage{amsmath,amsthm,amsxtra}
\usepackage[all]{xy}
\usepackage{color}
\usepackage{verbatim}

\usepackage{enumerate}

\makeatletter
\@addtoreset{equation}{section}
\makeatother

\newtheorem{theorem}{Theorem}[section]
\newtheorem{definition}[theorem]{Definition}
\newtheorem{lemma}[theorem]{Lemma}
\newtheorem{corollary}[theorem]{Corollary}
\newtheorem{proposition}[theorem]{Proposition}

\newtheorem*{lemma*}{Lemma}

\theoremstyle{remark}
\newtheorem{remark}[theorem]{Remark}
\newtheorem{example}[theorem]{Example}

\newcommand{\st}[1]{\ensuremath{^{\scriptstyle \textrm{#1}}}}

\newcommand\bigcheck[1]{#1 \raise1ex\hbox{$\hspace{-1ex}{}^\vee$}}
\newcommand\sucheck[1]{#1 \raise0.5ex\hbox{$\hspace{-1ex}{}^\vee$}}

\newcommand{\mc}[1]{{\mathcal #1}}
\newcommand{\mf}[1]{{\mathfrak #1}}
\newcommand{\mb}[1]{{\mathbb #1}}

\newcommand\tint{{\textstyle\int}}

\newcommand{\id}{{1 \mskip -5mu {\rm I}}}

\renewcommand{\tilde}{\widetilde}

\newcommand{\End}{\mathop{\rm End }}
\newcommand{\Hom}{\mathop{\rm Hom }}
\newcommand{\ad}{\mathop{\rm ad }}
\newcommand{\sign}{\mathop{\rm sign }}
\newcommand{\Der}{\mathop{\rm Der }}

\newcommand{\Mat}{\mathop{\rm Mat }}

\newcommand{\RCend}{\mathop{\rm RCend }}
\newcommand{\RCder}{\mathop{\rm RCder }}
\newcommand{\Cder}{\mathop{\rm Cder }}

\renewcommand{\ker}{\mathop{\rm Ker }}
\newcommand{\im}{\mathop{\rm Im }}

\newcommand{\Vect}{\mathop{\rm Vect }}
\newcommand{\CVect}{\mathop{\rm CVect }}

\newcommand{\sym}{\mathop{\rm sym }}

\newcommand{\as}{\mathop{\rm as }}
\newcommand{\op}{\mathop{\rm op }}
\newcommand{\var}{\mathop{\rm var }}

\newcommand{\ord}{\mathop{\rm ord }}

\newcommand{\codim}{\mathop{\rm codim }}

\definecolor{light}{gray}{.9}

\begin{document}


\title{The variational Poisson cohomology}

\author{
Alberto De Sole$^1$
\and
Victor G. Kac$^2$
}

\thanks{\!\!\!\!\!\!\!1. Dipartimento di Matematica, Universit\`a di Roma ``La Sapienza'',
00185 Roma, Italy\,
desole@mat.uniroma1.it
}

\thanks{\noindent 2. Department of Mathematics, M.I.T.,
Cambridge, MA 02139, USA\,
kac@math.mit.edu
}

\thanks{
$\phantom{.}$\\
A.D.S. was partially supported by PRIN and AST grants.
V.K. was partially supported by an NSF grant, and an ERC advanced grant.
Parts of this work were done while
A.D.S. was visiting the Department of Mathematics of M.I.T.,
while V.K. was visiting the newly created Center for Mathematics and Theoretical Physics in Rome,
and A.D.S and V.K. were visiting the MSC and the Department of Mathematics of Tsinghua University in Beijing.
}

\maketitle

\vspace{4pt}

\begin{center}
\emph{To the memory of Boris Kupershmidt (11/27/1946 -- 12/12/2010)}
\end{center}

\vspace{2pt}


\begin{abstract}
\noindent
It is well known that the validity of the so called Lenard-Magri scheme of integrability
of a bi-Hamiltonian PDE can be established
if one has some precise information on the corresponding 1st variational Poisson cohomology
for one of the two Hamiltonian operators.
In the first part of the paper we explain how to introduce various cohomology complexes,
including Lie superalgebra and Poisson cohomology complexes,
and basic and reduced Lie conformal algebra and Poisson vertex algebra cohomology complexes,
by making use of the corresponding universal Lie superalgebra or Lie conformal superalgebra.
The most relevant are certain subcomplexes of the basic and reduced Poisson vertex algebra
cohomology complexes, which we identify (non-canonically)
with the generalized de Rham complex and the generalized variational complex.
In the second part of the paper we compute the cohomology of the generalized de Rham complex,
and, via a detailed study of the long exact sequence, we compute the cohomology of the generalized variational complex
for any quasiconstant coefficient Hamiltonian operator with invertible leading coefficient.
For the latter we use some differential linear algebra developed in the Appendix.
\end{abstract}

\vfill\eject

\tableofcontents

\vfill\eject


\section{Introduction}
\label{sec:intro}

\noindent
The theory of Poisson vertex algebras is a very convenient
framework for the theory of Hamiltonian partial differential
equations \cite{BDSK}.

First, let us introduce some key notions.  Let $V$ be a unital
commutative associative  algebra with a derivation $\partial$.
The space $\mf g_{-1}=V/\partial V$ is called the space of {\em
  Hamiltonian functions}, the image of $h \in V$ in $\mf g_{-1}$
being denoted by $\tint h$.  The Lie algebra $\mf g_0$ of all
derivations of $V$, commuting with $\partial$, is called the Lie
algebra of {\em evolutionary vector fields}.  Its action on $V$
descends to $\mf g_{-1}$.

A $\lambda$-{\em bracket} on $V$ is a linear map $V \otimes V \to
\mb F [\lambda] \otimes V$, $a \otimes b \mapsto \{ a_\lambda b
\}$, satisfying the following three properties $(a,b \in V)$:\\

\begin{list}{}{}
\item (sesquilinearity)  \quad  $\{ \partial a_\lambda b \}
  =-\lambda \{ a_\lambda b \}$, $\{ a_\lambda \partial b \} =
  (\partial +\lambda) \{ a_\lambda b \}$,\\
\item (skewcommutativity) \quad $\{ b_\lambda a \} =-\{
  a_{-\partial-\lambda}b \}$,\\
 where $\partial$ is moved to the left,\\
\item (Leibniz rule)\quad $\{ a_\lambda bc \} = \{ a_\lambda b\}c+
  b \{a_\lambda c\}$.
\end{list}

\noindent
Denote by $\mf g_1$ the space of all $\lambda$-brackets on $V$.

One of the basic constructions of the present paper is the
$\mb Z$-graded Lie superalgebra
\begin{equation}
  \label{eq:0.1}
  W^{\partial ,\as} \left(\Pi V\right) = \left(\Pi
    \mf g_{-1}\right) \oplus \mf g_0
     \oplus \left(\Pi \mf g_1\right) \oplus  \dots \, ,
\end{equation}
where $\mf g_{-1}$, $\mf g_0$ and $\mf g_1$ are as above and $\Pi\mf g_i$
stands for the space $\mf g_i$ with reversed parity.  For
$\tint f, \tint g \in \Pi \mf g_{-1}$, $X,Y \in \mf g_0$ and $H \in
\Pi \mf g_1$ the commutators are defined as follows:
\begin{eqnarray}
  \label{eq:0.2}
\hspace*{15ex} \left[\tint f, \tint g \right] &=& 0 \, ,  \\
  \label{eq:0.3}
   \left[X, \tint f\right]  &=& \tint X(f)\, , \\
  \label{eq:0.4}
    [X,Y] &=& XY-YX\, , \\
  \label{eq:0.5}
    \left[\{ \, .\, {}_\lambda \,. \, \}_H, \tint f\right] (g)
       &=& \{ f_\lambda g \}_H\big|_{\lambda=0}\, , \\
  \label{eq:0.6}
    \{ f_\lambda g \}_{[X,H]} &=& X (\{ f_\lambda g \}_H)
        - \{ X (f)_\lambda g \}_H - \{ f_\lambda , X(g)\}_H\, .
\end{eqnarray}
In Section 5 we construct explicitly the whole Lie superalgebra
$W^{\partial,\as} (\Pi V)$, but for applications to Hamiltonian
PDE one needs only the condition $[H,K]=0$ for $H,K \in \Pi
\mf g_1$, which is as follows $(f,g,h \in V)$:
\begin{equation}
  \label{eq:0.7}
  \{\{ f_\lambda g \}_{K \, \lambda +\mu} h \}_H
    \!-\! \{ f_\lambda \{ g_\mu h \}_K \}_H \!+\!
     \{ g_\mu \{ f_\lambda h \}_K\}_H \!+\! (H \leftrightarrow K)=0\,.
\end{equation}

A $\lambda$-bracket $\{ \, .\, {}_\lambda \,. \, \} = \{ \, .\,
{}_\lambda \,. \, \}_H$ is called a {\em Poisson}
$\lambda$-bracket if $[H,H]=0$, i.e.,~one has\\

(Jacobi identity) \quad $\{ f_\lambda \{ g_\mu h \}\}-
  \{ g_\mu \{ f_\lambda h \}\} = \{\{ f_\lambda g \}_{\lambda
    +\mu}h \}$.\\

The differential algebra $V$, endowed with a Poisson $\lambda$-bracket
, is called a {\em Poisson vertex algebra} (PVA) \cite{DSK1}.  
Two Poisson
$\lambda$-brackets $\{ \, .\, {}_\lambda \,. \, \}_H$ and $\{ \,
.\, {}_\lambda \,. \, \}_K$ on $V$ are called {\em compatible} if
(\ref{eq:0.7}) holds, which means that their sum is a Poisson $\lambda$-bracket as well.

One of the key properties of a PVA $V$
 is that the vector space $V/\partial V$ carries a well-defined
 Lie algebra structure, given by
 \begin{equation}
   \label{eq:0.8}
   \left\{ \tint f , \tint g \right\} = \tint \left\{ f_\lambda g \right\}
      \big|_{\lambda     =0}\, , \quad f,g \in V \, .
 \end{equation}
Moreover, $V$ is a left module over the Lie algebra $V/\partial
V$ with the well-defined action
\begin{equation}
  \label{eq:0.9}
 \{ \tint f ,g \} = \{ f_\lambda g \} \big|_{\lambda =0}\, ,
  \quad f,g \in V \,,
\end{equation}
by derivations, commuting with $\partial$, of the associative
product in $V$ and of the $\lambda$-bracket.  In particular,  all
the derivations $X^f = \{ \tint f \, , \, . \, \}$ of
$V$ are evolutionary; they are called {\em Hamiltonian vector
  fields}.

Two Hamiltonian functions $\tint f$ and $\tint g$ are said to be
{\em in involution} if
\begin{equation}
  \label{eq:0.10}
  \left\{ \tint f , \tint g \right\} =0 \, .
\end{equation}Given a Hamiltonian function $\tint h \in V
/\partial V$ and a Poisson $\lambda$-bracket on $V$, the
corresponding Hamiltonian equation is defined by the Hamiltonian
vector field $X^h$:
\begin{equation}
  \label{eq:0.11}
  \frac{du}{dt}=\left\{\tint h,u\right\}\,,\quad u\in V\,.
\end{equation}

The equation (\ref{eq:0.11}) is called integrable if $\tint h$ is
contained in an infinite-dimensional abelian subalgebra of the
Lie algebra $V /\partial V$ with bracket (\ref{eq:0.8}).  Picking
a basis $\tint h_0 = \tint h \,, \tint h_1\, \tint h_2, \dots$ of
this abelian subalgebra, we obtain a hierarchy of integrable
Hamiltonian equations
\begin{equation}
  \label{eq:0.12}
  \frac{du}{dt_n} = \left\{ \tint h_n ,u \right\} \, , \quad n \in \mb Z_+\, ,
\end{equation}
which are compatible since the corresponding Hamiltonian vector
fields $X^{h_n}$ commute.

The basic device for proving integrability of a Hamiltonian
equation is the so-called Lenard-Magri scheme, which is the
following simple observation,
first mentioned in \cite{GGKM} and \cite{Lax};
a survey of related results up to the early 90's can be found in \cite{Dor}.
Suppose that the differential
algebra $V$ is endowed with two $\lambda$-brackets $\{ \, .\,
{}_\lambda \,. \, \}_H$ and $\{ \, .\, {}_\lambda \,. \, \}_K$
and assume that:
\begin{equation}
  \label{eq:0.13}
  \left\{ \tint h_n , u \right\}_H = \left\{ \tint h_{n+1},u \right\}_K \, , \quad n \in
  \mb Z_+\, , \, u \in V \, ,
\end{equation}
for some Hamiltonian functions $\tint h_n \in V /\partial V$.
Then all these Hamiltonian functions are in involution with
respect to both brackets $\{ \, .\, {}_\lambda \,. \, \}_H$ and
$\{ \, .\, {}_\lambda \,. \, \}_K$ on $V/\partial V$.

Note that we do not need to assume that the
$\lambda$-brackets are Poisson nor that they are compatible.
These assumptions enter when we try to prove the existence of the
sequence, $\tint h_n$, satisfying (\ref{eq:0.13}), as we explain below.

Indeed, let $H,K \in \Pi \mf g_1$ be two compatible Poisson
$\lambda$-brackets on $V$.  Since $[K,K] =0$, it follows that $(\ad K)^2=0$,
hence we may consider the {\em variational cohomology complex}
$(W^{\partial,\as} (\Pi V) = \bigoplus^\infty_{j=-1} W_j, \ad K )$,
where $W_{-1}=\Pi \mf g_{-1}$, $W_0 = \mf g_0$, $W_1 = \Pi \mf g_1,
\dots$.  By definition, $X \in W_0$ is \emph{closed}
if, in view of \eqref{eq:0.6}, it is a derivation of
the $\lambda$-bracket $\{ \, .\, {}_\lambda \,. \, \}_K$,
and it is exact if, in view of \eqref{eq:0.5}, $X = \{ h_\lambda \cdot \}_K
\big|_{\lambda =0}$ for some $h \in W_{-1}$.  Now we can find a
solution to (\ref{eq:0.13}), by induction on $n$ as follows.  By
Jacobi identity in $W^{\partial,\as}(\Pi V)$ we have $ [K,[H,h_n]] =
-[H, [K,h_n]]$, which, by inductive assumption,
equals $-[H,[H,h_{n-1}]]=0$,
since $[H,H]=0$ and $H\in W_1$ is odd.
Thus, the element $[H,h_{n-1}] \in W_0$ is closed.  If this
closed element is exact, i.e.,~it equals $[K,h_n]$ for some $h_n
\in W_{-1}$, we complete the $n$\st{th} step of induction.
In general we have
\begin{equation}\label{110320:eq1}
[H,h_{n-1}]=[K,h_n]+a_n\,,
\end{equation}
where $a_n$ is a representative of the corresponding cohomology class.
Looking at \eqref{110320:eq1} more carefully,
one often can prove that one can take $a_n=0$,
so the Lenard-Magri scheme still works.

The cohomological approach to the Lenard-Magri scheme was
proposed long ago in \cite{Kra} and \cite{Ol1}.  
However, no machinery has been developed in order to compute
this cohomology.  
In the present paper we develop such a machinery by introducing a
``covering'' complex $(\tilde{W}^{\partial,\as}, \ad K)$ of the
complex $(W^{\partial,\as}, \ad K)$, whose cohomology is much
easier to compute, and then study in detail the corresponding long
exact sequence.

What does this have to do with the classical Hamiltonian PDE, like
the KdV equation?  In order to explain this, consider the algebra
of differential polynomials $R_\ell = \mb F [u^{(n)}_i \big| i=1,
\dots ,\ell \, ; \, n \in \mb Z_+]$ with the derivation
$\partial$, defined on generators by $\partial
(u^{(n)}_i)=u^{(n+1)}_i$.  Here one should think of the $u_i$ as
functions, depending on a parameter~$t$ (time), in one independent
variable~$x$, which is a coordinate on a 1-dimensional
manifold~$M$, and of $\partial$ as the derivative by~$x$, so that
$u^{(n)}_i$ is the $n$\st{th} derivative of $u_i$.  Furthermore, one
should think of $\tint h \in R_\ell/\partial R_\ell$ as $\tint_M h
\, dx$ since $R_\ell /\partial R_\ell$ provides the universal
space in which integration by parts holds.

It is straightforward to check that  equation
(\ref{eq:0.11}) can be written in the following equivalent,
but more familiar, form:
\begin{equation}
  \label{eq:0.14}
  \frac{du}{dt} = H (\partial) \frac{\delta h}{\delta u}\, ,
\end{equation}
where $\frac{\delta h}{\delta u}$ is the vector of variational
derivatives
\begin{equation}
  \label{eq:0.15}
  \frac{\delta h}{\delta u_i} = \sum_{n \in \mb Z_+} (-\partial)^n
     \frac{\partial h}{\partial u^{(n)}_i}\, ,
\end{equation}
and $H (\partial) = (H_{ij}(\partial))^\ell_{i,j=1}$ is the $\ell
\times \ell$ matrix differential operator with entries $H_{ij}
(\partial) = \{ u_j {}_{\partial} u_i \}_\to$.  Here the arrow
means that $\partial$ should be  moved to the right.  It is not
difficult to show that the skewcommutativity of the
$\lambda$-bracket is equivalent to skewadjointness of the
differential operator $H(\partial)$, and, in addition, the
validity of the Jacobi identity of the $\lambda$-bracket is, by definition,
equivalent to $H(\partial)$ being a Hamiltonian operator.
Furthermore, the bracket (\ref{eq:0.8}) on $R_\ell /\partial
R_\ell$ takes the familiar form
\begin{equation}
  \label{eq:0.16}
  \{ \tint f , \tint g \} = \int \frac{\delta g}{\delta u} \cdot
     \left(H (\partial )\frac{\delta f}{\delta u}\right)\, ,
\end{equation}
and one can show that this is a Lie algebra bracket if and only
if $H(\partial)$ is a Hamiltonian operator \cite{BDSK}.

Given $\lambda$-brackets
$\{{u_{i}}_\lambda u_j \} =- \{ u_{j_{-\partial-\lambda}} u_i \} \in R_\ell[\lambda]$
of any pair of generators $u_i,u_j$,
one can extend them uniquely to a $\lambda$-bracket on $R_\ell$,
which is given by the following explicit formula \cite{DSK1}
\begin{equation}
  \label{eq:0.19}
  \{ f_\lambda g \} = \sum_{\substack{1 \leq i,j \leq \ell\\
      m,n \in \mb Z_+}} \frac{\partial g}{\partial u^{(n)}_j}
  (\partial  + \lambda)^n \{ u_i {}_{\partial +\lambda} u_j
  \}_\to (-\partial -\lambda)^m \frac{\partial f}{\partial
    u^{(m)}_i}\, .
\end{equation}
This $\lambda$-bracket defines a PVA structure on $R_\ell$ if and
only if the Jacobi identity holds for any triple of generators
$u_i$, $u_j$, $u_k$ \cite{BDSK}.

The simplest example of a Hamiltonian operator is the
Gardner-Faddeev-Zakharov (GFZ) operator $K (\partial) = \partial$.
It is the observation in \cite{Gar} that the KdV equation
\begin{equation}
  \label{eq:0.17}
  \frac{du}{dt}= 3uu' + cu''' \, , \quad  c \in \mb F \, ,
\end{equation}
can be written in a Hamiltonian form
\begin{equation}
  \label{eq:0.18}
  \frac{du}{dt} = D \frac{\delta h_1}{\delta u}\, , \,
\hbox{\,\,  where \,\,} h_1 = \frac12 (u^3 + cuu'')\, ,
\end{equation}
and it is the subsequent proof in \cite{FZ} that KdV is a completely integrable
Hamiltonian equation,
that triggered the theory of Hamiltonian PDE.  The corresponding
$\lambda$-bracket on $R_1$ is, of course, given by the formula
$\{u_\lambda u \} =\lambda$, extended to
$R_1 \otimes R_1 \to \mb F[\lambda] \otimes R_1$
by \eqref{eq:0.19}.

In a subsequent paper \cite{Mag}, Magri showed that the operator
$H(\partial) = u'+2u\partial +c\partial^3$ is Hamiltonian for all $c\in\mb F$,
that it is compatible with the GFZ operator,
and that the KdV equation can be written in a different Hamiltonian form
\begin{equation}
  \label{eq:0.20}
  \frac{du}{dt} = (u'+2u\partial + c\partial^3)
     \frac{\delta h_1}{\delta u}, \hbox{\,\, where\,\,} h_1 = \frac12 u^2.
\end{equation}
Moreover, he explained how to use this
to prove the validity of the Lenard-Magri scheme
\footnote{Since
in the literature the names Lenard and Magri scheme are
alternatively used, we decided to call it the Lenard-Magri scheme. The history
of Lenard's contribution is colorfully described in \cite{PS},
where one can also find an extensive list of subsequent publications on the subject.},
which gave a new proof of integrability of KdV and some other
equations.

Of course, the $\lambda$-bracket corresponding to the Magri operator
is given by
\begin{equation}
  \label{eq:0.21}
  \{ u_\lambda u \} = (\partial +2\lambda) u+c\lambda^3\,,
\end{equation}
which defines (via (\ref{eq:0.19})) a PVA structure on $R_1$ for all
values of $c \in \mb F$.

The reader can find a detailed exposition of the applications of
PVA to Hamiltonian PDE in the paper \cite{BDSK}, where, in
particular, some sufficient conditions for the validity of the
Lenard-Magri scheme and its generalizations are found and applied
to the proof of integrability of many important equations.
However many Hamiltonian equations remain out of reach of the
methods of \cite{BDSK}, but we think that the cohomological
approach is more powerful (though less elementary) and we are
planning to demonstrate this in a subsequent paper.

In order to make our ideas clearer (or, perhaps, more confusing)
we begin the paper with a long digression,
which goes from Section \ref{sec:2} through Section \ref{sec:10},
to a general approach to various cohomology theories
(in fact, the reader, interested only in applications to the theory
of integrable Hamiltonian PDE, can, without much difficulty,
jump to Section \ref{sec:11}).

In Section \ref{sec:2}, given a vector superspace $V$, we consider the universal
$\mb Z$-graded Lie superalgebra $W (V) = \oplus_{j \geq -1}
W_j(V)$ with $W_{-1}(V)=V$.  Universality here is understood in
the sense that, given any other $\mb Z$-graded Lie superalgebra
$\mf g = \oplus_{j \geq -1} \mf g_j$ with $\mf g_{-1}=V$, there exists
a unique, grading preserving homomorphism $\mf g \to W(V)$,
identical on $V$.  It is easy to show that $W_j (V) = \Hom
(S^{j+1} (V),V)$ for all $j \geq -1$, and one can write down explicitly
the Lie superalgebra bracket.  In particular, $W_0 (V) = \End V$
and $W_1 (V) =  \Hom (S^2 V,V)$, so that any even element of the
vector superspace $W_1 (V)$ defines a commutative superalgebra
structure on $V$ (and this correspondence is bijective).

On the other hand, as observed in \cite{CK}, any odd element $X$
of the vector superspace $W_1 (\Pi V)$ defines an
skewcommutative superalgebra structure on $V$ by the formula
\begin{equation}
  \label{eq:0.22}
  [a,b] = (-1)^{p(a)} X (a \otimes b)\, , \quad a,b \in V\, ,
\end{equation}
where $p$ is the parity on $V$.  Moreover, this is a Lie
superalgebra structure if and only if $[X,X]=0$ in $W(\Pi V)$.
Thus, given a Lie superalgebra structure on $V$, considering the
corresponding element $X \in W_1 (\Pi V)$, we obtain a
cohomology complex $(C^\bullet= \oplus_{j \in \mb Z} C^j , \ad X)$, where
$C^j=W_{j+1}(\Pi V)$, and it turns out that $C^\bullet = C^\bullet (V,V)$
coincides with
the cohomology complex of the Lie superalgebra $V$ with
coefficients in the adjoint representation.  More generally,
given a module $M$ over the Lie superalgebra $V$, one considers,
instead of $V$, the Lie superalgebra $V\ltimes M$ with $M$ an
abelian ideal, and by a simple reduction procedure constructs the
cohomology of the Lie superalgebra $V$ with coefficients in $M$.
This construction for $V$ purely even goes back to the paper
\cite{NR} on deformation theory.

In Section \ref{sec:3}, assuming that $V$ carries a structure of a commutative
associative superalgebra, we let $W^{\as}_{-1} (\Pi V) = \Pi
V$, $W^{\as}_0 (\Pi V) = \Der V$, the subalgebra of all
derivations of the superalgebra $V$ 
in the superalgebra $\End \Pi V= \End V$
(the superscript ``as'' stands for ``associative'').
Let $W^{\as} (\Pi V) = \oplus_{j \geq -1} W^{\as}_j
(\Pi V)$ be the  full prolongation in the Lie superalgebra $W(\Pi V)$,
defined inductively for $j\geq1$ by
\begin{displaymath}
  W^{\as}_j \left(\Pi V\right) = \left\{ a \in W_j \left(\Pi
      V\right)\!
\Big|\!  \left[a,W_{-1} \left(\Pi V\right)\right]\! \subset\! W^{\as}_{j-1} \left(\Pi V\right)\right\}.
\end{displaymath}
Then odd elements $X$ in $W^{\as}_1 (\Pi V)$, such that
$[X,X]=0$, bijectively correspond to Poisson algebra structures
on $V$ (with the given commutative associative superalgebra
structure).  In this case the complex $(W^{\as} (\Pi V), \ad X)$
is the Poisson cohomology complex of the Poisson superalgebra $V$
(introduced in \cite{Lic}).

Incidentally, one can introduce a commutative associative
product on $W^{\as} (\Pi V)$, making it (along with the Lie
superalgebra bracket) an odd Poisson (= Gerstenhaber)
superalgebra.  Here we observe a remarkable duality when passing
from $\Pi V$ to $V$:  $W^{\as} (V)$ is an (even) Poisson
superalgebra, whereas the odd elements of $W^{\as}_1 (V)$
correspond to odd Poisson superalgebra structures on $V$.

Next, in Section \ref{sec:4} we consider the case when $V$ carries a structure of an $\mb F
[\partial]$ -module.  Here and throughout the paper $\mb F
[\partial]$, as usual, denotes the algebra of polynomials in an
(even) indeterminate $\partial$.  Motivated by the construction
of the universal Lie superalgebra $W (\Pi V)$, we construct a
$\mb Z$-graded Lie superalgebra $W^\partial (\Pi V) =
\oplus^\infty_{k=-1} W^\partial_k (\Pi V)$, which, to some
extent, plays the same role in the theory of Lie conformal
algebra as $W (\Pi V)$ plays in the theory of Lie algebras
(explained above).

Recall that a Lie conformal algebra is an $\mb F
[\partial]$-module, endowed with the $\lambda$-bracket,
satisfying sesquilinearity, skewcommutativity and Jacobi identity
(introduced above).  In other words, a Lie conformal algebra is
an analogue of a Lie algebra in the same way as a Poisson vertex
algebra is an analogue of a Poisson algebra.

We let $W^\partial_{-1} (\Pi V) = \Pi (V/\partial V)$ and
$W^\partial_0 (\Pi V) = \End_{\mb F [\partial]}V$, and construct
$W^\partial (\Pi V)$ as a prolongation in $W (\Pi
(V/\partial V))$ (not necessarily full), so that odd elements $X
\in W^\partial_1 (\Pi V)$  parameterize sesquilinear
skewcommutative $\lambda$-brackets on $V$, and the
$\lambda$-bracket satisfies the Jacobi identity  (i.e.,~defines
on $V$ a Lie conformal algebra structure) if and only if $[X,X]=0$.

In the same way as in the Lie algebra case, we obtain a
cohomology complex $(W^\partial (\Pi V), \ad X)$, provided that
$[X,X]=0$ for an odd element $X \in W^\partial_1 (\Pi V)$,  and
this complex (after the shift by~$1$), is the Lie conformal
algebra cohomology complex with coefficients in the adjoint
representation.  In the same way, by a reduction, we recover the
Lie conformal algebra cohomology complex with coefficients in any
representation, studied in \cite{BKV}, \cite{BDAK}, \cite{DSK2}.

Next, in Section \ref{sec:5} we consider the case when $V$ carries both, a structure
of an $\mb F [\partial]$-module, and a compatible with it commutative
algebra structure, in other words, $V$ is a differential algebra.  Then in
the same way as above, we construct the Lie superalgebra
$W^{\partial,\as} (\Pi V)$ (cf. (\ref{eq:0.1})) as a
$\mb Z$-graded subalgebra of $W^\partial (\Pi V)$, for which
$W^{\partial,\as}_{-1} (\Pi V) =W^\partial_{-1} (\Pi V)=\Pi(V/\partial V)$,
$W^{\partial,\as}_0 (\Pi V) = \Der_{\mb F [\partial]} V \subset
W^\partial_0 (\Pi V)= \End_{\mb F [\partial]}V$,
and $W^{\partial,\as}_1 (\Pi V)$ is such that its odd elements $X$
parameterize all $\lambda$-brackets on the differential algebra
$V$, so that those satisfying $[X,X]=0$ correspond to PVA
structures on $V$.  This explains the strange notation
(\ref{eq:0.1}) of this Lie superalgebra.

In Section \ref{sec:6} we construct the universal Lie conformal superalgebra 
$\tilde{W}^\partial(V)$ for a finitely generated $\mb F[\partial]$-supermodule $V$,
and in Section \ref{sec:7} we construct the universal odd Poisson vertex algebra 
$\tilde{W}^{\partial,\as}(\Pi\mc V)$ for a finitely generated differential superalgebra $\mc V$.
These constructions are very similar in spirit to the constructions
of the universal Lie superalgebras $W^\partial(V)$ and $W^{\partial,\as}(\mc V)$,
from Sections \ref{sec:4} and \ref{sec:5} respectively.
The finitely generated assumption is needed in order for the corresponding 
$\lambda$-brackets to be polynomial in $\lambda$.

Note that in the definition of the Lie algebra bracket on
$V/\partial V$, and its representation on $V$, as well in the
discussion of the Lenard-Magri scheme, we needed only that $V$ is a Lie
conformal algebra.  However, for practical applications one usually uses
PVA's, and, in fact some special kind of PVA's, which are
differential algebra extensions of $R_\ell$ with the
$\lambda$-bracket given by formula (\ref{eq:0.19}).  For such a
PVA $V$ we construct, in Sections \ref{sec:9} a subalgebra
of the Lie algebra $W^{\partial,\as} (\Pi V)$
\begin{displaymath}
  W^{\var} (\Pi V) = \oplus_{j \geq -1} W^{\var}_j\, ,
\end{displaymath}
where $W^{\var}_{-1}= W^{\partial,\as}_{-1}$, but $W^{\var}_j$
for $j \geq 0$ may be smaller.  For example $W^{\var}_0$ consists
of derivations of the form $\sum_{\substack{1 \leq j \leq \ell
    \\ n \in \mb Z_+}} P_{j,n} \frac{\partial}{\partial u^{(n)}_j}$, commuting with $\partial$, and it is these
derivations that are called in variational calculus evolutionary
vector fields.  Next, $W^{\var}_1$ consists of all
$\lambda$-brackets of the form (\ref{eq:0.19}), etc.  We call
elements of $W^{\var}_k$ the {\em variational $k$-vector fields}.

There has been an extensive discussion of variational poly-vector
fields in the literature.  The earliest  reference we know of is
\cite{Kup}, see also the book \cite{Ol2}.
One of the later references is \cite{IVV};
the idea to use Cartan's prolongation comes from this paper.


In order to solve the Lenard-Magri scheme (\ref{eq:0.13}) over a
differential function extension $V$ of $R_\ell$ with
the $\lambda$-brackets $\{ \, . \, , \, . \, \}_H$
and  $\{ \, . \, ,\, . \, \}_K$ of the form as in (\ref{eq:0.19}), one has to compute the
cohomology of the complex $(W^{\var} (\Pi V), \ad K)$, where $K
\in W^{\var}_1$ is such that $[K,K]=0$, as we explained above.

In order to compute this {\em variational Poisson cohomology}, we
construct, in Section \ref{sec:10} a $\mb Z_+$-graded Lie conformal superalgebra (which is
actually a subalgebra of the odd PVA $\tilde{W}^{\partial,\as}(\Pi\mc V)$) 
$\tilde{W}^{\var} (\Pi V) = \oplus_{j \geq-1} \tilde{W}^{\var}_j$ with $\tilde{W}^{\var}_{-1}=V$, 
for which the associated Lie superalgebra is $W^{\var} (\Pi V)$.  Since
the Lie superalgebra $W^{\var}(\Pi V)$ acts on the Lie
conformal superalgebra $\tilde{W}^{\var} (\Pi V)$, in
particular, $K$ acts, providing it with a differential $d_K$,
commuting with the action of $\partial$.  We thus have an exact
sequence of complexes:
\begin{equation}
  \label{eq:0.23}
  0 \to (\partial \tilde{W}^{\var} (V), d_K)  \to
   ( \tilde{W}^{\var} (V), d_K) \to (W^{\var} (V), \ad K)\to 0\, ,
\end{equation}
so that we can study the corresponding cohomology long exact sequence.

To actually perform calculations, we identify (non-canonically)
the space $ \tilde{W}^{\var} (\Pi V)$ with the space
$\tilde{\Omega}^\bullet (V)$ of the de Rham complex, and the space
$W^{\var} (\Pi V)$ with the space of the reduced de Rham
complex = variational complex $\Omega^\bullet (V)$.

We thus get the ``generalized'' de Rham complex $(\tilde{\Omega}^\bullet (V), d_K)$
and the ``generalized'' variational complex
$(\Omega^\bullet (V), \ad K)$. The ordinary de Rham and variational
complexes are not, strictly speaking, special cases, since they
correspond to $K=I$, which is not a skewadjoint operator.
However, in the case when the differential operator $K$ is
quasiconstant, i.e.,~$\frac{\partial}{\partial u^{(n)}_i} (K) =0$
for all $i,n$, the construction of these complexes is still
valid.

In Section \ref{sec:11} we completely solve the problem of computation of
cohomology of the generalized de Rham complex 
$(\tilde{\Omega}^\bullet(V), d_K)$ in the case when $V$ is a normal algebra of
differential functions and $K$ is a quasiconstant matrix differential
operator with invertible leading coefficient.  For that we use
``local'' homotopy operators, similar to those introduced in
\cite{BDSK} for the de Rham complex.

After that, as in \cite{BDSK}, we study the cohomology long exact
sequence corresponding to the short exact sequence (\ref{eq:0.23}).
As a result
we get a complete description of the cohomology
of the generalized variational complex
for an arbitrary quasiconstant $\ell\times\ell$ matrix differential
operator $K$ of order $N$ with invertible leading coefficient.
In fact, we find simple explicit formulas for representatives of cohomology classes,
and we prove that
$$
\dim H^k(\Omega^\bullet(\mc V),\ad K)=\binom{N\ell}{k+1}\,,
$$
provided that quasiconstants form a linearly closed differential field.\
These results
lead to further progress in the application of the Lenard-Magri scheme (work in progress).

In the special case when $K$ is a constant coefficient order 1 skewadjoint matrix differential operator,
it is proved in \cite{Get} that the variational Poisson cohomology complex is formal.

Our explicit description of the long exact sequence in terms
of polydifferential operators leads to some problems on systems
of linear differential equations of arbitrary order
in the same number of unknowns.
In the Appendix we develop some differential linear algebra
in order to solve these problems.

All vector spaces are considered over a field $\mb F$ of characteristic zero.
Tensor products, direct sums, and Hom's are considered over $\mb F$,
unless otherwise specified.

We wish to thank A. Kiselev for drawing our attention to the cohomological approach
to the Lenard-Magri scheme, I. Krasilshchik for correspondence,
and A. Maffei for useful discussions.


\section{The universal Lie superalgebra $W(V)$ for a vector superspace $V$,
and Lie superalgebra cohomology}
\label{sec:2}

Recall that a vector superspace is a $\mb Z/2\mb Z$-graded vector space $U=U_{\bar 0}\oplus U_{\bar 1}$.
If $a\in U_\alpha$, where $\alpha\in\mb Z/2\mb Z=\{\bar0,\bar1\}$, one says that $a$ has parity $p(a)=\alpha$.
In this case we say that the superspace $U$ has parity $p$.
By a superalgebra structure on $U$ we always mean a parity preserving product
$U\otimes U\to U,\,a\otimes b\mapsto ab$,
which is called commutative (resp. skewcommutative) if
$ba=(-1)^{p(a)p(b)}ab$ (resp. $ba=-(-1)^{p(a)p(b)}ab$).

An endomorphism of $U$ is called even (resp. odd) if it preserves (resp. reverses) the parity.
The superspace $\End(U)$ of all endomorphisms of $U$ is endowed
with a Lie superalgebra structure by the formula:
$[A,B]=A\circ B-(-1)^{p(A)p(B)}B\circ A$.

One denotes by $\Pi U$ the superspace obtained from $U$ by reversing the parity,
namely $\Pi U=U$ as a vector space, with parity $\bar p(a)=p(a)+\bar1$.
One defines a structure of a vector superspace on the tensor algebra $\mc T(U)$ over $U$
by additivity.
The symmetric, (respectively exterior) superalgebra $S(U)$ (resp. $\bigwedge(U)$)
is defined as the quotient of the tensor superalgebra $\mc T(U)$ by the relations
$u\otimes v-(-1)^{p(u)p(v)}v\otimes u$ (resp. $u\otimes v+(-1)^{p(u)p(v)}v\otimes u$).
Note that $S(\Pi U)$ is the same as $\bigwedge U$ as a vector space, but not as a vector superspace.

\subsection{The universal Lie superalgebra $W(V)$}\label{sec:2.2}

Let $V$ be a vector superspace with parity $\bar p$
(the reason for this notation will be clear later).
We recall the construction of the universal Lie superalgebra $W(V)$
associated to $V$.
Let $W_k(V)=\Hom(S^{k+1}(V),V)$, the superspace of $(k+1)$-linear supersymmetric
functions on $V$ with values in $V$,
and let $W(V)=\bigoplus_{k=-1}^\infty W_k(V)$.
Again, we denote its parity by $\bar p$.
We endow this vector superspace  with a structure of a $\mb Z$-graded Lie superalgebra as follows.
If $X\in W_h(V)$, $Y\in W_{k-h}(V)$, with $h\geq-1,\,k\geq h-1$, we define $X\Box Y$ to be
the following element in $W_{k}(V)$:
\begin{equation}\label{100418:eq1}
\begin{array}{l}
X\Box Y(v_0, \dots, v_k) \\
\displaystyle{
= \sum_{\substack{
i_0<\dots <i_{k-h}\\
i_{k-h+1}<\dots< i_k}}
\epsilon_v(i_0,\dots,i_k)
X(Y(v_{i_0},\dots, v_{i_{k-h}}), v_{i_{k-h+1}},\dots, v_{i_k})\,.
}
\end{array}
\end{equation}
Here $\epsilon_v(i_0,\dots,i_k)=0$ if two indexes are equal,
and, for $i_0,\dots,i_k$ distinct,
$\epsilon_v(i_0,\dots,i_k)=(-1)^N$,  where $N$ is the number of interchanges of indexes
of odd $v_i$'s in the permutation.
For example, if $V$ is purely even, then $\epsilon_v(i_0,\dots,i_k)=1$ for every permutation,
while, if $V$ is purely odd, then $\epsilon_v(i_0,\dots,i_k)$ is the sign of the permutation
$\sigma\in S_{k+1}$ given by $\sigma(\ell)=i_\ell$.
The above formula, for $h=-1$ gives zero, while for $k=h-1$ gives $X(Y,v_0,\dots,v_k)$.
Clearly, $X\Box Y$ is a supersymmetric map if both $X$ and $Y$ are,
hence $X\Box Y$ is a well-defined element of $W_k(V)$.
We then define the bracket $[\cdot\,,\,\cdot]:\,W_h(V)\times W_{k-h}(V)\to W_k(V)$
by the following formula:
\begin{equation}\label{box}
[X,Y]=X\Box Y-(-1)^{\bar p(X)\bar p(Y)}Y\Box X\,.
\end{equation}
\begin{proposition}
The bracket \eqref{box} defines a Lie superalgebra structure on $W(V)$.
\end{proposition}
\begin{proof}
The bracket (\ref{box}) is skewcommutative by construction.
Moreover, it is easy to see that the
operation $\Box$ is right symmetric, i.e.,  $(X,Y,Z)=(-1)^{\bar p(Y)\bar p(Z)}(X,Z,Y)$, where
$(X,Y,Z)=(X\Box Y)\Box Z-X\Box(Y\Box Z)$. The right symmetry
of $\Box$ implies the Jacobi identity for  bracket (\ref{box}).
\end{proof}

According to the above definitions, $W_{-1}(V)=V$, $W_0(V)=\End(V)$,
and the bracket between $W_0(V)$ and $W_{-1}(V)$ is given by the action of $\End(V)$ on $V$.
Moreover, for $X\in W_k(V)$ and $Y\in W_{-1}(V)=V$, we have
\begin{equation}\label{100411:eq2}
[X,Y](v_1,\dots,v_k)=X(Y,v_1,\dots,v_k)
\,\,,\,\,\,\,
v_1,\dots,v_k\in V\,,
\end{equation}
while for $X\in W_0(V)$ and $Y\in W_k(V),\,k\geq-1$, we have
\begin{equation}\label{100411:eq3}
\begin{array}{l}
[X,Y](v_0,\dots,v_k)=
X\big(Y(v_0,\dots,v_k)\big) \\
\displaystyle{
-(-1)^{\bar p(X)\bar p(Y)} \sum_{i=0}^k (-1)^{\bar p(X)\bar s_{0,i-1}}
Y(v_0,\dots X(v_i)\dots,v_k)\,.
}
\end{array}
\end{equation}
Here and further we let, for $i\leq j$,
\begin{equation}\label{100412:eq5}
\bar s_{i,j}=\bar p(v_i)+\dots+\bar p(v_j)\,.
\end{equation}
Finally, if $X\in W_1(V)$ and $Y\in W_{k-1}(V),\,k\geq0$, we have
$[X,Y]=X\Box Y-(-1)^{\bar p(X)\bar p(Y)}Y\Box X$, where
\begin{equation}\label{100412:eq4}
\begin{array}{l}
\displaystyle{
X\Box Y(v_0,\dots,v_k)=
\sum_{i=0}^k (-1)^{\bar p(v_i)(\bar p(Y)+\bar s_{0,i-1})}
X\big(v_i,Y(v_0,\stackrel{i}{\check{\dots}},v_k)\big)\,,
} \\
\displaystyle{
Y\Box X(v_0,\dots,v_k)
} \\
\displaystyle{
=\!\!
\sum_{0\leq i<j\leq k} \!\!
(-1)^{\bar p(v_i)\bar s_{0,i-1}+\bar p(v_j)(\bar s_{0,j-1}+\bar p(v_i))}
Y\big(X(v_i,v_j),v_0,\stackrel{i}{\check{\dots}}\,\stackrel{j}{\check{\dots}},v_k\big)\,.
}
\end{array}
\end{equation}
In particular, if both $X$ and $Y$ are in $W_1(V)$,
we get
\begin{equation}\label{100411:eq4}
\begin{array}{c}
X\Box Y(v_0,v_1,v_2) =
X\big(Y(v_0,v_1),v_2\big)
+(-1)^{\bar p(Y)\bar p(v_0)} X\big(v_0,Y(v_1,v_2)\big) \\
\vphantom{\Big(}
+(-1)^{(\bar p(v_0)+\bar p(Y))\bar p(v_1)} X\big(v_1,Y(v_0,v_2)\big) \,.
\end{array}
\end{equation}

\begin{remark}\label{100412:rem1}
It follows from \eqref{100411:eq2} that we have the following universality property of the Lie superalgebra $W(V)$:
for any $\mb Z$-graded Lie superalgebra $\mf g=\bigoplus_{k=-1}^\infty\mf g_k$
with $\mf g_{-1}=V$ there is a canonical homomorphism of $\mb Z$-graded Lie superalgebras
$\phi:\,\mf g\to W(V)$, extending the identity map on $V$, given by
$$
\phi(a)(v_0,\dots,v_k)=[\dots[[a,v_0],v_1],\dots,v_k]
\,\,,\,\,\,\,
\text{ if } k\geq0\,.
$$
This map is an embedding if and only if $\mf g$ has no ideals in $\bigoplus_{k\geq0}\mf g_k$.
\end{remark}
\begin{remark}\label{100412:rem2}
If $V$ is a finite dimensional vector superspace,
then $W(V)$ coincides with the Lie superalgebra of all polynomial
vector fields on $V$ (this explains the letter $W$, for Witt).
\end{remark}
\begin{remark}\label{100412:rem3}
If $V$ is purely odd, then $W(V)$ coincides with the so called Nijenhuis-Richardson algebra,
which plays an important role in deformation theory \cite{NR}.
\end{remark}


\subsection{The space $W(V,U)$ as a reduction of $W(V\oplus U)$}\label{sec:2.3}

Let $V$ and $U$ be vector superspaces with parity $\bar p$.
We define the $\mb Z_+$-graded vector superspace (with parity still denoted by $\bar p$)
$W(V,U)=\bigoplus_{k\in\mb Z_+}W_k(V,U)$,
where $W_k(V,U)=\Hom(S^{k+1}(V),U)$.

It can be obtained as a reduction of the universal Lie superalgebra $W(V\oplus U)$
as follows.
We consider the subspace
\begin{equation}\label{100414:eq1}
\Hom(S^{k+1}(V\oplus U),U)\subset W_k(V\oplus U)\,,
\end{equation}
defined by the canonical direct sum decomposition
$$
W_k(V\oplus U)=\Hom(S^{k+1}(V\oplus U),U)\oplus \Hom(S^{k+1}(V\oplus U),V)\,.
$$
The kernel of the restriction map $X\mapsto X\big|_{S^{k+1}(V)}$
is the subspace
\begin{equation}\label{100414:eq2}
\big\{X\in\Hom(S^{k+1}(V\oplus U),U) \,\big|\, X(S^{k+1}(V))=0\big\}\subset W_k(V\oplus U)\,.
\end{equation}
Hence we get an induced isomorphism of superspaces
\begin{equation}\label{100414:eq3}
\Hom(S^{k+1}(V\oplus U),U)\big/\big\{X \,\big|\, X(S^{k+1}(V))=0\big\}
\stackrel{\sim}{\longrightarrow} W_k(V,U)\,.
\end{equation}

\begin{proposition}\label{100414:prop}
Let $X\in W_h(V\oplus U)$. Then the adjoint action of $X$ on $W(V\oplus U)$
leaves the subspaces \eqref{100414:eq1} and \eqref{100414:eq2} invariant
provided that
\begin{enumerate}[(i)]
\item $X(w_0,\dots,w_h)\in U$ if at least one of the arguments $w_i$ lies in $U$,
\item $X(v_0,\dots,v_h)\in V$ if all the arguments $v_i$ lie in $V$.
\end{enumerate}
In this case $\ad X$ induces a well-defined map on the reduction $W(V,U)$,
via the isomorphism \eqref{100414:eq3}.
\end{proposition}
\begin{proof}
The proof is immediate from the definition of the Lie bracket \eqref{box} on $W(V\oplus U)$.
\end{proof}

\begin{remark}\label{100414:rem}
An element $X\in W_1(V\oplus U)$ defines a commutative 
(not necessarily associatve)
product $\cdot$ on the superspace $V\oplus U$.
In this case, conditions (i) and (ii) in Proposition \ref{100414:prop}
exactly mean that $V\cdot V\subset V$ and that $(V\oplus U)\cdot U\subset U$.
Moreover, the induced action of $\ad X$ on $W(V,U)$ is independent of the product on $U$.
\end{remark}

\subsection{Prolongations}\label{sec:2.4}

Let $V$ be a vector superspace, and let $\mf g_0$ be a subalgebra of the Lie superalgebra $\End(V)$.
A \emph{prolongation} of $\mf g_0$ is a $\mb Z$-graded subalgebra $\mf g=\bigoplus_{k=-1}^\infty\mf g_k$
of the $\mb Z$-graded Lie superalgebra $W(V)=\bigoplus_{k=-1}^\infty W_k(V)$,
such that $\mf g_{-1}=W_{-1}(V)=V$ and $\mf g_0$ coincides with the given Lie superalgebra.

The \emph{full prolongation} $W^{\mf g_0}(V)=\bigoplus_{k=-1}^\infty W^{\mf g_0}_k(V)$ of $\mf g_0$ is
defined by letting $W^{\mf g_0}_{-1}(V)=V,\,W^{\mf g_0}_0(V)=\mf g_0$ and,
inductively, for $k\geq1$,
$$
W^{\mf g_0}_k(V)=\big\{X\in W_k(V)\,\big|\,[X,W_{-1}(V)]\subset W^{\mf g_0}_{k-1}(V)\big\}\,.
$$
It is immediate to check, by the Jacobi identity, that the above formula defines
a maximal prolongation of the Lie superalgebra $W(V)$.

\subsection{Lie superalgebra structures}\label{sec:2.5}

By definition, the even elements $X\in W_1(V)$
are exactly the commutative (not necessariy associative) superalgebra structures on $V$:
for $X\in W_1(V)_{\bar0}$
we get a commutative product on $V$ by letting $uv=X(u,v)$.
Similarly, the skewcommutative superalgebra structures on a superspace $L$ with parity $p$
are in bijective correspondence with the odd elements of $W_1(\Pi L)$:
for $X\in W_1(\Pi L)_{\bar1}$, we get a skewcommutative product on $L$ by letting
\begin{equation}\label{100412:eq2}
[a,b]=(-1)^{p(a)}X(a,b)\,\,,\,\,\,\,a,b\in L\,,
\end{equation}
and vice-versa \cite{CK}.

Furthermore, let $X\in W_1(\Pi L)_{\bar1}$,
and consider the corresponding skewcommutative product \eqref{100412:eq2} on $L$.
The Lie bracket of $X$ with itself then becomes, by \eqref{100411:eq4},
$$
\begin{array}{l}
[X,X](a,b,c) = 2 X\Box X(a,b,c) \\
= -(-1)^{p(b)} 2 \Big\{
[a,[b,c]]-(-1)^{p(a)p(b)} [b,[a,c]]
- [[a,b],c]
\Big\}\,.
\end{array}
$$
Hence, the Lie superalgebra structures on $L$
are in bijective correspondence, via \eqref{100412:eq2},
with the set
\begin{equation}\label{100412:eq6}
\big\{X\in W_1(\Pi L)_{\bar1}\,\big|\,[X,X]=0\big\}\,.
\end{equation}
Therefore, for any Lie superalgebra $L$, we have a differential $d_X=\ad X$,
where $X$ in \eqref{100412:eq6} is associated to the Lie superalgebra structure on $L$,
on the superspace $W(\Pi L)$.
This coincides, up to a sign in the formula of the differential,
with the usual Lie superalgebra cohomology complex $(C^\bullet(L,L),d)$ for the Lie superalgebra $L$
with coefficients in its adjoint representation
(in equation \eqref{100414:eq4} below we give an explicit formula for the differential $d_X$
for an arbitrary representation $M$).
Thus the complex $C^\bullet(L,L)$ has a canonical Lie superalgebra structure
for which the differential $d_X$ (but not $d$) is a derivation.

In the next section we construct, by reduction, the Lie superalgebra cohomology complex
$(C^\bullet(L,M),d)$ for the Lie superalgebra $L$ with coefficients in an arbitrary $L$-module $M$.

\subsection{Lie superalgebra modules and cohomology complexes}\label{sec:2.6}

Let $L$ and $M$ be vector superspaces with parity $p$,
and consider the reduced superspace
$W(\Pi L,\Pi M)=\bigoplus_{k=-1}^\infty\Hom(S^{k+1}(\Pi L),\Pi M)$
introduced in Section \ref{sec:2.3},
with parity denoted by $\bar p$.

Suppose now that $L$ is a Lie superalgebra and $M$ is an $L$-module.
This is equivalent to say that
we have a Lie superalgebra structure on the superspace $L\oplus M$
extending the Lie bracket on $L$,
such that $M$ is an abelian ideal, the bracket between $a\in L$ and $m\in M$
being $a(m)$.
According to the above observations,
such a structure corresponds, bijectively,
to an element $X$ of the following set:
\begin{equation}\label{100412:eq7}
\left\{X\in W_1(\Pi L\oplus\Pi M)_{\bar1}\,\left|\,
\begin{array}{l}
[X,X]=0\,,X(L,L)\subset L,\\
X(L,M)\subset M,\,X(M,M)=0
\end{array}
\right.\right\}\,.
\end{equation}
Explicitly, to $X$ in \eqref{100412:eq7}
we associate the corresponding Lie superalgebra bracket on $L$ given by \eqref{100412:eq2},
and the corresponding $L$-module structure on $M$ given by
\begin{equation}\label{100412:eq3}
a(m) = (-1)^{p(a)}X(a,m)
\,\,,\,\,\,\,
a\in L,\,m\in M\,.
\end{equation}

Note that every element $X$ in the set \eqref{100412:eq7}
satisfies conditions (i) and (ii) in Proposition \ref{100414:prop}.
Hence $\ad X$ induces a well-defined endomorphism $d_X$ of $W(\Pi L,\Pi M)$ such that $d_X^2=0$,
thus making $(W(\Pi L,\Pi M),d_X)$ a complex.
The explicit formula for the differential $d_X$ follows from equations \eqref{100412:eq4}
and from the identifications \eqref{100412:eq2} and \eqref{100412:eq3}.
For $Y\in W_{k-1}(\Pi L,\Pi M)$, we have
\begin{equation}\label{100414:eq4}
\begin{array}{c}
\displaystyle{
(d_X Y)(a_0,\dots,a_k)
=
\sum_{i=0}^k(-1)^{\alpha_i}
a_i\big(Y(a_0,\stackrel{i}{\check{\dots}},a_k)\big)
}\\
\displaystyle{
+ \sum_{0\leq i<j\leq k}(-1)^{\alpha_{ij}}
Y([a_i,a_j],a_0,\stackrel{i}{\check{\dots}}\,\stackrel{j}{\check{\dots}},a_k)\,,
}
\end{array}
\end{equation}
where,
\begin{equation}\label{100421:eq10}
\begin{array}{rcl}
\alpha_i &=& 1+(p(a_i)+1)(\bar p(Y)+s_{0,i-1}+i+1) \,,\\
\alpha_{i,j} &=& \bar p(Y)+ (p(a_i)+1)(s_{0,i-1}+i+1) \\
&& +(p(a_j)+1)(s_{0,j-1}+j+p(a_i)+1)\,,
\end{array}
\end{equation}
and, recalling \eqref{100412:eq5}, we let, for $i\leq j$,
\begin{equation}\label{100414:eq5}
s_{i,j}=p(a_i)+\dots+p(a_j)\,.
\end{equation}
Note that, in the special case when $M=L$ is the adjoint representation,
the complex $(W(\Pi L,\Pi M),d_X)$ coincides with the complex $(W(\Pi L),d_X)$
discussed in Section \ref{sec:2.5}.
In the special case when $L$ is a (purely even) Lie algebra and $M$ is a purely even $L$-module,
we have $\bar p(Y)\equiv k \mod 2$, and the above formula reduces to
\begin{equation}\label{100414:eq6}
\begin{array}{c}
\displaystyle{
(d_X Y)(a_0,\dots,a_k)
=(-1)^k\Big(
\sum_{i=0}^k(-1)^{i}
a_i\big(Y(a_0,\stackrel{i}{\check{\dots}},a_k)\big)
}\\
\displaystyle{
+ \sum_{0\leq i<j\leq k}(-1)^{i+j}
Y([a_i,a_j],a_0,\stackrel{i}{\check{\dots}}\,\stackrel{j}{\check{\dots}},a_k)
\Big)\,,
}
\end{array}
\end{equation}
which, up to the overall sign factor $(-1)^{k}$,
is the usual formula for the Lie algebra cohomology differential (see e.g. \cite{Bou}).

In conclusion, the cohomology complex $(C^\bullet(L,M)=\bigoplus_{k\in\mb Z_+}C^k(L,M),d)$
of a Lie superalgebra $L$ with coefficients in an $L$-module $M$
can be defined by letting $C^k(L,M)=W_{k-1}(\Pi L,\Pi M)$ and $d=d_X$.

\begin{remark}
We have a canonical representation of a Lie superalgebra $L$
on each $W_k(\Pi L,\Pi M)=\Hom(S^{k+1}(\Pi L),\Pi M)$,
that we denote by $a\mapsto L_a,\,a\in L$ (the Lie derivative).
It is easy to check that $L_a=\ad[a,X]$.
Hence, defining the contraction operators $\iota_a=\ad a$,
we have Cartan's formula $L_a=[\iota_a,d_X]$.
This, together with the observation that, when $M=L$ is the adjoint representation,
$d_X$ is a derivation of the Lie bracket in $W(\Pi L)$,
leads us to believe that our choice of signs for the differential, contractions and Lie derivatives
in the Lie superalgebra cohomology complex is the most natural one.
In fact, one checks that the most general choice of signs which
keeps Cartan's formula valid up to a sign is the following:
$$
d_X(Y) = \epsilon(\bar p(Y)) [X,Y]
\,\,,\,\,\,\,
\iota_a(Y) = \delta(p(a))\epsilon(\bar p(Y)+\bar1) [a,Y]\,,
$$
where $\epsilon$ and $\delta$ are arbitrary functions: $\mb Z/2\mb Z\to\{\pm1\}$.
In this case Cartan's formula has the form
$L_a=\delta(p(a))[\iota_a,d_X]$.
Our choice of signs is $\epsilon=\delta=+1$.
The usual choice, see e.g. \cite{Bou}, in the case when $L$ and $M$ are purely even,
is $\delta=+1$ and $\epsilon(\bar k)=(-1)^k$, which corresponds to
$(\iota_{a_0}(Y))(a_1,\dots,a_k)=Y(a_0,a_1,\dots,a_k)$.
But this choice cannot be extended to the super case, if we require Cartan's formula to hold
(up to a sign).
\end{remark}


\section{The universal (odd) Poisson superalgebra
for a commutative associative superalgebra $A$
and Poisson superalgebra cohomology}
\label{sec:3}

Recall that a Poisson (resp. odd Poisson (=Gerstenhaber)) superalgebra $\mc P$, with parity $p$,
is a commutative associative superalgebra endowed with a bracket $[\cdot\,,\,\cdot]$
which makes $\mc P$ (resp. $\Pi\mc P$) a Lie superalgebra,
satisfying the following Leibniz rule:
$$
[a,bc] - [a,b]c = (-1)^{p(a)p(b)}b[a,c]
\,\,\Big(\text{ resp. } = (-1)^{(p(a)+\bar1)p(b)}b[a,c]\Big)\,.
$$
Usually the commutative associative product on $\mc P$ is denoted
by $\cdot$ in the Poisson superalgebra case,
and by $\wedge$ in the odd Poisson superalgebra case.
For example, if $(A,\mf g)$ is a Lie superalgebroid
over a commutative associative algebra $A$,
then $S_A(\mf g)$  has a natural structure of a Poisson superalgebra,
and $S_A(\Pi\mf g)$ has a natural structure of an odd Poisson superalgebra,
with the bracket on $\mf g$ extended by the Leibniz rule.
\begin{remark}\label{100424:rem1}
If $\mc P$ is a Poisson (resp. odd Poisson) algebra,
we can consider the \emph{opposite} Poisson (resp. odd Poisson)
algebra $\mc P^{\op}$, with the reversed product, 
and the opposite bracket, $[a,b]^{\op}=-[b,a]$.
\end{remark}

\subsection{The universal odd Poisson superalgebra $\Pi W^{\as}(\Pi A)$}\label{sec:3.1}

Throughout this section, we let $A$ be a commutative associative superalgebra with parity $p$,
and let $\Der(A)$ be the Lie superalgebra of derivations of $A$,
i.e., linear maps $X:\,A\to A$ satisfying the Leibniz rule
$X(ab)=X(a)b+(-1)^{p(a)p(b)}X(b)a$.

Consider the universal Lie superalgebra $W(\Pi A)=\bigoplus_{k=-1}^\infty W_k(\Pi A)$
associated to the vector superspace $\Pi A$, with parity denoted by $\bar p$.
The Lie superalgebra $\Der(A)$ is a subalgebra of $W_0(\Pi A)=\End(\Pi A)$,
so we can consider its full prolongation, as defined in Section \ref{sec:2.4},
which we denote by
$$
W^{\as}(\Pi A)=\bigoplus_{k=-1}^\infty W^{\as}_k(\Pi A)\subset W(\Pi A)\,.
$$
\begin{proposition}\label{100422:prop1}
For $k\geq-1$, the superspace $W^{\as}_k(\Pi A)$ consists of linear maps
$X:\,S^{k+1}(\Pi A)\to\Pi A$ satisfying the following Leibniz rule 
(for $a_0,\dots,a_{k-1},b,c\in A$, $k\geq0$):
\begin{equation}\label{100422:eq1}
X(a_0,\dots,a_{k-1},bc)=X(a_0,\dots,a_{k-1},b)c+(-1)^{p(b)p(c)}X(a_0,\dots,a_{k-1},c)b\,.
\end{equation}
\end{proposition}
\begin{proof}
It follows by an easy induction on $k\geq0$.
\end{proof}
\begin{remark}\label{100422:rem1}
Equation \eqref{100422:eq1} and the symmetry relations imply the following more general
formula, for every $i=0,\dots,k$:
\begin{equation}\label{100422:eq2}
\begin{array}{c}
\vphantom{\Big(}
X(a_0,\dots ,b_ic_i,\dots,a_k)=(-1)^{p(c_i)(s_{i+1,k}+k-i)}X(a_0,\dots,b_i,\dots,a_k)c_i \\
\vphantom{\Big(}
+(-1)^{p(b_i)(\bar p(X)+s_{0,i-1}+i)}b_iX(a_0,\dots,c_i,\dots,a_k)\,,
\end{array}
\end{equation}
where $s_{ij}$ is defined in \eqref{100414:eq5}.
\end{remark}

We next define a structure of commutative associative superalgebra on the superspace $\Pi W^{\as}(\Pi A)$,
making it an odd Poisson superalgebra.
Let $X\in \Pi W^{\as}_{h-1}(\Pi A)$ and $Y\in\Pi W^{\as}_{k-h-1}(\Pi A)$, for $h\geq0,\,k-h\geq0$,
and denote by $p(X)$ and $p(Y)$ their parities in these spaces.
We define their \emph{concatenation product} $X\wedge Y\in \Pi W^{\as}_{k-1}(\Pi A)$
as the following map:
\begin{equation}\label{100422:eq4}
\begin{array}{c}
\displaystyle{
(X\wedge Y)(a_1,\dots,a_k)
=
\sum_{\substack{
i_1<\dots <i_{h}\\
i_{h+1}<\dots< i_k}}
\epsilon_{a}(i_1,\dots,i_k) (-1)^{p(Y)(\bar p(a_{i_1})+\dots+\bar p(a_{i_h}))}
} \\
\displaystyle{
\times
X(a_{i_1},\dots,a_{i_h})Y(a_{i_{h+1}},\dots,a_{i_k})\,,
}
\end{array}
\end{equation}
where $\epsilon_{a}(i_1,\dots,i_k)$ is as in \eqref{100418:eq1}
for the elements $a_1,\dots,a_k\in\Pi A$.
\begin{proposition}\label{100422:prop3}
The $\mb Z_+$-graded superspace $\mc G(A)=\bigoplus_{k=0}^\infty\mc G_k(A)$,
where $\mc G_k(A)=\Pi W_{k-1}^{\as}(\Pi A)$,
with parity denoted by $p$,
together with the concatenation product
$\wedge:\,\mc G_h(A)\times\mc G_{k-h}(A)\to\mc G_k(A)$ given by \eqref{100422:eq4},
and with the Lie superalgebra bracket on $\Pi\mc G(A)=W^{\as}(\Pi A)$,
is a $\mb Z_+$-graded odd Poisson superalgebra.
\end{proposition}
\begin{proof}
First, we prove that $X\wedge Y$ in \eqref{100422:eq4} is an element of $\Pi W^{\as}_k(\Pi A)$,
namely, it is a map $S^k(\Pi A)\to\Pi A$ and it satisfies the Leibniz rule \eqref{100422:eq1}.
For the first assertion,
it is convenient to rewrite equation \eqref{100422:eq4} as a sum
over all permutations:
$$
\begin{array}{l}
\displaystyle{
(X\wedge Y)(a_1,\dots,a_k)
=
\frac{1}{h!(k-h)!}\sum_{\sigma\in S_k}
\epsilon_{a}(\sigma(1),\dots,\sigma(k))
} \\
\displaystyle{
\times (-1)^{p(Y)(\bar p(a_{\sigma(1)})+\dots+\bar p(a_{\sigma(h)}))}
X(a_{\sigma(1)},\dots,a_{\sigma(h)})Y(a_{\sigma(h+1)},\dots,a_{\sigma(k)})\,.
}
\end{array}
$$
Here we used the symmetry relations for $X$ and $Y$.
Using this, it is then immediate to check that
$$
(X\wedge Y)(a_{\tau(1)},\dots,a_{\tau(k)})=\epsilon_a(\tau(1),\dots,\tau(k))(X\wedge Y)(a_1,\dots,a_k)\,,
$$
for all permutations $\tau\in S_k$,
namely, $(X\wedge Y)(a_1,\dots,a_k)$
is supersymmetric in the variables $a_1,\dots,a_k\in\Pi A$, as we wanted.
Next, we show that $X\wedge Y$ satisfies the Leibniz rule \eqref{100422:eq1}.
We have
\begin{equation}\label{100423:eq1}
\begin{array}{l}
\displaystyle{
(X\wedge Y)(a_1,\dots,a_{k-1},bc)
=
\sum_{\substack{
i_1<\dots <i_{h}<k\\
i_{h+1}<\dots< i_k=k}}
\epsilon_{a_1,\dots,a_{k-1},bc}(i_1,\dots,i_k)
} \\
\displaystyle{
\vphantom{\Big(}
\,\,\,\,\,\,
\times
(-1)^{p(Y)(\bar p(a_{i_1})+\dots+\bar p(a_{i_h}))}
X(a_{i_1},\dots,a_{i_h})Y(a_{i_{h+1}},\dots,a_{i_{k-1}},bc)
} \\
\displaystyle{
+
\sum_{\substack{
i_1<\dots <i_{h}=k\\
i_{h+1}<\dots< i_k<k}}
\epsilon_{a_1,\dots,a_{k-1},bc}(i_1,\dots,i_k) (-1)^{p(Y)(\bar p(a_{i_1})+\dots+\bar p(a_{i_{h-1}})+\bar p(bc))}
} \\
\displaystyle{
\,\,\,\,\,\,\,\,\,\,\,\,\,\,\,\,\,\,\,\,\,\,\,\,\,\,\,\,\,\,\,\,\,\,\,\,\,\,\,\,\,\,\,\,\,
\times
X(a_{i_1},\dots,a_{i_{h-1}},bc)Y(a_{i_{h+1}},\dots,a_{i_k})\,.
}
\end{array}
\end{equation}
In the first term of the RHS of \eqref{100423:eq1}, we have, since $i_k=k$,
\begin{equation}\label{101109:eq1}
\epsilon_{a_1,\dots,a_{k-1},bc}(i_1,\dots,i_k)
=\epsilon_{a_1,\dots,a_{k-1},b}(i_1,\dots,i_k)
=\epsilon_{a_1,\dots,a_{k-1},c}(i_1,\dots,i_k)\,,
\end{equation}
and we can use the Leibniz rule for $Y$ to get:
$$
\begin{array}{c}
Y(a_{i_{h+1}},\dots,a_{i_{k-1}},bc)
=Y(a_{i_{h+1}},\dots,a_{i_{k-1}},b)c \\
+(-1)^{p(b)p(c)}Y(a_{i_{h+1}},\dots,a_{i_{k-1}},c)b\,.
\end{array}
$$
On the other hand,
in the second term of the RHS of \eqref{100423:eq1}, we have, since $i_h=k$,
\begin{equation}\label{101109:eq2}
\begin{array}{c}
\epsilon_{a_1,\dots,a_{k-1},bc}(i_1,\dots,i_k)
= \epsilon_{a_1,\dots,a_{k-1},b}(i_1,\dots,i_k)(-1)^{p(c)(\bar p(a_{i_{h+1}})+\dots+\bar p(a_{i_k}))} \\
= \epsilon_{a_1,\dots,a_{k-1},c}(i_1,\dots,i_k)(-1)^{p(b)(\bar p(a_{i_{h+1}})+\dots+\bar p(a_{i_k}))}\,,
\end{array}
\end{equation}
and we can use the Leibniz rule for $X$ and the commutativity of the product on $A$, to get
$$
\begin{array}{l}
X(a_{i_1},\dots,a_{i_{h-1}},bc)Y(a_{i_{h+1}},\dots,a_{i_k}) \\
\vphantom{\bigg(}
=
(-1)^{p(c)(p(Y)+\bar p(a_{i_{h+1}})+\dots+\bar p(a_{i_k}))}
X(a_{i_1},\dots,a_{i_{h-1}},b) Y(a_{i_{h+1}},\dots,a_{i_k}) c \\
+(-1)^{p(b)(p(c)+p(Y)+\bar p(a_{i_{h+1}})+\dots+\bar p(a_{i_k}))}
X(a_{i_1},\dots,a_{i_{h-1}},c) Y(a_{i_{h+1}},\dots,a_{i_k}) b .
\end{array}
$$
Furthermore, we have
\begin{equation}\label{101109:eq3}
\bar p(bc)=\bar p(b)+p(c)=p(b)+\bar p(c)\,\,\text{ in } \mb Z/2\mb Z\,.
\end{equation}
Putting the above formulas in equation \eqref{100423:eq1},
we get the desired Leibniz rule for $X\wedge Y$.

We now prove that formula \eqref{100422:eq4} for $X\wedge Y$ defines a product on the superspace
$\Pi W^{\as}(\Pi A)$ which is both commutative and associative.
Recall that we are denoting the parity on $W^{\as}(\Pi A)$ by $\bar p$
and on $\Pi W^{\as}(\Pi A)$ by $p$.
Hence, given $a_1,\dots,a_h\in\Pi A$ and $X\in W^{\as}_h(\Pi A)\subset\Hom(S^h\Pi A,\Pi A)$,
the element $X(a_1,\dots,a_h)\in\Pi A$ has parity $\bar p(X)+\bar p(a_1)+\dots+\bar p(a_h)$.
It follows that
$X(a_{i_1},\dots,a_{i_h})Y(a_{i_{h+1}},\dots,a_{i_k})$
has parity $\bar 1+p(X)+p(Y)+\bar p(a_1)+\dots+\bar p(a_k)$ as an element of $\Pi A$.
Hence, $X\wedge Y$ has parity $\bar 1+p(X)+p(Y)$ as an element of $W^{\as}(\Pi A)$,
or, equivalently, it has parity $p(X)+p(Y)$ as an element of $\Pi W^{\as}(\Pi A)$.
This shows that $\Pi W^{\as}(\Pi A)$, endowed with the wedge product \eqref{100422:eq4},
is a superalgebra.
Since $A$ is a commutative superalgebra, we have
$X(a_{i_1},\dots,a_{i_h})Y(a_{i_{h+1}},\dots,a_{i_k})=\pm Y(a_{i_{h+1}},\dots,a_{i_k})X(a_{i_1},\dots,a_{i_h})$,
where
$\pm = (-1)^{(p(X)+\bar p(a_{i_1})+\dots+\bar p(a_{i_h}))(p(Y)+\bar p(a_{i_{h+1}})+\dots+\bar p(a_{i_k}))}$.
This immediately implies the commutativity of the wedge product \eqref{100422:eq4}.
Moreover, given $X\in\Pi W^{\as}_{h-1}(\Pi A),\,
Y\in\Pi W^{\as}_{k-h-1}(\Pi A),\,Z\in\Pi W^{\as}_{\ell-k-1}(\Pi A)$,
and $a_1,\dots,a_\ell\in\Pi A$, we have, using associativity of $A$,
that both $(X\wedge(Y\wedge Z))(a_1,\dots,a_\ell)$ and $((X\wedge Y)\wedge Z)(a_1,\dots,a_\ell)$ are equal to
$$
\begin{array}{c}
\displaystyle{
\sum_{\substack{
i_1<\dots <i_{h}\\
i_{h+1}<\dots< i_k\\
i_{k+1}<\dots< i_\ell
}}
\epsilon_{a}(i_1,\dots,i_\ell)
(-1)^{p(Y)(\bar p(a_{i_1})+\dots+\bar p(a_{i_h})) + p(Z)(\bar p(a_{i_1})+\dots+\bar p(a_{i_k}))}
} \\
\displaystyle{
\times
X(a_{i_1},\dots,a_{i_h})Y(a_{i_{h+1}},\dots,a_{i_k})Z(a_{i_{k+1}},\dots,a_{i_\ell})
\,,
}
\end{array}
$$
proving associativity of the wedge product.

To complete the proof of the proposition,
we are left to prove that the Lie bracket on $W^{\as}(\Pi A)$ satisfies the odd Leibniz rule,
\begin{equation}\label{100423:eq2}
[X,Y\wedge Z]=[X,Y]\wedge Z+(-1)^{\bar p(X) p(Y)}Y\wedge[X,Z]\,,
\end{equation}
thus making $\Pi W^{\as}(\Pi A)$ an odd Poisson superalgebra.
This follows immediately from the following lemma.
\begin{lemma}\label{100423:lem}
The left and right Leibniz formulas for the box product \eqref{100418:eq1} of $W^{\as}(\Pi A)$ hold:
\begin{equation}\label{100423:eq3}
\begin{array}{rcl}
X\Box(Y\wedge Z) &=& (X\Box Y)\wedge Z+(-1)^{\bar p(X)p(Y)}Y\wedge(Z\Box X)\,,\\
(X\wedge Y)\Box Z &=& X\wedge (Y\Box Z)+(-1)^{p(Y)\bar p(Z)}(X\Box Z)\wedge Y\,.
\end{array}
\end{equation}
\end{lemma}
\begin{proof}
The first formula in \eqref{100423:eq3} is obtained, by a straightforward computation,
using the definitions \eqref{100418:eq1} and \eqref{100422:eq4} of the box product and of the wedge product,
and the Leibniz rule \eqref{100422:eq2} for $X:\,S^{h+1}(\Pi A)\to\Pi A$.
The second formula in \eqref{100423:eq3} is also obtained by a straightforward computation,
using \eqref{100418:eq1} and \eqref{100422:eq4}.
\end{proof}
\end{proof}

\begin{remark}\label{100424:rem2}
Assuming that $A$ is a purely even commutative associative algebra,
we may consider the Lie algebroid $(A,\Der(A))$,
the associated odd Poisson superalgebra of polyvector fields $S_A(\Pi\Der(A))$,
and also the opposite odd Poisson superalgebra $S_A(\Pi\Der(A))^{\op}$,
defined in Remark \ref{100424:rem1}.
Then, we have a homomorphism of $\mb Z_+$-graded odd Poisson superalgebras
$\phi:\,S_A(\Pi\Der(A))^{\op}\to\mc G(A)=\Pi W^{\as}(\Pi A)$, given by
$$
\phi(X_1\wedge\dots\wedge X_k)(a_1,\dots,a_k)
=
\det(X_i(a_j))_{i,j=1}^k\,.
$$
Indeed, it is easy to check that the map $\phi$ is a homomorphism of associative superalgebras.
Moreover, since it is the identity on $A\oplus\Pi\Der(A)$, it is automatically
a Lie superalgebra homomorphism, due to the Leibniz rule.
In fact, the map $\phi$ is an isomorphism provided that $\Der(A)$ is a free module over $A$
of finite rank, for example when $A$ is the algebra of polynomials in finitely many variables.
In general, though, this map is neither injective nor surjective.
\end{remark}

\subsection{Poisson superalgebra structures and Poisson superalgebra cohomology complexes}
\label{sec:3.1.5}

\begin{proposition}\label{100422:prop2}
The Poisson superalgebra structures on
a commutative associative superalgebra $A$
are in bijective correspondence, via \eqref{100412:eq2},
with the set
\begin{equation}\label{100422:eq3}
\big\{X\in W^{\as}_1(\Pi A)_{\bar1}\,\big|\,[X,X]=0\big\}\,.
\end{equation}
\end{proposition}
\begin{proof}
By the results in Section \ref{sec:2.5}
the elements $X\in W_1(\Pi A)_{\bar1}$ such that $[X,X]=0$
correspond, via \eqref{100412:eq2}, to the Lie superalgebra structures on $A$.
Moreover, to say that $X$ lies in $W_1^{\as}(\Pi A)$ means that the corresponding Lie bracket
satisfies the Leibniz rule, hence $A$ is a Poisson superalgebra.
\end{proof}

It follows from the above Proposition that, for any Poisson superalgebra structure on $A$,
we have a differential $d_X=\ad X$ on the superspace $W^{\as}(\Pi A)$,
where $X$ in \eqref{100422:eq3}
is associated to the Lie superalgebra structure on $A$.
This differential is obviously  an odd derivation of the Lie bracket on $W^{\as}(\Pi A)$,
and an odd derivation of the concatenation product on $\Pi W^{\as}(\Pi A)$.
We thus get a cohomology complex $(\mc G(A),d_X)$.
\begin{remark}\label{100424:rem3}
If $A$ is a purely even Poisson algebra,
then the odd Poisson superalgebra $S_A(\Pi\Der(A))$
of polyvector fields on $A$ is the usual Poisson cohomology complex,
with the differential $d=\ad X$, where $X$ is the bivector field defining
the Poisson algebra structure on $A$ \cite{Lic}.
In this case, the odd Poisson superalgebra homomorphism $\phi$ defined in Remark \ref{100424:rem2}
is a homomorphism of cohomology complexes.
\end{remark}


\subsection{The universal Poisson superalgebra $W^{\as}(A)$
and odd Poisson superalgebra structures on $A$}\label{sec:3.2}

As in the previous section, let $A$ be a commutative associative superalgebra, with parity $p$
and let $\Der(A)$ be the Lie superalgebra of derivations of $A$.
Instead of $W(\Pi A)$, we may consider the universal Lie superalgebra
$W(A)=\bigoplus_{k=-1}^\infty W_k(A)$, with parity still denoted by $p$.
As we shall see below, we arrive at a ``dual'' picture:
$W^{\as}(A)\subset W(A)$ (defined below)
has a natural structure of a Poisson superalgebra,
while elements $X\in W_1(A)_{\bar 1}$ such that $[X,X]=0$
correspond to the odd Poisson superalgebra structures on $A$.

The Lie superalgebra $\Der(A)$ is a subalgebra of $W_0(A)=\End(A)$,
so we can consider its full prolongation,
which we denote by
$$
W^{\as}(A)=\bigoplus_{k=-1}^\infty W^{\as}_k(A)\subset W(A)\,.
$$
\begin{proposition}\label{100422b:prop1}
For $k\geq-1$, the superspace $W^{\as}_k(A)$ consists of linear maps
$X:\,S^{k+1}(A)\to A$ satisfying the following Leibniz rule (for $a_0,\dots,a_{k-1},b,c\in A$):
\begin{equation}\label{100422b:eq1}
X(a_0,\dots,a_{k-1},bc)=X(a_0,\dots,a_{k-1},b)c+(-1)^{p(b)p(c)}X(a_0,\dots,a_{k-1},c)b\,.
\end{equation}
\end{proposition}
\begin{proof}
It follows by an easy induction on $k\geq0$.
\end{proof}
%

We next define a structure of commutative associative superalgebra on the superspace $W^{\as}(A)$,
making it a Poisson superalgebra.
Given $X\in W^{\as}_{h-1}(A)$ and $Y\in W^{\as}_{k-h-1}(A)$, for $h\geq0,\,k-h\geq0$,
we let their \emph{concatenation product} $X\cdot Y\in W^{\as}_{k-1}(A)$
be the following map:
\begin{equation}\label{100422b:eq4}
\begin{array}{c}
\displaystyle{
(X\cdot Y)(a_1,\dots,a_k)
=
\sum_{\substack{
i_1<\dots <i_{h}\\
i_{h+1}<\dots< i_k}}
\epsilon_a(i_1,\dots,i_k) (-1)^{p(Y)(p(a_{i_1})+\dots+p(a_{i_h}))}
} \\
\displaystyle{
\times
X(a_{i_1},\dots,a_{i_h})Y(a_{i_{h+1}},\dots,a_{i_k})\,.
}
\end{array}
\end{equation}
Note that $\epsilon_a(i_1,\dots,i_k)$ in \eqref{100422b:eq4} is not the same as
in \eqref{100422:eq4}, since here we consider $a_1,\dots,a_k$ as elements of $A$,
not of $\Pi A$.
\begin{proposition}\label{100422b:prop3}
The $\mb Z_+$-graded superspace $\mc P(A)=\bigoplus_{k=0}^\infty\mc P_k(A)$,
where $\mc P_k(A)=W_{k-1}^{\as}(A)$,
together with the concatenation product
$\cdot:\,\mc P_h(A)\times\mc P_{k-h}\to\mc P_k(A)$ given by \eqref{100422b:eq4},
and with the Lie superalgebra bracket on $\mc P(A)=W^{\as}(A)$,
is a Poisson superalgebra.
\end{proposition}
\begin{proof}
First, we prove, in the same way as in Proposition \ref{100422:prop3},
that $X\cdot Y$ in \eqref{100422b:eq4} is an element of $W^{\as}_k(A)$,
namely, it is a map $S^k(A)\to A$ and it satisfies the Leibniz rule \eqref{100422b:eq1}.
Moreover, it is immediate to check that \eqref{100422b:eq4} makes $W^{\as}(A)$
a commutative associative superalgebra, using commutativity and associativity
of the product on $A$.
To complete the proof of the proposition,
we are left to prove that the Lie bracket on $W^{\as}(A)$ satisfies the usual Leibniz rule,
\begin{equation}\label{100423b:eq2}
[X,Y\cdot Z]=[X,Y]\cdot Z+(-1)^{p(X) p(Y)}Y\cdot[X,Z]\,,
\end{equation}
thus making $W^{\as}(A)$ a Poisson superalgebra.
This follows immediately from the following lemma.
\begin{lemma}\label{100423b:lem}
The left and right Leibniz formulas for the box product \eqref{100418:eq1} of $W^{\as}(A)$ hold:
\begin{equation}\label{100423b:eq3}
\begin{array}{rcl}
X\Box(Y\cdot Z) &=& (X\Box Y)\cdot Z+(-1)^{p(X)p(Y)}Y\cdot(Z\Box X)\,,\\
(X\cdot Y)\Box Z &=& X\cdot (Y\Box Z)+(-1)^{p(Y)p(Z)}(X\Box Z)\cdot Y\,.
\end{array}
\end{equation}
\end{lemma}
\begin{proof}
The proof is the same as that of Lemma \ref{100423:lem}.
\end{proof}
\end{proof}

\begin{remark}\label{100424b:rem2}
Assuming that $A$ is a purely even commutative associative algebra,
we may consider the associated Poisson algebra $S_A(\Der(A))$.
Then, we have a homomorphism of $\mb Z_+$-graded Poisson algebras
$\phi:\,S_A(\Der(A))\to\mc P(A)=W^{\as}(A)$, given by
$$
\phi(X_1,\dots,X_k)(a_1,\dots,a_k)
=
\sum_{\sigma\in S_k}X_1(a_{\sigma(1)})\dots X_k(a_{\sigma(k)})\,.
$$
Indeed, it is easy to check that the map $\phi$ is a homomorphism of associative algebras.
Moreover, since it is the identity on $A\oplus\Der(A)$, it is automatically
a Lie algebra homomorphism, due to the Leibniz rule.
\end{remark}

\begin{proposition}\label{100422b:prop2}
The odd Poisson superalgebra structures on $A$
are in bijective correspondence, via \eqref{100412:eq2},
with the set
\begin{equation}\label{100422b:eq3}
\big\{X\in W^{\as}_1(A)_{\bar1}\,\big|\,[X,X]=0\big\}\,.
\end{equation}
\end{proposition}
\begin{proof}
The proof is similar to that of Proposition \ref{100422:prop2}.
\end{proof}

It follows from the above Proposition that, for any odd Poisson superalgebra $A$,
we have a differential $d_X=\ad X$ on the superspace $W^{\as}(A)$,
where $X$ in \eqref{100422b:eq3}
is associated to the Lie superalgebra structure on $\Pi A$.
This differential is obviously  an odd derivation of the Lie bracket on $W^{\as}(A)$,
and an odd derivation of the concatenation product on $W^{\as}(A)$.
We thus get a cohomology complex $(\mc P(A),d_X)$.



\section{The Lie superalgebra $W^{\partial}(V)$ for an $\mb F[\partial]$-module $V$,
and Lie conformal superalgebra cohomology}
\label{sec:4}

\noindent
In this section we repeat the discussion of Section \ref{sec:2}
in the case of conformal superalgebras.
Recall that a conformal superalgebra is a vector superspace $R$
endowed with a structure of an $\mb F[\partial]$-module
compatible with parity, and with a
$\lambda$-product, i.e., a linear map
$R\otimes R\to\mb F[\lambda]\otimes R,\,a\otimes b\mapsto a_\lambda b$,
satisfying the sesquilinearity relations:
\begin{equation}\label{100421:eq2}
(\partial a)_\lambda b=-\lambda a_\lambda b
\,\,,\,\,\,\,
a_\lambda(\partial b)=(\lambda+\partial) (a_\lambda b)\,.
\end{equation}
A conformal superalgebra is called commutative (resp. skewcommutative) if
\begin{equation}\label{100421:eq3}
b_\lambda a=(-1)^{p(a)p(b)} a_{-\lambda-\partial} b
\,\,\,\,
\Big( \text{ resp. }
= - (-1)^{p(a)p(b)} a_{-\lambda-\partial} b
\Big)\,,
\end{equation}
where $\partial$ is moved to the left.
In the skewcommutative case the $\lambda$-product is usually called $\lambda$-bracket
and it is denoted by $[a_\lambda b]$.
If in addition to skewcommutativity the $\lambda$-bracket satisfies the Jacobi identity,
\begin{equation}\label{100421:eq4}
[a_\lambda[b_\mu c]]-(-1)^{p(a)p(b)}[b_\mu[a_\lambda c]]=[[a_\lambda b]_{\lambda+\mu}c]\,,
\end{equation}
$R$ is called a Lie conformal superalgebra \cite{K}.

For an $\mb F[\partial]$-module $R$,
one usually denotes by $\tint$ the canonical map $R\to R/\partial R$.

Recall that, if $R$ is a Lie conformal superalgebra,
then the vector superspace $R/\partial R$ has a well-defined
structure of a Lie superalgebra,
given by $[\tint a,\tint b]=\tint [a_\lambda b]\,\big|_{\lambda=0},\,a,b\in R$.
Moreover, $R$ is a left module over the Lie superalgebra $R/\partial R$,
with the well-defined action $\big(\tint a\big)(b)=[a_\lambda b]\,\big|_{\lambda=0},\,a,b\in R$,
by derivations of the Lie conformal superalgebra $R$.

\subsection{The Lie superalgebra $W^\partial(V)$}\label{sec:4.1}

Let $V$ be a vector superspace with parity $\bar p$,
endowed with a structure of an $\mb F[\partial]$-module, compatible with the parity.
Motivated by the construction of $W(V)$ introduced in Section \ref{sec:2.2},
we construct in this section the $\mb Z$-graded Lie superalgebra
$W^\partial(V)=\bigoplus_{k=-1}^\infty W^\partial_k(V)$,
which, to some extent, plays the same role in the theory of Lie conformal algebras
as $W(V)$ plays in the theory of Lie algebras.

For $k\geq-1$, we let $W^\partial_k(V)$ be the superspace of $(k+1)$-$\lambda$-brackets on $V$,
as defined in \cite{DSK2}, namely,
\begin{equation}\label{100516:eq1}
W^\partial_k(V)
=
\Hom{}^{\sym}_{\mb F[\partial]^{\otimes (k+1)}}(V^{\otimes(k+1)},
\mb F_-[\lambda_0,\dots,\lambda_k]\otimes_{\mb F[\partial]}V)\,,
\end{equation}
with the parity $\bar p$ induced by the parity on $V$.
Here and further $\mb F_{-}[\lambda_0,\dots,\lambda_k]$
denotes the space of polynomials in the $k+1$ variables $\lambda_0,\dots,\lambda_k$,
endowed with a structure of $\mb F[\partial]^{\otimes (k+1)}$-module,
by letting $P_0(\partial)\otimes\dots\otimes P_k(\partial)$
act as $P_0(-\lambda_0)\dots P_k(-\lambda_k)$.
Using the embedding $\mb F[\partial]\subset\mb F[\partial]^{\otimes (k+1)}$
given by the standard comultiplication, we get the corresponding $\mb F[\partial]$-module
structure on $\mb F_{-}[\lambda_0,\dots,\lambda_k]$,
namely $\partial$ acts by multiplication
by $-(\lambda_0+\dots+\lambda_k)$.
Thus, $W^\partial_k(V)$ consists of all linear maps
$$
\begin{array}{rcl}
X\,&:&\,\,
V^{\otimes(k+1)}
\,\,\,\,\,\,\,\,\,\,\,\,
\longrightarrow
\,\,\,\,\,\,
\mb F_-[\lambda_0,\dots,\lambda_k]\otimes_{\mb F[\partial]} V
\,,\\
&& v_0\otimes\dots\otimes v_k
\,\,\,
\mapsto
\,\,\,\,\,\,
X_{\lambda_0,\dots,\lambda_k}(v_0,\dots,v_k)\,,
\end{array}
$$
satisfying the following conditions:
\begin{enumerate}[]
\item
(sesquilinearity)
$X_{\lambda_0,\dots,\lambda_k}(v_0,\dots,\partial v_i,\dots,v_k)
=-\lambda_iX_{\lambda_0,\dots,\lambda_k}(v_0,\dots,v_k)$,
\item
(symmetry)\,\,
$X_{\lambda_0,\dots,\lambda_k}(v_0,\dots,v_k)
=
\epsilon_v(i_0,\dots,i_k)
X_{\lambda_{i_0},\dots,\lambda_{i_k}}(v_{i_0},\dots,v_{i_k})$,
for all permutations $(i_0,\dots,i_k)$ of $(0,\dots,k)$.
\end{enumerate}
Here $\epsilon_v(i_0,\dots,i_k)$ is the same as in \eqref{100418:eq1}.

For example, $W^\partial_{-1}(V)=V/\partial V$.
Furthermore, $W^\partial_0(V)=\End_{\mb F[\partial]}(V)$,
namely the map $X_\lambda:\,V\to\mb F_-[\lambda]\otimes_{\mb F[\partial]}V$
is identified with the $\mb F[\partial]$-module endomorphism $X:\,V\to V$
given by $X(v)=X_{-\partial}(v)$, with $\partial$ moved to the left.

For $k\geq0$, we can identify, in a non-canonical way,
$\mb F_-[\lambda_0,\dots,\lambda_k]\otimes_{\mb F[\partial]} V
=
\mb F_-[\lambda_0,\dots,\lambda_{k-1}]\otimes V$,
by letting $\lambda_k=-\lambda_0-\dots-\lambda_{k-1}-\partial$.
For example, $W^\partial_1(V)$ is identified with the space of
$\lambda$-brackets
$$
\{\cdot\,_\lambda\,\cdot\}\,:\,\,V\otimes V\to V[\lambda]
\,\,,\,\,\,\,
u\otimes v\mapsto\{u_\lambda v\}\,,
$$
satisfying:
\begin{enumerate}[]
\item
(sesquilinearity)
$\{\partial u_\lambda v\}=-\lambda\{u_\lambda v\}
\,,\,\,
\{u_\lambda\partial v\}=(\lambda+\partial)\{u_\lambda v\}$,
\item
(commutativity)
$\{v_\lambda u\}=(-1)^{\bar p(u)\bar p(v)}\{u_{-\lambda-\partial}v\}$.
\end{enumerate}

Keeping in mind the construction in Section \ref{sec:2.2},
we next endow $W^\partial(V)$ with a structure of a $\mb Z$-graded Lie superalgebra as follows.
If $X\in W^\partial_h(V)$, $Y\in W^\partial_{k-h}(V)$, with $h\geq-1,\,k\geq h-1$,
we define $X\Box Y$ to be the following element of $W^\partial_{k}(V)$:
\begin{equation}\label{100418:eq2}
\begin{array}{l}
\displaystyle{
\big(X\Box Y\big)_{\lambda_0,\dots,\lambda_k}(v_0, \dots, v_k)
= \sum_{\substack{
i_0<\dots <i_{k-h}\\
i_{k-h+1}<\dots< i_k}}
\epsilon_v(i_0,\dots,i_k)
} \\
\displaystyle{
\times X_{\lambda_{i_0}+\dots+\lambda_{i_{k-h}},\lambda_{i_{k-h+1}},\dots,\lambda_{i_k}}
(Y_{\lambda_{i_0},\dots,\lambda_{i_{k-h}}}(v_{i_0},\dots, v_{i_{k-h}}), v_{i_{k-h+1}},\dots, v_{i_k})\,.
}
\end{array}
\end{equation}
The above formula, for $h=-1$ gives zero, while for $k=h-1$ gives
$X_{0,\lambda_0,\dots,\lambda_k}(Y,v_0,\dots,v_k)$,
which is well defined for $Y\in V/\partial V$ by the sesquilinearity condition on $X$.
We observe
that the box product \eqref{100418:eq2} is well defined,
namely it preserves the defining relations
of $\mb F_-[\lambda_0,\dots,\lambda_k]\otimes_{\mb F[\partial]} V$.
Indeed, if
$X_{\lambda_0,\dots,\lambda_h}(v_0,\dots,v_h)
=(\lambda_0+\dots+\lambda_h+\partial)F(\lambda_0,\dots,\lambda_h;v_0,\dots,v_h)$,
for some polynomial $F$ in $\lambda_0,\dots,\lambda_h$ with coefficients in $V$, then
$$
\big(X\Box Y\big)_{\lambda_0,\dots,\lambda_k}(v_0, \dots, v_k)
= (\lambda_0+\dots+\lambda_k+\partial)\sum \pm F(\dots)\,,
$$
which is zero in $\mb F_-[\lambda_0,\dots,\lambda_k]\otimes_{\mb F[\partial]} V$.
Similarly, if
$Y_{\lambda_0,\dots,\lambda_{k-h}}(v_0,\dots,v_{k-h})
=(\lambda_0+\dots+\lambda_{k-h}+\partial)G(\lambda_0,\dots,\lambda_{k-h};v_0,\dots,v_{k-h})$,
for some polynomial $G$ in $\lambda_0,\dots,\lambda_{k-h}$ with coefficients in $V$,
then $\big(X\Box Y\big)_{\lambda_0,\dots,\lambda_k}(v_0, \dots, v_k) = 0$,
by the sesquilinearity condition for $X$.
Moreover, it is straightforward to check that $X\Box Y$
satisfies both the sesquilinearity and the symmetry conditions, if $X$ and $Y$ do.
Hence, $\Box$ is a well-defined map $W^\partial_h(V)\times W^\partial_{k-h}(V)\to W^\partial_k(V)$.

We then define the bracket $[\cdot\,,\,\cdot]:\,W^\partial_h(V)\times W^\partial_{k-h}(V)\to W^\partial_k(V)$
by the same formula as before:
\begin{equation}\label{box2}
[X,Y]=X\Box Y-(-1)^{\bar p(X)\bar p(Y)}Y\Box X\,.
\end{equation}
\begin{proposition}
\begin{enumerate}[(a)]
\item
The bracket \eqref{box2} defines a Lie superalgebra structure on $W^\partial(V)$.
\item
We have the following canonical homomorphism of $\mb Z$-graded Lie superalgebras:
\begin{equation}\label{100419:eq3}
W^\partial(V)\to W(V/\partial V)
\,\,,\,\,\,\,
X_{\lambda_0,\dots,\lambda_k}\mapsto \tint X_{0,\dots,0}\,,
\end{equation}
where, as usual, for $v\in V$, $\tint v$ denotes its image in $V/\partial V$.
\end{enumerate}
\end{proposition}
\begin{proof}
The bracket (\ref{box2}) is skewcommutative by construction.
To prove Jacobi identity, it suffices to check that
the box-product \eqref{100418:eq2} is right symmetric, i.e.,
$(X,Y,Z)=(-1)^{\bar p(Y)\bar p(Z)}(X,Z,Y)$, where
$(X,Y,Z)=(X\Box Y)\Box Z-X\Box(Y\Box Z)$ is the associator of $X,Y,Z$.
Let then $X\in W^\partial_h(V),\,Y\in W^\partial_{k-h}(V)$ and $Z\in W^\partial_{\ell-k}(V)$.
We have, by the definition \eqref{100418:eq2} of the box-product,
\begin{equation}\label{100419:eq1}
\begin{array}{l}
\displaystyle{
\big(X\Box (Y\Box Z)\big)_{\lambda_0,\dots,\lambda_\ell}(v_0, \dots, v_\ell)
= \sum_{\substack{
i_0<\dots <i_{\ell-k}\\
i_{\ell-k+1}<\dots< i_{\ell-h}\\
i_{\ell-h+1}<\dots< i_\ell
}}
\epsilon_v(i_0,\dots,i_\ell)
} \\
\displaystyle{
\times X_{\lambda_{i_0}+\dots+\lambda_{i_{\ell-h}},\lambda_{i_{\ell-h+1}},\dots,\lambda_{i_\ell}}
\Big(
Y_{\lambda_{i_0}+\dots+\lambda_{i_{\ell-k}},\lambda_{i_{\ell-k+1}},\dots,\lambda_{i_{\ell-h}}}
} \\
\displaystyle{
\,\,\,\,\,\,\,\,\,\,\,\,
\big(
Z_{\lambda_{i_0},\dots,\lambda_{i_{\ell-k}}}(v_{i_0},\dots, v_{i_{\ell-k}}),
v_{i_{\ell-k+1}},\dots, v_{i_{\ell-h}}
\big),
v_{i_{\ell-h+1}},\dots, v_{i_\ell}
\Big)\,,
}
\end{array}
\end{equation}
and
\begin{equation}\label{100419:eq2}
\begin{array}{l}
\vphantom{\bigg(}
\displaystyle{
\big((X\Box Y)\Box Z\big)_{\lambda_0,\dots,\lambda_\ell}(v_0, \dots, v_\ell)
- \big(X\Box (Y\Box Z)\big)_{\lambda_0,\dots,\lambda_\ell}(v_0, \dots, v_\ell)
} \\
\displaystyle{
=\sum_{\substack{
i_0<\dots <i_{k-h}\\
i_{k-h+1}<\dots< i_{\ell-h+1}\\
i_{\ell-h+2}<\dots< i_\ell
}}
\epsilon_v(i_0,\dots,i_\ell)
(-1)^{\bar p(Z)(\bar p(v_{i_0})+\dots+\bar p(v_{i_{k-h}}))}
} \\
\displaystyle{
\times X_{\lambda_{i_0}+\dots+\lambda_{i_{k-h}},
\lambda_{i_{k-h+1}}+\dots+\lambda_{i_{\ell-h+1}},
\lambda_{i_{\ell-h+2}},\dots,\lambda_{i_\ell}}
\Big(
Y_{\lambda_{i_0},\dots,\lambda_{i_{k-h}}}(v_{i_0},\dots, v_{i_{k-h}}),
} \\
\displaystyle{
\,\,\,\,\,\,\,\,\,\,\,\,\,\,\,\,\,\,\,\,\,\,\,\,\,\,\,\,\,\,\,\,\,\,\,\,\,\,\,\,\,\,\,\,\,\,\,
Z_{\lambda_{i_{k-h+1}},\dots,\lambda_{i_{\ell-h+1}}}(v_{i_{k-h+1}},\dots, v_{i_{\ell-h+1}}),
v_{i_{\ell-h+2}},\dots, v_{i_\ell}
\Big)\,.
}
\end{array}
\end{equation}
We then observe that the RHS above is supersymmetric under the exchange of $Y$ and $Z$.
This concludes the proof of part (a).
For part (b), one easily checks that the map \eqref{100419:eq3} is a well-defined
linear map of $\mb Z$-graded superspaces, and it is a homomorphism for the box-products
in $W^\partial(V)$ and in $W(V/\partial V)$.
Hence it is a $\mb Z$-graded Lie superalgebra homomorphism.
\end{proof}

For $X\in W^\partial_{k}(V)$ and $Y\in V/\partial V=W^\partial_{-1}(V)$, we have
\begin{equation}\label{100419:eq4}
[X,Y]_{\lambda_1,\dots,\lambda_k}(v_1,\dots,v_k)=X_{0,\lambda_1,\dots,\lambda_k}(Y,v_1,\dots,v_k)
\,\,,\,\,\,\,
v_1,\dots,v_k\in V\,.
\end{equation}
For $X\in W^\partial_0(V)=\End_{\mb F[\partial]}(V)$ and $Y\in W^\partial_k(V),\,k\geq-1$, we have
\begin{equation}\label{100419:eq5}
\begin{array}{l}
[X,Y]_{\lambda_0,\dots,\lambda_k}(v_0,\dots,v_k)=
X\big(Y_{\lambda_0,\dots,\lambda_k}(v_0,\dots,v_k)\big) \\
\displaystyle{
-(-1)^{\bar p(X)\bar p(Y)} \sum_{i=0}^k (-1)^{\bar p(X)\bar s_{0,i-1}}
Y_{\lambda_0,\dots,\lambda_k}(v_0,\dots X(v_i)\dots,v_k)\,,
}
\end{array}
\end{equation}
where $\bar s_{ij}$ is as in \eqref{100412:eq5}.
Finally, if $X\in W^\partial_1(V)$ and $Y\in W^\partial_{k-1}(V),\,k\geq0$, we have
$[X,Y]=X\Box Y-(-1)^{\bar p(X)\bar p(Y)}Y\Box X$, where
\begin{equation}\label{100419:eq6}
\begin{array}{l}
\displaystyle{
(X\Box Y)_{\lambda_0,\dots,\lambda_k}(v_0,\dots,v_k)
} \\
\displaystyle{
\,\,\,\,\,\,\,\,\,
= \sum_{i=0}^k (-1)^{\bar p(v_i)\bar s_{i+1,k}}
X_{\lambda_0+\stackrel{i}{\check{\dots}}+\lambda_k,\lambda_i}
\big(Y_{\lambda_0,\stackrel{i}{\check{\dots}},\lambda_k}(v_0,\stackrel{i}{\check{\dots}},v_k),v_i\big)\,,
} \\
\displaystyle{
(Y\Box X)_{\lambda_0,\dots,\lambda_k}(v_0,\dots,v_k)
=\!\!
\sum_{0\leq i<j\leq k} \!\!
(-1)^{\bar p(v_i)\bar s_{0,i-1}+\bar p(v_j)(\bar s_{0,i-1}+\bar s_{i+1,j-1})}
} \\
\displaystyle{
\,\,\,\,\,\,\,\,\,\,\,\,\,\,\,\,\,\,\,\,\,\,\,\,\,\,\,\,\,\,\,\,\,\,\,\,
\times Y_{\lambda_i+\lambda_j,\lambda_0,\stackrel{i}{\check{\dots}}\stackrel{j}{\check{\dots}},\lambda_k}
\big(X_{\lambda_i,\lambda_j}(v_i,v_j),v_0,\stackrel{i}{\check{\dots}}\,\stackrel{j}{\check{\dots}},v_k\big)\,.
}
\end{array}
\end{equation}
In particular, if both $X$ and $Y$ are in $W^\partial_1(V)$,
we get
\begin{equation}\label{100419:eq7}
\begin{array}{c}
(X\Box Y)_{\lambda_0,\lambda_1,\lambda_2}(v_0,v_1,v_2) =
X_{\lambda_0+\lambda_1,\lambda_2}\big(Y_{\lambda_0,\lambda_1}(v_0,v_1),v_2\big) \\
\vphantom{\Bigg(}
+(-1)^{\bar p(v_1)(\bar p(Y)+\bar p(v_0))}
X_{\lambda_1,\lambda_0+\lambda_2}\big(v_1,Y_{\lambda_0,\lambda_2}(v_0,v_2)\big) \\
+(-1)^{\bar p(v_0)\bar p(Y)}
X_{\lambda_0,\lambda_1+\lambda_2}\big(v_0,Y_{\lambda_1,\lambda_2}(v_1,v_2)\big) \,.
\end{array}
\end{equation}

\subsection{The space $W^\partial(V,U)$ as a reduction of $W^\partial(V\oplus U)$}
\label{sec:4.2}

Let $V$ and $U$ be vector superspaces with parity $\bar p$, endowed with a structure of $\mb F[\partial]$-modules.
We define the $\mb Z$-graded vector superspace (with parity still denoted by $\bar p$)
$W^\partial(V,U)=\bigoplus_{k\geq-1}W^\partial_k(V,U)$,
where
$$
W^\partial_k(V,U)
=\Hom{}^{\sym}_{\mb F[\partial]^{\otimes(k+1)}}(V^{\otimes(k+1)},
\mb F_-[\lambda_0,\dots,\lambda_k]\otimes_{\mb F[\partial]}U)\,.
$$

In the same way as in Section \ref{sec:2.3},
$W^\partial(V,U)$ is obtained as a subquotient of the universal
Lie superalgebra $W^\partial(V\oplus U)$,
via the canonical isomorphism of superspaces
\begin{equation}\label{100421:eq1}
\begin{array}{c}
\mc U/\mc K
\stackrel{\sim}{\longrightarrow} W^\partial_k(V, U)\,,
\end{array}
\end{equation}
where $\mc U$ and $\mc K$ are the following subspaces of $W^\partial_k(V\oplus U)$:
$$
\begin{array}{rcl}
\mc U
&=&
\Hom{}^{\sym}_{\mb F[\partial]^{\otimes(k+1)}}((V\oplus U)^{\otimes(k+1)},
\mb F_-[\lambda_0,\dots,\lambda_k]\otimes_{\mb F[\partial]}U) \,, \\
\mc K &=& \big\{Y \,\big|\, Y(V^{\otimes(k+1)})=0\big\} \,.
\end{array}
$$
The following analogue of Proposition \ref{100414:prop} holds:
\begin{proposition}\label{100421:prop}
Let $X\in W^\partial_h(V\oplus U)$. Then the adjoint action of $X$ on $W^\partial(V\oplus U)$
leaves the subspaces $\mc U$ and $\mc K$ invariant
provided that
\begin{enumerate}[(i)]
\item
$X_{\lambda_0,\dots,\lambda_h}(w_0,\dots,w_h)\in \mb F_-[\lambda_0,\dots,\lambda_h]\otimes_{\mb F[\partial]} U$
if one of the arguments $w_i$ lies in $U$,
\item
$X_{\lambda_0,\dots,\lambda_h}(v_0,\dots,v_h)\in \mb F_-[\lambda_0,\dots,\lambda_h]\otimes_{\mb F[\partial]} V$
if all the arguments $v_i$ lie in $V$.
\end{enumerate}
In this case, $\ad X$ induces a well-defined linear map on the reduced space $W^\partial(V,U)$,
via the isomorphism \eqref{100421:eq1}.
\end{proposition}

%
%

\subsection{Lie conformal superalgebra structures}\label{sec:4.4}

By definition, the even elements $X\in W^\partial_1(V)$
are exactly the commutative conformal superalgebra structures on $V$:
for $X\in W^\partial_1(V)_{\bar0}$
we get a commutative $\lambda$-product on $V$ by letting $u_\lambda v=X_{\lambda,-\lambda-\partial}(u,v)$.
Similarly, the skewcommutative conformal superalgebra structures on an $\mb F[\partial]$-module $R$
with parity $p$ are in bijective correspondence with the odd elements of $W^\partial_1(\Pi R)$:
for $X\in W^\partial_1(\Pi R)_{\bar1}$, we get a skewcommutative $\lambda$-bracket on $R$ by letting
\begin{equation}\label{100421:eq5}
[a_\lambda b]=(-1)^{p(a)}X_{\lambda,-\lambda-\partial}(a,b)\,\,,\,\,\,\,a,b\in R\,,
\end{equation}
and vice-versa.

Furthermore, let $X\in W^\partial_1(\Pi R)_{\bar1}$,
and consider the corresponding skewcommutative $\lambda$-product \eqref{100421:eq5} on $R$.
The Lie bracket of $X$ with itself then becomes, by \eqref{100419:eq7},
$$
\begin{array}{l}
[X,X]_{\lambda,\mu,-\lambda-\mu-\partial}(a,b,c)
= 2(X\Box X)_{\lambda,\mu,-\lambda-\mu-\partial}(a,b,c) \\
= -(-1)^{p(b)} 2 \Big\{
[a_\lambda[b_\mu c]]-(-1)^{p(a)p(b)} [b_\mu[a_\lambda c]]
- [[a_\lambda b]_{\lambda+\mu}c]
\Big\}\,.
\end{array}
$$
Hence, the Lie conformal superalgebra structures on $R$
are in bijective correspondence, via \eqref{100421:eq5},
with the set
\begin{equation}\label{100421:eq6}
\big\{X\in W^\partial_1(\Pi R)_{\bar1}\,\big|\,[X,X]=0\big\}\,.
\end{equation}
Therefore, for any Lie conformal superalgebra $R$, we have a differential $d_X=\ad X$,
where $X$ in \eqref{100421:eq6} is associated to the Lie conformal superalgebra structure on $R$,
on the superspace $W^\partial(\Pi R)$, which makes it a cohomology complex
so that the differential $d_X$ is a derivation of the Lie bracket.

\subsection{Lie conformal superalgebra modules and cohomology complexes}\label{sec:4.5}

Let $R$ and $M$ be vector superspaces with parity $p$,
endowed with $\mb F[\partial]$-module structures.
Consider the reduced superspace
$W^\partial(\Pi R,\Pi M)$
introduced in Section \ref{sec:4.2}, with parity denoted by $\bar p$.

Suppose now that $R$ is a Lie conformal superalgebra and $M$ is an $R$-module.
This is equivalent to say that
we have a Lie conformal superalgebra structures on the $\mb F[\partial]$-module $R\oplus M$
extending the $\lambda$-bracket on $R$,
and such that $M$ is an abelian ideal, the bracket between $a\in R$ and $m\in M$
being $a_\lambda(m)$, the $\lambda$-action of $R$ in $M$.
According to the above observations,
such a structure corresponds, bijectively,
to an element $X$ of the following set:
\begin{equation}\label{100421:eq7}
\begin{array}{r}
\big\{X\in W^\partial_1(\Pi R\oplus\Pi M)_{\bar1}\,\big|\,
[X,X]=0\,,X_{\lambda,\mu}(R,R)\subset \mb F_-[\lambda,\mu]\otimes_{\mb F[\partial]}R,\\
X_{\lambda,\mu}(R,M)\subset \mb F_-[\lambda,\mu]\otimes_{\mb F[\partial]}M,\,X_{\lambda,\mu}(M,M)=0
\big\}
\end{array}
\,.
\end{equation}
Explicitly, to $X$ in \eqref{100421:eq7}
we associate the corresponding $\lambda$-bracket on $R$ given by \eqref{100421:eq5},
and the corresponding $R$-module structure on $M$ given by
\begin{equation}\label{100421:eq8}
a_\lambda(m) = (-1)^{p(a)}X_{\lambda,-\lambda-\partial}(a,m)
\,\,,\,\,\,\,
a\in R,\,m\in M\,.
\end{equation}

Note that every element $X$ in the set \eqref{100421:eq7}
satisfies conditions (i) and (ii) in Proposition \ref{100421:prop}.
Hence $\ad X$ induces a well-defined endomorphism $d_X$ of $W^\partial(\Pi R,\Pi M)$
such that $d_X^2=0$,
thus making $(W^\partial(\Pi R,\Pi M),d_X)$ a cohomology complex.
The explicit formula for the differential $d_X$ follows from equations \eqref{100419:eq6}
and from the identifications \eqref{100421:eq5} and \eqref{100421:eq8}.
For $Y\in W^\partial_{k-1}(\Pi R,\Pi M)$, we have
\begin{equation}\label{100421:eq9}
\begin{array}{c}
\displaystyle{
(d_X Y)_{\lambda_0,\dots,\lambda_k}(a_0,\dots,a_k)
=
\sum_{i=0}^k(-1)^{\alpha_i}
{a_i}_{\lambda_i}\big(Y_{\lambda_0,\stackrel{i}{\check{\dots}},\lambda_k}(a_0,\stackrel{i}{\check{\dots}},a_k)\big)
}\\
\displaystyle{
+ \sum_{0\leq i<j\leq k}(-1)^{\alpha_{ij}}
Y_{\lambda_i+\lambda_j,\lambda_0,\stackrel{i}{\check{\dots}}\stackrel{j}{\check{\dots}},\lambda_k}
([{a_i}_{\lambda_i}a_j],a_0,\stackrel{i}{\check{\dots}}\,\stackrel{j}{\check{\dots}},a_k)\,,
}
\end{array}
\end{equation}
where $\alpha_i$ and $\alpha_{ij}$ are defined in \eqref{100421:eq10}.
Note that, in the special case when $M=R$ is the adjoint representation,
the complex $(W^\partial(\Pi R,\Pi M),d_X)$ coincides with the complex $(W^\partial(\Pi R),d_X)$
discussed in Section \ref{sec:4.4}.

In the special case when $R$ is a (purely even) Lie conformal algebra and $M$ is a purely even $R$-module, we have
$\bar p(Y)\equiv k \mod 2$, and the above formula reduces to
\begin{equation}\label{100421:eq11}
\begin{array}{c}
\displaystyle{
(d_X Y)_{\lambda_0,\dots,\lambda_k}(a_0,\dots,a_k)
=(-1)^k\Big(
\sum_{i=0}^k(-1)^i
{a_i}_{\lambda_i}\big(Y_{\lambda_0,\stackrel{i}{\check{\dots}},\lambda_k}(a_0,\stackrel{i}{\check{\dots}},a_k)\big)
}\\
\displaystyle{
+ \sum_{0\leq i<j\leq k}(-1)^{i+j}
Y_{\lambda_i+\lambda_j,\lambda_0,\stackrel{i}{\check{\dots}}\stackrel{j}{\check{\dots}},\lambda_k}
([{a_i}_{\lambda_i}a_j],a_0,\stackrel{i}{\check{\dots}}\,\stackrel{j}{\check{\dots}},a_k)
\Big)\,,}
\end{array}
\end{equation}
which, up to the overall sign factor $(-1)^{k}$,
is the usual formula for the Lie conformal algebra cohomology differential
(see \cite{BKV}, \cite{BDAK} and \cite{DSK2}).
In conclusion, the cohomology complex $(C^\bullet(R,M)=\bigoplus_{k\in\mb Z_+}C^k(R,M),d)$
of a Lie conformal superalgebra $R$ with coefficients in an $R$-module $M$
can be defined by letting $C^k(R,M)=W^\partial_{k-1}(\Pi R,\Pi M)$ and $d=d_X$.

\begin{remark}\label{100516:rem}
One can replace $\mb F[\partial]$ by $U(\mf d)$, where $\mf d$ is a Lie algebra.
Then, following the same reasoning as above,
for any $\mf d$-module $R$ one constructs the $\mb Z$-graded Lie superalgebra $W^{\mf d}(\Pi R)$,
so that an odd element $X\in W^{\mf d}_1(\Pi R)$ such that $[X,X]=0$
defines a pseudoalgebra structure on $R$ and its cohomology complex, cf. \cite{BDAK}.
This, and its relation to the variational bicomplex,
will be discussed in a forthcoming publication.
\end{remark}


\section{The Lie superalgebra $W^{\partial,\as}(\mc V)$
for a commutative associative differential superalgebra $\mc V$
and PVA cohomology}
\label{sec:5}

Recall that a \emph{Poisson vertex superalgebra} (abbreviated PVA)
$\mc V$, with parity $p$,
is a unital commutative associative differential
superalgebra endowed with a $\lambda$-bracket $[\cdot\,_\lambda\,\cdot]$
which makes $\mc V$ a Lie conformal superalgebra,
satisfying the following Leibniz rule:
\begin{equation}\label{100426:eq1}
[a_\lambda bc] = [a_\lambda b]c + (-1)^{p(a)p(b)}b[a_\lambda c]\,.
\end{equation}
For example, if $R$ is a Lie conformal superalgebra,
then the symmetric algebra $S(R)$  has a natural structure of a Poisson vertex superalgebra,
with the $\lambda$-bracket on $R$ extended to $S(R)$ by the Leibniz rule \eqref{100426:eq1}.
In analogy with the notion of an odd Poisson superalgebra from Section \ref{sec:3},
we also introduce the notion of an \emph{odd Poisson vertex superalgebra}.
This is a unital commutative associative differential superalgebra $\mc V$, with parity $p$,
endowed with a $\lambda$-bracket
$[\cdot\,_\lambda\,\cdot]$ which makes $\Pi \mc V$ a Lie conformal superalgebra,
satisfying the following odd Leibniz rule:
\begin{equation}\label{100426:eq2}
[a_\lambda bc] = [a_\lambda b]c + (-1)^{(p(a)+\bar1)p(b)}b[a_\lambda c]\,.
\end{equation}
For example, if $R$ is a Lie conformal superalgebra,
then the symmetric algebra $S(\Pi R)$  has a natural structure of an odd Poisson vertex superalgebra,
with the $\lambda$-bracket on $R$ extended to $S(\Pi R)$ by the Leibniz rule \eqref{100426:eq2}.

\subsection{Poisson vertex superalgebra structures}\label{sec:5.1}

Throughout this section, we let $\mc V$ be a unital commutative associative differential superalgebra,
with a given even derivation $\partial$, and with parity denoted by $p$.
We let $\Der^\partial(\mc V)$ be the Lie superalgebra of derivations of $\mc V$ commuting with
$\partial$.

Consider the universal Lie superalgebra $W^\partial(\Pi \mc V)=\bigoplus_{k=-1}^\infty W^\partial_k(\Pi \mc V)$
associated to the $\mb F[\partial]$-module $\Pi \mc V$,
defined in Section \ref{sec:4.1}, with parity denoted by $\bar p$.
\begin{proposition}\label{100422c:prop1}
Let, for $k\geq-1$, $W^{\partial,\as}_k(\Pi \mc V)$ be the subspace of $W^{\partial}_k(\Pi \mc V)$
consisting of maps
$X:\,(\Pi \mc V)^{\otimes(k+1)}\to\mb F_-[\lambda_0,\dots,\lambda_k]\otimes_{\mb F[\partial]}\Pi \mc V$,
denoted by $a_0\otimes\dots\otimes a_k\mapsto X_{\lambda_0,\dots,\lambda_k}(a_0,\dots,a_k)$,
satisfying the following Leibniz rule (for $a_0,\dots,a_{k-1},b_i,c_i\in \mc V,\,i=0,\dots,k$):
\begin{equation}\label{100422c:eq1}
\begin{array}{l}
\vphantom{\Big(}
X_{\lambda_0,\dots,\lambda_k}(a_0,\dots ,b_ic_i,\dots,a_k) \\
=(-1)^{p(c_i)(s_{i+1,k}+k-i))}
X_{\lambda_0,\dots,\lambda_i+\partial,\dots,\lambda_k}
(a_0,\dots,b_i,\dots,a_k)_\to c_i \\
\vphantom{\Big(}
+(-1)^{p(b_i)(p(c_i)+s_{i+1,k}+k-i))}
X_{\lambda_0,\dots,\lambda_i+\partial,\dots,\lambda_k}
(a_0,\dots,c_i,\dots,a_k)_\to b_i\,,
\end{array}
\end{equation}
where $s_{ij}$ are defined in \eqref{100414:eq5}.
Then $W^{\partial,\as}(\Pi \mc V)=\bigoplus_{k=-1}^\infty W^{\partial,\as}_k(\Pi \mc V)$
is a subalgebra of the $\mb Z$-graded Lie superalgebra $W^\partial(\Pi \mc V)$
such that $W_{-1}^{\partial,\as}(\Pi \mc V)=\Pi \mc V/\partial\Pi \mc V$,
and $W_0^{\partial,\as}(\Pi \mc V)=\Der^\partial(\mc V)$ is the Lie superalgebra
of derivations of $\mc V$ commuting with $\partial$.
\end{proposition}
\begin{proof}
Clearly, by definition, $W_{-1}^{\partial,\as}(\Pi \mc V)=\Pi \mc V/\partial\Pi \mc V$.
Recall that, for $k=0$, we identify $\mb F[\lambda_0]\otimes_{\mb F[\partial]}\Pi \mc V=\Pi \mc V$,
and, via this identification,
$W^\partial_0(\Pi \mc V)=\End_{\mb F[\partial]}(\Pi \mc V)=\End_{\mb F[\partial]}(\mc V)$.
One easily checks that an element $X\in W^\partial_0(\Pi \mc V)$
satisfies \eqref{100422c:eq1} if and only if
the corresponding element in $\End_{\mb F[\partial]}(\mc V)$ is a derivation of $\mc V$
of parity $\bar p(X)$.
Hence, $W^{\partial,\as}_0(\Pi \mc V)=\Der^\partial(\mc V)$.
We are left to prove that for $X\in W^{\partial,\as}_h(\Pi \mc V)$
and $Y\in W^{\partial,\as}_{k-h}(\Pi \mc V)$,
their bracket $[X,Y]=X\Box Y-(-1)^{\bar p(X)\bar p(Y)} Y\Box X$ lies in $W^{\partial,\as}_k(\Pi \mc V)$,
namely it satisfies the Leibniz rule \eqref{100422c:eq1}.
Since $(X\Box Y)_{\lambda_0,\dots,\lambda_k}(a_0,\dots,a_k)$,
hence $[X,Y]_{\lambda_0,\dots,\lambda_k}(a_0,\dots,a_k)$,
is symmetric with respect to simultaneous permutations of the elements $a_0,\dots,a_k$
and the variables $\lambda_0,\dots,\lambda_k$,
it suffices to prove that $[X,Y]$ satisfies the Leibniz rule \eqref{100422c:eq1} for $i=0$.
In this case we have, by a straightforward computation using the definition \eqref{100418:eq2} of the box product,
$$
\begin{array}{l}
\displaystyle{
\vphantom{\Big(}
\big(X\Box Y\big)_{\lambda_0,\dots,\lambda_k}(bc,a_1, \dots, a_k)
} \\
\displaystyle{
\vphantom{\Big(}
= (-1)^{p(c)(\bar p(a_1)+\dots+\bar p(a_k))}
\big(X\Box Y\big)_{\lambda_0+\partial,\lambda_1,\dots,\lambda_k}(b,a_1, \dots, a_k)_\to c
} \\
\displaystyle{
\vphantom{\Big(}
+ (-1)^{p(b)p(c)+p(b)(\bar p(a_1)+\dots+\bar p(a_k))}
\big(X\Box Y\big)_{\lambda_0+\partial,\lambda_1,\dots,\lambda_k}(c,a_1, \dots, a_k)_\to b
}
\end{array}
$$
$$
\begin{array}{l}
\displaystyle{
+ \sum_{\substack{
i_1<\dots <i_h\\
i_{h+1}<\dots< i_k}}
\epsilon_a(i_1,\dots,i_k) (-1)^{p(b)p(c)+(p(c)+\bar p(Y))(p(b)+\bar p(a_{i_1})+\dots+\bar p(a_{i_h}))}
} \\
\displaystyle{
\vphantom{\Big(}
\,\,\,\,\,\,\,\,\,\,\,\,\,\,\,\,\,\,
\times X_{-\lambda_{i_1}-\dots-\lambda_{i_h}-\partial,\lambda_{i_1},\dots,\lambda_{i_h}}(b,a_{i_1},\dots,a_{i_h})
} \\
\displaystyle{
\vphantom{\Big(}
\,\,\,\,\,\,\,\,\,\,\,\,\,\,\,\,\,\,
\times Y_{-\lambda_{i_{h+1}}-\dots-\lambda_{i_k}-\partial,\lambda_{i_{h+1}},\dots,\lambda_{i_k}}
 (c,a_{i_{h+1}},\dots,a_{i_k})
} \\
\displaystyle{
+
\!\!\!\!\!\!\!\!\!
\sum_{\substack{
i_1<\dots< i_{k-h} \\
i_{k-h+1}<\dots <i_k
}}
\!\!\!\!\!\!\!\!\!
\epsilon_a(i_1,\dots,i_k)
(-1)^{\bar p(Y)\bar p(X)
+ p(b)p(c)
+ (p(c)+\bar p(X))(p(b)+\bar p(a_{i_1})+\dots+\bar p(a_{i_{k-h}}))}
} \\
\displaystyle{
\vphantom{\Big(}
\,\,\,\,\,\,\,\,\,\,\,\,\,\,\,\,\,\,
\times Y_{-\lambda_{i_1}-\dots-\lambda_{i_{k-h}}-\partial,\lambda_{i_1},\dots,\lambda_{i_{k-h}}}
 (b,a_{i_1},\dots,a_{i_{k-h}})
} \\
\displaystyle{
\vphantom{\Big(}
\,\,\,\,\,\,\,\,\,\,\,\,\,\,\,\,\,\,
\times X_{-\lambda_{i_{k-h+1}}-\dots-\lambda_{i_k}-\partial,\lambda_{i_{k-h+1}},\dots,\lambda_{i_k}}
 (c,a_{i_{k-h+1}},\dots,a_{i_k})
\,.
}
\end{array}
$$
To complete the proof, we just observe that, if we exchange $X$ and $Y$,
the two sums in the RHS get multiplied by $(-1)^{\bar p(X)\bar p(Y)}$,
hence they do not contribute to $[X,Y]$.
\end{proof}

%

\begin{proposition}\label{100422c:prop2}
The Poisson vertex superalgebra structures on $\mc V$
are in bijective correspondence, via \eqref{100421:eq5},
with the set
\begin{equation}\label{100422c:eq3}
\big\{X\in W^{\partial,\as}_1(\Pi \mc V)_{\bar1}\,\big|\,[X,X]=0\big\}\,.
\end{equation}
\end{proposition}
\begin{proof}
By the results in Section \ref{sec:4.4}
the elements $X\in W_1(\Pi \mc V)_{\bar1}$ such that $[X,X]=0$
correspond, via \eqref{100421:eq5}, to the Lie conformal superalgebra structures on $\mc V$.
Moreover, to say that $X$ lies in $W_1^{\partial,\as}(\Pi \mc V)$ means that the corresponding
$\lambda$-bracket satisfies the Leibniz rule, hence $\mc V$ is a Poisson vertex superalgebra.
\end{proof}

It follows from the above Proposition that, for any Poisson vertex superalgebra $\mc V$,
we have a differential $d_X=\ad X$ on the superspace $W^{\partial,\as}(\Pi \mc V)$,
where $X$ in \eqref{100422c:eq3}
is associated to the Lie conformal superalgebra structure on $\mc V$.
This differential is obviously  an odd derivation of the Lie bracket on $W^{\partial,\as}(\Pi \mc V)$.
We thus get a cohomology complex $(W^{\partial,\as}(\Pi \mc V),d_X)$.

%


\subsection{Odd Poisson vertex superalgebra structures}\label{sec:5.2}

As in the previous Section, let $\mc V$ be a commutative associative differential superalgebra,
with derivation $\partial$ and parity $p$,
and let $\Der^\partial(\mc V)$ be the Lie superalgebra of derivations of $\mc V$ commuting
with $\partial$.
As in Section \ref{sec:3.2}, we consider here the picture ``dual'' to that discussed
in the previous section.
\begin{proposition}\label{100422d:prop1}
Let, for $k\geq-1$, $W^{\partial,\as}_k(\mc V)$ be the superspace of
elements $X\in W^{\partial}_k(\mc V)$
satisfying the Leibniz rule (for $a_0,\dots,a_{k-1},b_i,c_i\in \mc V,\,i=0,\dots,k$):
\begin{equation}\label{100422d:eq1}
\begin{array}{l}
\vphantom{\Big(}
X_{\lambda_0,\dots,\lambda_k}(a_0,\dots ,b_ic_i,\dots,a_k) \\
=(-1)^{p(c_i)(s_{i+1,k}))}
X_{\lambda_0,\dots,\lambda_i+\partial,\dots,\lambda_k}
(a_0,\dots,b_i,\dots,a_k)_\to c_i \\
\vphantom{\Big(}
+(-1)^{p(b_i)(p(c_i)+s_{i+1,k}))}
X_{\lambda_0,\dots,\lambda_i+\partial,\dots,\lambda_k}
(a_0,\dots,c_i,\dots,a_k)_\to b_i\,,
\end{array}
\end{equation}
where $s_{ij}$ is defined in \eqref{100414:eq5}.
Then $W^{\partial,\as}(\mc V)=\bigoplus_{k=-1}^\infty W^{\partial,\as}_k(\mc V)$
is a subalgebra of the $\mb Z$-graded Lie superalgebra $W^\partial(\mc V)$
such that $W_{-1}^{\partial,\as}(\mc V)=\mc V/\partial \mc V$,
and $W_0^{\partial,\as}(\mc V)=\Der^\partial(\mc V)$.
\end{proposition}
\begin{proof}
It is the same as for Proposition \ref{100422c:prop1}.
\end{proof}

\begin{proposition}\label{100422d:prop2}
The odd Poisson vertex superalgebra structures on $\mc V$
are in bijective correspondence, via \eqref{100421:eq5},
with the set
\begin{equation}\label{100422d:eq3}
\big\{X\in W^{\partial,\as}_1(\mc V)_{\bar1}\,\big|\,[X,X]=0\big\}\,.
\end{equation}
\end{proposition}
\begin{proof}
It is the same as for Proposition \ref{100422c:prop2}.
\end{proof}

It follows from the above Proposition that, for any odd Poisson vertex superalgebra $\mc V$,
we have a differential $d_X=\ad X$ on the superspace $W^{\partial,\as}(\mc V)$,
where $X$ in \eqref{100422d:eq3}
is associated to the Lie conformal superalgebra structure on $\Pi\mc V$.
This differential is obviously  an odd derivation of the Lie bracket on $W^{\partial,\as}(\mc V)$.
We thus get a cohomology complex $(W^{\partial,\as}(\mc V),d_X)$.



\section{The universal Lie conformal superalgebra $\tilde W^{\partial}(V)$
for an $\mb F[\partial]$-module $V$,
and the basic Lie conformal superalgebra cohomology complexes}
\label{sec:6}

In this section we study the universal Lie conformal superalgebra
$\tilde W^{\partial}(V)$ associated to a finitely generated $\mb F[\partial]$-module $V$.
%

\subsection{The universal Lie conformal superalgebra $\tilde W^{\partial}(V)$}
\label{sec:6.1}

As in Section \ref{sec:4}, let $V$ be a vector superspace with parity $\bar p$,
endowed with a structure of an $\mb F[\partial]$-module, compatible with the parity.
Assume, moreover, that $V$ is finitely generated over $\mb F[\partial]$.

The Lie superalgebra $W^\partial(V)$ does not have the universality property
similar to that of $W(V)$, described in Remark \ref{100412:rem1}.
In this section we construct the \emph{universal} $\mb Z$-\emph{graded Lie conformal superalgebra}
$\tilde W^{\partial}(V)=\bigoplus_{k=-1}^\infty \tilde W^{\partial}_k(V)$
associated to the finitely generated $\mb F[\partial]$-module $V$ as follows.

For $k\geq-1$, we let, cf. \eqref{100516:eq1},
$$
\tilde W^{\partial}_k(V)
=
\Hom{}^{\sym}_{\mb F[\partial]^{\otimes (k+1)}}(V^{\otimes(k+1)},
\mb F_-[\lambda_0,\dots,\lambda_k]\otimes V)\,,
$$
with its natural superspace structure, denoted by $\bar p$.
For example, we have $\tilde W^{\partial}_{-1}(V)=V$,
$\tilde W^{\partial}_0(V)=\RCend(V)$, the space of right conformal endomorphisms of $V$,
namely the linear maps $X_\lambda:\, V\to \mb F[\lambda]\otimes V$,
such that
\begin{equation}\label{100528:eq1}
X_\lambda(\partial v)=-\lambda X_\lambda(v)\,,
\end{equation}
and, for $k\geq1$, $\tilde W^\partial_k(V)$ consists of linear maps
$X:\,V^{\otimes{k+1}}\to\mb F_-[\lambda_0,\dots,\lambda_k]\otimes V$
satisfying sesquilinearity in each argument and the skewsymmetry condition.

The superspace $\tilde W^{\partial}_k(V)$ has a natural $\mb F[\partial]$-module structure given by
\begin{equation}\label{100517:eq3}
(\partial X)_{\lambda_0,\dots,\lambda_k}(v_0,\dots,v_k)
=
(\lambda_0+\dots+\lambda_k+\partial) X_{\lambda_0,\dots,\lambda_k}(v_0,\dots,v_k)\,.
\end{equation}
Next we endow $\tilde W^{\partial}(V)$ with a structure
of a $\mb Z$-graded Lie conformal superalgebra as follows.
If $X\in \tilde W^{\partial}_h(V)$, $Y\in\tilde W^{\partial}_{k-h}(V)$, with $h\geq-1,\,k\geq h-1$,
we define $X\Box_\lambda Y$ to be the following element of $\mb F[\lambda]\otimes\tilde W^{\partial}_{k}(V)$:
\begin{equation}\label{100420:eq1}
\begin{array}{l}
\displaystyle{
\big(X\Box_\lambda Y\big)_{\lambda_0,\dots,\lambda_k}(v_0, \dots, v_k)
= \sum_{\substack{
i_0<\dots <i_{k-h}\\
i_{k-h+1}<\dots< i_k}}
\epsilon_v(i_0,\dots,i_k)
} \\
\displaystyle{
\times X_{-\lambda-\lambda_{i_{k-h+1}}-\dots-\lambda_{i_k}-\partial,
\lambda_{i_{k-h+1}},\dots,\lambda_{i_k}}
(Y_{\lambda_{i_0},\dots,\lambda_{i_{k-h}}}(v_{i_0},\dots, v_{i_{k-h}}),
}\\
\displaystyle{
\,\,\,\,\,\,\,\,\,\,\,\,\,\,\,\,\,\,\,\,\,\,\,\,\,\,\,\,\,\,\,\,\,\,\,\,\,\,\,\,\,\,\,\,\,
\,\,\,\,\,\,\,\,\,\,\,\,\,\,\,\,\,\,\,\,\,\,\,\,\,\,\,\,\,\,\,\,\,\,\,\,\,\,\,\,\,\,\,\,\,
\,\,\,\,\,\,\,\,\,\,\,\,\,\,\,\,\,\,\,\,\,\,\,\,\,\,\,\,\,\,\,\,\,\,\,
v_{i_{k-h+1}},\dots, v_{i_k})\,,
}
\end{array}
\end{equation}
where $\partial$ is moved to the left.
Since, by assumption, $V$ is finitely generated over $\mb F[\partial]$,
and since elements of $\tilde W^{\partial}(V)$ are determined by theirs values 
on a set of generators of $V$,
$X\Box_\lambda Y$ is indeed a polynomial in $\lambda$ with coefficients 
in $\tilde W^{\partial}(V)$.

\begin{lemma}\label{100516:lem}
\begin{enumerate}[(a)]
\item
The $\lambda$-product
$\Box_\lambda$ given by \eqref{100420:eq1} gives a well-defined map
$\tilde W^{\partial}_h(V)\times\tilde W^{\partial}_{k-h}(V)\to \mb F[\lambda]\otimes\tilde W^{\partial}_k(V)$,
and it makes $\tilde W^{\partial}(V)$ a conformal superalgebra.
\item
The $\lambda$-product $\Box_\lambda$ is right symmetric, i.e.,
defining the associator of $X,Y,Z\in\tilde W^{\partial}(V)$ as
$$
(X_\lambda Y_\mu Z)=(X\Box_\lambda Y)\Box_{\lambda+\mu} Z-X\Box_\lambda(Y\Box_\mu Z)\,,
$$
we have the following symmetry relation:
\begin{equation}\label{eq:rightassoc}
(X_\lambda Y_\mu Z)=(-1)^{\bar p(Y)\bar p(Z)}(X_\lambda Z_{-\lambda-\mu-\partial}Y)\,,
\end{equation}
where, as usual, $\partial$ is moved to the left.
\end{enumerate}
\end{lemma}
\begin{proof}
First, we need to prove that, for
$X\in\tilde W^{\partial}_h(V)$ and $Y\in\tilde W^{\partial}_{k-h}(V)$,
the $\lambda$-product $X\Box_\lambda Y$ defined by \eqref{100420:eq1}
is a polynomial in $\lambda$ with coefficients in $\tilde W^{\partial}_k(V)$,
i.e. it satisfies the sesquilinearity and symmetry conditions:
$$
\begin{array}{l}
(X\Box_\lambda Y)_{\lambda_0,\dots,\lambda_k}(v_0,\dots,\partial v_i,\dots,v_k)
=-\lambda_i (X\Box_\lambda Y)_{\lambda_0,\dots,\lambda_k}(v_0,\dots,v_k)\,, \\
(X\Box_\lambda Y)_{\lambda_0,\dots,\lambda_k}(v_0,\dots,v_k)
=
\epsilon_v(i_0,\dots,i_k)
X_{\lambda_{i_0},\dots,\lambda_{i_k}}(v_{i_0},\dots,v_{i_k})\,,
\end{array}
$$
for all permutations $(i_0,\dots,i_k)$ of $(0,\dots,k)$.
Both these relations
follow immediately from the definition \eqref{100420:eq1} of the $\lambda$-product $\Box_\lambda$
and by the sesquilinearity and symmetry conditions  on $X$ and $Y$.
To complete the proof of part (a) we need to check that the $\lambda$-product $\Box_\lambda$ is sesquilinear,
making $\tilde W^{\partial}(V)$ a conformal superalgebra.
The first sesquilinearity condition is straightforward,
using the definition \eqref{100517:eq3}
of the $\mb F[\partial]$-module structure of $\tilde W^{\partial}(V)$.
For the second sesquilinearity condition, we have
$$
\begin{array}{l}
\displaystyle{
\big(X\Box_\lambda(\partial Y)\big)_{\lambda_0,\dots,\lambda_k}(v_0, \dots, v_k)
= \sum_{\substack{
i_0<\dots <i_{k-h}\\
i_{k-h+1}<\dots< i_k}}
\epsilon_v(i_0,\dots,i_k)
} \\
\displaystyle{
\,\,\,\,\,\,\,\,\,
\times X_{-\lambda-\lambda_{i_{k-h+1}}-\dots-\lambda_{i_k}-\partial,
\lambda_{i_{k-h+1}},\dots,\lambda_{i_k}}
((\partial Y)_{\lambda_{i_0},\dots,\lambda_{i_{k-h}}}(v_{i_0},\dots, v_{i_{k-h}}),
}\\
\,\,\,\,\,\,\,\,\,\,\,\,\,\,\,\,\,\,\,\,\,\,\,\,\,\,\,\,\,\,\,\,\,\,\,\,\,\,\,\,\,\,\,\,\,
\,\,\,\,\,\,\,\,\,\,\,\,\,\,\,\,\,\,\,\,\,\,\,\,\,\,\,\,\,\,\,\,\,\,\,\,\,\,\,\,\,\,\,\,\,
\,\,\,\,\,\,\,\,\,\,\,\,\,\,\,\,\,\,\,\,\,\,\,\,\,\,\,\,\,\,\,\,\,\,\,
v_{i_{k-h+1}},\dots, v_{i_k}) \\
\displaystyle{
= \sum_{\substack{
i_0<\dots <i_{k-h}\\
i_{k-h+1}<\dots< i_k}}
\epsilon_v(i_0,\dots,i_k)
X_{-\lambda-\lambda_{i_{k-h+1}}-\dots-\lambda_{i_k}-\partial,
\lambda_{i_{k-h+1}},\dots,\lambda_{i_k}}
} \\
\,\,\,\,\,\,\,\,\,\,\,\,\,\,\,\,\,\,\,
\big(
(\lambda_{i_0}+\dots+\lambda_{i_{k-h}}+\partial) Y_{\lambda_{i_0},\dots,\lambda_{i_{k-h}}}
(v_{i_0},\dots, v_{i_{k-h}}),
v_{i_{k-h+1}},\dots, v_{i_k}
\big) \\
\displaystyle{
=
(\lambda+\lambda_0+\dots+\lambda_k+\partial)
\sum_{\substack{
i_0<\dots <i_{k-h}\\
i_{k-h+1}<\dots< i_k}}
\epsilon_v(i_0,\dots,i_k)
}\\
\,\,\,\,\,\,\,\,\,
\times X_{-\lambda-\lambda_{i_{k-h+1}}-\dots-\lambda_{i_k}-\partial,
\lambda_{i_{k-h+1}},\dots,\lambda_{i_k}}
\big(
Y_{\lambda_{i_0},\dots,\lambda_{i_{k-h}}}
(v_{i_0},\dots, v_{i_{k-h}}),
\\
\,\,\,\,\,\,\,\,\,\,\,\,\,\,\,\,\,\,\,\,\,\,\,\,\,\,\,\,\,\,\,\,\,\,\,\,\,\,\,\,\,\,\,\,\,
\,\,\,\,\,\,\,\,\,\,\,\,\,\,\,\,\,\,\,\,\,\,\,\,\,\,\,\,\,\,\,\,\,\,\,\,\,\,\,\,\,\,\,\,\,
\,\,\,\,\,\,\,\,\,\,\,\,\,\,\,\,\,\,\,\,\,\,\,\,\,\,\,\,\,\,\,\,\,\,\,
v_{i_{k-h+1}},\dots, v_{i_k} \big) \\
=
\big((\lambda+\partial)(X\Box_\lambda Y)\big)_{\lambda_0,\dots,\lambda_k}(v_0, \dots, v_k) \,.
\end{array}
$$

For part (b), let
$X\in\tilde W^{\partial}_\alpha(V)$, $Y\in\tilde W^{\partial}_\beta(V)$
and $Z\in\tilde W^{\partial}_\gamma(V)$.
We have, with a straightforward computation using the definition \eqref{100420:eq1}
of the $\lambda$-product $\Box_\lambda$,
\begin{equation}\label{101011:eq1}
\begin{array}{l}
\big((X\Box_\lambda Y)\Box_{\lambda+\mu}Z\big)_{\lambda_0,\dots,\lambda_{\alpha+\beta+\gamma}}
(v_0, \dots, v_{\alpha+\beta+\gamma})
\\
\displaystyle{
= \sum_{\substack{
i_0<\dots <i_\gamma \\
i_{\gamma+1}<\dots< i_{\beta+\gamma} \\
i_{\beta+\gamma+1}<\dots< i_{\alpha+\beta+\gamma}
}}
\epsilon_v(i_0,\dots,i_{\alpha+\beta+\gamma})
} \\
\displaystyle{
\,\,\,\,\,\,\,\,\,
\times X_{-\lambda-\lambda_{i_{\beta+\gamma+1}}-\dots-\lambda_{i_{\alpha+\beta+\gamma}}-\partial,
\lambda_{i_{\beta+\gamma+1}},\dots,\lambda_{i_{\alpha+\beta+\gamma}}}
}\\
\,\,\,\,\,\,\,\,\,\,\,\,\,\,
\Big(Y_{-\mu-\lambda_{i_{\gamma+1}}-\dots-\lambda_{i_{\beta+\gamma}}-\partial,
\lambda_{i_{\gamma+1}},\dots,\lambda_{i_{\beta+\gamma}}}
\\
\,\,\,\,\,\,\,\,\,\,\,\,\,\,\,\,\,\,\,\,\,\,\,\,\,\,
\big(Z_{\lambda_{i_0},\dots,\lambda_{i_\gamma}}
(v_{i_0},\dots, v_{i_\gamma}),
v_{i_{\gamma+1}},\dots, v_{i_{\beta+\gamma}} \big),
v_{i_{\beta+\gamma+1}},\dots, v_{i_{\alpha+\beta+\gamma}} \Big)
\\
\displaystyle{
+ \sum_{\substack{
i_0<\dots <i_\gamma \\
i_{\gamma+1}<\dots< i_{\beta+\gamma+1} \\
i_{\beta+\gamma+2}<\dots< i_{\alpha+\beta+\gamma}
}}
(-1)^{\bar p(Y)(\bar p(Z)+\bar p(v_{i_0})+\dots+\bar p(v_{i_\gamma})}
\epsilon_v(i_0,\dots,i_{\alpha+\beta+\gamma})
} \\
\displaystyle{
\,\,\,\,\,\,\,\,\,
\times X_{
-\lambda-\mu-\lambda_{i_{\gamma+1}}-\dots-\lambda_{i_{\alpha+\beta+\gamma}}-\partial,
\mu+\lambda_{i_{\gamma+1}}+\dots+\lambda_{i_{\beta+\gamma+1}},
\lambda_{i_{\beta+\gamma+2}},\dots,\lambda_{i_{\alpha+\beta+\gamma}}}
}\\
\,\,\,\,\,\,\,\,\,\,\,\,\,\,\,\,\,\,\,\,\,\,\,
\Big(
Z_{\lambda_{i_0},\dots,\lambda_{i_\gamma}}(v_{i_0},\dots, v_{i_\gamma}),
Y_{\lambda_{i_{\gamma+1}},\dots,\lambda_{i_{\beta+\gamma+1}}}
(v_{i_{\gamma+1}},\dots, v_{i_{\beta+\gamma+1}}),
\\
\,\,\,\,\,\,\,\,\,\,\,\,\,\,\,\,\,\,\,\,\,\,\,\,\,\,\,\,\,\,\,\,\,\,\,\,\,\,\,\,\,\,\,\,\,
\,\,\,\,\,\,\,\,\,\,\,\,\,\,\,\,\,\,\,\,\,\,\,\,\,\,\,\,\,\,\,\,\,\,\,\,\,\,\,\,\,\,\,\,\,
\,\,\,\,\,\,\,\,\,\,\,\,\,\,\,\,\,\,
v_{i_{\beta+\gamma+2}},\dots, v_{i_{\alpha+\beta+\gamma}} \Big)\,.
\end{array}
\end{equation}
Similarly, we have
\begin{equation}\label{101011:eq2}
\begin{array}{l}
\big(X\Box_\lambda (Y\Box_\mu Z)\big)_{\lambda_0,\dots,\lambda_{\alpha+\beta+\gamma}}
(v_0, \dots, v_{\alpha+\beta+\gamma})
\\
\displaystyle{
= \sum_{\substack{
i_0<\dots <i_\gamma \\
i_{\gamma+1}<\dots< i_{\beta+\gamma} \\
i_{\beta+\gamma+1}<\dots< i_{\alpha+\beta+\gamma}
}}
\epsilon_v(i_0,\dots,i_{\alpha+\beta+\gamma})
} \\
\displaystyle{
\,\,\,\,\,\,\,\,\,
\times X_{-\lambda-\lambda_{i_{\beta+\gamma+1}}-\dots-\lambda_{i_{\alpha+\beta+\gamma}}-\partial,
\lambda_{i_{\beta+\gamma+1}},\dots,\lambda_{i_{\alpha+\beta+\gamma}}}
}\\
\,\,\,\,\,\,\,\,\,\,\,\,\,\,
\Big(Y_{-\mu-\lambda_{i_{\gamma+1}}-\dots-\lambda_{i_{\beta+\gamma}}-\partial,
\lambda_{i_{\gamma+1}},\dots,\lambda_{i_{\beta+\gamma}}}
\\
\,\,\,\,\,\,\,\,\,\,\,\,\,\,\,\,\,\,\,\,\,\,
\big(Z_{\lambda_{i_0},\dots,\lambda_{i_\gamma}}
(v_{i_0},\dots, v_{i_\gamma}),
v_{i_{\gamma+1}},\dots, v_{i_{\beta+\gamma}} \big),
v_{i_{\beta+\gamma+1}},\dots, v_{i_{\alpha+\beta+\gamma}} \Big)\,.
\end{array}
\end{equation}
We then observe that the RHS in \eqref{101011:eq2}
is equal to the first term in the LHS of \eqref{101011:eq1}.
Hence, combining the above equations, we get
\begin{equation}\label{101011:eq3}
\begin{array}{l}
\big(X_\lambda Y_\mu Z\big)_{\lambda_0,\dots,\lambda_{\alpha+\beta+\gamma}}
(v_0, \dots, v_{\alpha+\beta+\gamma})
\\
\displaystyle{
= \sum_{\substack{
i_0<\dots <i_\gamma \\
i_{\gamma+1}<\dots< i_{\beta+\gamma+1} \\
i_{\beta+\gamma+2}<\dots< i_{\alpha+\beta+\gamma}
}}
(-1)^{\bar p(Y)(\bar p(Z)+\bar p(v_{i_0})+\dots+\bar p(v_{i_\gamma})}
\epsilon_v(i_0,\dots,i_{\alpha+\beta+\gamma})
} \\
\displaystyle{
\,\,\,\,\,\,\,\,\,
\times X_{
-\lambda-\mu-\lambda_{i_{\gamma+1}}-\dots-\lambda_{i_{\alpha+\beta+\gamma}}-\partial,
\mu+\lambda_{i_{\gamma+1}}+\dots+\lambda_{i_{\beta+\gamma+1}},
\lambda_{i_{\beta+\gamma+2}},\dots,\lambda_{i_{\alpha+\beta+\gamma}}}
}\\
\,\,\,\,\,\,\,\,\,\,\,\,\,\,\,\,\,\,\,\,\,\,\,
\Big(
Z_{\lambda_{i_0},\dots,\lambda_{i_\gamma}}(v_{i_0},\dots, v_{i_\gamma}),
Y_{\lambda_{i_{\gamma+1}},\dots,\lambda_{i_{\beta+\gamma+1}}}
(v_{i_{\gamma+1}},\dots, v_{i_{\beta+\gamma+1}}),
\\
\,\,\,\,\,\,\,\,\,\,\,\,\,\,\,\,\,\,\,\,\,\,\,\,\,\,\,\,\,\,\,\,\,\,\,\,\,\,\,\,\,\,\,\,\,
\,\,\,\,\,\,\,\,\,\,\,\,\,\,\,\,\,\,\,\,\,\,\,\,\,\,\,\,\,\,\,\,\,\,\,\,\,\,\,\,\,\,\,\,\,
\,\,\,\,\,\,\,\,\,\,\,\,\,\,\,\,\,\,
v_{i_{\beta+\gamma+2}},\dots, v_{i_{\alpha+\beta+\gamma}} \Big)\,.
\end{array}
\end{equation}
To conclude, we observe that,
if we exchange $Y$ and $Z$  (and $\beta$ and $\gamma$),
and we replace $\mu$
by $-\lambda-\mu-\lambda_0-\dots-\lambda_{\alpha+\beta+\gamma}-\partial$,
the RHS of \eqref{101011:eq3}
gets multiplied by the factor $(-1)^{\bar p(Y)\bar p(Z)}$.
\end{proof}
\begin{lemma}\label{100516:lemb}
If $R$ is a conformal superalgebra with right symmetric $\lambda$-product
$a_\lambda b,\,a,b\in R$, and parity $p$,
then the $\lambda$-bracket
\begin{equation}\label{box4}
[a_\lambda b]=a_\lambda b-(-1)^{p(a)p(b)}b_{-\lambda-\partial} a\,,
\end{equation}
defines on $R$ a structure of a Lie conformal superalgebra.
\end{lemma}
\begin{proof}
Recall that right symmetry means the following identity
$(a_\lambda b_\mu c)=(-1)^{p(b)p(c)}(a_\lambda c_{-\lambda-\mu-\partial} b)$,
where $(a_\lambda b_\mu c)=(a_\lambda b)_{\lambda+\mu}c-a_\lambda(b_\mu c)$
is the associator.
The statement follows by the following identity, which is easily derived
from \eqref{box4}:
$$
\begin{array}{l}
[a_\lambda[b_\mu c]]-(-1)^{p(a)p(b)}[b_\mu [a_\lambda c]]
-[[a_\lambda b]_{\lambda+\mu}c] \\
=
-\Big(
(a_\lambda b_\mu c)-(-1)^{p(b)p(c)}(a_\lambda c_{-\lambda-\mu-\partial}b)
\Big) \\
+(-1)^{p(a)p(b)}\Big(
(b_\mu a_\lambda c)-(-1)^{p(a)p(c)}(b_\mu c_{-\lambda-\mu-\partial}a)
\Big) \\
-(-1)^{p(c)(p(a)+p(b))}\Big(
(c_{-\lambda-\mu-\partial} a_\lambda b)
-(-1)^{p(a)p(b)} (c_{-\lambda-\mu-\partial} b_\mu a)
\Big)\,.
\end{array}
$$
\end{proof}
\begin{corollary}\label{100516:prop}
The $\lambda$-bracket
\begin{equation}\label{box3}
[X_\lambda Y]=X\Box_\lambda Y-(-1)^{\bar p(X)\bar p(Y)}Y\Box_{-\lambda-\partial} X\,,
\end{equation}
defines a structure of a Lie conformal superalgebra on $\tilde W^{\partial}(V)$.
\end{corollary}
\begin{proof}
It follows immediately from Lemmas \ref{100516:lem} and \ref{100516:lemb}.
\end{proof}

For $X\in\tilde W^{\partial}_{k}(V)$ and $Y\in V=\tilde W^{\partial}_{-1}(V)$, we have,
for $v_1,\dots,v_k\in V$,
$$
[X_{\lambda_0} Y]_{\lambda_1,\dots,\lambda_k}(v_1,\dots,v_k)
=X_{-\lambda_0-\dots-\lambda_k-\partial,\lambda_1,\dots,\lambda_k}(Y,v_1,\dots,v_k)
\,,
$$
or, equivalently,
\begin{equation}\label{101125:eq1}
\begin{array}{c}
X_{\lambda_0,\lambda_1,\dots,\lambda_k}(Y,v_1,\dots,v_k)
= [X_{-\lambda_0-\partial} Y]_{\lambda_1,\dots,\lambda_k}(v_1,\dots,v_k) \\
= (-1)^{1+\bar p(X)\bar p(Y)}[Y_{\lambda_0} X]_{\lambda_1,\dots,\lambda_k}(v_1,\dots,v_k)
\,.
\end{array}
\end{equation}
It follows that we have the following universality property
of the Lie conformal superalgebra $\tilde W^{\partial}(V)$:
for any $\mb Z$-graded Lie conformal superalgebra $R=\bigoplus_{k=-1}^\infty R_k$
with $R_{-1}=V$, there is a canonical homomorphism of $\mb Z$-graded Lie conformal superalgebras
$\phi:\,R\to\tilde W^{\partial}(V)$, extending the identity map on $V$ by
$$
\phi(a)_{\lambda_0,\dots,\lambda_k}(v_0,\dots,v_k)
=\pm [{v_k}_{\lambda_k}\dots[{v_1}_{\lambda_1}[{v_0}_{\lambda_0}a]]\dots]
\,\,,\,\,\,\,
\text{ if } k\geq0\,,
$$
where $\pm=(-1)^{k+1+\bar p(a)(\bar p(v_0)+\dots+\bar p(v_k))}\epsilon_v(k,k-1,\dots,0)$.

For a right conformal endomorphism
$X\in \tilde W^\partial_0(V)$
and for $Y\in\tilde W^\partial_k(V)$, where $k\geq-1$,
we have
\begin{equation}\label{100419x:eq5}
\begin{array}{l}
[X_\lambda Y]_{\lambda_0,\dots,\lambda_k}(v_0,\dots,v_k)=
X_{-\lambda-\partial}\big(Y_{\lambda_0,\dots,\lambda_k}(v_0,\dots,v_k)\big)
\\
\displaystyle{
-(-1)^{\bar p(X)\bar p(Y)} \sum_{i=0}^k (-1)^{\bar p(X)\bar s_{0,i-1}}
Y_{\lambda_0,\dots,\lambda+\lambda_i,\dots,\lambda_k}(v_0,\dots,X_{\lambda_i}(v_i),\dots,v_k)\,,
}
\end{array}
\end{equation}
where $\bar s_{ij}$ is as in \eqref{100412:eq5}.
In particular, for $k=0$, the above formula gives the following Lie conformal superalgebra structure
on the $\mb F[\partial]$-module $\tilde W^\partial_0(V)=\RCend(V)$
of all right conformal endomorphisms of $V$:
\begin{equation}\label{100517:eq1}
[X_\lambda Y]_{\mu}(v)=
X_{-\lambda-\partial}(Y_{\mu}(v))
-(-1)^{\bar p(X)\bar p(Y)} Y_{\lambda+\mu}(X_{\mu}(v))\,.
\end{equation}
\begin{remark}\label{20110528:rem}
This Lie conformal superalgebra, which is natural to denote $Rgc(V)$,
is isomorphic to the Lie conformal superalgebra $gc(V)$
of all (left) conformal endomorphisms \cite{K}, via the map
\begin{equation}\label{100528:eq3}
*:\,Rgc(V)\to gc(V)\,,
\quad\text{ where }\,\,
X^*_\lambda(v)= X_{-\lambda-\partial}(v)\,.
\end{equation}
\end{remark}
Furthermore, if $X\in\tilde W^\partial_1(V)$ and $Y\in\tilde W^\partial_{k-1}(V),\,k\geq0$, we have
\begin{equation}\label{100419x:eq6}
\begin{array}{l}
\displaystyle{
(X\Box_\lambda Y)_{\lambda_0,\dots,\lambda_k}(v_0,\dots,v_k)
} \\
\displaystyle{
\,\,\,\,\,\,\,\,\,
= \sum_{i=0}^k (-1)^{\bar p(v_i)\bar s_{i+1,k}}
X_{-\lambda-\lambda_i-\partial,\lambda_i}
\big(Y_{\lambda_0,\stackrel{i}{\check{\dots}},\lambda_k}(v_0,\stackrel{i}{\check{\dots}},v_k),v_i\big)\,,
} \\
\displaystyle{
(Y\Box_{-\lambda-\partial} X)_{\lambda_0,\dots,\lambda_k}(v_0,\dots,v_k)
=\!\!
\sum_{0\leq i<j\leq k} \!\!
(-1)^{\bar p(v_i)\bar s_{0,i-1}+\bar p(v_j)(\bar s_{0,i-1}+\bar s_{i+1,j-1})}
} \\
\displaystyle{
\,\,\,\,\,\,\,\,\,\,\,\,\,\,\,\,\,\,\,\,\,\,\,\,\,\,\,\,\,\,\,\,\,\,\,\,
\times Y_{\lambda+\lambda_i+\lambda_j,\lambda_0,\stackrel{i}{\check{\dots}}\stackrel{j}{\check{\dots}},\lambda_k}
\big(X_{\lambda_i,\lambda_j}(v_i,v_j),v_0,\stackrel{i}{\check{\dots}}\,\stackrel{j}{\check{\dots}},v_k\big)\,.
}
\end{array}
\end{equation}
In particular, if both $X$ and $Y$ are in $\tilde W^\partial_1(V)$,
we get
$$
\begin{array}{c}
(X\Box_\lambda Y)_{\lambda_0,\lambda_1,\lambda_2}(v_0,v_1,v_2) =
X_{-\lambda-\lambda_2-\partial,\lambda_2}\big(Y_{\lambda_0,\lambda_1}(v_0,v_1),v_2\big) \\
\vphantom{\Bigg(}
+(-1)^{\bar p(v_1)(\bar p(Y)+\bar p(v_0))}
X_{\lambda_1,-\lambda-\lambda_1-\partial}\big(v_1,Y_{\lambda_0,\lambda_2}(v_0,v_2)\big) \\
+(-1)^{\bar p(v_0)\bar p(Y)}
X_{\lambda_0,-\lambda-\lambda_0-\partial}\big(v_0,Y_{\lambda_1,\lambda_2}(v_1,v_2)\big) \,,
\end{array}
$$
and
$$
\begin{array}{c}
(Y\Box_{-\lambda-\partial} X)_{\lambda_0,\lambda_1,\lambda_2}(v_0,v_1,v_2) =
Y_{\lambda+\lambda_0+\lambda_1,\lambda_2}\big(X_{\lambda_0,\lambda_1}(v_0,v_1),v_2\big) \\
\vphantom{\Bigg(}
+(-1)^{\bar p(v_1)(\bar p(Y)+\bar p(v_0))}
Y_{\lambda_1,\lambda+\lambda_0+\lambda_2}\big(v_1,X_{\lambda_0,\lambda_2}(v_0,v_2)\big) \\
+(-1)^{\bar p(v_0)\bar p(Y)}
Y_{\lambda_0,\lambda+\lambda_1+\lambda_2}\big(v_0,X_{\lambda_1,\lambda_2}(v_1,v_2)\big) \,.
\end{array}
$$

There is a close relation between the universal Lie conformal superalgebra $\tilde W^\partial(V)$
and the Lie superalgebra $W^\partial(V)$ associated to the finitely generated 
$\mb F[\partial]$-module $V$.
In order to describe this connection, we consider the quotient map
$\tint:\,\mb F_-[\lambda_0,\dots,\lambda_k]\otimes V
\to \mb F_-[\lambda_0,\dots,\lambda_k]\otimes_{\mb F[\partial]} V$.
For $k=-1$, this coincides with the usual map $V\to V/\partial V,\,v\mapsto\tint v$.
\begin{proposition}\label{100517:prop}
\begin{enumerate}[(a)]
\item
We have a linear map
$\tint:\, \tilde W^{\partial}(V)\to W^\partial(V)$,
induced by the quotient map
$\tint:\,\mb F_-[\lambda_0,\dots,\lambda_k]\otimes
V\to \mb F_-[\lambda_0,\dots,\lambda_k]\otimes_{\mb F[\partial]} V$,
which induces an injective homomorphism of Lie superalgebras
$\tint:\, \tilde W^{\partial}(V)/\partial\tilde W^{\partial}(V) \to W^\partial(V)$.
\item
We have a representation of the Lie superalgebra $W^\partial(V)$ on $\tilde W^{\partial}(V)$,
with the action of $X\in W^\partial_h(V)$ on $\tilde Y\in\tilde W^\partial_{k-h}(V)$
denoted by $[X,\tilde Y]\in\tilde W^\partial_k(V)$, given by the following formula:
\begin{equation}\label{100517:eq2}
\begin{array}{l}
\displaystyle{
[X, \tilde Y]_{\lambda_0,\dots,\lambda_k}(v_0, \dots, v_k)
= \sum_{\substack{
i_0<\dots <i_{k-h}\\
i_{k-h+1}<\dots< i_k}}
\epsilon_v(i_0,\dots,i_k)
} \\
\displaystyle{
\vphantom{\Bigg(}
\times X_{-\lambda_{i_{k-h+1}}-\dots-\lambda_{i_{k}}-\partial,\lambda_{i_{k-h+1}},\dots,\lambda_{i_k}}
(\tilde Y_{\lambda_{i_0},\dots,\lambda_{i_{k-h}}}(v_{i_0},\dots, v_{i_{k-h}}),
} \\
\displaystyle{
v_{i_{k-h+1}},\dots, v_{i_k})
- (-1)^{\bar p(X)\bar p(\tilde Y)}
\sum_{\substack{
i_0<\dots <i_{h}\\
i_{h+1}<\dots< i_k}}
\epsilon_v(i_0,\dots,i_k)
} \\
\displaystyle{
\times \tilde Y_{\lambda_{i_0}+\dots+\lambda_{i_h},\lambda_{i_{h+1}},\dots,\lambda_{i_k}}
(X_{\lambda_{i_0},\dots,\lambda_{i_h}}(v_{i_0},\dots, v_{i_h}), v_{i_{h+1}},\dots, v_{i_k})\,.
}
\end{array}
\end{equation}
This action of the Lie superalgebra $W^\partial(V)$
is by derivations of the $\lambda$-bracket on $\tilde W^\partial(V)$
and it commutes with the action of $\partial$.
\item
The canonical map $\tint:\, \tilde W^{\partial}(V)\to W^\partial(V)$
from part (a) is a homomorphism of representations of the Lie superalgebra $W^\partial(V)$.
Moreover,
the representation of $W^\partial(V)$ on $\tilde W^{\partial}(V)$ is compatible,
via the map $\tint$ in (a), with the representation of the Lie superalgebra
$\tilde W^{\partial}(V)/\partial\tilde W^{\partial}(V)$ on $\tilde W^{\partial}(V)$.
\end{enumerate}
\end{proposition}
\begin{proof}
For $X\in \tilde W^\partial_k(V)$,
let $\tint X$ be the map
$V^{\otimes(k+1)}\to\mb F_-[\lambda_0,\dots,\lambda_k]\otimes_{\mb F[\partial]} V$,
given by
$(\tint X)_{\lambda_0,\dots,\lambda_k}(v_0,\dots,v_k)=\tint X_{\lambda_0,\dots,\lambda_k}(v_0,\dots,v_k)$,
where, in the RHS, $\tint$ denotes the map
$\mb F_-[\lambda_0,\dots,\lambda_k]\otimes V
\to\mb F_-[\lambda_0,\dots,\lambda_k]\otimes_{\mb F[\partial]} V$.
It is immediate to check that $\tint X$ lies in $W^\partial_k(V)$,
i.e., it satisfies the sesquilinearity and symmetry conditions.
Hence, we get a well-defined linear map
$\tint:\,\tilde W^\partial_k(V)\to W^\partial_k(V)$.
Next, we prove that
$\ker\big(\tint\,\big|_{\tilde W^{\partial}_k(V)}\big)=\partial\big(\tilde W^{\partial}_k(V)\big)$,
so that
$\tint$ factors through an injective linear map
$\tint:\,\tilde W^{\partial}_k(V)/\partial\tilde W^{\partial}_k(V) \to W^\partial_k(V)$.
For $k=-1$, $\tint$ coincides with the quotient map $V\to V/\partial\ V$,
so there is nothing to prove.
Let then $k\geq0$.
The inclusion
$\partial\big(\tilde W^{\partial}_k(V)\big)\subset \ker\big(\tint\,\big|_{\tilde W^{\partial}_k(V)}\big)$
is immediate by the definition \eqref{100517:eq3} of the $\mb F[\partial]$-module
structure on $\tilde W^\partial(V)$.
Conversely,
suppose $X\in\tilde W^\partial_k(V)$ lies in $\ker(\tint)$,
namely, for every $v_0,\dots,v_k\in V$, we have
$X_{\lambda_0,\dots,\lambda_k}(v_0,\dots,v_k)
=(\partial+\lambda_0+\dots+\lambda_k)Y_{\lambda_0,\dots,\lambda_k}(v_0,\dots,v_k)$,
for some polynomial
$Y_{\lambda_0,\dots,\lambda_k}(v_0,\dots,v_k)$ in $\lambda_0,\dots,\lambda_k$
with coefficients in $V$.
Since $\partial+\lambda_0+\dots+\lambda_k$ is injective on $\mb F[\lambda_0,\dots,\lambda_k]\otimes V$
for every $k\geq0$,
the sesquilinearity and symmetry relations for $X$ imply those for $Y$.
Hence, $X=\partial Y\in\partial\tilde W^\partial_k(V)$, proving the claim.
The fact that the induced map
$\tint:\,\tilde W^{\partial}(V)/\partial\tilde W^{\partial}(V) \to W^\partial(V)$
is a Lie algebra homomorphism follows by comparing the explicit expressions \eqref{box2} and \eqref{box3}
for the Lie bracket on $W^\partial(V)$
and the $\lambda$-bracket on $\tilde W^{\partial}(V)$ respectively.
This proves part (a).

It is immediate to check that formula \eqref{100517:eq2} does not depend on the choice
of representative of $X_{\lambda_{i_0},\dots,\lambda_{i_k}}(v_{i_0},\dots,v_{i_k})
\in \mb F[\lambda_{i_0},\dots,\lambda_{i_k}]\otimes_{\mb F[\partial]} V$
in $\mb F[\lambda_{i_0},\dots,\lambda_{i_k}]\otimes V$.
Moreover, if $X$ and $\tilde Y$ satisfy the sesquilinearity and symmetry relations, so does $[X,\tilde Y]$.
Hence, we get a well-defined map
$W^\partial_h(V)\times\tilde W^\partial_{k-h}(V)\to\tilde W^\partial_k(V)$.
We next prove that \eqref{100517:eq2} defines a representation
of the Lie superalgebra $W^\partial(V)$ on $\tilde W^\partial(V)$.
Introduce the left and right box products
$\Box^L:\,W^\partial(V)\times\tilde W^\partial(V)\to\tilde W^\partial(V)$
and
$\Box^R:\,\tilde W^\partial(V)\times W^\partial(V)\to\tilde W^\partial(V)$,
given, respectively, by the first and (without the sign in front) the second term
in the RHS of \eqref{100517:eq2}, i.e.
\begin{equation}\label{101024:eq1}
\begin{array}{c}
(X\Box^L \tilde Y)_{\!\lambda_0,\dots,\lambda_k}\!\!(v_0, \dots, v_k)\!
=
\!\!\!\!\!\!\!\!\!\!\!
\displaystyle{
\sum_{\substack{
i_0<\dots <i_{k-h}\\
i_{k-h+1}<\dots< i_k}}
\!\!\!\!\!\!\!\!\!\!\!
\epsilon_v(i_0,\dots,i_k)
X_{-\lambda_{i_{k-h+1}}\!\!\!\dots-\lambda_{i_{k}}-\partial,\lambda_{i_{k-h+1}},\dots,\lambda_{i_k}\!\!\!}
} \\
\displaystyle{
\vphantom{\Bigg(}
\Big(\tilde Y_{\lambda_{i_0},\dots,\lambda_{i_{k-h}}}(v_{i_0},\dots, v_{i_{k-h}}),
v_{i_{k-h+1}},\dots, v_{i_k}
\Big)
} \\
(\tilde Y\Box^R X)_{\lambda_0,\dots,\lambda_k}(v_0, \dots, v_k)
=
\!\!\!\displaystyle{
\sum_{\substack{
i_0<\dots <i_{h}\\
i_{h+1}<\dots< i_k}}
\epsilon_v(i_0,\dots,i_k)
\tilde Y_{\lambda_{i_0}+\dots+\lambda_{i_h},\lambda_{i_{h+1}},\dots,\lambda_{i_k}}
} \\
\displaystyle{
\Big(X_{\lambda_{i_0},\dots,\lambda_{i_h}}(v_{i_0},\dots, v_{i_h}), v_{i_{h+1}},\dots, v_{i_k}\Big)\,.
}
\end{array}
\end{equation}
We claim that they satisfy the following right symmetry identities:
\begin{equation}\label{101024:eq2}
(X\Box Y)\Box^L\tilde Z - X\Box^L(Y\Box^L\tilde Z)
=
(-1)^{\bar p(Y)\bar p(\tilde Z)}
\Big((X\Box^L\tilde Z)\Box^R Y - X\Box^L(\tilde Z\Box^R Y)\Big),\!
\end{equation}
for $X\in W^\partial_\alpha(V),\, Y\in W^\partial_\beta(V),\,\tilde Z\in\tilde W^\partial_\gamma(V)$,
and
\begin{equation}\label{101024:eq3}
(\tilde X\Box^R Y)\Box^R Z - \tilde X\Box^R(Y\Box Z)
=
(-1)^{\bar p(Y)\bar p(Z)}
\Big((\tilde X\Box^R Z)\Box^R Y - \tilde X\Box^R(Z\Box Y)\Big)\,,
\end{equation}
for $\tilde X\in\tilde W^\partial_\alpha(V),\, Y\in W^\partial_\beta(V),\,Z\in W^\partial_\gamma(V)$.
For \eqref{101024:eq2} we have
\begin{equation}\label{101024:eq4}
\begin{array}{l}
\big((X\Box Y)\Box^L \tilde Z\big)_{\lambda_0,\dots,\lambda_{\alpha+\beta+\gamma}}
(v_0, \dots, v_{\alpha+\beta+\gamma}) =
\!\!\!\!\!\!\!
\displaystyle{
\sum_{\substack{
i_0<\dots <i_\gamma \\
i_{\gamma+1}<\dots< i_{\beta+\gamma} \\
i_{\beta+\gamma+1}<\dots< i_{\alpha+\beta+\gamma}
}}
\!\!\!\!\!\!
\epsilon_v(i_0,\dots,i_{\alpha+\beta+\gamma})
} \\
\times X_{-\lambda_{i_{\beta+\gamma+1}}-\dots-\lambda_{i_{\alpha+\beta+\gamma}}-\partial,
\lambda_{i_{\beta+\gamma+1}},\dots,\lambda_{i_{\alpha+\beta+\gamma}}}
\Big(Y_{-\lambda_{i_{\gamma+1}}-\dots-\lambda_{i_{\beta+\gamma}}-\partial,
\lambda_{i_{\gamma+1}},\dots,\lambda_{i_{\beta+\gamma}}}
\\
\,\,\,\,\,\,\,\,\,\,\,\,\,\,\,\,\,\,\,\,\,\,\,\,\,\,\,\,\,\,\,\,
\vphantom{\Bigg(}
\big(
\tilde Z_{\lambda_{i_0},\dots,\lambda_{i_\gamma}}
(v_{i_0},\dots, v_{i_\gamma}),
v_{i_{\gamma+1}},\dots, v_{i_{\beta+\gamma}} \big),
v_{i_{\beta+\gamma+1}},\dots, v_{i_{\alpha+\beta+\gamma}} \Big)
\end{array}%
\end{equation}%
$$%
\begin{array}{l}%
\displaystyle{
+ \sum_{\substack{
i_0<\dots <i_\beta \\
i_{\beta+1}<\dots< i_{\beta+\gamma+1} \\
i_{\beta+\gamma+2}<\dots< i_{\alpha+\beta+\gamma}
}}
(-1)^{\bar p(\tilde Z)(\bar p(v_{i_0})+\dots+\bar p(v_{i_\beta})}
\epsilon_v(i_0,\dots,i_{\alpha+\beta+\gamma})
} \\
\displaystyle{
\times X_{
\lambda_{i_0}+\dots+\lambda_{i_\beta},
-\lambda_{i_0}-\dots-\lambda_{i_\beta}
-\lambda_{i_{\beta+\gamma+1}}-\dots-\lambda_{i_{\alpha+\beta+\gamma}}-\partial,
\lambda_{i_{\beta+\gamma+2}},\dots,\lambda_{i_{\alpha+\beta+\gamma}}}
\Big(
}\\
Y_{\lambda_{i_0},\dots,\lambda_{i_\beta}}
\!\!\!(\!v_{i_0},\dots, v_{i_\beta}\!),
\tilde Z_{\lambda_{i_{\beta+1}},\dots,\lambda_{i_{\beta+\gamma+1}}}
\!\!\!(\!v_{i_{\beta+1}},\dots, v_{i_{\beta+\gamma+1}}\!),
\!v_{i_{\beta+\gamma+2}},\dots,\! v_{i_{\alpha+\beta+\gamma}}\!\! \Big).\!\!\!\!
\end{array}
$$%
Similarly, we have
\begin{equation}\label{101024:eq5}
\begin{array}{l}
\big((X\Box \tilde Z)\Box^R Y\big)_{\lambda_0,\dots,\lambda_{\alpha+\beta+\gamma}}
(v_0, \dots, v_{\alpha+\beta+\gamma}) =
\!\!\!\!\!\!\!
\displaystyle{
\sum_{\substack{
i_0<\dots <i_\beta \\
i_{\beta+1}<\dots< i_{\beta+\gamma} \\
i_{\beta+\gamma+1}<\dots< i_{\alpha+\beta+\gamma}
}}
\!\!\!\!\!\!
\epsilon_v(i_0,\dots,i_{\alpha+\beta+\gamma})
} \\
\times X_{-\lambda_{i_{\beta+\gamma+1}}-\dots-\lambda_{i_{\alpha+\beta+\gamma}}-\partial,
\lambda_{i_{\beta+\gamma+1}},\dots,\lambda_{i_{\alpha+\beta+\gamma}}}
\Big(
\tilde Z_{\lambda_{i_0}+\dots+\lambda_{i_\beta},
\lambda_{i_{\beta+1}},\dots,\lambda_{i_{\beta+\gamma}}}
\\
\,\,\,\,\,\,\,\,\,\,\,\,\,\,\,\,\,\,\,\,\,\,\,\,\,\,\,\,\,\,\,\,
\vphantom{\Bigg(}
\big(
Y_{\lambda_{i_0},\dots,\lambda_{i_\beta}}
(v_{i_0},\dots, v_{i_\beta}),
v_{i_{\beta+1}},\dots, v_{i_{\beta+\gamma}} \big),
v_{i_{\beta+\gamma+1}},\dots, v_{i_{\alpha+\beta+\gamma}} \Big) \\
\displaystyle{
+ (-1)^{\bar p(Y)\bar p(\tilde Z)}
\sum_{\substack{
i_0<\dots <i_\beta \\
i_{\beta+1}<\dots< i_{\beta+\gamma+1} \\
i_{\beta+\gamma+2}<\dots< i_{\alpha+\beta+\gamma}
}}
(-1)^{\bar p(\tilde Z)(\bar p(v_{i_0})+\dots+\bar p(v_{i_\beta})}
\epsilon_v(i_0,\dots,i_{\alpha+\beta+\gamma})
} \\
\displaystyle{
\times X_{
\lambda_{i_0}+\dots+\lambda_{i_\beta},
-\lambda_{i_0}-\dots-\lambda_{i_\beta}
-\lambda_{i_{\beta+\gamma+1}}-\dots-\lambda_{i_{\alpha+\beta+\gamma}}-\partial,
\lambda_{i_{\beta+\gamma+2}},\dots,\lambda_{i_{\alpha+\beta+\gamma}}}
\Big(
}\\
Y_{\lambda_{i_0},\dots,\lambda_{i_\beta}}
\!\!\!(\!v_{i_0},\dots, v_{i_\beta}\!),
\tilde Z_{\lambda_{i_{\beta+1}},\dots,\lambda_{i_{\beta+\gamma+1}}}
\!\!\!(\!v_{i_{\beta+1}},\dots, v_{i_{\beta+\gamma+1}}\!),
\!v_{i_{\beta+\gamma+2}},\dots,\! v_{i_{\alpha+\beta+\gamma}}\!\! \Big).\!\!\!\!
\end{array}
\end{equation}
It is easy to check that
the first term in the RHS of \eqref{101024:eq4}
is equal to
$\big(X\Box^L(Y\Box^L \tilde Z)\big)_{\lambda_0,\dots,\lambda_{\alpha+\beta+\gamma}}
(v_0, \dots, v_{\alpha+\beta+\gamma})$,
while the first term in the RHS of \eqref{101024:eq5}
is equal to
$\big(X\Box^L(\tilde Z\Box^R Y)\big)_{\lambda_0,\dots,\lambda_{\alpha+\beta+\gamma}}
(v_0, \dots, v_{\alpha+\beta+\gamma})$.
Equation \eqref{101024:eq2} then follows from the
observation that the second terms in the RHS of \eqref{101024:eq4} and \eqref{101024:eq5}
differ by a factor  $(-1)^{\bar p(Y)\bar p(\tilde Z)}$.
Next, let us prove equation \eqref{101024:eq3}.
We have
\begin{equation}\label{101024:eq6}
\begin{array}{l}
\big((\tilde X\Box^R Y)\Box^R Z\big)_{\lambda_0,\dots,\lambda_{\alpha+\beta+\gamma}}
(v_0, \dots, v_{\alpha+\beta+\gamma}) =
\!\!\!\!\!\!\!
\displaystyle{
\sum_{\substack{
i_0<\dots <i_\gamma \\
i_{\gamma+1}<\dots< i_{\beta+\gamma} \\
i_{\beta+\gamma+1}<\dots< i_{\alpha+\beta+\gamma}
}}
\!\!\!\!\!\!\!\!\!\!
\epsilon_v(i_0,\dots,i_{\alpha+\beta+\gamma})
} \\
\times \tilde X_{\lambda_{i_0}+\dots+\lambda_{i_{\beta+\gamma}},
\lambda_{i_{\beta+\gamma+1}},\dots,\lambda_{i_{\alpha+\beta+\gamma}}}
\Big(Y_{\lambda_{i_0}+\dots+\lambda_{i_{\gamma}},
\lambda_{i_{\gamma+1}},\dots,\lambda_{i_{\beta+\gamma}}}
\\
\,\,\,\,\,\,\,\,\,\,\,\,\,\,\,\,\,\,\,\,\,\,\,\,\,\,\,\,\,\,\,\,
\vphantom{\Bigg(}
\big(
Z_{\lambda_{i_0},\dots,\lambda_{i_\gamma}}
(v_{i_0},\dots, v_{i_\gamma}),
v_{i_{\gamma+1}},\dots, v_{i_{\beta+\gamma}} \big),
v_{i_{\beta+\gamma+1}},\dots, v_{i_{\alpha+\beta+\gamma}} \Big)
\end{array}
\end{equation}
$$
\begin{array}{l}
\displaystyle{
+ \sum_{\substack{
i_0<\dots <i_\beta \\
i_{\beta+1}<\dots< i_{\beta+\gamma+1} \\
i_{\beta+\gamma+2}<\dots< i_{\alpha+\beta+\gamma}
}}
(-1)^{\bar p(\tilde Z)(\bar p(v_{i_0})+\dots+\bar p(v_{i_\beta})}
\epsilon_v(i_0,\dots,i_{\alpha+\beta+\gamma})
} \\
\displaystyle{
\times \tilde X_{
\lambda_{i_0}+\dots+\lambda_{i_\beta},
\lambda_{i_{\beta+1}}+\dots+\lambda_{i_{\beta+\gamma+1}},
\lambda_{i_{\beta+\gamma+2}},\dots,\lambda_{i_{\alpha+\beta+\gamma}}}
\Big(
Y_{\lambda_{i_0},\dots,\lambda_{i_\beta}}
(v_{i_0},\dots, v_{i_\beta}),
}\\
\,\,\,\,\,\,\,\,\,\,\,\,\,\,\,\,\,\,\,\,\,\,\,\,\,\,\,\,\,\,\,\,\,\,\,\,
Z_{\lambda_{i_{\beta+1}},\dots,\lambda_{i_{\beta+\gamma+1}}}
(v_{i_{\beta+1}},\dots, v_{i_{\beta+\gamma+1}}),
v_{i_{\beta+\gamma+2}},\dots, v_{i_{\alpha+\beta+\gamma}} \Big)\,.
\end{array}
$$
It is then easy to check that
the first term in the RHS of \eqref{101024:eq6}
is equal to
$\big(\tilde X\Box^R(Y\Box Z)\big)_{\lambda_0,\dots,\lambda_{\alpha+\beta+\gamma}}
(v_0, \dots, v_{\alpha+\beta+\gamma})$,
while the second term, if we exchange $Y$ and $Z$,
stays unchanged up to a factor $(-1)^{\bar p(Y)\bar p(\tilde Z)}$.
This proves \eqref{101024:eq3}.
We then have, by \eqref{101024:eq2} and \eqref{101024:eq3},
$$
\begin{array}{l}
[X,[Y,\tilde Z]]-(-1)^{\bar p(X)\bar p(Y)}[Y,[X,\tilde Z]]-
[[X,Y],\tilde Z]
=
-\Bigg(
(X\Box Y)\Box^L\tilde Z + \\
- X\Box^L(Y\Box^L\tilde Z)
-(-1)^{\bar p(Y)\bar p(\tilde Z)}
\Big((X\Box^L\tilde Z)\Box^R Y - X\Box^L(\tilde Z\Box^R Y)\Big)
\Bigg) + \\
+(\!-1\!)^{\bar p(X)\bar p(Y)}
\!\Bigg(\!\!
(Y\Box X)\Box^L\tilde Z - Y\Box^L(X\Box^L\tilde Z)
-(\!-1\!)^{\bar p(X)\bar p(\tilde Z)}
\!\Big(\!(Y\Box^L\tilde Z)\Box^R X + \\
- Y\Box^L(\tilde Z\Box^R X)\Big)
\Bigg)
-(-1)^{\bar p(\tilde Z)(\bar p(X)+\bar p(Y))}
\Bigg(
(\tilde Z\Box^R X)\Box^R Y - \tilde Z\Box^R(X\Box Y) +\\
-(-1)^{\bar p(X)\bar p(Y)}
\Big((\tilde Z\Box^R Y)\Box^R X - \tilde Z\Box^R(Y\Box X)\Big)
\Bigg)
=0\,.
\end{array}
$$
This proves that \eqref{100517:eq2} defines a representation of the Lie superalgebra $W^\partial(V)$
on $\tilde W^\partial(V)$.
With similar computations one can show that the Lie action of $W^\partial(V)$
is by derivations of the $\lambda$-brackets in $\tilde W^\partial(V)$.
Moreover, it is immediate to check that the Lie action of $W^\partial(V)$
commutes with the action $\partial$ on $\tilde W^\partial(V)$,
proving part (b).

Comparing equations \eqref{box2} and \eqref{100517:eq2}, we immediately get that
$\tint[X,\tilde Y]=[X,\tint \tilde Y]$
for every $X\in W^\partial(V)$ and $\tilde Y\in\tilde W^\partial(V)$,
proving that the map $\tint:\,\tilde W^\partial(V)\to W^\partial(V)$ is a homomorphism
of representations of the Lie superalgebra $W^\partial(V)$.
Moreover,
comparing \eqref{box3} and \eqref{100517:eq2},
we immediately get that
$[\tint\tilde X,\tilde Y]=[\tilde X_\lambda\tilde Y]\,|_{\lambda=0}$,
for every $\tilde X,\tilde Y\in\tilde W^\partial(V)$,
i.e., the Lie superalgebra action of $W^\partial(V)$
on $\tilde W^\partial_{k-h}(V)\to\tilde W^\partial_k(V)$
is compatible,
via the map $\tint:\,\tilde W^\partial(V)\to W^\partial(V)$,
with the Lie superalgebra action of
$\tilde W^{\partial}(V)/\partial\tilde W^{\partial}(V)$ on $\tilde W^{\partial}(V)$,
proving part (c).
\end{proof}
\begin{remark}\label{100517:rem}
The Lie superalgebra homomorphism $\tint$ defined in Proposition \ref{100517:prop}
in general is not surjective.
For example, if $V$ is a torsion module over $\mb F[\partial]$,
then $\tilde W^\partial_0(V)=0$ due to sesquilinearity,
while $W^\partial_0(V)=\End_{\mb F[\partial]}(V)$ needs not be zero.
However, if the $\mb F[\partial]$-module $V$ decomposes as
$V=\mc T\oplus(\mb F[\partial]\otimes U)$,
where $\mc T$ is the torsion submodule
and $\mb F[\partial]\otimes U$ is a finitely generated free submodule,
then $\tint:\,\tilde W^{\partial}_k(V)/\partial\tilde W^{\partial}_k(V) \to W^\partial_k(V)$
is a bijection for each $k\neq0$,
and, if $\mc T=0$, then
$\tint:\,\tilde W^{\partial}_0(V)/\partial\tilde W^{\partial}_0(V) \to W^\partial_0(V)$ is bijective as well.
Indeed, due to sesquilinearity,
if $\tilde X\in\tilde W^\partial_k(V)$,
then $\tilde X_{\lambda_0,\dots,\lambda_k}(v_0,\dots,v_k)$ vanishes
if one of the arguments $v_i$ lies in the torsion $\mc T\subset V$,
and $\tilde X$ is uniquely determined by its values on $U^{\otimes(k+1)}$.
Hence, we can identify $\tilde W^\partial_k(V)$ with the space of linear maps
$\tilde X:\,U^{\otimes(k+1)}\to\mb F_-[\lambda_0,\dots,\lambda_k] V$
satisfying the symmetry condition
\begin{equation}\label{101030:eq2}
\tilde X_{\lambda_0,\dots,\lambda_k}(u_0,\dots,u_k)
=\epsilon_u(i_0,\dots,i_k) \tilde X_{\lambda_{i_0},\dots,\lambda_{i_k}}(u_{i_0},\dots,u_{i_k})\,,
\end{equation}
for all permutations $(i_0,\dots,i_k)$ of $(0,\dots,k)$.
Similarly, if $X\in W^\partial_k(V)$ with $k\neq0$,
then $X_{\lambda_0,\dots,\lambda_k}(v_0,\dots,v_k)$ vanishes
if one of the arguments $v_i$ lies in the torsion $\mc T\subset V$,
and we can identify $W^\partial_k(V)$ with the space of linear maps
$X:\,U^{\otimes(k+1)}\to
\mb F[\lambda_0,\dots,\lambda_{k-1}]\otimes V$
(by identifying $\mb F_-[\lambda_0,\dots,\lambda_k]\otimes_{\mb F[\partial]}V$
and $\mb F_-[\lambda_0,\dots,\lambda_{k-1}]\otimes V$,
replacing $\lambda_k$
by $\lambda_k^\dagger=-\lambda_0-\dots-\lambda_{k-1}-\partial$)
satisfying the symmetry condition
\begin{equation}\label{101030:eq1}
X_{\lambda_0,\dots,\lambda_{k-1}}(u_0,\dots,u_k)
=\epsilon_u(i_0,\dots,i_k) X_{\lambda_{i_0},\dots,\lambda_{i_{k-1}}}(u_{i_0},\dots,u_{i_k})
\,\big|_{\lambda_k\mapsto\lambda_k^\dagger}\,,
\end{equation}
for all permutations $(i_0,\dots,i_k)$ of $(0,\dots,k)$.
Given $X\in W^\partial_k(V),\, k\neq0$,
a preimage $\tilde X\in\tilde W^\partial_k(V)$ of $X$ is
obtained by letting
$$
\tilde X_{\lambda_0,\dots,\lambda_k}(u_0,\dots,u_k)
=\frac1{k+1}\sum_{i=0}^k
X_{\lambda_0,\dots,\lambda_{k-1}}(u_0,\dots,u_k)
\,\big|_{\lambda_i\mapsto\lambda_i^\dagger}\,,
$$
where $\lambda_i^\dagger=-\lambda_0-\stackrel{i}{\check{\dots}}-\lambda_k-\partial$.
Indeed, it is immediate to check that, if $X$ satisfies the symmetry condition \eqref{101030:eq1},
then $\tilde X$ satisfies the symmetry condition \eqref{101030:eq2}.
\end{remark}

Let $V$ and $U$ be vector superspaces with parity $\bar p$, endowed with a structure
of finitely generated $\mb F[\partial]$-modules.
In analogy with the reduction introduced in Section \ref{sec:4.2},
we define the $\mb Z_+$-graded vector superspace (with parity still denoted by $\bar p$)
$\tilde W^\partial(V,U)=\bigoplus_{k\in\mb Z_+}W^\partial_k(V,U)$,
where
$$
\tilde W^\partial_k(V,U)
=\Hom{}^{\sym}_{\mb F[\partial]^{\otimes(k+1)}}(V^{\otimes(k+1)},
\mb F_-[\lambda_0,\dots,\lambda_k]\otimes U)\,.
$$
One checks that the analogue of Proposition \ref{100421:prop} holds,
if we replace tensor products over $\mb F[\partial]$ by tensor products
over the field $\mb F$.

The reduced space $\tilde W^\partial(V,U)$ is obtained as a subquotient of the universal
Lie conformal superalgebra $\tilde W^\partial(V\oplus U)$,
via the canonical isomorphism of superspaces
\begin{equation}\label{100421z:eq1}
\begin{array}{c}
\tilde{\mc U}/\tilde{\mc K}
\stackrel{\sim}{\longrightarrow} \tilde W^\partial_k(V, U)\,,
\end{array}
\end{equation}
where $\tilde{\mc U}$ and $\tilde{\mc K}$ are the following subspaces
of $\tilde W^\partial_k(V\oplus U)$:
$$
\begin{array}{rcl}
\tilde{\mc U}
&=&
\Hom{}^{\sym}_{\mb F[\partial]^{\otimes(k+1)}}((V\oplus U)^{\otimes(k+1)},
\mb F_-[\lambda_0,\dots,\lambda_k]\otimes U) \,, \\
\tilde{\mc K} &=& \big\{\tilde Y \,\big|\, \tilde Y(V^{\otimes(k+1)})=0\big\} \,.
\end{array}
$$
The following analogue of Proposition \ref{100421:prop} holds (the proof is similar):
\begin{proposition}\label{100421z:prop}
Let $X\in W^\partial_h(V\oplus U)$. Then the action of $X$ on $\tilde W^\partial(V\oplus U)$
given by Proposition \ref{100517:prop}(b)
leaves the subspaces $\tilde{\mc U}$ and $\tilde{\mc K}$ invariant
provided that
\begin{enumerate}[(i)]
\item
$X_{\lambda_0,\dots,\lambda_h}(w_0,\dots,w_h)\in \mb F_-[\lambda_0,\dots,\lambda_h]\otimes U$
if one of the arguments $w_i$ lies in $U$,
\item
$X_{\lambda_0,\dots,\lambda_h}(v_0,\dots,v_h)\in \mb F_-[\lambda_0,\dots,\lambda_h]\otimes_{\mb F[\partial]} V$
if all the arguments $v_i$ lie in $V$.
\end{enumerate}
In this case, the action of $X$ induces a well-defined linear map on the reduced space $\tilde W^\partial(V,U)$,
via the isomorphism \eqref{100421z:eq1}.
\end{proposition}

Furthermore, as in Section \ref{sec:2.4},
given a subalgebra $R_0$ of the Lie conformal superalgebra
$\tilde W^\partial_0(V)=\RCend(V)$,
we define a \emph{prolongation} of $R_0$ in $\tilde W^\partial(V)$
as a $\mb Z$-graded subalgebra $R=\bigoplus_{k=-1}^\infty R_k$
of the $\mb Z$-graded Lie conformal superalgebra
$\tilde W^\partial(V)=\bigoplus_{k=-1}^\infty \tilde W^\partial_k(V)$,
such that $R_{-1}=V$
and $R_0$ coincides with the given Lie conformal superalgebra.
The \emph{full prolongation}
$\tilde W^{\partial,R_0}(V)=\bigoplus_{k=-1}^\infty \tilde W^{\partial,R_0}_k(V)$
of $R_0$ is defined by letting
$\tilde W^{\partial,R_0}_{-1}(V)=V,\,\tilde W^{\partial,R_0}_0(V)=R_0$ and,
inductively, for $k\geq1$,
\begin{equation}\label{100528:eq4}
\tilde W^{\partial,R_0}_k(V)=\big\{X\in\tilde W^\partial_k(V)\,\big|\,[X_\lambda V]
\subset \mb F[\lambda]\otimes\tilde W^{\partial,R_0}_{k-1}(V)\big\}\,.
\end{equation}
It is immediate to check, by the Jacobi identity, that the above formula defines
a maximal prolongation of the Lie conformal superalgebra $\tilde W^\partial(V)$.

\subsection{The basic Lie conformal superalgebra cohomology complex}
\label{sec:6.2}

Suppose that $R$ is a Lie conformal superalgebra and $M$ is an $R$-module,
with parity $p$,
assume that $R$ and $M$ are finitely generated as $\mb F[\partial]$-modules,
and consider the corresponding element $X$ in the subset \eqref{100421:eq7} of
$W^\partial_1(\Pi R\oplus\Pi M)_{\bar1}$.
Consider the reduced superspace
$\tilde W^\partial(\Pi R,\Pi M)$
introduced in Section \ref{sec:6.1}, with parity denoted by $\bar p$.

Note that the element $X$
satisfies conditions (i) and (ii) in Proposition \ref{100421z:prop}.
Hence the action of $X$ on $\tilde W^\partial(\Pi R\oplus\Pi M)$,
induces a well-defined endomorphism $d_X$ of the reduced space $\tilde W^\partial(\Pi R,\Pi M)$
such that $d_X^2=0$,
thus making $(\tilde W^\partial(\Pi R,\Pi M),d_X)$ a cohomology complex.
The explicit formula for the differential $d_X$
is the same as \eqref{100421:eq9},
except that we view both sides as elements of
$\mb F[\lambda_0,\dots,\lambda_k]\otimes M$.
If $R$ is a (purely even) Lie conformal algebra and $M$ is a purely even $R$-module,
we recover, up to an overall sign,
the basic Lie conformal algebra cohomology complex
as defined in \cite{BKV}, \cite{BDAK} and \cite{DSK2}.

Note that, in the special case when $M=R$ is the adjoint representation,
the complex $(\tilde W^\partial(\Pi R,\Pi M),d_X)$ coincides
with the complex $(\tilde W^\partial(\Pi R),d_X)$,
where $d_X$ here is the differential given by the Lie superalgebra action of $W^\partial(\Pi R)$
on $\tilde W^\partial(\Pi R)$ given by Proposition \ref{100517:prop}(b).
By Proposition \ref{100517:prop}(c),
the canonical map $\tint:\,\tilde W^\partial(\Pi R)\to W^\partial(\Pi R)$
defined in Proposition \ref{100517:prop}(a) is a homomorphism of cohomology complexes.
The same holds for the map $\tint:\,\tilde W^\partial(\Pi R,\Pi M)\to W^\partial(\Pi R,\Pi M)$.

\subsection{Extension to infinitely generated $\mb F[\partial]$-modules}
\label{sec:6.3}

If $V$ is not necessarily finitely generated as $\mb F[\partial]$-module,
we may still consider the $\mb F[\partial]$-module $\tilde W^\partial(V)$
endowed with the $\lambda$-product $X\Box_\lambda Y$
defined by the same formula \eqref{100420:eq1}.
The problem here is that in general $X\Box_\lambda Y$ will be a formal power series in $\lambda$ 
(not anymore a polynomial) with coefficients in $\tilde W^\partial(V)$:
\begin{equation}\label{box-infinite}
\Box_\lambda:\,
\tilde W^\partial(V)\times\tilde W^\partial(V)\to\mb F[[\lambda]]\otimes\tilde W^\partial(V)\,.
\end{equation}
Note that when dealing with formal power series in $\lambda$,
the corresponding $\lambda$-bracket 
$[X_\lambda Y]=X\Box_\lambda Y-Y\Box_{-\lambda-\partial}X$
would seem ill-defined (since the coefficient of a given power of $\lambda$
will be an infinite sum).
However,
for every fixed collection of vectors $v_0,\dots,v_k\in V$,
the element $(X\Box_\lambda Y)_{\lambda_0,\dots,\lambda_k}(v_0,\dots,v_k)$
defined by \eqref{100420:eq1} is polynomial in $\lambda$
(and all the other variables $\lambda_0,\dots,\lambda_k$).
Hence,
$(Y\Box_{-\lambda-\partial} X)_{\lambda_0,\dots,\lambda_k}(v_0,\dots,v_k)$
makes perfect sense, 
by replacing $\mu$ by $-\lambda-\lambda_0-\dots-\lambda_k-\partial$ 
($\partial$ acting from the left)
in the polynomial $(Y\Box_\mu X)_{\lambda_0,\dots,\lambda_k}(v_0,\dots,v_k)$.
We can then define their $\lambda$-bracket
$$
[X_\lambda Y]
\,:\,\,
V^{\otimes k+1}\to\mb F[\lambda,\lambda_0,\dots,\lambda_k]\otimes V\,,
$$
which is well defined on a every given collection of vectors $v_0,\dots,v_k\in V$:
\begin{equation}\label{20120801:eq1}
\begin{array}{l}
\vphantom{\Big(}
\displaystyle{
[X_\lambda Y]_{\lambda_0,\dots,\lambda_k}(v_0,\dots,v_k)
=(X\Box_\lambda Y)_{\lambda_0,\dots,\lambda_k}(v_0,\dots,v_k)
} \\
\vphantom{\Big(}
\displaystyle{
-(-1)^{\bar p(X)\bar p(Y)} 
(Y\Box_{-\lambda-\lambda_0-\dots-\lambda_k-\partial}X)_{\lambda_0,\dots,\lambda_k}(v_0,\dots,v_k)\,,
}
\end{array}
\end{equation}
(where $\partial$ in the second term is moved to the left),
and, on any such collection of vectors $v_0,\dots,v_k$, 
it is clearly polynomial in all the variables, including $\lambda$.
Alternatively, $[X_\lambda Y]$ can be described as the formal power series in $\lambda$
with coefficients in $\tilde{W}^\partial(V)$ such that \eqref{20120801:eq1} holds.

Note that all the identities in Lemma \ref{100516:lem} and Corollary \ref{100516:prop}
are proved on given collection of vectors from $V$.
Hence, the same computations show that the $\lambda$-bracket
$[X_\lambda Y]$ on $\tilde W^\partial(V)$,
even for non finitely generated $\mb F[\partial]$-module $V$,
satisfies all the Lie conformal algebra axioms,
sesquilinearity, skewsymmetry and Jacobi identity,
in the sense that  each axiom holds (polynomially in the variables $\lambda,\lambda_0,
\dots,\lambda_k$) on every fixed collections of vectors from $V$.
We thus have the following
\begin{lemma}\label{infinite:lem1}
For $X\in\tilde W^{\partial}_h(V),\,Y\in\times\tilde W^{\partial}_{k-h}(V)$, their $\lambda$-bracket
$[X_\lambda Y]:\,V^{\otimes k+1}\to\mb F[\lambda,\lambda_0,\dots,\lambda_k]\otimes V$
satisfies all the Lie conformal algebra axioms on every given set of vectors from $V$:\\
sesquilinearity
$$
\begin{array}{l}
\vphantom{\Big(}
[\partial X_\lambda Y]_{\lambda_0,\dots,\lambda_k}(v_0,\dots,v_k)
=-\lambda
[X_\lambda Y]_{\lambda_0,\dots,\lambda_k}(v_0,\dots,v_k)\,,\\
\vphantom{\Big(}
[X_\lambda \partial Y]_{\lambda_0,\dots,\lambda_k}\!\!(v_0,\dots,v_k)
\!=\!
(\lambda+\lambda_0+\dots+\lambda_k+\partial)
[X_\lambda Y]_{\lambda_0,\dots,\lambda_k}\!\!(v_0,\dots,v_k),
\end{array}
$$
skewsymmetry
$$
\vphantom{\Big(}
[X_\lambda Y]_{\lambda_0,\dots,\lambda_k}\!\!(v_0,\dots,v_k)
\!=\!-(\!-1\!)^{\bar p(X)\bar p(Y)}
[Y_{-\lambda-\lambda_0-\dots-\lambda_k-\partial} Y]_{\lambda_0,\dots,\lambda_k}\!\!(v_0,\dots,v_k),
$$
where in the RHS $\partial$ is moved to the left,
and Jacobi identity (for 
$X\in\tilde W^{\partial}_h(V),\,Y\in\times\tilde W^{\partial}_{k-h}(V)$ 
and $Z\in\tilde W^{\partial}_{\ell-k}(V)$)
$$
\begin{array}{l}
\vphantom{\Big(}
[X_\lambda [Y_\mu Z]]_{\lambda_0,\dots,\lambda_\ell}(v_0,\dots,v_\ell)\!
-(-1)^{\bar p(X)\bar p(Y)}
[Y_\mu [X_\lambda Z]]_{\lambda_0,\dots,\lambda_\ell}(v_0,\dots,v_\ell)
\\
\vphantom{\Big(}
=
[[X_\lambda Y]_{\lambda+\mu} Z]_{\lambda_0,\dots,\lambda_\ell}(v_0,\dots,v_\ell)
\,.
\end{array}
$$
\end{lemma}
\begin{remark}\label{generalized-lca}
It follows from the above lemma that,
for $X,Y,Z\in \tilde{W}^\partial(V)$,
the sesquilinearity conditions $[\partial X_\lambda Y]=-\lambda[X_\lambda Y]$
and $[X_\lambda \partial Y]=(\lambda+\partial)[X_\lambda Y]$
hold in the ring of formal power series $\mb F[[\lambda]]\otimes\tilde{W}^\partial(V)$,
and similarly the Jacobi identity
$[X_\lambda [Y_\mu Z]]
-(-1)^{\bar p(X)\bar p(Y)}[Y_\mu [X_\lambda Z]]
=[[X_\lambda Y]_{\lambda+\mu} Z]$
holds in the ring $\mb F[[\lambda,\mu]]\otimes\tilde{W}^\partial(V)$.
As for the skewsymmetry relation 
$[X_\lambda Y]=-(-1)^{\bar p(X)\bar p(Y)}[Y_{-\lambda-\partial}X]$,
we can only say that it holds, for every $N\geq 0$,
in the quotient space 
$\mb F[[\lambda]]\otimes\big(\tilde{W}^\partial(V)/\partial^N\tilde{W}^\partial(V)\big)$
(we need to do this to avoid diverging series).
One may talk, in this sense, of a ``generalized'' Lie conformal superalgebra.
\end{remark}
\begin{corollary}\label{infinite:cor1}
If $R\subset\tilde W^\partial(V)$ is an $\mb F[\partial]$-submodule
with the property that the $\lambda$-bracket of every two elements $X,Y\in R$
is actually polynomial in $\lambda$ and with coefficients in $R$,
then $R$ is a (honest) Lie conformal superalgebra.
\end{corollary}
\begin{proof}
Obvious.
\end{proof}
The above result will be applied in the next sections,
when studying the universal odd PVAs $\tilde W^{\partial,\as}(\Pi\mc V)$ 
and $\tilde W^{\var}(\Pi\mc V)$.

We can also extend, to the case of infinitely generated $\mb F[\partial]$-modules $V$,
the notions of prolongation and full prolongation.
\begin{definition}\label{infinite:def1}
Let $R_0\subset\tilde W^\partial_0(V)=\RCend(V)$
be an $\mb F[\partial]$-submodule with the property that, for every $X,Y\in R_0$,
the formal power series $[X_\lambda Y]$ has coefficients in $R_0$.
\begin{enumerate}[(a)]
\item
A \emph{prolongation} of $R_0$ in $\tilde W^\partial(V)$
is a $\mb Z$-graded $\mb F[\partial]$-submodule 
$R=\bigoplus_{k=-1}^\infty R_k\subset\tilde W^\partial(V)$, with $R_k\subset\tilde W^\partial_k(V)$,
such that $R_{-1}=V$, $R_0$ coincides with the given $\mb F[\partial]$-module,
and, for every $X\in R_h,Y\in R_{k-h}$,
the formal power series $[X_\lambda Y]$ has coefficients in $R_k$.
\item
The \emph{full prolongation}
$\tilde W^{\partial,R_0}(V)=\bigoplus_{k=-1}^\infty \tilde W^{\partial,R_0}_k(V)$
of $R_0$ is defined by letting
$\tilde W^{\partial,R_0}_{-1}(V)=V,\,\tilde W^{\partial,R_0}_0(V)=R_0$ and,
inductively, for $k\geq1$,
\begin{equation}\label{100528:eq4-infinite}
\tilde W^{\partial,R_0}_k(V)=\big\{X\in\tilde W^\partial_k(V)\,\big|\,[X_\lambda V]
\subset \mb F[[\lambda]]\otimes\tilde W^{\partial,R_0}_{k-1}(V)\big\}\,.
\end{equation}
\end{enumerate}
\end{definition}
One easily checks, by the Jacobi identity, that \eqref{100528:eq4-infinite}
defines indeed a prolongation of $\tilde W^\partial(V)$.


\section{The universal (odd) Poisson vertex superalgebra
for a differential superalgebra $\mc V$
and basic PVA cohomology}
\label{sec:7}

\subsection{The universal odd PVA $\tilde W^{\partial,\as}(\Pi\mc V)$}
\label{sec:7.1}

Throughout this section, we let $\mc V$ be a commutative associative differential superalgebra,
with a given even derivation $\partial$, and with parity denoted by $p$.
We assume moreover that $\mc V$ is finitely generated as a differential algebra,
i.e. there are finitely many elements which, along with all their derivatives, generate $\mc V$.

We let $\RCder(\mc V)$ be the Lie conformal superalgebra of right conformal derivations
of $\mc V$, namely the linear maps
$X_\lambda:\, \mc V\to \mb F[\lambda]\otimes\mc V$,
satisfying \eqref{100528:eq1} and
\begin{equation}\label{100528:eq2}
X_\lambda(uv)=X_{\lambda+\partial}(u)_\to v+(-1)^{p(u)p(v)}X_{\lambda+\partial}(v)_\to u\,.
\end{equation}
This is a subalgebra of the space $\RCend(\mc V)$ of right conformal endomorphisms,
$X_\mu:\,\mc V\to\mb F[\mu]\otimes\mc V$,
satisfying \eqref{100528:eq1}
with $\lambda$-bracket given by
$[X_\lambda Y]_\mu(v)
=X_{-\lambda-\partial}(Y_\mu(v))-(-1)^{\bar p(X)\bar p(Y)} Y_\lambda(X_\mu(v))$
(cf. \eqref{100420:eq1}).
Though the $\lambda$-bracket on $\RCend(\mc V)$ has values in formal power series
(see Section \ref{sec:6.3}),
when restricted to $\RCder(\mc V)$ it is polynomial in $\lambda$,
due to the assumption that $\mc V$ is a finitely generated differential algebra.
Indeed, due to the sesquilinearity assumption \eqref{100528:eq1} 
and the Leibniz rule \eqref{100528:eq2},
an element of $\RCder(\mc V)$ is determined by its values on a set 
of differential generators of $\mc V$.
\begin{remark}\label{20110528:rem2}
Isomorphism \eqref{100528:eq3} restricts to an isomorphism
of $\RCder(\mc V)$ to the Lie conformal superalgebra $\Cder(\mc V)$
of all conformal derivations of $\mc V$,
namely the conformal endomorphisms of $\mc V$, satisfying
$X_\lambda(uv)=X_\lambda(u)v+(-1)^{p(u)p(v)}X_\lambda(v)u$.
\end{remark}

Recall from Section \ref{sec:6.3} the definition of
the $\mb F[\partial]$-module
$\tilde W^\partial(\Pi \mc V)=\bigoplus_{k=-1}^\infty \tilde W^\partial_k(\Pi \mc V)$
together with the $\lambda$-bracket
which makes it a ``generalized'' Lie conformal superalgebra
(in the sense of Remark \ref{generalized-lca}).
We denote its parity by $\bar p$.
Consider the full prolongation (cf. Definition \ref{infinite:def1}),
associated to the Lie conformal superalgebra 
$\RCder(\mc V)\subset\RCend(\mc V)=\tilde W^\partial_0(\Pi\mc V)$,
which we denote by
$$
\tilde W^{\partial,\as}(\Pi\mc V)=\bigoplus_{k=-1}^\infty\tilde W^{\partial,\as}_k(\Pi\mc V)\subset
\tilde W^\partial(\Pi\mc V)\,.
$$
\begin{proposition}\label{100422z:prop1}
\begin{enumerate}[(a)]
\item
For every $k\geq-1$, the superspace $\tilde W^{\partial,\as}_k(\Pi\mc V)$
is the subspace of $\tilde W^{\partial}_k(\Pi\mc V)$,
consisting of linear maps
$X:\,(\Pi\mc V)^{\otimes(k+1)}\to\mb F[\lambda_0,\dots,\lambda_k]\otimes\Pi \mc V$
satisfying the symmetry and sesquilinearity conditions,
and the Leibniz rule \eqref{100422c:eq1}
(where both sides are interpreted as elements 
of $\mb F[\lambda_0,\dots,\lambda_k]\otimes\Pi \mc V$).
\item
$\tilde W^{\partial,\as}(\Pi\mc V)$ is a Lie conformal superalgebra.
\end{enumerate}
\end{proposition}
\begin{proof}
Part (a) follows by an easy induction on $k\geq0$.
Let us next prove part (b).
Due to the sesquilinearity and the Leibniz rule \eqref{100422c:eq1},
an element $X\in\tilde W^{\partial,\as}_k(\Pi\mc V)$
is determined by its values on a set of differential generators of $\mc V$.
On the other hand, for $X\in\tilde W^{\partial,\as}_h(\Pi\mc V)$
and $Y\in\tilde W^{\partial,\as}_{k-h}(\Pi\mc V)$,
we have that
$[X_\lambda Y]_{\lambda_0,\dots,\lambda_k}(v_{i_0},\dots,v_{i_k})$
is polynomial in $\lambda$ (and all the other variables $\lambda_0,\dots,\lambda_k$)
for every $k$-tuple $(v_{i_0},\dots,v_{i_k})$ consisting of differential generators of $V$.
Since such $k$-tuples are finitely many,
we deduce that the $\lambda$-bracket $[X_\lambda Y]$
is polynomial in $\lambda$ 
(and with coefficients in $\tilde W^{\partial,\as}_k(\Pi\mc V)$, by definition
of prolongation).
Hence, the statement follows from Corollary \ref{infinite:cor1}.
\end{proof}

We next define a structure of a commutative associative superalgebra
on the superspace $\Pi \tilde W^{\partial,\as}(\Pi\mc V)$,
making it an odd Poisson vertex superalgebra.
Let $X\in \Pi \tilde W^{\partial,\as}_{h-1}(\Pi \mc V)$
and $Y\in \Pi \tilde W^{\partial,\as}_{k-h-1}(\Pi \mc V)$, for $h\geq0,\,k-h\geq0$,
and denote by $p(X)$ and $p(Y)$ their parities in these spaces.
We define their \emph{concatenation product} $X\wedge Y\in \Pi \tilde W^{\partial,\as}_{k-1}(\Pi \mc V)$
as the following map:
\begin{equation}\label{100422z:eq4}
\begin{array}{c}
\displaystyle{
(X\wedge Y)_{\lambda_1,\dots,\lambda_k}(a_1,\dots,a_k)
=
\!\!\!\!
\sum_{\substack{
i_1<\dots <i_{h}\\
i_{h+1}<\dots< i_k}}
\!\!\!\!
\epsilon_{a}(i_1,\dots,i_k) (-1)^{p(Y)(\bar p(a_{i_1})+\dots+\bar p(a_{i_h}))}
} \\
\displaystyle{
\times
X_{\lambda_{i_1},\dots,\lambda_{i_h}}(a_{i_1},\dots,a_{i_h})
Y_{\lambda_{i_{h+1}},\dots,\lambda_{i_k}}(a_{i_{h+1}},\dots,a_{i_k})\,,
}
\end{array}
\end{equation}
where $\epsilon_{a}(i_1,\dots,i_k)$ is as in \eqref{100418:eq1}
for the elements $a_1,\dots,a_k\in\Pi \mc V$.
\begin{proposition}\label{100422z:prop3}
\begin{enumerate}[(a)]
\item
The $\mb Z_+$-graded superspace, with parity $p$,
$\tilde{\mc G}(\mc V)=\bigoplus_{k=0}^\infty\tilde{\mc G}_k(\mc V)$,
where $\tilde{\mc G}_k(\mc V)=\Pi\tilde W_{k-1}^{\partial,\as}(\Pi \mc V)$,
together with the concatenation product
$\wedge:\,\tilde{\mc G}_h(\mc V)\times\tilde{\mc G}_{k-h}(\mc V)\to\tilde{\mc G}_k(\mc V)$
given by \eqref{100422z:eq4},
and with the Lie conformal superalgebra $\lambda$-bracket
on $\Pi\tilde{\mc G}(\mc V)=\tilde W^{\partial,\as}(\Pi \mc V)$,
is a $\mb Z_+$-graded odd Poisson vertex superalgebra.
\item
The representation of the Lie superalgebra $W^{\partial}(\Pi\mc V)$
on $\tilde W^{\partial}(\Pi\mc V)$
defined by Proposition \ref{100517:prop}(b)
restricts to a representation of its subalgebra $W^{\partial,\as}(\Pi\mc V)$
on the odd Poisson vertex superalgebra
$\tilde W^{\partial,\as}(\Pi\mc V)\subset \tilde W^{\partial}(\Pi\mc V)$,
commuting with $\partial$ and
acting by derivations of both
the concatenation product and the $\lambda$-bracket.
\item
The canonical map $\tint:\,\tilde W^\partial(\Pi\mc V)\to W^\partial(\Pi\mc V)$
defined in Proposition \ref{100517:prop}(a)
restricts to a map $\tint:\,\tilde W^{\partial,\as}(\Pi\mc V)\to W^{\partial,\as}(\Pi\mc V)$,
which is a homomorphism of representations of the Lie superalgebra $W^{\partial,\as}(\Pi\mc V)$.
Moreover, this map induces an injective Lie algebra homomorphism
$\tint:\,\tilde W^{\partial,\as}(\Pi\mc V)/\partial\tilde W^{\partial,\as}(\Pi\mc V)
\to W^{\partial,\as}(\Pi\mc V)$.
\end{enumerate}
\end{proposition}
\begin{proof}
First, we prove that $X\wedge Y$ in \eqref{100422z:eq4} is an element
of $\Pi \tilde W^{\partial,\as}_k(\Pi\mc V)$,
namely, it is a map $(\Pi\mc V)^{\otimes k}\to\mb F_-[\lambda_1,\dots,\lambda_k]\otimes\Pi\mc V$
satisfying the symmetry and sesquilinearity conditions, and
the Leibniz rule \eqref{100422c:eq1}.
The symmetry condition can be checked directly with the usual argument
(see e.g. the proof of Proposition \ref{100422:prop3}),
while the sesquilinearity condition is immediate by the definition \eqref{100422z:eq4}
of the concatenation product.
As for the Leibniz rule, by the symmetry condition it is enough to prove
\eqref{100422c:eq1} for $i=k$, i.e.
\begin{equation}\label{100422c:eq1bis}
\begin{array}{c}
\vphantom{\Big(}
(X\wedge Y)_{\lambda_1,\dots,\lambda_k}(a_1,\dots ,a_{k-1},bc)
\!=\!(X\wedge Y)_{\lambda_1,\dots,\lambda_{k-1},\lambda_k+\partial}
(a_1,\dots,a_{k-1},b)_\to\! c \\
\vphantom{\Big(}
+(-1)^{p(b)p(c)}
(X\wedge Y)_{\lambda_1,\dots,\lambda_{k-1},\lambda_k+\partial}
(a_1,\dots,a_{k-1},c)_\to b\,.
\end{array}
\end{equation}
We have, by the definition of the concatenation product,
\begin{equation}\label{100423z:eq1bis}
\begin{array}{l}
(X\wedge Y)_{\lambda_1,\dots,\lambda_k}(a_1,\dots,a_{k-1},bc)
= \\
\displaystyle{
\sum_{\substack{
i_1<\dots <i_{h}<k\\
i_{h+1}<\dots< i_k=k}}
\epsilon_{a_1,\dots,a_{k-1},bc}(i_1,\dots,i_k)
(-1)^{p(Y)(\bar p(a_{i_1})+\dots+\bar p(a_{i_h}))}
} \\
\displaystyle{
\vphantom{\Big(}
\,\,\,\,\,\,
\times
X_{\lambda_{i_1},\dots,\lambda_{i_h}}(a_{i_1},\dots,a_{i_h})
Y_{\lambda_{i_{h+1}},\dots,\lambda_{i_{k-1}},\lambda_k}(a_{i_{h+1}},\dots,a_{i_{k-1}},bc)
} \\
\displaystyle{
+
\sum_{\substack{
i_1<\dots <i_{h}=k\\
i_{h+1}<\dots< i_k<k}}
\epsilon_{a_1,\dots,a_{k-1},bc}(i_1,\dots,i_k)
(-1)^{p(Y)(\bar p(a_{i_1})+\dots+\bar p(a_{i_{h-1}})+\bar p(bc))}
} \\
\displaystyle{
\,\,\,\,\,\,
\times
X_{\lambda_{i_1},\dots,\lambda_{i_{h-1}},\lambda_k}(a_{i_1},\dots,a_{i_{h-1}},bc)
Y_{\lambda_{i_{h+1}},\dots,\lambda_{i_k}}(a_{i_{h+1}},\dots,a_{i_k})\,.
}
\end{array}
\end{equation}
Equation \eqref{100422c:eq1bis} can be derived from \eqref{100423z:eq1bis}
using the Leibniz rules \eqref{100422c:eq1} for $X$ and $Y$,
together with the sign identities \eqref{101109:eq1}
(valid in the first term of the RHS of \eqref{100423z:eq1bis}),
\eqref{101109:eq2}
(valid in the second term of the RHS of \eqref{100423z:eq1bis}),
and \eqref{101109:eq3}.

We now prove that the concatenation product \eqref{100422z:eq4}
makes $\Pi \tilde W^{\partial,\as}(\Pi\mc V)$ into
a commutative, associative, differential superalgebra.
First, one easily checks that $X\wedge Y$ has parity
$p(X)+p(Y)$ as an element of $\Pi\tilde W^{\partial,\as}(\Pi\mc V)$,
so that $\Pi\tilde W^{\partial,\as}(\Pi\mc V)$,
endowed with the concatenation product \eqref{100422z:eq4},
is a superalgebra.
Recalling the definition \eqref{100517:eq3} of the $\mb F[\partial]$-module structure
of $\Pi\tilde W^{\partial,\as}(\Pi\mc V)$,
it is immediate to check that $\partial$ is an even derivation
of the concatenation product \eqref{100422z:eq4},
making $\Pi\tilde W^{\partial,\as}(\Pi\mc V)$ a differential superalgebra.
Moreover,
since $\mc V$ is a commutative superalgebra, we have
$$
\begin{array}{l}
X_{\lambda_{i_{1}},\dots,\lambda_{i_h}}(a_{i_1},\dots,a_{i_h})
Y_{\lambda_{i_{h+1}},\dots,\lambda_{i_k}}(a_{i_{h+1}},\dots,a_{i_k}) \\
=\pm Y_{\lambda_{i_{h+1}},\dots,\lambda_{i_k}}(a_{i_{h+1}},\dots,a_{i_k})
X_{\lambda_{i_{1}},\dots,\lambda_{i_h}}(a_{i_1},\dots,a_{i_h})\,,
\end{array}
$$
where
$\pm = (-1)^{(p(X)+\bar p(a_{i_1})+\dots+\bar p(a_{i_h}))(p(Y)+\bar p(a_{i_{h+1}})
+\dots+\bar p(a_{i_k}))}$.
This immediately implies commutativity of the concatenation product \eqref{100422z:eq4}.
Finally,
given $X\in\Pi\tilde W^{\partial,\as}_{h-1}(\Pi\mc V),\,
Y\in\Pi\tilde W^{\partial,\as}_{k-h-1}(\Pi\mc V),\,Z\in\Pi\tilde W^{\partial,\as}_{\ell-k-1}(\Pi\mc V)$,
and $a_1,\dots,a_\ell\in\Pi\mc V$, we have, using associativity of $\mc V$,
that both $(X\wedge(Y\wedge Z))_{\lambda_1,\dots,\lambda_\ell}(a_1,\dots,a_\ell)$
and $((X\wedge Y)\wedge Z)_{\lambda_1,\dots,\lambda_\ell}(a_1,\dots,a_\ell)$ are equal to
$$
\begin{array}{l}
\displaystyle{
\sum_{\substack{
i_1<\dots <i_{h}\\
i_{h+1}<\dots< i_k\\
i_{k+1}<\dots< i_\ell
}}
\epsilon_{a}(i_1,\dots,i_\ell)
(-1)^{p(Y)(\bar p(a_{i_1})+\dots+\bar p(a_{i_h})) + p(Z)(\bar p(a_{i_1})+\dots+\bar p(a_{i_k}))}
} \\
\displaystyle{
X_{\lambda_{i_{1}},\dots,\lambda_{i_h}}(a_{i_1},\dots,a_{i_h}\!)
Y_{\lambda_{i_{h+1}},\dots,\lambda_{i_k}}(a_{i_{h+1}},\dots,a_{i_k}\!)
Z_{\lambda_{i_{k+1}},\dots,\lambda_{i_\ell}}(a_{i_{k+1}},\dots,a_{i_\ell}\!),
}
\end{array}
$$
proving associativity of the concatenation product.

To complete the proof of part (a),
we need to prove that the $\lambda$-bracket on $\tilde W^{\partial,\as}(\Pi\mc V)$
satisfies the odd Leibniz rule,
\begin{equation}\label{100423z:eq2bis}
[X_\lambda Y\wedge Z]=[X_\lambda Y]\wedge Z+(-1)^{\bar p(X) p(Y)}Y\wedge[X_\lambda Z]\,,
\end{equation}
thus making $\tilde{\mc G}(\mc V)=\Pi\tilde W^{\partial,\as}(\Pi\mc V)$
into an odd Poisson vertex superalgebra.
This follows by the following two identities,
which can be checked directly:
\begin{equation}\label{100423z:eq3bis}
\begin{array}{rcl}
X\Box_\lambda(Y\wedge Z) &=& (X\Box_\lambda Y)\wedge Z
+(-1)^{\bar p(X)p(Y)}Y\wedge(Z\Box_\lambda X)\,,\\
(X\wedge Y)\Box_\lambda Z &=& \big(e^{\partial\partial_\lambda}X\big)\wedge (Y\Box_\lambda Z)
+(-1)^{p(Y)\bar p(Z)}(X\Box_{\lambda+\partial} Z)_\to\wedge Y\,.
\end{array}
\end{equation}

Let us next prove part (b). Given $X\in W_h^{\partial,\as}(\Pi\mc V)$
and $\tilde Y\in\tilde W_{k-h}^{\partial,\as}(\Pi\mc V)$, we want to prove that
$[X,\tilde Y]$, defined by \eqref{100517:eq2}, belongs to $\tilde W_{k}^{\partial,\as}(\Pi\mc V)$,
i.e. it satisfies
\begin{equation}\label{101110:eq1}
\begin{array}{c}
\vphantom{\Big(}
[X,\tilde Y]_{\lambda_0,\dots,\lambda_k}(a_0,\dots ,a_{k-1},bc)
=[X,\tilde Y]_{\lambda_0,\dots,\lambda_{k-1},\lambda_k+\partial}
(a_0,\dots,a_{k-1},b)_\to c \\
\vphantom{\Big(}
+(-1)^{p(b)p(c)}
[X,\tilde Y]_{\lambda_0,\dots,\lambda_{k-1},\lambda_k+\partial}
(a_0,\dots,a_{k-1},c)_\to b\,.
\end{array}
\end{equation}
Recall, from Section \ref{sec:6.1}, that the left
and right box products defined in \eqref{101024:eq1} are such that
$[X,\tilde Y]=X\Box^L\tilde Y-(-1)^{\bar p(X)\bar p(Y)}\tilde Y\Box^R X$.
Since, by assumption, $X$ and $\tilde Y$
satisfy the Leibniz rule \eqref{100422c:eq1},
and using the sign identities \eqref{101109:eq1} and \eqref{101109:eq2},
we get, after a straightforward computation, that
\begin{equation}\label{101110:eq2}
\begin{array}{l}
\vphantom{\bigg(}
\displaystyle{
(X\Box^L\tilde Y)_{\lambda_0,\dots,\lambda_k}(a_0,\dots ,a_{k-1},bc)
= (X\Box^L\tilde Y)_{\lambda_0,\dots,\lambda_k+\partial}(a_0,\dots ,a_{k-1},b)_\to c
} \\
\vphantom{\bigg(}
\displaystyle{
+(-1)^{p(b)p(c)} (X\Box^L\tilde Y)_{\lambda_0,\dots,\lambda_k+\partial}(a_0,\dots ,a_{k-1},c)_\to b
} \\
\displaystyle{
+ \sum_{\substack{
i_0<\dots <i_{k-h}=k\\
i_{k-h+1}<\dots< i_k<k}}
\epsilon_{a_0,\dots,a_{k-1},bc}(i_0,\dots,i_k)
} \\
\displaystyle{
\Bigg\{
(-1)^{\bar p(X)(\bar p(\tilde Y)+\bar p(a_{i_0})+\dots+\bar p(a_{i_{k-h-1}})+p(b))}
\tilde Y_{\lambda_{i_0},\dots,\lambda_{i_{k-h-1}},\lambda_{i_{k-h}}+\dots+\lambda_{i_k}+\partial}
} \\
\vphantom{\bigg(}
\displaystyle{
\,\,\,(a_{i_0},\dots,a_{i_{k-h-1}},b)_\to
X_{-\lambda_{i_{k-h+1}}-\dots-\lambda_{i_k}-\partial,\lambda_{i_{k-h+1}},\dots,\lambda_{i_k}}
(c,a_{i_{k-h+1}},\dots,a_{i_k})
} \\
\vphantom{\bigg(}
\displaystyle{
+(\!-1\!)^{p(b)p(c)+\bar p(X)(\bar p(\tilde Y)+\bar p(a_{i_0})+\dots+\bar p(a_{i_{k-h-1}})+p(c))}
\tilde Y_{\lambda_{i_0},\dots,\lambda_{i_{k-h-1}},\lambda_{i_{k-h}}\!+\dots+\lambda_{i_k}\!+\partial}
} \\
\displaystyle{
\,\, (a_{i_0},\dots,a_{i_{k-h-1}},c)_\to
X_{-\lambda_{i_{k-h+1}}-\dots-\lambda_{i_k}-\partial,\lambda_{i_{k-h+1}},\dots,\lambda_{i_k}}
(b,a_{i_{k-h+1}},\dots,a_{i_k})
\!\Bigg\}.\!\!
}
\end{array}
\end{equation}
Similarly, for the right box product we have
\begin{equation}\label{101110:eq3}
\begin{array}{l}
\vphantom{\bigg(}
\displaystyle{
(\tilde Y\Box^RX)_{\lambda_0,\dots,\lambda_k}(a_0,\dots ,a_{k-1},bc)
= (\tilde Y\Box^RX)_{\lambda_0,\dots,\lambda_k+\partial}(a_0,\dots ,a_{k-1},b)_\to c
} \\
\vphantom{\bigg(}
\displaystyle{
+(-1)^{p(b)p(c)} (\tilde Y\Box^RX)_{\lambda_0,\dots,\lambda_k+\partial}(a_0,\dots ,a_{k-1},c)_\to b
} 
\end{array}%
\end{equation}%
\begin{equation*}%
\begin{array}{l}
\displaystyle{
+ \sum_{\substack{
i_0<\dots <i_{h}=k\\
i_{h+1}<\dots< i_k<k}}
\epsilon_{a_0,\dots,a_{k-1},bc}(i_0,\dots,i_k)
} \\
\displaystyle{
\Bigg\{
(-1)^{(\bar p(X)+\bar p(a_{i_0})+\dots+\bar p(a_{i_{h-1}})+p(b))
(p(c)+\bar p(a_{i_{h+1}})+\dots+\bar p(a_{i_{k}}))}
} \\
\vphantom{\bigg(}
\displaystyle{
\,\,\,\,\,\,\,\,\,
\times\tilde Y_{\lambda_{i_0}+\dots+\lambda_{i_{h}}+\partial,\lambda_{i_{h+1}},\dots,\lambda_{i_k}}
\,\,\,(c,a_{i_{h+1}},\dots,a_{i_{k}})_\to
} \\
\vphantom{\bigg(}
\displaystyle{
\,\,\,\,\,\,\,\,\,
\times X_{\lambda_{i_0},\dots,\lambda_{i_{h-1}},-\lambda_{i_0}-\dots-\lambda_{i_{h-1}}-\partial}
(a_{i_0},\dots,a_{i_{h-1}},b)
} \\
\vphantom{\bigg(}
\displaystyle{
+(-1)^{p(b)p(c)+(\bar p(X)+\bar p(a_{i_0})+\dots+\bar p(a_{i_{h-1}})+p(c))
(p(b)+\bar p(a_{i_{h+1}})+\dots+\bar p(a_{i_{k}}))}
} \\
\vphantom{\bigg(}
\displaystyle{
\,\,\,\,\,\,\,\,\,
\times\tilde Y_{\lambda_{i_0}+\dots+\lambda_{i_{h}}+\partial,\lambda_{i_{h+1}},\dots,\lambda_{i_k}}
\,\,\,(b,a_{i_{h+1}},\dots,a_{i_{k}})_\to
} \\
\vphantom{\bigg(}
\displaystyle{
\,\,\,\,\,\,\,\,\,
\times X_{\lambda_{i_0},\dots,\lambda_{i_{h-1}},-\lambda_{i_0}-\dots-\lambda_{i_{h-1}}-\partial}
(a_{i_0},\dots,a_{i_{h-1}},c)
\!\Bigg\}.\!\!
}
\end{array}
\end{equation*}
Note that, by the symmetry conditions on $X$ and $\tilde Y$,
the sum in the RHS of \eqref{101110:eq2} and the sum in the RHS of \eqref{101110:eq3}
differ by the factor $(-1)^{\bar p(X)\bar p(Y)}$.
Hence, combining \eqref{101110:eq2} and \eqref{101110:eq3}, we get \eqref{101110:eq1}.
This shows that we have a well-defined representation
of the Lie superalgebra $W^{\partial,\as}(\Pi\mc V)$
on $\tilde W^{\partial,\as}(\Pi\mc V)$.

Thanks to Proposition \ref{100517:prop}(b),
the action of $W^{\partial,\as}(\Pi\mc V)$ on $\tilde W^{\partial,\as}(\Pi\mc V)$
commutes with $\partial$ and
it is given by derivations of the $\lambda$-bracket.
To complete the proof of part (b) we only have to check that the Lie superalgebra action
of $W^{\partial,\as}(\Pi\mc V)$ on $\tilde W^{\partial,\as}(\Pi\mc V)$
is by derivations of the concatenation product, i.e.
\begin{equation}\label{20101115:eq1}
[X,\tilde Y\wedge\tilde Z] = [X,\tilde Y]\wedge\tilde Z
+(-1)^{\bar p(X) p(\tilde Y)}\tilde Y\wedge[X,\tilde Z]\,.
\end{equation}
This follows by the following two identities (similar to \eqref{100423z:eq3bis})
which can be checked directly:
\begin{equation}\label{100423z:eq3tris}
\begin{array}{rcl}
X\Box^L(\tilde Y\wedge\tilde Z) &=& (X\Box^L\tilde Y)\wedge\tilde Z
+(-1)^{\bar p(X)p(\tilde Y)}\tilde Y\wedge(\tilde Z\Box^L X)\,,\\
(\tilde X\wedge\tilde Y)\Box^R Z &=& \tilde X\wedge (\tilde Z\Box^R Z)
+(-1)^{p(\tilde Y)\bar p(Z)}(\tilde X\Box^R Z)\wedge \tilde Y\,.
\end{array}
\end{equation}
Finally, part (c) is immediate from Proposition \ref{100517:prop}.
\end{proof}

It follows from Propositions \ref{100422c:prop2}
and \ref{100422z:prop3} that, for any Poisson vertex superalgebra $\mc V$,
which is finitely generated as differential algebra,
we have a differential $d_X$ on the superspace $\tilde W^{\partial,\as}(\Pi \mc V)$,
given by the action of $X$ on this space,
where $X$ is the element in the set \eqref{100422c:eq3}
associated to the Lie conformal superalgebra structure on $\mc V$.
Moreover, by Proposition \ref{100422z:prop3}(b),
the differential $d_X$ is an odd derivation of the odd PVA structure
of $\tilde W^{\partial,\as}(\Pi \mc V)$.
We thus get the \emph{basic PVA cohomology complex} $(\tilde W^{\partial,\as}(\Pi \mc V),d_X)$.
By Proposition \ref{100422z:prop3}(c),
the canonical map $\tint$
is a homomorphism of cohomology complexes
$(\tilde W^{\partial,\as}(\Pi \mc V),d_X)\to (W^{\partial,\as}(\Pi \mc V),\ad X)$.


\subsection{The universal PVA $\tilde W^{\partial,\as}(\mc V)$}
\label{sec:7.3}

As in Section \ref{sec:3.2},
instead of $\tilde W^{\partial}(\Pi\mc V)$,
we may consider the universal ``generalized'' 
Lie conformal superalgebra $\tilde W^{\partial}(\mc V)$,
with parity $p$, and, inside it,
the full prolongation of $\RCder(\mc V)\subset \tilde W^{\partial}_0(\mc V)=\RCend(\mc V)$,
which we denote by
$$
\tilde W^{\partial,\as}(\mc V)=\bigoplus_{k=-1}^\infty \tilde W^{\partial,\as}_k(\mc V)\,.
$$

As in Proposition \ref{100422z:prop1}, one proves that
$\tilde W^{\partial,\as}_k(\mc V)$ consists of linear maps
$X:\,\mc V^{\otimes(k+1)}\to\mb F[\lambda_0,\dots,\lambda_k]\otimes \mc V$
satisfying the symmetry and sesquilinearity conditions,
and the Leibniz rule \eqref{100422d:eq1}.

As in Proposition \ref{100422z:prop3},
assuming that $\mc V$ is finitely generated as differential algebra,
one can define on $\tilde W^{\partial,\as}(\mc V)$ a structure of a Poisson vertex algebra,
where the $\lambda$-bracket comes 
from the  $\lambda$-bracket \eqref{20120801:eq1} on $\tilde W^{\partial}(\mc V)$,
and the commutative associative product is given by a concatenation product
as in \eqref{100422z:eq4}, with $\bar p$ replaced by $p$.
Moreover,
the representation \eqref{100517:eq2} of Lie superalgebra $W^{\partial}(\mc V)$
on $\tilde W^{\partial}(\mc V)$
induces a representation of its subalgebra $W^{\partial,\as}(\mc V)$
on the Poisson vertex superalgebra
$\tilde W^{\partial,\as}(\mc V)\subset \tilde W^{\partial}(\mc V)$,
acting by derivations of both
the concatenation product and the $\lambda$-bracket.

It follows from Proposition \ref{100422d:prop2} that,
for any odd Poisson vertex superalgebra $\mc V$,
which is finitely generated as differential algebra,
we have a differential $d_X$ on the superspace $\tilde W^{\partial,\as}(\mc V)$,
given by the action of $X$ via the representation of $W^{\partial,\as}(\mc V)$,
where $X$ in \eqref{100422d:eq3}
is associated to the Lie conformal superalgebra structure on $\Pi\mc V$.
We thus get a cohomology complex $(\tilde W^{\partial,\as}(\mc V),d_X)$.


\section{Algebras of differential functions and the variational complex}
\label{sec:8}

\subsection{Algebras of differential functions}
\label{sec:8.1}

An \emph{algebra of differential functions} $\mc V$
in one independent variable $x$ and $\ell$ dependent variables $u_i$,
indexed by the set $I=\{1,\dots,\ell\}$ ($\ell$ may be infinite),
is, by definition, a differential algebra
(i.e. a unital commutative associative algebra with a derivation $\partial$),
endowed with commuting derivations
$\frac{\partial}{\partial u_i^{(n)}}\,:\,\,\mc V\to\mc V$, for all $i\in I$ and $n\in\mb Z_+$,
such that, given $f\in\mc V$,
$\frac{\partial}{\partial u_i^{(n)}}f=0$ for all but finitely many $i\in I$ and $n\in\mb Z_+$,
and the following commutation rules with $\partial$ hold:
\begin{equation}\label{eq:comm_frac}
\Big[\frac{\partial}{\partial u_i^{(n)}} , \partial\Big] = \frac{\partial}{\partial u_i^{(n-1)}}\,,
\end{equation}
where the RHS is considered to be zero if $n=0$.
An equivalent way to write the identities \eqref{eq:comm_frac} is in terms of generating series:
\begin{equation}\label{eq:comm_frac_b}
\sum_{n\in\mb Z_+}z^n\frac{\partial}{\partial u_i^{(n)}}\circ \partial =
(z+\partial)\circ\sum_{n\in\mb Z_+}z^n\frac{\partial}{\partial u_i^{(n)}}\,.
\end{equation}
\begin{remark}\label{100503:rem}
It would be natural in this paper to consider a commutative differential superalgebra $\mc V$,
with an even derivation $\partial$.
However, we restricted ourselves to the purely even case for the sake of simplicity.
The generalization to the superalgebra case is straightforward.
\end{remark}

We call $\mc C=\ker(\partial)\subset\mc V$ the subalgebra of \emph{constant functions},
and we denote by $\mc F\subset\mc V$ the subalgebra of \emph{quasiconstant functions},
defined by
\begin{equation}\label{eq:4.2}
\mc F
=
\big\{f\in\mc V\,\big|\,\frac{\partial f}{\partial u_i^{(n)}}=0\,\,\forall i\in I,\,n\in\mb Z_+\big\}\,.
\end{equation}
It follows from \eqref{eq:comm_frac}
by downward induction that a constant function is quasiconstant: $\mc C\subset\mc F$.
Also, clearly, $\partial\mc F\subset\mc F$.
One says that $f\in\mc V$ has \emph{differential order} $n$ in the variable $u_i$
if $\frac{\partial f}{\partial u_i^{(n)}}\neq0$
and $\frac{\partial f}{\partial u_i^{(m)}}=0$ for all $m>n$.

Typical examples of algebras of differential functions are:
the ring of translation invariant differential polynomials,
$R_\ell\,=\,\mb F[u_i^{(n)}\,|\,i\in I,n\in\mb Z_+]$,
where $\partial(u_i^{(n)})=u_i^{(n+1)}$,
and the ring of differential polynomials,
$R_\ell[x]=\mb F[x,u_i^{(n)}\,|\,i\in I,n\in\mb Z_+]$,
where $\partial x=1$ and $\partial u_i^{(n)}=u_i^{(n+1)}$.
Other examples can be constructed starting from $R_\ell$ or $R_\ell[x]$
by taking a localization by some multiplicative subset $S$,
or an algebraic extension obtained by adding solutions of some polynomial equations,
or a differential extension obtained by adding solutions of some differential equations.
In all these examples,
and more generally in any algebra of differential functions extension of $R_\ell$,
the action of $\partial:\, \mc V\to\mc V$
is given by
$\displaystyle{
\partial=\frac\partial{\partial x}+\sum_{i\in I,n\in\mb Z_+} u_i^{(n+1)}\frac{\partial}{\partial u_i^{(n)}}
}$,
which implies that
\begin{equation}\label{101024:eqv1}
\mc F\cap\partial\mc V=\partial\mc F\,.
\end{equation}
Indeed, if $f\in\mc V$ has, in some variable $u_i$, differential order $n\geq0$,
then $\partial f$ has differential order $n+1$,
hence it does not lie in $\mc F$.

The \emph{variational derivative}
$\frac\delta{\delta u}:\,\mc V\to\mc V^{\oplus \ell}$
is defined by
\begin{equation}\label{eq:varder}
\frac{\delta f}{\delta u_i}\,:=\,\sum_{n\in\mb Z_+}(-\partial)^n\frac{\partial f}{\partial u_i^{(n)}}\,.
\end{equation}
It follows immediately from \eqref{eq:comm_frac_b} that
\begin{equation}\label{eq:mar11_2}
\frac{\delta}{\delta u_i}(\partial f) = 0\,,
\end{equation}
for every $i\in I$ and $f\in\mc V$,
namely, $\partial\mc V\subset\ker \frac{\delta}{\delta u}$.

A \emph{vector field} is, by definition, a derivation of $\mc V$ of the form
\begin{equation}\label{2006_X}
X=\sum_{i\in I,n\in\mb{Z}_+} P_{i,n} \frac{\partial}{\partial u_i^{(n)}}\,\,,  \quad P_{i,n} \in \mc V\,.
\end{equation}
We denote by $\Vect(\mc V)$ the space of all vector fields,
which is clearly a subalgebra of the Lie algebra $\Der(\mc V)$ of all derivations of $\mc V$.
A vector field $X$ is called \emph{evolutionary} if $[\partial,X]=0$,
and we denote by $\Vect^\partial(\mc V)$ the Lie subalgebra of all evolutionary vector fields.
Namely, $\Vect^\partial(\mc V)=\Vect(\mc V)\cap\Der^\partial(\mc V)$.
By \eqref{eq:comm_frac}, a vector field $X$ is evolutionary
if and only if it has the form
\begin{equation}\label{2006_X2}
X_P=\sum_{i\in I,n\in\mb{Z}_+} (\partial^n P_i) \frac{\partial}{\partial u_i^{(n)}}\,,
\end{equation}
where $P=(P_i)_{i\in I}\in\mc V^\ell$, is called the \emph{characteristic} of $X_P$.
As in \cite{BDSK}, we denote by $\mc V^\ell$ the space of $\ell\times1$ column vectors with entries in $\mc V$,
and by $\mc V^{\oplus\ell}$ the subspace of $\ell\times1$ column vectors with only finitely many
non-zero entries.

\subsection{de Rham complex $\tilde\Omega^\bullet(\mc V)$ and variational complex $\Omega^\bullet(\mc V)$}
\label{sec:8.2}

Here we describe the explicit construction of the complex of variational calculus
following \cite{DSK2}.

Recall that the \emph{de Rham complex} $\tilde\Omega^\bullet(\mc V)$
is defined as the
free commutative superalgebra over $\mc V$
with odd generators $\delta u_i^{(n)},\,i\in I,n\in\mb{Z}_+$
and the differential $\delta$ defined further.
The algebra $\tilde\Omega^\bullet(\mc V)$ consists of finite sums of the form
\begin{equation}\label{eq:apr24_1}
\tilde\omega=\sum_{
\substack{i_1,\dots,i_k\in I \\ m_1,\dots,m_k\in\mb{Z}_+}
}
P^{m_1\dots m_k}_{i_1\dots i_k}\,
\delta u_{i_1}^{(m_1)}\wedge\dots\wedge\delta u_{i_k}^{(m_k)}
\,\,,  \quad P^{m_1\dots m_k}_{i_1\dots i_k} \in \mc V\,.
\end{equation}
We have a natural $\mb Z_+$-grading
$\tilde\Omega^\bullet(\mc V)=\bigoplus_{k\in\mb Z_+}\tilde\Omega^k(\mc V)$
defined by letting elements in $\mc V$ have degree 0,
while the generators $\delta u_i^{(n)}$ have degree 1.
The space $\tilde\Omega^k(\mc V)$ is a free module over $\mc V$ with
a basis consisting of the elements
$\delta u_{i_1}^{(m_1)}\wedge\dots\wedge\delta u_{i_k}^{(m_k)}$,
with $(m_1,i_1)>\dots>(m_k,i_k)$ (with respect to the lexicographic order).
In particular $\tilde\Omega^0(\mc V)=\mc V$
and $\tilde\Omega^1(\mc V)=\bigoplus_{i\in I,n\in\mb Z_+}\mc V\delta u_i^{(n)}$.

We let $\delta$ be an odd derivation of degree 1 of $\tilde\Omega^\bullet(\mc V)$,
such that $\delta f=\sum_{i\in I,\,n\in\mb Z_+}\frac{\partial f}{\partial u_i^{(n)}}\delta u_i^{(n)}$
for $f\in\mc V$, and $\delta(\delta u_i^{(n)})=0$.
It is immediate to check that $\delta^2=0$ and that,
for $\tilde\omega\in\tilde\Omega^k$ as in \eqref{eq:apr24_1}, we have
\begin{equation}\label{eq:apr24_2}
\delta(\tilde\omega)
\,=\,
\sum_{j\in I,n\in\mb Z_+}\sum_{
\substack{i_1,\dots,i_k\in I \\ m_1,\dots,m_k\in\mb{Z}_+}
}
\frac{\partial P^{m_1\dots m_k}_{i_1\dots i_k}}{\partial u_j^{(n)}}\,
\delta u_j^{(n)}\wedge \delta u_{i_1}^{(m_1)}\wedge\dots\wedge\delta u_{i_k}^{(m_k)}\,.
\end{equation}

The superspace $\tilde\Omega^\bullet(\mc V)$ has a structure of an $\mb F[\partial]$-module,
where $\partial$ acts as an even derivation of the wedge product,
which extends the action on $\mc V=\tilde\Omega^0(\mc V)$, and commutes with $\delta$.
Since $\partial$ commutes with $\delta$,
we may consider the corresponding reduced complex
$\Omega^\bullet(\mc V)=\tilde\Omega^\bullet(\mc V)/\partial\tilde\Omega^\bullet(\mc V)
=\bigoplus_{k\in\mb Z_+}\Omega^k(\mc V)$,
known as the \emph{variational complex}.
By an abuse of notation, we denote by $\delta$ the corresponding differential
on $\Omega^\bullet(\mc V)$.


We identify the space $\tilde\Omega^k(\mc V)$
with the space of \emph{skewsymmetric arrays}, i.e. arrays of polynomials
\begin{equation}\label{100518:eq2}
P=\big(P_{i_1,\dots,i_k}(\lambda_1,\dots,\lambda_k)\big)_{i_1,\dots,i_k\in I}\,,
\end{equation}
where $P_{i_1,\dots,i_k}(\lambda_1,\dots,\lambda_k)\in\mb F[\lambda_1,\dots,\lambda_k]\otimes\mc V$
are zero for all but finitely many choices of indexes,
and are skewsymmetric with respect to simultaneous permutations of the variables $\lambda_1,\dots,\lambda_k$
and the indexes $i_1,\dots,i_k$.
The identification is obtained by associating $P$ in \eqref{100518:eq2}
to $\tilde\omega$ in \eqref{eq:apr24_1}, where $P^{m_1,\dots,m_k}_{i_1,\dots,i_k}$
is the coefficient of $\lambda_1^{m_1}\dots\lambda_k^{m_k}$ in $P_{i_1,\dots,i_k}(\lambda_1,\dots,\lambda_k)$.
The formula for the differential $\delta:\,\tilde\Omega^k(\mc V)\to\tilde\Omega^{k+1}(\mc V)$
gets translated as follows:
\begin{equation}\label{100518:eq3}
(\delta P)_{i_0,\dots,i_k}(\lambda_0,\dots,\lambda_k)
=
\sum_{\alpha=0}^k(-1)^\alpha \sum_{n\in\mb Z_+}
\frac{\partial
P_{i_0,\stackrel{\alpha}{\check{\dots}},i_k}
(\lambda_0,\stackrel{\alpha}{\check{\dots}},\lambda_k)}
{\partial u_{i_\alpha}^{(n)}}
\lambda_\alpha^n \,.
\end{equation}

In this language the $\mb F[\partial]$-module structure of $\tilde\Omega^\bullet(\mc V)$
is given by
\begin{equation}\label{100518:eq4}
(\partial P)_{i_1,\dots,i_k}(\lambda_1,\dots,\lambda_k)
=
(\partial+\lambda_1+\dots+\lambda_k)
P_{i_1,\dots,i_k}(\lambda_1,\dots,\lambda_k)\,,
\end{equation}
so that the reduced space $\Omega^k(\mc V)=\tilde\Omega^k(\mc V)/\partial\tilde\Omega^k(\mc V)$
gets naturally identified with the space of arrays \eqref{100518:eq2},
where $P_{i_1,\dots,i_k}(\lambda_1,\dots,\lambda_k)$
are considered as elements of $\mb F_-[\lambda_1,\dots,\lambda_k]\otimes_{\mb F[\partial]}\mc V$.
The differential $\delta$ on $\Omega^\bullet(\mc V)$ is given by the same formula \eqref{100518:eq3}.

For example, $\Omega^0(\mc V)=\mc V/\partial\mc V$,
and $\Omega^1(\mc V)$
is naturally identified with $\mc V^{\oplus\ell}$,
thanks to the canonical isomorphism $\mb F[\lambda]\otimes_{\mb F[\partial]}\mc V\simeq\mc V$.
Under these identifications,
the map $\delta:\,\mc V/\partial\mc V\to\mc V^{\oplus\ell}$ coincides with the variational
derivative \eqref{eq:varder}:
$$
\delta(\tint f)=\frac{\delta f}{\delta u}\,,
$$
where, as in the previous sections,
we denote by $f\mapsto\tint f$ the canonical quotient map $\mc V\to\mc V/\partial\mc V$.
Furthermore, $\Omega^2(\mc V)$ is naturally identified with the space
of skewadjoint $\ell\times\ell$-matrix differential operators
$S(\partial)=\big(S_{ij}(\partial)\big)_{i,j\in I}$.
This identification is obtained by mapping
$P=\big(\sum_{m,n\in\mb Z_+}P^{m,n}_{i,j}\lambda^m\mu^n\big)_{i,j\in I}$
to the operator $S(\partial)$
given by $S_{ij}(\partial)=\sum_{m,n\in\mb Z_+}(-\partial)^n\circ P_{i,j}^{m,n}\partial^m$.
With these identifications,
formula \eqref{100518:eq3} for the differential
of $F\in\mc V^{\oplus\ell}=\Omega^1(\mc V)$ becomes
$$
\delta F = - D_F(\partial) + D_F^*(\partial)\,,
$$
where
\begin{equation}\label{frechet}
D_F(\partial)=\Big(\sum_{n\in\mb Z_+}\frac{\partial F_i}{\partial u_j}\partial^n\Big)_{i,j\in I}
\end{equation}
is the \emph{Frechet derivative} of $F$,
and $D_F^*(\partial)$ is the adjoint differential operator.

\subsection{Exactness of the variational complex}
\label{sec:8.3}

Recall from \cite{BDSK} that an algebra of differential functions $\mc V$
is called \emph{normal} if we have
$\frac\partial{\partial u_i^{(m)}}\big(\mc V_{m,i}\big)=\mc V_{m,i}$
for all $i\in I,m\in\mb Z_+$,
where we let
\begin{equation}\label{eq:july21_1}
\mc V_{m,i}\,:=\,\Big\{ f\in\mc V\,\Big|\,
\frac{\partial f}{\partial u_j^{(n)}}=0\,\,
\text{ if } (n,j)>(m,i) \text{ in lexicographic order }\Big\}\,.
\end{equation}
We also denote $\mc V_{m,0}=\mc V_{m-1,\ell}$, and $\mc V_{0,0}=\mc F$.

The algebras $R_\ell$ and $R_\ell[x]$ are obviously normal.
Moreover, any their extension $\mc V$ can be further extended to a normal algebra.
Conversely, in \cite{DSK2} it is proved that any normal algebra of differential functions $\mc V$
is automatically a differential algebra extension of $R_\ell$.

In \cite{BDSK} we proved the following result (see also \eqref{101024:eqv1}):
\begin{theorem}\label{100430:th}
If $\mc V$ is a normal algebra of differential functions, then
\begin{enumerate}[(a)]
\item
$H^k(\tilde\Omega^\bullet(\mc V),\delta)=0$ for $k\geq1$,
and $H^0(\tilde\Omega^\bullet(\mc V),\delta)=\mc F$,
\item
$H^k(\Omega^\bullet(\mc V),\delta)=0$ for $k\geq1$,
and $H^0(\Omega^\bullet(\mc V),\delta)=\mc F/\partial\mc F$.

\noindent In particular, $\frac{\delta f}{\delta u}=0$ if and only if $f\in\partial\mc V+\mc F$,
and $F\in\mc V^{\oplus\ell}$ is in the image of $\frac\delta{\delta u}$
if and only if
its Frechet derivative $D_F(\partial)$ is selfadjoint.
\end{enumerate}
\end{theorem}


\section{The Lie superalgebra of variational polyvector fields and PVA cohomology}
\label{sec:9}

Let $\mc V$ be an algebra of differential functions
extension of the algebra of differential polynomials
$R_\ell=\mb F[u_i^{(n)}\,|\,i\in I,n\in\mb Z_+]$.
Recall from Section \ref{sec:5.1} the $\mb Z$-graded Lie superalgebra $W^{\partial,\as}(\Pi\mc V)$,
obtained as a prolongation of the Lie algebra $\Der^\partial(\mc V)$
of derivations of $\mc V$, commuting with $\partial$,
in the universal Lie superalgebra $W^\partial(\Pi\mc V)$.
In Section \ref{sec:9.1}
we introduce a smaller $\mb Z$-graded subalgebra of $W^\partial(\Pi\mc V)$,
which we call the Lie superalgebra of variational polyvector fields,
denoted by $W^{\var}(\Pi \mc V)=\bigoplus_{k=-1}^\infty W^{\var}_k$.
It is obtained as a prolongation of the Lie subalgebra of evolutionary vector fields
$\Vect^\partial(\mc V)\subset\Der^\partial(\mc V)$, introduced in Section \ref{sec:8.1}.
We then identify in Section \ref{sec:9.3}
the space $W^{\var}(\Pi \mc V)$ with the space $\Omega^\bullet(\mc V)$
introduced in Section \ref{sec:8.2},
and we relate the corresponding cohomology complexes.

\subsection{The Lie superalgebra of variational polyvector fields $W^{\var}(\Pi\mc V)$}
\label{sec:9.1}

Recall that the superspace $W^\partial_k(\Pi\mc V)$, of parity $k$ mod $2$, consists of maps
$X:\,\mc V^{\otimes(k+1)}\to\mb F_-[\lambda_0,\dots,\lambda_k]\otimes_{\mb F[\partial]}\mc V$,
satisfying sesquilinearity:
\begin{equation}\label{100530:eq1}
X_{\lambda_0,\dots,\lambda_k}(f_0,\dots\partial f_i\dots,f_k)
=-\lambda_i X_{\lambda_0,\dots,\lambda_k}(f_0,\dots,f_k)
\,\,,\,\,\,\,
i=0,\dots,k\,,
\end{equation}
and skewsymmetry:
\begin{equation}\label{100530:eq2}
X_{\lambda_{\sigma(0)},\dots,\lambda_{\sigma(k)}}(f_{\sigma(0)},\dots,f_{\sigma(k)})
=\sign(\sigma)X_{\lambda_0,\dots,\lambda_k}(f_0,\dots,f_k)
\,\,,\,\,\,\,
\sigma\in S_{k+1}\,.
\end{equation}

In the special case when $\mc V=R_\ell=\mb F[u_i^{(n)}\,|\,i\in I,n\in\mb Z_+]$,
the Leibniz rule \eqref{100422c:eq1} implies the following \emph{master equation}
for an element $X\in W^{\partial,\as}_k(\Pi\mc V)$,
which expresses the action of $X$ on $\mc V^{\otimes(k+1)}$
in terms of its action on $(k+1)$-tuples of generators:
\begin{equation}\label{100430:eq2}
\begin{array}{l}
\displaystyle{
X_{\lambda_0,\dots,\lambda_k}(f_0,\dots,f_k)
=
\sum_{\substack{i_0,\dots,i_k\in I \\ m_0,\dots,m_k\in\mb Z_+}}
\bigg(e^{\partial\partial_{\lambda_0}}\frac{\partial f_0}{\partial u_{i_0}^{(m_0)}}\bigg)
\dots
}\\
\displaystyle{
\,\,\,\,\,\,\,\,\,\dots
\bigg(e^{\partial\partial_{\lambda_k}}\frac{\partial f_k}{\partial u_{i_k}^{(m_k)}}\bigg)
(-\lambda_0)^{m_0}\dots(-\lambda_k)^{m_k}
X_{\lambda_0,\dots,\lambda_k}(u_{i_0},\dots,u_{i_k})\,.
}
\end{array}
\end{equation}
Here we are using the following formula:
\begin{equation}\label{100624:eq1}
\big(e^{\partial\partial_\lambda}f\big)P(\lambda)
:= \sum_{n\in\mb Z_+}\frac1{n!} (\partial^n f)\frac{\partial^n P(\lambda)}{\partial\lambda^n}
= P(\lambda+\partial)_\to f\,.
\end{equation}
In general, for an arbitrary algebra of differential functions $\mc V$ containing $R_\ell$,
we define the space $W^{\var}_k$ of \emph{variational} $k$-\emph{vector fields}
as the subspace of $W^\partial_k(\Pi\mc V)$ consisting of elements $X$
satisfying the master equation \eqref{100430:eq2}.
By \eqref{100624:eq1},
another form of equation \eqref{100430:eq2} is the following (for each $s=0,\dots,k$):
\begin{equation}\label{100502:eq1}
\begin{array}{l}
\displaystyle{
X_{\lambda_0,\dots,\lambda_k}(f_0,\dots,f_k)
} \\
\displaystyle{
=
\sum_{i\in I,\,m\in\mb Z_+}
X_{\lambda_0,\dots,\lambda_s+\partial,\dots,\lambda_k}(f_0,\dots,\stackrel{s}{\check{u_i}},\dots,f_k)_\to
(-\lambda_s-\partial)^{m} \frac{\partial f_s}{\partial u_i^{(m)}}
\,,
}
\end{array}
\end{equation}
where $\stackrel{s}{\check{u_i}}$ means that $u_i$ is put in place of $f_s$.

Note that the master equation implies that $X$ satisfies the Leibniz rule \eqref{100422c:eq1}.
Thus, $W^{\var}_k$ is a subspace of $W^{\partial,\as}_k(\Pi\mc V)$.

\begin{proposition}\label{100430:prop}
The superspace $W^{\var}(\Pi\mc V)$ is a $\mb Z$-graded subalgebra
of the Lie superalgebra $W^{\partial,\as}(\Pi\mc V)$.
\end{proposition}
\begin{proof}
The proof of this statement is similar to that of Proposition \ref{100422c:prop1}.
We need to prove that, if $X\in W^{\var}_h$ and $Y\in W^{\var}_{k-h}$,
then $[X,Y]=X\Box Y-(-1)^{h(k-h)}Y\Box X$ lies in $W^{\var}_k$,
namely, it satisfies the master equation \eqref{100430:eq2},
or, equivalently, equation \eqref{100502:eq1} for $s=0,\dots,k$.
We observe that, since $[X,Y]_{\lambda_0,\dots,\lambda_k}(f_0,\dots,f_k)$
is skewsymmetric with respect to simultaneous permutations of the variables $\lambda_i$
and the elements $f_i$,
it suffices to prove that $[X,Y]$ satisfies equation \eqref{100502:eq1} for $s=0$.
We have, by a straightforward computation,
$$
\begin{array}{l}
\displaystyle{
\vphantom{\Big(}
\big(X\Box Y\big)_{\lambda_0,\dots,\lambda_k}(f_0,f_1, \dots, f_k)
} \\
\displaystyle{
\vphantom{\Big(}
= \sum_{j\in I,n\in\mb Z_+}
\big(X\Box Y\big)_{\lambda_0+\partial,\lambda_1,\dots,\lambda_k}(u_j,f_1, \dots, f_k)_\to
(-\lambda_0-\partial)^n\frac{\partial f_0}{\partial u_j^{(n)}}
} \\
\displaystyle{
+ \sum_{i,j\in I,\,m,n\in\mb Z_+}
\sum_{\substack{
i_1<\dots <i_{k-h}\\
i_{k-h+1}<\dots< i_k}}
\pm
\frac{\partial^2 f}{\partial u_i^{(m)}\partial u_j^{(n)}}
} \\
\displaystyle{
\,\,\,\,\,\,\,\,\,\,\,\,\,\,\,\,\,\,\,\,\,\,\,\,\,\,\,
\times\Big(
(\nu+\partial)^n
X_{-\nu-\partial,\lambda_{i_{k-h+1}},\dots,\lambda_{i_k}}
(u_j,f_{i_{k-h+1}},\dots,f_{i_k})
\Big)
} \\
\displaystyle{
\,\,\,\,\,\,\,\,\,\,\,\,\,\,\,\,\,\,\,\,\,\,\,\,\,\,\,\,\,\,\,\,\,\,\,\,\,\,\,
\times
\Big(
(\mu+\partial)^m
Y_{-\mu-\partial,\lambda_{i_1},\dots,\lambda_{i_{k-h}}}
(u_i,f_{i_1},\dots,f_{i_{k-h}})
\Big)\,,
}
\end{array}
$$
where $\pm$ is the sign of the permutation $(i_1,\dots,i_k)$ of the set $\{1,\dots,k\}$,
$\mu=\lambda_{i_1}+\dots+\lambda_{i_{k-h}}$ and $\nu=\lambda_{i_{k-h+1}}+\dots+\lambda_{i_k}$.
To complete the proof we just observe that
the second sum in the RHS is supersymmetric in $X$ and $Y$: if we exchange them it gets multiplied by
$(-1)^{h(k-h)}$.
Hence it does not contribute to $[X,Y]$.
\end{proof}

We can describe explicitly the spaces $W^{\var}_k$ for $k=-1,0,1$.
Clearly, $W^{\var}_{-1}=\mc V/\partial\mc V$.
Identifying $\mb F_-[\lambda]\otimes_{\mb F[\partial]}\mc V$ with $\mc V$,
the master equation \eqref{100430:eq2} for $X\in W^{\var}_0$ reads
$X(f)=\sum_{i\in I,n\in\mb Z_+}(\partial^nX(u_i))\frac{\partial f}{\partial u_i^{(n)}}$,
i.e., $X$ is an evolutionary vector field, see \eqref{2006_X2}.
Hence $W^{\var}_0=\Vect^\partial(\mc V)$.
Next, recall from Section \ref{sec:4.1} that $W^\partial_1(\Pi\mc V)$
is identified with the space of sesquilinear skewcommutative $\lambda$-brackets
$\{\cdot\,_\lambda\,\cdot\}:\,\mc V\otimes\mc V\to\mb F[\lambda]\otimes\mc V$.
For $X\in W^\partial_1(\Pi\mc V)$, the corresponding $\lambda$-bracket
is $\{f_\lambda g\}=X_{\lambda,-\lambda-\partial}(f,g)$.
Under this identification, the master equation \eqref{100430:eq2} translates into the usual
formula for $\lambda$-brackets (cf. \cite{DSK1}):
\begin{equation}\label{100502:eq2}
\{f_\lambda g\}=
\sum_{i,j\in I,m,n\in\mb Z_+}
\frac{\partial g}{\partial u_j^{(n)}}(\lambda+\partial)^n
\{{u_i}_{\lambda+\partial}u_j\}_\to
(-\lambda-\partial)^m\frac{\partial f}{\partial u_i^{(m)}}\,.
\end{equation}
Hence, $W^{\var}_1$ is identified with the space of skewcommutative $\lambda$-brackets on $\mc V$
satisfying equation \eqref{100502:eq2}.

We can also write some explicit formulas for the Lie brackets in $W^{\var}(\Pi\mc V)$.
Recall that $W^{\var}_{-1}$ is an abelian subalgebra.
If $X\in W^{\var}_k,\,k\geq0$, and $\tint f_0\in\mc V/\partial\mc V= W^{\var}_{-1}$,
recalling equation \eqref{100419:eq4} and applying the master equation \eqref{100502:eq1}, we have,
\begin{equation}\label{100503:eq1}
[X,\tint f]_{\lambda_1,\dots,\lambda_k}(f_1,\dots,f_k)
=
\sum_{i\in I}X_{\partial,\lambda_1,\dots,\lambda_k}(u_i,f_1,\dots,f_k)_\to\frac{\delta f}{\delta u_i}\,.
\end{equation}
In particular, for $k=0$, recalling that $W^{\var}_0=\Vect^\partial(\mc V)$
is identified with $\mc V^\ell$ via $P\mapsto X_P$ (cf. \eqref{2006_X2}), we have
\begin{equation}\label{100503:eq2}
[X_P,\tint f]
=
\tint X_P(f)
=
\sum_{i\in I}\tint P_i\frac{\delta f}{\delta u_i}\,.
\end{equation}
Moreover, for $k=1$, recalling that $W^{\var}_1$ is identified with the space
of skewcommutative $\lambda$-brackets, we have
\begin{equation}\label{100503:eq3}
[\{\cdot\,_\lambda\,\cdot\},\tint f](g)
=
\{f_\lambda g\}\,\big|_{\lambda=0}\,.
\end{equation}
We next identify skewcommutative $\lambda$-brackets with skewadjoint
differential operators
by associating to a given $\lambda$-bracket $\{\cdot\,_\lambda\,\cdot\}$
the differential operator
$H(\partial)=\big(H_{ij}(\partial)\big)_{i,j\in I}:\,\mc V^{\oplus\ell}\to\mc V^\ell$,
given by
\begin{equation}\label{100430:eq1b}
H_{ij}(\partial)=\{{u_j}_\partial u_i\}_\to\,.
\end{equation}
Furthermore, identifying the space of evolutionary vector fields
with $\mc V^\ell$, we can rewrite \eqref{100503:eq3} as follows
\begin{equation}\label{100503:eq4}
[H(\partial),\tint f]
=
H(\partial)\frac{\delta f}{\delta u}
\,\,\,\Big(=
X_{H(\partial)\frac{\delta f}{\delta u}}
\Big)\,.
\end{equation}
Next, let $X\in\Vect^\partial(\mc V)=W^{\var}_0$ and $Y\in W^{\var}_k,\,k\geq0$.
We have, by \eqref{100419:eq5},
\begin{equation}\label{100503:eq5}
\begin{array}{c}
[X,Y]_{\lambda_0,\dots,\lambda_k}(f_0,\dots,f_k)=
X\big(Y_{\lambda_0,\dots,\lambda_k}(f_0,\dots,f_k)\big)
\\
\displaystyle{
-\sum_{i=0}^k Y_{\lambda_0,\dots,\lambda_k}(f_0,\dots X(f_i)\dots,f_k)\,.
}
\end{array}
\end{equation}
%
%
For $k=0$ this reduces to the usual commutator of evolutionary vector fields:
\begin{equation}\label{100503:eq7}
[P,Q]_i=[X_P,X_Q](u_i)
=X_P(Q_i)-X_Q(P_i)\,,
\end{equation}
while for $k=1$ it gives, identifying $H\in W^{\var}_1$ with the corresponding
skewcommutative $\lambda$-bracket $\{\cdot\,_\lambda\,\cdot\}_H$,
\begin{equation}\label{100503:eq8}
\{f_\lambda g\}_{[X,H]}
=
X\big(\{f_\lambda g\}_H\big)
-\{X(f)_\lambda g\}_H-\{f_\lambda X(g)\}_H\,,
\end{equation}
or equivalently, restricting to generators and using the notation in \eqref{100430:eq1b},
\begin{equation}\label{100503:eq9}
\begin{array}{l}
[X_P,H]_{ij}(\lambda)
=
\displaystyle{
X_P\big(H_{ij}(\lambda)\big)
} \\
\displaystyle{
-\sum_{k\in I,n\in\mb Z_+}
H_{ik}(\lambda+\partial)(-\lambda-\partial)^n\frac{\partial P_j}{\partial u_k^{(n)}}
-\sum_{k\in I,n\in\mb Z_+}
\frac{\partial P_i}{\partial u_k^{(n)}}(\lambda+\partial)^n H_{kj}(\lambda)
} \\
\displaystyle{
=
X_P\big(H_{ij}(\lambda)\big)
-\sum_{k\in I}
H_{ik}(\lambda+\partial)\big(D_P^*(\lambda)\big)_{kj}
-\sum_{k\in I}
\big(D_P(\lambda+\partial)\big)_{ik} H_{kj}(\lambda) \,.
}
\end{array}
\end{equation}
In the last identity we used the definition \eqref{frechet} of the Frechet derivative $D_P(\partial)$.
Equivalently, in terms of differential operators, we have
\begin{equation}\label{100503:eq10}
[X_P,H](\partial)
=
X_P(H)(\partial)
- H(\partial)\circ D_P^*(\partial)
- D_P(\partial)\circ H(\partial) \,,
\end{equation}
where $X_P(H)$ means applying the derivation $X_P$ to the coefficients
of the differential operator $H(\partial)$.
Finally, let $H\in W^{\var}_1$ be associated to the skewcommutative
$\lambda$-bracket $\{\cdot\,_\lambda\,\cdot\}_H$,
and let $Y\in W^{\var}_{k-1},\,k\geq1$.
We have, by \eqref{100419:eq6},
\begin{equation}\label{100503:eq11}
\begin{array}{c}
\displaystyle{
[H,Y]_{\lambda_0,\dots,\lambda_k}(f_0,\dots,f_k)
= (-1)^{k+1}\bigg(
\sum_{i=0}^k (-1)^{i}
\big\{{f_i}_{\lambda_i}
 Y_{\lambda_0,\stackrel{i}{\check{\dots}},\lambda_k}(f_0,\stackrel{i}{\check{\dots}},f_k)\big\}_H
} \\
\displaystyle{
+ \sum_{0\leq i<j\leq k}
(-1)^{i+j}
Y_{\lambda_i+\lambda_j,\lambda_0,\stackrel{i}{\check{\dots}}\stackrel{j}{\check{\dots}},\lambda_k}
\big(\{{f_i}_{\lambda_i}f_j\}_H,f_0,\stackrel{i}{\check{\dots}}\,\stackrel{j}{\check{\dots}},f_k\big)
\bigg)\,.
}
\end{array}
\end{equation}
In particular, for $k=2$,
\begin{equation}\label{100503:eq12}
\begin{array}{l}
[H,K]_{\lambda,\mu,\nu}(f,g,h) \\
\displaystyle{
\vphantom{\Big(}
= \{{\{f_{\lambda}g\}_K}_{\lambda+\mu} h\}_H
-\{f_{\lambda}\{g_{\mu}h\}_K\}_H
+\{g_{\mu}\{f_{\lambda}h\}_K\}_H
}\\
\vphantom{\Big(}
\displaystyle{
+ \{{\{f_{\lambda}g\}_H}_{\lambda+\mu} h\}_K
-\{f_{\lambda}\{g_{\mu}h\}_H\}_K
+\{g_{\mu}\{f_{\lambda}h\}_H\}_K
\,.
}
\end{array}
\end{equation}

\begin{remark}\label{100502:rem}
As we pointed out above, $W^{\var}=W^{\partial,\as}$ if $\mc V=R_\ell$.
This is not always the case for an extension of $R_\ell$.
For example, consider the algebra of differential functions in one variable
$\mc V=R_1[e^u]$, with $\partial e^u=e^uu^\prime$.
In this case, $W^{\var}_0=\Vect^\partial(\mc V)$,
while $W^{\partial,\as}_0(\Pi\mc V)=\Vect^\partial(\mc V)+\mb F Z$,
where $Z$ is the derivation of $\mc V$ commuting with $\partial$
given by: $Z(Pe^{mu})=mPe^{mu},\,P\in R_1,\,m\in\mb Z_+$.
\end{remark}

\subsection{PVA structures on an algebra of differential functions and cohomology complexes}
\label{sec:9.2}

Let $K\in W^{\var}_1$ be such that $[K,K]=0$,
and denote by $\{\cdot\,_\lambda\,\cdot\}_K$
the corresponding Poisson $\lambda$-bracket on $\mc V$,
and by $K(\partial)$ the corresponding Hamiltonian operator given by \eqref{100430:eq1b}.
Then $(\ad K)^2=0$, hence we can consider
the associated \emph{Poisson cohomology complex} $(W^{\var}(\Pi\mc V),\ad K)$.
Let $\mc Z_K^\bullet(\mc V)=\bigoplus_{k=-1}^\infty\mc Z_K^k$,
where $\mc Z_K^k=\ker\big(\ad K\big|_{W^{\var}_k}\big)$,
and $\mc B_K^\bullet(\mc V)=\bigoplus_{k=-1}^\infty\mc B_K^k$,
where $\mc B_K^k=(\ad K)\big(W^{\var}_{k-1}\big)$.
Then $\mc Z_K^\bullet(\mc V)$ is a $\mb Z$-graded subalgebra
of the Lie superalgebra $W^{\var}(\Pi\mc V)$,
and $\mc B_K^\bullet(\mc V)\subset\mc Z_K^\bullet(\mc V)$ is an ideal.
Hence, the corresponding \emph{Poisson cohomology}
$$
\mc H_K^\bullet(\mc V)=\bigoplus_{k=-1}^\infty \mc H_K^k
\,\,,\,\,\,\,
\mc H_K^k=\mc Z_K^k\big/\mc B_K^k\,,
$$
is a $\mb Z$-graded Lie superalgebra.

By equation \eqref{100503:eq4}, We have
$$
\begin{array}{c}
\vphantom{\Big(}
\mc H_K^{-1}=\mc Z_K^{-1}
= \Big\{\tint f\in\mc V/\partial\mc V\,\big|\, K(\partial)\frac{\delta f}{\delta u}=0\Big\} \\
\vphantom{\Big(}
= \Big\{\tint f\in\mc V/\partial\mc V\,\big|\,\,\,\,\,\, \{f _\lambda \mc V\}_K\big|_{\lambda=0}=0\Big\} \,.
\end{array}
$$
%

Next, recalling that $W^{\var}_0=\Vect^\partial(\mc V)$,
we have
$$
\mc B_K^0
=\Big\{X_{K(\partial)\frac{\delta f}{\delta u}}
\,\,\big|\,\,
\tint f\in\mc V/\partial\mc V\Big\}
=\Big\{\{f_\lambda\,\cdot\}_K\,\big|_{\lambda=0}
\,\,\big|\,\,
\tint f\in\mc V/\partial\mc V\Big\}\,,
$$
where $X_P\in\Vect^\partial(\mc V)$ was defined in \eqref{2006_X2}.
Moreover, recalling equation \eqref{100503:eq8}, we get
$$
\mc Z_K^0=
\Big\{
X\in\Vect{}^\partial(\mc V)
\,\,\big|\,\,
X(\{f_\lambda g\}_K)
=\{X(f)_\lambda g\}_K+\{f_\lambda X(g)\}_K,\,f,g\in\mc V
\Big\}
$$
or, identifying $\Vect^\partial(\mc V)=\mc V^\ell$ via \eqref{2006_X2}, we have by \eqref{100503:eq10}
$$
\mc Z_K^0=
\Big\{
P\in\mc V^\ell
\,\,\big|\,\,
X_P(K(\partial)) = K(\partial)\circ D_P^*(\partial) + D_P(\partial)\circ K(\partial)
\Big\}\,.
$$
Finally, recalling \eqref{100503:eq11} we get that, for $k\geq1$, $\mc Z^k_K$ consists of elements
$X\in W^{\var}_{k}$ satisfying the following equation:
$$
\begin{array}{l}
\displaystyle{
\sum_{i=0}^{k+1} (-1)^{i}
\big\{{f_i}_{\lambda_i}
 X_{\lambda_0,\stackrel{i}{\check{\dots}},\lambda_{k+1}}(f_0,\stackrel{i}{\check{\dots}},f_{k+1})\big\}_K
} \\
\displaystyle{
+ \sum_{0\leq i<j\leq k+1}
(-1)^{i+j}
X_{\lambda_i+\lambda_j,\lambda_0,\stackrel{i}{\check{\dots}}\stackrel{j}{\check{\dots}},\lambda_{k+1}}
\big(\{{f_i}_{\lambda_i}f_j\}_K,f_0,\stackrel{i}{\check{\dots}}\,\stackrel{j}{\check{\dots}},f_{k+1}\big)
=0\,,
}
\end{array}
$$
and $\mc B_K^k$ consists of elements $X\in W^{\var}_k$ of the form
$$
\begin{array}{c}
\displaystyle{
X_{\lambda_0,\dots,\lambda_k}(f_0,\dots,f_k)
=
\sum_{i=0}^k (-1)^{i}
\big\{{f_i}_{\lambda_i}
 Y_{\lambda_0,\stackrel{i}{\check{\dots}},\lambda_k}(f_0,\stackrel{i}{\check{\dots}},f_k)\big\}_K
} \\
\displaystyle{
+ \sum_{0\leq i<j\leq k}
(-1)^{i+j}
Y_{\lambda_i+\lambda_j,\lambda_0,\stackrel{i}{\check{\dots}}\stackrel{j}{\check{\dots}},\lambda_k}
\big(\{{f_i}_{\lambda_i}f_j\}_K,f_0,\stackrel{i}{\check{\dots}}\,\stackrel{j}{\check{\dots}},f_k\big)
\,.
}
\end{array}
$$
for some $Y\in W^{\var}_{k-1}$.

\begin{remark}\label{100505:th}
One can show that
the identity map on $\mc V/\partial\mc V=\Omega^0(\mc V)=W^{\var}_{-1}$
extends to a homomorphism of cohomology complexes
$\Phi_K:\,(\Omega^\bullet(\mc V),\delta)\to(W^{\var}(\Pi\mc V),\ad K)\,,\,\,
P\in \Omega^{k+1}(\mc V)\mapsto\Phi_K^k P\in W^{\var}_k$,
defined by the following formula:
\begin{equation}\label{100505:eq1}
\begin{array}{l}
\displaystyle{
(\Phi_K^k P)_{\lambda_0,\dots,\lambda_k}(f_0,\dots,f_k)
=
(-1)^{k+1}
\sum_{i_0,\dots,i_k\in I}
}\\
\displaystyle{
P_{i_0,\dots,i_k}(\lambda_0+\partial_0,\dots,\lambda_k+\partial_k)
\{{f_0}_{\lambda_0}u_{i_0}\}_K\dots\{{f_k}_{\lambda_k}u_{i_k}\}_K\,,
}
\end{array}
\end{equation}
where $\partial_\alpha$ is $\partial$ applied to $\{{f_\alpha}_{\lambda_\alpha}u_{i_\alpha}\}_K$.
For $k=0,1$ the map $\Phi_K^k$ reduces to
$\Phi_K^0(F)=K(\partial)F,\,F\in\mc V^{\oplus\ell}$,
where $K(\partial)$ is the differential operator associated to $K$ via \eqref{100430:eq1b},
and, identifying elements in $\Omega^2(\mc V)$ and $W^{\var}_1$ with
the corresponding skewadjoint differential operators, we have
\begin{equation}\label{100510:eq2}
(\Phi_K^1S)(\partial)=-K(\partial)\circ S(\partial)\circ K(\partial)\,.
\end{equation}
One can also show that, for $k\geq1$, the map $\Phi_K^k:\,\Omega^{k+1}(\mc V)\to W^{\var}_k$ is injective
provided that the map $K(\partial+\lambda):\,\mc V[\lambda]^{\oplus\ell}\to\mc V[\lambda]^\ell$ is injective.
\end{remark}

\subsection{Identification of $W^{\var}(\Pi\mc V)$ with $\Omega^\bullet(\mc V)$}
\label{sec:9.3}

In this section we will construct an explicit identification of the superspaces
$W^{\var}_k$ and $\Omega^{k+1}(\mc V)$ for finite $\ell$.
However, this identification is not covariant, i.e. it depends on the choice of generators
$u_i,\,i\in I$, of the algebra of differential functions $\mc V$.
Indeed these two spaces play completely different roles.
In a forthcoming publication
we will give a covariant description of the variational complex $\Omega^\bullet(\mc V)$,
and clarify the relation between these two complexes, as well as the relation
with the calculus structure discussed in \cite{DSHK}.

Note that if $X:\,\mc V^{\otimes(k+1)}
\to\mb F_-[\lambda_0,\dots,\lambda_k]\otimes_{\mb F[\partial]}\mc V$
satisfies the master equation \eqref{100430:eq2},
then automatically it satisfies sesquilinearity.
Moreover, it satisfies skewsymmetry provided that it is skewsymmetric on the generators $u_i,\,i\in I$.
Conversely, any map $X:\,\big(\bigoplus_{i\in I}\mb F u_i\big)^{\otimes(k+1)}
\to\mb F_-[\lambda_0,\dots,\lambda_k]\otimes_{\mb F[\partial]}\mc V$
satisfying skewsymmetry
can be extended uniquely to an element of $W^{\var}_k$ by the master equation.
Hence, $X\in W^{\var}_k$ is completely determined by the collection of polynomials
$X_{\lambda_0,\dots,\lambda_k}(u_{i_0},\dots,u_{i_k})
\in\mb F_-[\lambda_0,\dots,\lambda_k]\otimes_{\mb F[\partial]}\mc V$,
with $i_1,\dots,i_k\in I$.

Thanks to the above observations, we construct an injective linear map
$$
\Phi:\,\Omega^\bullet(\mc V)\to W^{\var}(\Pi\mc V)\,,
$$
sending $P\in\Omega^{k+1}(\mc V)$ to $\Phi^kP\in W^{\var}_k$,
defined on generators by
\begin{equation}\label{100518:eq5}
(\Phi^kP)_{\lambda_0,\dots,\lambda_k}(u_{i_0},\dots,u_{i_k})
=
P_{i_0,\dots,i_k}(\lambda_0,\dots,\lambda_k)\,,
\end{equation}
and extended to a map
$\Phi^kP:\,\mc V^{\otimes(k+1)}\to\mb F_-[\lambda_0,\dots,\lambda_k]\otimes_{\mb F[\partial]}\mc V$
by the master equation \eqref{100430:eq2}.
Clearly, $\Phi$ is surjective for finite $\ell$, and it is injective in general.
Note that $\Phi^k$ formally coincides, up to a sign, with $\Phi^k_K$ in \eqref{100505:eq1}
if we let $K=\id$.

Let $\ell$ be finite. Let $K\in W^{\var}_1$ be such that $[K,K]=0$,
and consider the Poisson cohomology complex $(W^{\var}(\Pi\mc V),\ad K)$.
Using the bijection $\Phi:\,\Omega^\bullet(\mc V)\to W^{\var}(\Pi\mc V)$,
we get a differential $d_K:\,\Omega^k(\mc V)\to\Omega^{k+1}(\mc V)$
induced by the action of $\ad K$ on $W^{\var}(\Pi\mc V)$.
Explicitly, recalling equation \eqref{100503:eq11}, we have
\begin{equation}\label{100518:eq1}
\begin{array}{l}
\displaystyle{
(d_KP)_{i_0,\dots,i_k}(\lambda_0,\dots,\lambda_k)
} \\
\displaystyle{
= (-1)^{k+1} \sum_{j\in I,n\in\mb Z_+} \bigg(
\sum_{\alpha=0}^k (-1)^{\alpha}
\frac{\partial
P_{i_0,\stackrel{\alpha}{\check{\dots}},i_k}(\lambda_0,\stackrel{\alpha}{\check{\dots}},\lambda_k)
}{\partial u_j^{(n)}}
(\lambda_\alpha+\partial)^n
K_{j,i_\alpha}(\lambda_\alpha)
} \\
\displaystyle{
+ \sum_{0\leq \alpha<\beta\leq k}
(-1)^{\alpha+\beta}
P_{j,i_0,\stackrel{\alpha}{\check{\dots}}\,\stackrel{\beta}{\check{\dots}},i_k}
(\lambda_\alpha+\lambda_\beta+\partial,
\lambda_0,\stackrel{\alpha}{\check{\dots}}\,\stackrel{\beta}{\check{\dots}},\lambda_k)_\to
} \\
\displaystyle{
\,\,\,\,\,\,\,\,\,\,\,\,\,\,\,\,\,\,\,\,\,\,\,\,\,\,\,\,\,\,\,\,\,\,\,\,\,\,\,\,\,\,
\,\,\,\,\,\,\,\,\,\,\,\,\,\,\,\,\,\,\,\,\,\,\,\,\,\,\,\,\,\,\,\,\,\,\,\,\,\,\,\,\,\,
(-\lambda_\alpha-\lambda_\beta-\partial)^n
\frac{\partial K_{i_\beta,i_\alpha}(\lambda_\alpha)}{\partial u_j^{(n)}}
\bigg)\,.
}
\end{array}
\end{equation}

\subsection{The variational and Poisson cohomology complexes in terms of local polydifferential operators}
\label{sec:9.4}

Let $\mc V$ be an algebra of differentiable functions extension
of $R_\ell=\mb F[u_i^{(n)}\,|\,i\in I,n\in\mb Z_+]$.
Recall \cite{DSK2} that a $k$-\emph{local differential operator of rank} $r$
is an $\mb F$-linear map $S:\,(\mc V^{\oplus r})^{k}\to\mc V/\partial\mc V$, of the form
\begin{equation}\label{eq:dic15_1}
S(P^1,\dots,P^k)
\,=\,
\tint \sum_{\substack{m_1,\dots,m_{k}\in\mb Z_+ \\ i_1,\dots,i_{k}\in\{1,\dots,r\}}}
f_{i_1,\dots,i_k}^{m_1,\dots,m_{k}}
(\partial^{m_1}P^1_{i_1})\dots(\partial^{m_{k}}P^{k}_{i_{k}})\,,
\end{equation}
where the coefficients $f_{i_1,\dots,i_k}^{m_1,\dots,m_{k}}$ lie in $\mc V$,
and, for each $k$-tuple $(i_1,\dots,i_k)$ they are zero for all but finitely many
choices of $(m_1,\dots,m_k)\in\mb Z_+^k$.
When $r$ is infinite, $S$ is called of \emph{finite type} if
all but finitely many of the coefficients $f_{i_1,\dots,i_k}^{m_1,\dots,m_{k}}$ are zero.
In this case, $S$ extends to a map $S:\,(\mc V^r)^k\to\mc V/\partial\mc V$.
The operator $S$ is called skewsymmetric if
$$
S(P^1,\dots,P^k)=\text{sign}(\sigma) S(P^{\sigma(1)},\dots,P^{\sigma(k)})\,,
$$
for every $P^1,\dots,P^k\in\mc V^\ell$ and every permutation $\sigma\in S_k$.
In this section we describe both the variational complex $\Omega^\bullet(\mc V)$
and the space of variational polyvector fields $W^{\var}(\Pi\mc V)$
in terms of local polydifferential operators.

\begin{lemma}\label{100506:lem}
We have a canonical isomorphism of
$\mb F_-[\lambda_1,\dots,\lambda_k]\otimes_{\mb F[\partial]}\mc V$
to the space of local $k$-differential operators of rank 1,
which associates to the polynomial
$f(\lambda_1,\dots,\lambda_k)
=\displaystyle{\sum_{m_1,\dots,m_k\in\mb Z_+}}f^{m_1,\dots,m_k}\lambda_1^{m_1}\dots\lambda_k^{m_k}$,
the map $f:\,\mc V^k\to\mc V/\partial\mc V$, given by
\begin{equation}\label{100506:eq1}
f(g_1,\dots,g_k)
=
\sum_{m_1,\dots,m_k\in\mb Z_+}
\tint f^{m_1,\dots,m_k}
(\partial^{m_1}g_1)\dots(\partial^{m_k}g_k)\,.
\end{equation}
\end{lemma}
\begin{proof}
Note first that, if the polynomial
$f(\lambda_1,\dots,\lambda_k)$ lies in the image of $(\partial+\lambda_1+\dots+\lambda_k)$,
then the corresponding local $k$-differential operator $f:\,\mc V^k\to\mc V/\partial\mc V$ is zero.
Hence, we have a well-defined linear map from
$\mb F_-[\lambda_1,\dots,\lambda_k]\otimes_{\mb F[\partial]}\mc V$
to the space of local $k$-differential operators of rank 1, given by \eqref{100506:eq1}.
Clearly any local $k$-differential operator of rank 1 is in the image of this map.
We are left to prove injectivity.
Let $f\in\mb F_-[\lambda_1,\dots,\lambda_k]\otimes_{\mb F[\partial]}\mc V$
be such that $f(g_1,\dots,g_k)=0$ for all $g_1,\dots,g_k$.
Recall that we have a (non-canonical) isomorphism
$\mb F_-[\lambda_1,\dots,\lambda_k]\otimes_{\mb F[\partial]}\mc V\simeq
\mb F_-[\lambda_1,\dots,\lambda_{k-1}]\otimes\mc V$,
obtained by replacing $\lambda_k$ with $-\lambda_1-\dots-\lambda_{k-1}-\partial$.
Hence, we can write the polynomial $f$ in the form
$$
f(\lambda_1,\dots,\lambda_{k-1})
=\sum_{m_1,\dots,m_{k-1}\in\mb Z_+}f^{m_1,\dots,m_{k-1}}\lambda_1^{m_1}\dots\lambda_{k-1}^{m_{k-1}}\,.
$$
Since, by assumption, the corresponding map $f:\,\mc V^k\to\mc V/\partial\mc V$ is zero, we get
$$
\sum_{m_1,\dots,m_{k-1}\in\mb Z_+}
\tint f^{m_1,\dots,m_{k-1}}
(\partial^{m_1}g_1)\dots(\partial^{m_{k-1}}g_{k-1}) g_k=0\,,
$$
for every $g_1,\dots,g_k\in\mc V$.
By the non-degeneracy of the pairing $\mc V\times\mc V\to\mc V/\partial\mc V,\,(f,g)=\tint fg$
(cf. Lemma 10(c) from Section 5.2 of \cite{DSK2}),
it follows that
$$
\sum_{m_1,\dots,m_{k-1}\in\mb Z_+}
f^{m_1,\dots,m_{k-1}}
(\partial^{m_1}g_1)\dots(\partial^{m_{k-1}}g_{k-1})=0\,,
$$
for all $g_1,\dots,g_{k-1}\in\mc V$.
This is equivalent to $f(\lambda_1,\dots,\lambda_{k-1})=0$.
\end{proof}
\begin{proposition}\label{100506:prop}
\begin{enumerate}[(a)]
\item
We have a canonical isomorphism of $\Omega^k(\mc V)$
to the space of skewsymmetric local $k$-differential operators of rank $\ell$ of finite type,
which associates to the element
$P\in \Omega^k(\mc V)$
the local $k$-differential operator
$S:\,(\mc V^\ell)^k\to\mc V/\partial\mc V$, given by
\begin{equation}\label{100506:eq2}
S(P^1,\dots,P^k)
=
\sum_{i_1,\dots,i_k\in I}
\tint
P_{i_1,\dots,i_k}(\partial_1,\dots,\partial_k)
P^1_{i_1}\dots P^{k}_{i_{k}}\,,
\end{equation}
where $\partial_\alpha$ means $\partial$ acting on $P^\alpha_{i_\alpha}$.
\item
We have a canonical isomorphism
of the space of variational $k$-vector fields $W^{\var}_{k-1}$
to the space of skewsymmetric local $k$-differential operators of rank $\ell$,
which associates to the element
$X\in W^{\var}_{k-1}$
the local $k$-differential operator
$X:\,(\mc V^{\oplus\ell})^k\to\mc V/\partial\mc V$, given by
\begin{equation}\label{100506:eq3}
X(F^1,\dots,F^k)
=
\sum_{i_1,\dots,i_k\in I}
\tint X_{\partial_1,\dots,\partial_k}(u_{i_1},\dots,u_{i_k})
F^1_{i_1}\dots F^{k}_{i_{k}}\,.
\end{equation}
\end{enumerate}
\end{proposition}
\begin{proof}
Recalling the definition of $\Omega^k(\mc V)$ in Section \ref{sec:8.2},
an element $P\in\Omega^k(\mc V)$ is an array of polynomials
$\big(P_{i_1,\dots,i_k}(\lambda_1,\dots,\lambda_k)\big)_{i_1,\dots,i_k\in I}$,\
with finitely many non-zero entries
$P_{i_1,\dots,i_k}(\lambda_1,\dots,\lambda_k)
\in\mb F_-[\lambda_1,\dots,\lambda_k]\otimes_{\mb F[\partial]}\mc V$,
skewsymmetric with respect to simultaneous permutations
of the variables $\lambda_0,\dots,\lambda_k$
and the indexes $i_0,\dots,i_k$.
Using Lemma \ref{100506:lem} this
corresponds, bijectively,
to the local $k$-differential operator of rank $\ell$ \eqref{100506:eq2}.
The skewsymmetry condition on $P\in\Omega^k(\mc V)$
translates to the skewsymmetry of the corresponding local $k$-differential operator,
and the finiteness condition on $P$ translates into saying that
the corresponding local $k$-differential operator is of finite type.
This proves part (a).
The proof of part (b) follows by arguments similar to those in Section \ref{sec:9.3}.
\end{proof}
\begin{remark}\label{100507:rem1}
The following identity is immediate from \eqref{100506:eq3} and the master equation \eqref{100502:eq1}:
$$
\tint X_{0,\dots,0}(f_0,\dots,f_k)
=X\Big(\frac{\delta f_0}{\delta u},\dots,\frac{\delta f_k}{\delta u}\Big)\,.
$$
\end{remark}

We can write down the expression of the differential $\delta:\,\Omega^k(\mc V)\to \Omega^{k+1}(\mc V)$
in terms of local polydifferential operators of rank $\ell$.
Recalling \eqref{100518:eq3}, we have,
for a local finite type $(k+1)$-differential operator $S$,
$$
(\delta S)(P^0,\dots,P^k)
=
\sum_{i=0}^k(-1)^{k+i}
(X_{P^i}S)\big(P^0,\stackrel{i}{\check{\dots}},P^k)
\big)\,,
$$
where $X_P$ denotes the evolutionary vector field of characteristic $P\in\mc V^\ell$,
defined in \eqref{2006_X2}, and, if $S$ is the local $k$-differential operator \eqref{eq:dic15_1},
$X_PS$ denotes the local $k$-differential operator obtained from $S$
by replacing the coefficients $f_{i_1,\dots,i_k}^{n_1,\dots,n_{k}}$
by $X_P(f_{i_1,\dots,i_k}^{n_1,\dots,n_{k}})$ \cite{DSK2}.

Next, in view of Proposition \ref{100506:prop}(b), we can write the Lie superalgebra structure
of $W^{\var}(\Pi\mc V)$ in terms of local polydifferential operators.
Given $X\in W^{\var}_h$ and $Y\in W^{\var}_{k-h}$,
we have $[X,Y]=X\Box Y-(-1)^{h(k-h)}Y\Box X$ and,
recalling \eqref{100418:eq2},
the local $k$-differential operator corresponding to $X\Box Y\in W^{\var}_k$ is given by:
$$
\begin{array}{l}
\displaystyle{
(X\Box Y)\big(F^0,\dots,F^k\big)
} \\
\displaystyle{
= \sum_{\substack{
i_0<\dots <i_{k-h}\\
i_{k-h+1}<\dots< i_k}}
\pm X\Big(
\frac\delta{\delta u} Y \big(F^{\sigma(0)},\dots,F^{\sigma(k-h)}\big),
F^{\sigma(k-h+1)},\dots,F^{\sigma(k)}
\Big)\,,
}
\end{array}
$$
where $\pm$ is the sign of the permutation $(i_0,\dots,i_k)$ of the set $\{0,\dots,k\}$.

Let $K\in W^{\var}_1$ be such that $[K,K]=0$,
and consider the corresponding Poisson vertex algebra structure on $\mc V$,
or, equivalently, the corresponding Hamiltonian map $K(\partial)$ defined by \eqref{100430:eq1b}.
We can write the formula of the differential $\ad K$ for the Poisson cohomology complex $W^{\var}$,
identified with the space of local polydifferential operators.
Given $X\in W^{\var}_{k-1}$, the local $k$-differential operator
corresponding to $(\ad K)(X)\in W^{\var}_k$ is given by the following formula,
equivalent to \eqref{100503:eq11}:
$$
\begin{array}{c}
\displaystyle{
((\ad K)X)\big(F^0,\dots,F^k\big)
=
\sum_{i=0}^k (-1)^{k+i}
\tint F^i\cdot K(\partial)\frac{\delta}{\delta u} X(F^0,\stackrel{i}{\check{\dots}},F^k)
} \\
\displaystyle{
+ \sum_{0\leq i<j\leq k} (-1)^{k+i+j}
X\Big(
\frac\delta{\delta u} \tint \big(F^i\cdot K(\partial)F^j\big),
F^0,\stackrel{i}{\check{\dots}}\,\stackrel{j}{\check{\dots}},F^k
\Big)\,.
}
\end{array}
$$
Here $\cdot$ denotes the usual pairing $\mc V^{\oplus\ell}\times\mc V^\ell\to\mc V$.

\begin{remark}
We can translate the homomorphism $\Phi_K$ defined in Remark \ref{100505:th}
into the language of local polydifferential operators.
Given a finite type local $(k+1)$-differential operator $S$,
$\Phi_K^k(S)$ is the following local $(k+1)$-differential operator:
\begin{equation}\label{100507:eq8}
(\Phi_K^k S)(F^0,\dots,F^k)
= (-1)^{k+1} S(K(\partial)F^0,\dots,K(\partial)F^k)\,.
\end{equation}
Theorem \ref{100430:th}, Remark \ref{100505:th} and Proposition \ref{100506:prop}
imply that,
if $\mc V$ is a normal algebra of differential functions,
$K\in W^{\var}_1$ is such that $[K,K]=0$, and
the map $K(\lambda+\partial):\,\mc V[\lambda]^{\oplus\ell}\to\mc V[\lambda]^{\ell}$
is injective,
then, $\mc H_K^{k}(\mc V)=0$ for $k\geq0$
provided that the following condition holds:
\begin{equation}\label{100510:eq4}
\ker\big(\ad K:\,W^{\var}_k\to W^{\var}_{k+1}\big)
\,\subset\,
\Phi_K^k\big(\Omega^{k+1}(\mc V)\big)\,.
\end{equation}
Indeed,
let $\bar X\in\mc H_K^{k}(\mc V)$, and let $X\in \ker\big(\ad K:\,W^{\var}_k\to W^{\var}_{k+1}\big)$
be a representative of it.
By assumption \eqref{100510:eq4}, there exists $P\in\Omega^{k+1}(\mc V)$
such that $X=\Phi_K^k(P)$, and,
by Remark \ref{100505:th}, we have that $\Phi_K^{k+1}(\delta P)=[K,X]=0$.
Since $\Phi_K$ is injective, we have that $\delta P=0$,
hence, by Theorem \ref{100430:th}, $P=\delta Q$ for some $Q\in\Omega^k(\mc V)$.
In conclusion, $X=\Phi_K^k(\delta Q)=(\ad K)\big(\Phi_K^{k-1}(Q)\big)$, completing the proof.
\end{remark}

\begin{remark}\label{100507:rem2}
It is useful to see what the analogue of the homomorphism $\Phi_K$ 
is in the finite dimensional setup.
Let $A=\mb F[u_i\,|\,i\in I]$.
Recall from Section \ref{sec:3.1} that $W^{\as}(\Pi A)=\bigoplus_{k=-1}^\infty W^{\as}_k(\Pi A)$,
where $W^{\as}_k(\Pi A)$ consists of linear maps
$X:\,\bigwedge^{k+1}A\to A$ satisfying the Leibniz rule in all arguments,
or, equivalently, the following analogue of the master equation \eqref{100430:eq2}:
$$
X(f_0,\dots,f_k)=\sum_{i_0,\dots,i_k\in I}X(u_{i_0},\dots,u_{i_k})
\frac{\partial f_0}{\partial u_{i_0}}\dots \frac{\partial f_k}{\partial u_{i_k}}\,.
$$
To the map $X$ we associate the map $\tilde X:\,\bigwedge^{k+1}(A^{\oplus\ell})\to A$,
given by
$$
\tilde X(F^0,\dots,F^k)=\sum_{i_0,\dots,i_k\in I}X(u_{i_0},\dots,u_{i_k})
F^0_{i_0}\dots F^k_{i_k}\,,
$$
and we have the following analogue of the identity in Remark \ref{100507:rem1}:
$$
X(f_0,\dots,f_k)=
\tilde X(\nabla_u f_0,\dots,\nabla_u f_k)\,,
$$
where $\nabla_u f$ denotes the vector of partial derivatives of $f$.
Next, let $\Omega^\bullet(A)$ be the algebra of differential forms over $A$,
and let $d$ be the de Rham differential on $\Omega^\bullet(A)$.
We can associate to $\omega=\sum f_{i_0,\dots,i_k}du_{i_0}\wedge du_{i_k}\in\Omega^{k+1}(A)$
a map $\omega:\,(A^\ell)^{\otimes(k+1)}\to A$, given by
$\omega(P^0,\dots,P^k)=\sum f_{i_0,\dots,i_k} P^0_{i_0}\dots P^k_{i_k}$.
Then, for $K\in W^{\as}_1(\Pi A)$, such that $[K,K]=0$,
we have a homomorphism of complexes $\Phi_K:\,(\Omega^\bullet(A),d)\to(W^{\as}(\Pi A),\ad K)$,
given by (cf. \eqref{100507:eq8}):
$$
(\Phi_K^k\omega)\big(F^0,\dots,F^k\big)
=
\pm
\omega\big(K F^0,\dots,K F^k\big)\,,
$$
where $K$ is the $\ell\times\ell$ skewsymmetric matrix $K_{ij}=K(u_i,u_j)$.
It follows, in particular, that if $K$ is surjective, then the Poisson cohomology is trivial.
\end{remark}

\subsection{Generalized variational complexes}
\label{sec:9.5}

Let $\ell$ be finite, and let $K(\partial)=\big(K_{ij}(\partial)\big)_{i,j\in I}$
be an $\ell\times\ell$ matrix differential operator with quasiconstant coefficients.
Then, we define $\delta_K:\,\Omega^k(\mc V)\to\Omega^{k+1}(\mc V)$
by the following formula:
\begin{equation}\label{100519:eq1}
\begin{array}{l}
\displaystyle{
(\delta_KP)_{i_0,\dots,i_k}(\lambda_0,\dots,\lambda_k)
} \\
\displaystyle{
=
\sum_{\alpha=0}^k (-1)^{\alpha} \sum_{j\in I,n\in\mb Z_+}
\frac{\partial
P_{i_0,\stackrel{\alpha}{\check{\dots}},i_k}(\lambda_0,\stackrel{\alpha}{\check{\dots}},\lambda_k)
}{\partial u_j^{(n)}}
(\lambda_\alpha+\partial)^n
K_{j,i_\alpha}(\lambda_\alpha)\,.
}
\end{array}
\end{equation}
Note that, when $K=\id$ this coincides with the differential $\delta$ of the variational complex
defined in equation \eqref{100518:eq3},
and when $K(\partial)$ is a skewadjoint operator
it coincides with $(-1)^{k+1}d_K$, where $d_K$ is given by equation \eqref{100518:eq1}.
\begin{proposition}\label{100519:prop}
If $K(\partial)$ is an $\ell\times\ell$ matrix differential operator with quasiconstant coefficients,
then $\delta_K$ in \eqref{100519:eq1} is a well-defined map $\Omega^k(\mc V)\to\Omega^{k+1}(\mc V)$,
and it makes $(\Omega^\bullet(\mc V),\delta_K)$ a cohomology complex.
\end{proposition}
\begin{proof}
It follows from Proposition \ref{100519b:prop} in Section \ref{sec:10.4}.
\end{proof}
\begin{remark}\label{100531:rem}
One can show that $d_K$ given by formula \eqref{100518:eq1} is well defined
only if $K(\partial)$ is a sum of a skewadjoint operator and a quasiconstant operator.
\end{remark}

%


\section{The universal odd PVA $\tilde W^{\var}(\Pi\mc V)$
for an algebra of differential functions, and basic PVA cohomology}
\label{sec:10}

\subsection{The Lie conformal algebra $\CVect(\mc V)$ of conformal vector fields}
\label{sec:10.05}

As in the previous section,
let $\mc V$ be an algebra of differential functions,
extension of the algebra of differential polynomials
$R_\ell=\mb F[u_i^{(n)}\,|\,i\in I,n\in\mb Z_+]$.
We assume moreover, in this section, that $\ell$ is finite.

For $i\in I$, introduce the linear map $E^i_\lambda:\,\mc V\to\mb F[\lambda]\otimes\mc V$,
given by
\begin{equation}\label{100529:eq1a}
E^i_\lambda\,=\,\sum_{n\in\mb Z_+}(-\lambda-\partial)^n\frac{\partial}{\partial u_i^{(n)}}\,.
\end{equation}
Note that $E^i_\lambda=\sum_{n\in\mb Z_+}\lambda^n E^i_{(n)}$,
is the generating series of the \emph{higher Euler operators} (see \cite{Ol2}),
$$
E^i_{(n)}=\sum_{m=n}^\infty \binom{m}{n}(-1)^m \partial^{m-n} \frac{\partial}{\partial u_i^{(m)}}\,.
$$
In particular, $E^i_0=\frac\delta{\delta u_i}$ is the variational derivative \eqref{eq:varder}.
\begin{lemma}\label{100529:prop1}
$E^i_\lambda$ is a right conformal derivation of $\mc V$ (see Section \ref{sec:7.1}).
\end{lemma}
\begin{proof}
We need to check that, for $f,g\in\mc V$,
we have $E^i_\lambda(\partial f)=-\lambda E^i_\lambda(f)$
and $E^i_\lambda(fg)=(E^i_{\lambda+\partial}f)_\to g+(E^i_{\lambda+\partial}g)_\to f$.
The first identity follows immediately from \eqref{eq:comm_frac_b},
and the second one is straightforward.
\end{proof}
Given an $\ell$-tuple of polynomials $P=(P_i(\lambda))_{i\in I}\in(\mb F[\lambda]\otimes\mc V)^\ell$,
we define the \emph{conformal vector field} of characteristic $P$
as the map
$\tilde X^P:\,\mc V\to\mb F[\lambda]\otimes\mc V$, given by
\begin{equation}\label{100529:eq1b}
\tilde X^P_\lambda(f)
=
\sum_{i\in I}P_i(\lambda+\partial)E^i_\lambda(f)
\end{equation}
We denote by $\CVect(\mc V)$ the space of all conformal vector fields.
It follows immediately from Lemma \ref{100529:prop1} that $\CVect(\mc V)$
is a subspace of the space $\RCder(\mc V)$ of right conformal derivations of $\mc V$.
\begin{proposition}\label{100529:prop2}
\begin{enumerate}[(a)]
\item
The $\mb F[\partial]$-module structure of $\RCder(\mc V)$ restricts
to the following $\mb F[\partial]$-module structure on $\CVect(\mc V)$:
for $P\in(\mb F[\lambda]\otimes\mc V)^\ell$,
$\partial \tilde X^P$ is the conformal vector field with characteristics
$$
\big((\partial+\lambda)P_i(\lambda)\big)_{i\in I}\,.
$$
\item
the  formal power series valued $\lambda$-bracket 
on $\RCend(\mc V)$ (cf. Section \ref{sec:6.3})
restricts to the following polynomial valued 
$\lambda$-bracket on $\CVect(\mc V)$:
$[\tilde X^P{}_\lambda\tilde X^Q]=\tilde X^{[P_\lambda Q]}$
for $P,Q\in(\mb F[\lambda]\otimes\mc V)^\ell$,
where
\begin{equation}\label{100529:eq2}
\begin{array}{rcl}
[P_\lambda Q]_i(\mu)
&=&
\displaystyle{
\sum_{j\in I,n\in\mb Z_+}
\big((\lambda+\partial)^nP^*_j(\lambda)\big)\frac{\partial Q_i(\mu)}{\partial u_j^{(n)}}
} \\
&& \displaystyle{
-\sum_{j\in I,n\in\mb Z_+}
Q_j(\lambda+\mu+\partial)(-\lambda-\mu-\partial)^n\frac{\partial P_i(\mu)}{\partial u_j^{(n)}}
\,.
}
\end{array}
\end{equation}
Here, for $\displaystyle{P_i(\lambda)=\sum_{n\in\mb Z_+}\lambda^nP_i^n}$,
we denote $\displaystyle{P^*_i(\lambda)=\sum_{n\in\mb Z_+}(-\lambda-\partial)^nP_i^n}$.
\item
$\CVect(\mc V)$ is a Lie conformal algebra $\RCder(\mc V)$.
\end{enumerate}
\end{proposition}
\begin{proof}
Part (a) follows from the definition \eqref{100517:eq3} of the $\mb F[\partial]$-module structure
of the space $\RCend(\mc V)=\tilde W^\partial_0(\Pi\mc V)$ of right conformal endomorphisms of $\mc V$.
Part (b) is obtained, by a straightforward computation,
using the definition \eqref{100517:eq1} of the $\lambda$-bracket in $\RCend(\mc V)$
and identity \eqref{eq:comm_frac_b}.
Finally, the $\lambda$-bracket \eqref{100529:eq2}
is clearly polynomial valued,
and it satisfies all the Lie conformal algebra axioms 
by Lemma \ref{infinite:lem1}.
\end{proof}
\begin{remark}\label{100529:rem}
The image of $\CVect(\mc V)$ via the isomorphism \eqref{100528:eq3}
is a subalgebra of $\Cder(\mc V)$, the space 
of (left) conformal derivations of $\mc V$.
Its elements, which is natural to call \emph{left conformal vector fields},
are of the form
$$
{}^P\tilde X_\lambda=\sum_{i\in I,n\in\mb Z_+}
\big((\partial+\lambda)^nP_i(\lambda)\big) \frac{\partial}{\partial u_i^{(n)}}\,.
$$
Equivalently, letting ${}^iE_\lambda=\sum_{n\in\mb Z_+}\lambda^n\frac{\partial}{\partial u_i^{(n)}}$,
we have
$$
{}^P\tilde X_\lambda(f)=\sum_{i\in I}\big({}^iE_{\lambda+\partial}f\big)_\to P_i(\lambda)\,.
$$
In fact, we have $(\tilde X^P)^*={}^{P^*}\tilde X$,
where $P^*$ is defined in Proposition \ref{100529:prop2}(b).
The induced Lie conformal algebra structure on the space of left conformal vector fields
is as follows:
$\partial({}^P\tilde X)$ is the left conformal vector field of characteristics
$\big(-\lambda P_i(\lambda)\big)_{i\in I}$,
and
$[{}^P\tilde X_\lambda{}^Q\tilde X]={}^{[P_\lambda Q]^L}\tilde X$,
where
$$
[P_\lambda Q]^L_i(\mu)={}^P\tilde X_\lambda(Q_i(\mu-\lambda))
-{}^Q\tilde X_{\mu-\lambda}(P_i(\lambda))\,.
$$
In fact, formula \eqref{100529:eq2} for the $\lambda$-bracket in $\CVect(\mc V)$
can be written, using this notation, in a similar form:
$$
[P_\lambda Q]_i(\mu)={}^{P^*}\tilde X_\lambda(Q_i(\mu))
-\tilde X^Q_{\lambda+\mu}(P_i(\mu))\,.
$$
\end{remark}

\subsection{The universal odd PVA $\tilde W^{\var}(\Pi\mc V)$}
\label{sec:10.1}

Recall the definition of the $\mb F[\partial]$-module
$\tilde W^\partial(\Pi \mc V)=\bigoplus_{k=-1}^\infty \tilde W^\partial_k(\Pi \mc V)$,
with parity denoted by $\bar p$ (see Section \ref{sec:6.1}),
together with its formal power series values $\lambda$-bracket,
defined in Section \ref{sec:6.3}.
Consider the full prolongation (cf. Definition \ref{infinite:def1})
of the Lie conformal superalgebra
$\CVect(\mc V)\subset\RCend(\mc V)=\tilde W^\partial_0(\Pi\mc V)$,
which we denote
$$
\tilde W^{\var}(\Pi\mc V)=\bigoplus_{k=-1}^\infty\tilde W^{\var}_k\subset
\tilde W^\partial(\Pi\mc V)\,.
$$
Its elements are called conformal polyvector fields.
We will show that the restriction of the $\lambda$-bracket on 
$\Pi\tilde W^{\var}(\Pi \mc V)$ is polynomial valued,
and this makes $\Pi\tilde W^{\var}(\Pi \mc V)$ an odd PVA.
\begin{proposition}\label{100422x:prop1}
For $k\geq-1$, the superspace $\tilde W^{\var}_k$
is the subspace of $\tilde W^{\partial}_k(\Pi\mc V)$,
consisting of linear maps
$X:\,\mc V^{\otimes(k+1)}\to\mb F[\lambda_0,\dots,\lambda_k]\otimes \mc V$
satisfying the sesquilinearity and skewsymmetry conditions \eqref{100530:eq1} and \eqref{100530:eq2},
and the master equation \eqref{100430:eq2}, where both sides are interpreted as elements
of $\mb F[\lambda_0,\dots,\lambda_k]\otimes \mc V$
(not, as in Section \ref{sec:9.1}, of $\mb F[\lambda_0,\dots,\lambda_k]\otimes_{\mb F[\partial]} \mc V$).
\end{proposition}
\begin{proof}
First, we observe that the master equation \eqref{100430:eq2} implies
the sesquilinearity condition \eqref{100530:eq1}.
Recall also that the master equation \eqref{100430:eq2} is equivalent to
the equations \eqref{100502:eq1},
and if, moreover, $X$ satisfies the skewsymmetry condition \eqref{100530:eq1},
it is enough to have equation \eqref{100502:eq1} for $s=k$.
Hence, an element $X\in\tilde W^{\partial}_k(\Pi\mc V)$ satisfies all
conditions \eqref{100530:eq1}, \eqref{100530:eq2} and \eqref{100430:eq2}
if and only if it satisfies condition \eqref{100530:eq2} and the equation
\begin{equation}\label{100502:eq1bis}
\begin{array}{l}
\displaystyle{
X_{\lambda_0,\dots,\lambda_k}(f_0,\dots,f_k)
} \\
\displaystyle{
=
\sum_{i\in I}
X_{\lambda_0,\dots,\lambda_{k-1},\lambda_k+\partial}(f_0,\dots,f_{k-1},u_i)_\to
E^i_{\lambda_k}(f_k)\,.
}
\end{array}
\end{equation}

Clearly, for $k=0$ the maps $X:\,\mc V\to\mb F[\lambda]\otimes\mc V$
satisfying \eqref{100502:eq1bis} are exactly the conformal vector fields of $\mc V$,
the characteristics being $\big(X_\lambda(u_i)\big)_{i\in I}$.
Let $X\in\tilde W^{\var}_k$
and let us prove by induction on $k$ that it satisfies equation \eqref{100502:eq1bis}.
By \eqref{101125:eq1} we have
$$
\begin{array}{l}
\vphantom{\Big(}
X_{\lambda_0,\dots,\lambda_k}(f_0,\dots,f_k)
=
(-1)^{1+k}[{f_0}_{\lambda_0}X]_{\lambda_1,\dots,\lambda_k}(f_1,\dots,f_k) \\
\vphantom{\Big(}
\displaystyle{
=
\sum_{i\in I}(-1)^{1+k}[{f_0}_{\lambda_0}X]_{\lambda_1,\dots,\lambda_{k-1},\lambda_k+\partial}
(f_1,\dots,f_{k-1},u_i)_\to E^i_{\lambda_k}(f_k)
} \\
\vphantom{\Big(}
\displaystyle{
=
\sum_{i\in I}X_{\lambda_0,\dots,\lambda_{k-1},\lambda_k+\partial}
(f_0,\dots,f_{k-1},u_i)_\to E^i_{\lambda_k}(f_k)\,,
}
\end{array}
$$
proving \eqref{100502:eq1bis}.
\end{proof}
\begin{remark}\label{100502x:rem}
For $\mc V=R_\ell$, we have
$\tilde W^{\var}(\Pi\mc V)=\tilde W^{\partial,\as}(\Pi\mc V)$,
but for arbitrary $\mc V$ this is not always the case (see Remark \ref{100502:rem}).
\end{remark}
Note that the master equation implies that $X$ satisfies the Leibniz rule \eqref{100422c:eq1}.
Consequently, $\tilde W^{\var}_k$ is a subspace of $\tilde W^{\partial,\as}_k(\Pi\mc V)$.
\begin{proposition}\label{100422x:prop3}
\begin{enumerate}[(a)]
\item
For $X,Y\in\tilde W^{\var}(\Pi\mc V)$, $[X_\lambda Y]$
is a polynomial in $\lambda$ with coefficients in $\tilde W^{\var}(\Pi\mc V)$.
Moreover, the subspace $\tilde W^{\var}(\Pi\mc V)\subset\tilde W^{\partial,\as}(\Pi\mc V)$
is closed under the concatenation product \eqref{100422z:eq4}.
\item
$\Pi\tilde W^{\var}(\Pi\mc V)$, together with the $\lambda$-bracket and the concatenation product,
is a $\mb Z_+$-graded odd PVA.
\item
The representation of the Lie superalgebra $W^{\partial}(\Pi\mc V)$
on $\tilde W^{\partial}(\Pi\mc V)$
restricts to a representation of its subalgebra $W^{\var}(\Pi\mc V)$
on the odd PVA
$\tilde W^{\var}(\Pi\mc V)\subset \tilde W^{\partial}(\Pi\mc V)$,
commuting with $\partial$ and
acting by derivations of both
the concatenation product and the $\lambda$-bracket.
\item
The canonical map $\tint:\,\tilde W^\partial(\Pi\mc V)\to W^\partial(\Pi\mc V)$,
defined in Proposition \ref{100517:prop}(a),
restricts to a map $\tint:\,\tilde W^{\var}(\Pi\mc V)\to W^{\var}(\Pi\mc V)$,
which is a homomorphism of representations of the Lie superalgebra $W^{\var}(\Pi\mc V)$.
Moreover, this map induces a Lie algebra isomorphism
$$
\tint\,:\,\,\tilde W^{\var}(\Pi\mc V)/\partial\tilde W^{\var}(\Pi\mc V)
\stackrel{\sim}{\longrightarrow} W^{\var}(\Pi\mc V)\,.
$$
\end{enumerate}
\end{proposition}
\begin{proof}
First note that, by Proposition \ref{100422x:prop1} and formula \eqref{100430:eq2},
an element $X\in\tilde W^{\var}_k$ is determined by its values on the generators
$u_1,\dots,u_\ell$.
Since, by assumption, $\ell$ is finite, it follows that $[X_\lambda Y]$ 
is a polynomial in $\lambda$.
Let $X\in\tilde W^{\var}_{h-1}$ and $Y\in\tilde W^{\var}_{k-h-1}$,
with $k\geq h\geq 0$.
For part (a) it remains to prove that $X\wedge Y$, defined by \eqref{100422z:eq4},
lies in $\tilde W^{\var}_{k-1}$,
namely, it satisfies the master equation \eqref{100430:eq2}:
$$
\begin{array}{l}
\displaystyle{
(X\wedge Y)_{\lambda_1,\dots,\lambda_k}(f_1,\dots,f_k)
=
\sum_{\substack{i_1,\dots,i_k\in I \\ m_1,\dots,m_k\in\mb Z_+}}
\bigg(e^{\partial\partial_{\lambda_1}}\frac{\partial f_1}{\partial u_{i_1}^{(m_1)}}\bigg)
\dots
}\\
\displaystyle{
\,\,\,\,\,\,\,\,\,\dots
\bigg(e^{\partial\partial_{\lambda_k}}\frac{\partial f_k}{\partial u_{i_k}^{(m_k)}}\bigg)
(-\lambda_1)^{m_1}\dots(-\lambda_k)^{m_k}
(X\wedge Y)_{\lambda_1,\dots,\lambda_k}(u_{i_1},\dots,u_{i_k})\,.
}
\end{array}
$$
This is easily checked by a straightforward computation.
Part (b) can be proved in the same way as Proposition \ref{100422z:prop3}(a).
Recall from Proposition \ref{100422z:prop3}(b)
that we have a representation of the Lie superalgebra $W^{\partial,\as}(\Pi\mc V)$
on $\tilde W^{\partial,\as}(\Pi\mc V)$ commuting with $\partial$,
and acting by derivations of both the concatenation product and the $\lambda$-bracket.
Recall also, from Proposition \ref{100430:prop}, that
$W^{\var}(\Pi\mc V)\subset W^{\partial,\as}(\Pi\mc V)$ is a Lie subalgebra,
and, from part (b), that
$\tilde W^{\var}(\Pi\mc V)\subset \tilde W^{\partial,\as}(\Pi\mc V)$ is an odd Poisson vertex algebra.
Hence, in order to prove part (c), we only need to check that,
if $X\in W^{\var}_{h}$ and $\tilde Y\in\tilde W^{\var}_{k-h}$,
then $[X,\tilde Y]$, defined by \eqref{100517:eq2}, lies in $\tilde W^{\var}_k$.
By the observations in the proof of Proposition \ref{100422x:prop1},
it is therefore enough to prove that $[X,\tilde Y]$
satisfies condition \eqref{100502:eq1bis}:
\begin{equation}\label{101206:eq4}
\begin{array}{l}
\displaystyle{
[X,\tilde Y]_{\lambda_0,\dots,\lambda_k}(f_0,\dots,f_k)
} \\
\displaystyle{
-\sum_{i\in I}
[X,\tilde Y]_{\lambda_0,\dots,\lambda_{k-1},\lambda_k+\partial}(f_0,\dots,f_{k-1},u_i)_\to
E^i_{\lambda_k}(f_k) = 0\,.
}
\end{array}
\end{equation}
We consider separately the left and right box products $X\Box^L\tilde Y$
and $\tilde Y\Box^R X$, defined by \eqref{101024:eq1}.
We get, after a long but straightforward computation,
\begin{equation}\label{101206:eq1}
\begin{array}{l}
\displaystyle{
(X\Box^L \tilde Y)_{\lambda_0,\dots,\lambda_k}(f_0,\dots,f_k)
=\sum_{\substack{
\alpha_0<\dots <\alpha_{k-h}=k\\
\alpha_{k-h+1}<\dots< \alpha_k}}
\sign(\alpha)
\sum_{i\in I}
} \\
\displaystyle{
\tilde Y_{\lambda_{\alpha_0},\dots,\lambda_{\alpha_{k-h-1}},
\lambda_k+\lambda_{\alpha_{k-h+1}}+\dots+\lambda_{\alpha_k}+\partial}
(f_{\alpha_0},\dots, f_{\alpha_{k-h-1}},u_i)_\to
} \\
\displaystyle{
X_{-\lambda_{\alpha_{k-h+1}}-\dots-\lambda_{\alpha_k}-\partial,
\lambda_{\alpha_{k-h+1}},\dots,\lambda_{\alpha_k}}
\Big(E^i_{\lambda_k}(f_k),f_{\alpha_{k-h+1}},\dots, f_{\alpha_k}\Big)
} \\
\displaystyle{
+ \sum_{i\in I}
(X\Box^L \tilde Y)_{\lambda_0,\dots,\lambda_{k-1},\lambda_k+\partial}(f_0,\dots,f_{k-1},u_i)_\to
E^i_{\lambda_k}(f_k) \,,
}
\end{array}
\end{equation}
where the first sum in the RHS runs over all permutations $\alpha$ of $\{0,\dots,k\}$
satisfying the specified inequalities.
Similarly, we have
\begin{equation}\label{101206:eq2}
\begin{array}{l}
\displaystyle{
(\tilde Y\Box^R X)_{\lambda_0,\dots,\lambda_k}(f_0,\dots,f_k)
=\sum_{\substack{
\alpha_0<\dots <\alpha_{h}=k\\
\alpha_{h+1}<\dots< \alpha_k}}
\sign(\alpha)
\sum_{i\in I}
} \\
\displaystyle{
\tilde Y_{\lambda_{\alpha_0}+\dots+\lambda_{\alpha_{h-1}}+\lambda_{\alpha_k}+\partial
\lambda_{\alpha_{h+1}},\dots,\lambda_{\alpha_k}}
(E^i_{\lambda_k}(f_k),f_{\alpha_{h+1}},\dots, f_{\alpha_k})_\to
} \\
\displaystyle{
X_{\lambda_{\alpha_0},\dots,\lambda_{\alpha_{h-1}},
-\lambda_{\alpha_0}-\dots-\lambda_{\alpha_{h-1}}-\partial
}
\Big(
f_{\alpha_0},\dots, f_{\alpha_{h-1}},u_i\Big)
} \\
\displaystyle{
+ \sum_{i\in I}
(\tilde Y\Box^R \tilde X)_{\lambda_0,\dots,\lambda_{k-1},\lambda_k+\partial}(f_0,\dots,f_{k-1},u_i)_\to
E^i_{\lambda_k}(f_k) \,.
}
\end{array}
\end{equation}
Combining \eqref{101206:eq1} and \eqref{101206:eq2}, we get that
the LHS of \eqref{101206:eq4} is
$$
\begin{array}{l}
\displaystyle{
\sum_{\substack{
\alpha_0<\dots <\alpha_{k-h}=k\\
\alpha_{k-h+1}<\dots< \alpha_k}}
\!\!\!\!\!\!\!\!\!
\sign(\alpha)
\sum_{i\in I}
\tilde Y_{\lambda_{\alpha_0},\dots,\lambda_{\alpha_{k-h-1}},
\lambda_{\alpha_{k-h}}+\dots+\lambda_{\alpha_k}+\partial}
(f_{\alpha_0},\dots, f_{\alpha_{k-h-1}},u_i)_\to
} \\
\displaystyle{
\,\,\,\,\,\,\,\,\,\,\,\,\,\,\,\,\,\,
\vphantom{\Bigg(}
X_{-\lambda_{\alpha_{k-h+1}}-\dots-\lambda_{\alpha_k}-\partial,
\lambda_{\alpha_{k-h+1}},\dots,\lambda_{\alpha_k}}
\Big(E^i_{\lambda_k}(f_k),f_{\alpha_{k-h+1}},\dots, f_{\alpha_k}\Big)
} \\
\displaystyle{
- (-1)^{h(k-h)}
\!\!\!\!\!\!\!\!\!
\sum_{\substack{
\alpha_0<\dots <\alpha_{h}=k\\
\alpha_{h+1}<\dots< \alpha_k}}
\!\!\!\!\!\!\!
\sign(\alpha)
\sum_{i\in I}
\tilde Y_{\lambda_{\alpha_0}+\dots+\lambda_{\alpha_h}+\partial,
\lambda_{\alpha_{h+1}},\dots,\lambda_{\alpha_k}}
(E^i_{\lambda_k}(f_k),f_{\alpha_{h+1}},\dots
} \\
\displaystyle{
\,\,\,\,\,\,\,\,\,\,\,\,\,\,\,\,\,\,
\vphantom{\Bigg(}
\dots, f_{\alpha_k})_\to
X_{\lambda_{\alpha_0},\dots,\lambda_{\alpha_{h-1}},
-\lambda_{\alpha_0}-\dots-\lambda_{\alpha_{h-1}}-\partial
}
\Big(
f_{\alpha_0},\dots, f_{\alpha_{h-1}},u_i\Big)\,.
}
\end{array}
$$
We can then use conditions \eqref{100530:eq2} and \eqref{100430:eq2} on $X$ and $\tilde Y$
to rewrite the above expression as follows
$$
\begin{array}{c}
\displaystyle{
\sum_{\substack{
\alpha_0<\dots <\alpha_{k-h}=k\\
\alpha_{k-h+1}<\dots< \alpha_k}}
\!\!\!\!\!\!\!\!\!
\sign(\alpha)\!
\sum_{i,j\in I}\!
\Bigg(\!
\tilde Y_{\lambda_{\alpha_0},\dots,\lambda_{\alpha_{k-h-1}},
\lambda_{\alpha_{k-h}}+\dots+\lambda_{\alpha_k}+\partial}
(f_{\alpha_0},\dots, f_{\alpha_{k-h-1}},u_j)_\to
\!\!\!\!
} \\
\displaystyle{
\Big(
{}^iE_{\lambda_{\alpha_{k-h+1}}+\dots+\lambda_{\alpha_k}+\partial}\big(E^j_{\lambda_k}(f_k)\big)_\to
-
E^j_{\lambda_{\alpha_{k-h}}+\dots+\lambda_{\alpha_k}+\partial}\big(E^i_{\lambda_k}(f_k)\big)_\to
\Big)
} \\
\displaystyle{
\vphantom{\Bigg(}
X_{-\lambda_{\alpha_{k-h+1}}-\dots-\lambda_{\alpha_k}-\partial,
\lambda_{\alpha_{k-h+1}},\dots,\lambda_{\alpha_k}}
\Big(u_i,f_{\alpha_{k-h+1}},\dots, f_{\alpha_k}\Big)
\Bigg)\,,
}
\end{array}
$$
where ${}^iE_\lambda$ was introduced in Remark \ref{100529:rem}.
To conclude, we observe that the above expression is zero, thanks to the following identity,
$$
{}^iE_\mu\big(E^j_\lambda(f)\big)=E^j_{\lambda+\mu}\big(E^i_\lambda(f)\big)\,,
$$
which can be easily checked: both sides above are equal to
$$
\displaystyle{\sum_{m,n\in\mb Z_+}\mu^m(-\lambda-\mu-\partial)^n
\frac{\partial^2 f}{\partial u_i^{(m)}\partial u_j^{(n)}}}\,.
$$

The first statement in part (d) is obvious from  Proposition \ref{100517:prop}.
We thus only need to prove that the map
$\tint:\,\tilde W^{\var}(\Pi\mc V)\to W^{\var}(\Pi\mc V)$
is surjective, so that the induced map
$\tint:\,\tilde W^{\var}(\Pi\mc V)/\partial\tilde W^{\var}(\Pi\mc V)\to W^{\var}(\Pi\mc V)$
is bijective.
Let $X\in W^{\var}_k$. We can construct a representative $\tilde X\in\tilde W^{\var}_k$
as follows.
Recall that we can identify
$\mb F_-[\lambda_0,\dots,\lambda_k]\otimes_{\mb F[\partial]}\mc V
\simeq\mb F[\lambda_0,\dots,\lambda_{k-1}]\otimes\mc V$,
by letting $\lambda_k=-\lambda_0-\dots-\lambda_{k-1}-\partial$.
For $i_0,\dots,i_k\in I$, we then have
$X_{\lambda_0,\dots,\lambda_{k-1}}(u_{i_0},\dots,u_{i_k})
\in
\mb F[\lambda_0,\dots,\lambda_{k-1}]\otimes\mc V$.
We let
$$
\tilde X_{\lambda_0,\dots,\lambda_k}(u_{i_0},\dots,u_{i_k})
=
\frac{1}{k+1}\sum_{\alpha=0}^k(-1)^{k-\alpha}
X_{\lambda_0,\stackrel{\alpha}{\check{\dots}},\lambda_{k}}
(u_{i_0},\stackrel{\alpha}{\check{\dots}},u_{i_k},u_{i_\alpha})\,,
$$
and we extend it to a map
$\tilde X:\,\mc V^{\otimes(k+1)}\to\mb F[\lambda_0,\dots,\lambda_k]\otimes\mc V$
by the master formula \eqref{100430:eq2}.
By construction $\tilde X$ lies in $\tilde W^{\var}_k$, and we clearly have $\tint\tilde X=X$,
since they agree on generators.
\end{proof}

Recall that $W^{\var}_{-1}=\mc V/\partial\mc V,\,W^{\var}_{0}=\Vect^\partial(\mc V)$
and $W^{\var}_1$ coincides with the space of skewcommutative $\lambda$-brackets on $\mc V$
satisfying the master equation \eqref{100502:eq2}.
We can then compute explicitly, using \eqref{100517:eq2},
the action of $X\in W^{\var}_h$ on $\tilde W^{\var}_{k-h}$, for $h=-1,0,1$.
The formulas for these actions coincide, respectively, with formulas
\eqref{100503:eq1}, \eqref{100503:eq5} and \eqref{100503:eq11},
where all the terms are considered as elements of
$\mb F[\lambda_1,\dots,\lambda_k]\otimes\mc V$
(not of $\mb F[\lambda_1,\dots,\lambda_k]\otimes_{\mb F[\partial]}\mc V$).

\subsection{PVA structures on an algebra of differential functions and basic cohomology complexes}
\label{sec:10.2}

Let $K\in W^{\var}_1$ be such that $[K,K]=0$,
and denote by $\{\cdot\,_\lambda\,\cdot\}_K$
the corresponding Poisson $\lambda$-bracket on $\mc V$,
and by $K(\partial)$ the corresponding Hamiltonian operator given by \eqref{100430:eq1b}.
Then $(\ad K)^2=0$, hence the action of $K$ on $\tilde W^{\var}(\Pi\mc V)$,
given by Proposition \ref{100422x:prop3}(b),
provides a differential on $\tilde W^{\var}(\Pi\mc V)$, which we denote by $d_K$.
We call $(\tilde W^{\var}(\Pi\mc V),d_K)$
the \emph{basic Poisson cohomology complex} of the PVA $\mc V$
with the $\lambda$-bracket $\{\cdot\,_\lambda\,\cdot\}_K$.
Note that, by Proposition \ref{100422x:prop3}(b),
the differential $d_K$ is an odd derivation of both the concatenation product
and the $\lambda$-bracket of $\tilde W^{\var}(\Pi \mc V)$.
Moreover, by Proposition \ref{100422x:prop3}(c),
the linear map $\tint$ is a homomorphism of cohomology complexes
$(\tilde W^{\var}(\Pi \mc V),d_K)\to(W^{\var}(\Pi \mc V),\ad K)$.

Recalling \eqref{100503:eq11} we get, by the observations at the end of Section \ref{sec:10.1},
the following explicit formula for the differential $d_K$ associated to $K\in W^{\var}_1$,
acting on $Y\in\tilde W^{\var}_{k-1}$:
\begin{equation}\label{100503x:eq11}
\begin{array}{c}
\displaystyle{
(d_K Y)_{\lambda_0,\dots,\lambda_k}(f_0,\dots,f_k)
= (-1)^{k+1}\bigg(
\sum_{i=0}^k (-1)^{i}
\big\{{f_i}_{\lambda_i}
 Y_{\lambda_0,\stackrel{i}{\check{\dots}},\lambda_k}(f_0,\stackrel{i}{\check{\dots}},f_k)\big\}_K
} \\
\displaystyle{
+ \sum_{0\leq i<j\leq k}
(-1)^{i+j}
Y_{\lambda_i+\lambda_j,\lambda_0,\stackrel{i}{\check{\dots}}\stackrel{j}{\check{\dots}},\lambda_k}
\big(\{{f_i}_{\lambda_i}f_j\}_K,f_0,\stackrel{i}{\check{\dots}}\,\stackrel{j}{\check{\dots}},f_k\big)
\bigg)\,.
}
\end{array}
\end{equation}
\begin{remark}\label{100530:rem}
The identity map on $\mc V=\tilde \Omega^0(\mc V)=\tilde W^{\var}_{-1}$
extends to a homomorphism of cohomology complexes
$\tilde \Phi_K:\,(\tilde \Omega^\bullet(\mc V),\delta)\to(\tilde W^{\var}(\Pi\mc V),d_K)$
given by formula \eqref{100505:eq1},
where both sides are interpreted as elements in
$\mb F[\lambda_0,\dots,\lambda_k]\otimes\mc V$,
not in $\mb F[\lambda_0,\dots,\lambda_k]\otimes_{\mb F[\partial]}\mc V$
as for $\Phi_K$ from Remark \ref{100505:th}.
We thus have the following commutative diagram of homomorphisms of cohomology complexes:
$$
\UseTips
\xymatrix{
{\big(\tilde\Omega^\bullet(\mc V),\delta\big)} \ar@{>>}[d]
& \ar[r]^{\tilde\Phi_K} & & {\big(\tilde W^{\var}(\Pi\mc V),d_K\big)} \ar@{>>}[d]
\\
{\big(\tilde\Omega^\bullet(\mc V)/\partial \tilde\Omega^\bullet(\mc V),\delta\big)} \ar@{=}[d]
& & & {\big(\tilde W^{\var}(\Pi\mc V)/\partial\tilde W^{\var}(\Pi\mc V),d_K\big)} \ar@{^{(}->}[d]
\\
{\big(\Omega^\bullet(\mc V),\delta\big)}
& \ar[r]^{\Phi_K} & & {\big(W^{\var}(\Pi\mc V),d_K\big)}
}
$$
\end{remark}

\subsection{Identification of $\tilde W^{\var}(\Pi\mc V)$ with $\tilde\Omega^\bullet(\mc V)$}
\label{sec:10.3}

In this section we will construct a (non covariant) explicit identification of the superspaces
$\tilde W^{\var}_k$ and $\tilde \Omega^{k+1}(\mc V)$ for finite $\ell$,
along the lines of the discussion in Section \ref{sec:9.3}.

An element $X\in\tilde W^{\var}_k$ is completely determined,
via the master equation \eqref{100430:eq2}, by the collection of polynomials
$X_{\lambda_0,\dots,\lambda_k}(u_{i_0},\dots,u_{i_k})
\in\mb F[\lambda_0,\dots,\lambda_k]\otimes\mc V$,
with $i_1,\dots,i_k\in I$.
Hence, we construct a linear map
$$
\tilde \Phi:\,\tilde \Omega^\bullet(\mc V)\to\tilde W^{\var}(\Pi\mc V)\,,
$$
sending $P\in\tilde \Omega^{k+1}(\mc V)$ to $\tilde \Phi^kP\in W^{\var}_k$,
such that
\begin{equation}\label{100518x:eq5}
(\tilde \Phi^kP)_{\lambda_0,\dots,\lambda_k}(u_{i_0},\dots,u_{i_k})
=
P_{i_0,\dots,i_k}(\lambda_0,\dots,\lambda_k)\,,
\end{equation}
and it is extended to a map
$\tilde \Phi^kP:\,\mc V^{\otimes(k+1)}\to\mb F[\lambda_0,\dots,\lambda_k]\otimes\mc V$
by the master equation \eqref{100430:eq2}.
Clearly, $\tilde \Phi$ is surjective for finite $\ell$, and it is injective in general.

Let $\ell$ be finite. Let $K\in W^{\var}_1$ be such that $[K,K]=0$,
and consider the Poisson cohomology complex $(\tilde W^{\var}(\Pi\mc V),d_K)$.
Using the bijection $\tilde \Phi:\,\tilde\Omega^\bullet(\mc V)\to \tilde W^{\var}(\Pi\mc V)$,
we get a differential $d_K:\,\tilde \Omega^k(\mc V)\to\tilde \Omega^{k+1}(\mc V)$
induced by the action of $d_K$ on $\tilde W^{\var}(\Pi\mc V)$.
Recalling equation \eqref{100503x:eq11}, we get that the explicit formula for the differential
$d_K$ on $\tilde\Omega^k(\mc V)$ is given by \eqref{100518:eq1},
where both sides are considered as elements of
$\mb F[\lambda_1,\dots,\lambda_k]\otimes\mc V$
(not of $\mb F[\lambda_1,\dots,\lambda_k]\otimes_{\mb F[\partial]}\mc V$):
\begin{equation}\label{100518b:eq1}
\begin{array}{l}
\displaystyle{
(d_KP)_{i_0,\dots,i_k}(\lambda_0,\dots,\lambda_k)
} \\
\displaystyle{
= (-1)^{k+1} \sum_{j\in I,n\in\mb Z_+} \bigg(
\sum_{\alpha=0}^k (-1)^{\alpha}
\frac{\partial
P_{i_0,\stackrel{\alpha}{\check{\dots}},i_k}(\lambda_0,\stackrel{\alpha}{\check{\dots}},\lambda_k)
}{\partial u_j^{(n)}}
(\lambda_\alpha+\partial)^n
K_{j,i_\alpha}(\lambda_\alpha)
} \\
\displaystyle{
+ \sum_{0\leq \alpha<\beta\leq k}
(-1)^{\alpha+\beta}
P_{j,i_0,\stackrel{\alpha}{\check{\dots}}\,\stackrel{\beta}{\check{\dots}},i_k}
(\lambda_\alpha+\lambda_\beta+\partial,
\lambda_0,\stackrel{\alpha}{\check{\dots}}\,\stackrel{\beta}{\check{\dots}},\lambda_k)_\to
} \\
\displaystyle{
\,\,\,\,\,\,\,\,\,\,\,\,\,\,\,\,\,\,\,\,\,\,\,\,\,\,\,\,\,\,\,\,\,\,\,\,\,\,\,\,\,\,
\,\,\,\,\,\,\,\,\,\,\,\,\,\,\,\,\,\,\,\,\,\,\,\,\,\,\,\,\,\,\,\,\,\,\,\,\,\,\,\,\,\,
(-\lambda_\alpha-\lambda_\beta-\partial)^n
\frac{\partial K_{i_\beta,i_\alpha}(\lambda_\alpha)}{\partial u_j^{(n)}}
\bigg)\,.
}
\end{array}
\end{equation}

\subsection{Generalized de Rham complexes}
\label{sec:10.4}

Let $\ell$ be finite, and let $K(\partial)=\big(K_{ij}(\partial)\big)_{i,j\in I}$
be an $\ell\times\ell$ matrix differential operator with quasiconstant coefficients.
We define the map $\delta_K:\,\tilde\Omega^k(\mc V)\to\tilde\Omega^{k+1}(\mc V)$
by the same formula \eqref{100519:eq1}, interpreting both sides as elements of
$\mb F[\lambda_1,\dots,\lambda_k]\otimes\mc V$
(not of $\mb F[\lambda_1,\dots,\lambda_k]\otimes_{\mb F[\partial]}\mc V$):
\begin{equation}\label{100519b:eq1}
\begin{array}{l}
\displaystyle{
(\delta_KP)_{i_0,\dots,i_k}(\lambda_0,\dots,\lambda_k)
} \\
\displaystyle{
=
\sum_{\alpha=0}^k (-1)^{\alpha} \sum_{j\in I,n\in\mb Z_+}
\frac{\partial
P_{i_0,\stackrel{\alpha}{\check{\dots}},i_k}(\lambda_0,\stackrel{\alpha}{\check{\dots}},\lambda_k)
}{\partial u_j^{(n)}}
(\lambda_\alpha+\partial)^n
K_{j,i_\alpha}(\lambda_\alpha)\,.
}
\end{array}
\end{equation}

Note that, when $K=\id$ this coincides with the differential $\delta$ of the de Rham complex
defined in equation \eqref{100518:eq3},
and when $K(\partial)$ is a skewadjoint operator
it coincides with $(-1)^{k+1}d_K$, where $d_K$ is given by equation \eqref{100518b:eq1}.
\begin{proposition}\label{100519b:prop}
\begin{enumerate}[(a)]
\item
If $K(\partial)$ is an $\ell\times\ell$ matrix differential operator with quasiconstant coefficients,
then $(\tilde\Omega^\bullet(\mc V),\delta_K)$ is a cohomology complex, i.e. $\delta_K^2=0$.
\item
The differential $\delta_K$ in \eqref{100519b:eq1} commutes with the action of $\partial$
on $\tilde\Omega^\bullet(\mc V)$ given by \eqref{100518:eq4}.
Moreover, the canonical quotient map $\tilde\Omega^\bullet(\mc V)\to\Omega^\bullet(\mc V)$
gives a homomorphism of complexes
$(\tilde\Omega^\bullet(\mc V),\delta_K)\to(\Omega^\bullet(\mc V),\delta_K)$.
\end{enumerate}
\end{proposition}
\begin{proof}
Let $P\in\tilde\Omega^k(\mc V)$ and let $\sigma\in S_{k+1}= Perm(0,\dots,k)$.
We have
$$
\begin{array}{l}
\displaystyle{
(\delta_KP)_{i_{\sigma(0)},\dots,i_{\sigma(k)}}(\lambda_{\sigma(0)},\dots,\lambda_{\sigma(k)})
\sum_{\alpha=0}^k (-1)^{\alpha}
} \\
\displaystyle{
=
\sum_{j\in I,n\in\mb Z_+}
\frac{\partial
P_{i_{\sigma(0)},\stackrel{\alpha}{\check{\dots}},i_{\sigma(k)}}(\lambda_{\sigma(0)},\stackrel{\alpha}{\check{\dots}},\lambda_{\sigma(k)})
}{\partial u_j^{(n)}}
(\lambda_{\sigma(\alpha)}+\partial)^n
K_{j,i_{\sigma(\alpha)}}(\lambda_{\sigma(\alpha)})\,.
}
\end{array}
$$
Note that the permutation $(\sigma(0),\stackrel{\alpha}{\check{\dots}},\sigma(k))$
of the set $\{0,\stackrel{\sigma(\alpha)}{\check{\dots}},k\}$ has sign $\sign(\sigma)(-1)^{\alpha+\sigma(\alpha)}$.
Hence, by the skewsymmetry condition on $P$, we can write the RHS above as
$$
\begin{array}{l}
\displaystyle{
\sum_{\alpha=0}^k (-1)^{\alpha} \sum_{j\in I,n\in\mb Z_+}
\sign(\sigma)(-1)^{\alpha+\sigma(\alpha)}
\frac{\partial
P_{i_0,\stackrel{\sigma(\alpha)}{\check{\dots}},i_k}(\lambda_0,\stackrel{\sigma(\alpha)}{\check{\dots}},\lambda_k)
}{\partial u_j^{(n)}}
} \\
\displaystyle{
\times(\lambda_{\sigma(\alpha)}+\partial)^n
K_{j,i_{\sigma(\alpha)}}(\lambda_{\sigma(\alpha)})
=\sign(\sigma)
(\delta_KP)_{i_0,\dots,i_k}(\lambda_0,\dots,\lambda_k)\,.
}
\end{array}
$$
This shows that $\delta_KP$ is a skewsymmetric array, i.e. $\delta_K$
is a well defined map from $\tilde\Omega^k(\mc V)$ to $\tilde\Omega^{k+1}(\mc V)$.

Since, by assumption, $K$ has quasiconstant coefficients, we have
$$
\begin{array}{l}
\displaystyle{
(\delta^2_KP)_{i_0,\dots,i_k}(\lambda_0,\dots,\lambda_k)
} \\
\displaystyle{
=
\sum_{\beta=0}^k (-1)^{\beta} \sum_{j\in I,n\in\mb Z_+}
\frac{\partial}{\partial u_j^{(n)}}
(\delta_K P)_{i_0,\stackrel{\beta}{\check{\dots}},i_k}(\lambda_0,\stackrel{\beta}{\check{\dots}},\lambda_k)
(\lambda_\beta+\partial)^n
K_{j,i_\beta}(\lambda_\beta)
} \\
\displaystyle{
=
\sum_{0\leq\alpha<\beta\leq k} (-1)^{\alpha+\beta} \sum_{i,j\in I,m,n\in\mb Z_+}
\frac{\partial}{\partial u_j^{(n)}} \frac{\partial}{\partial u_i^{(m)}}
P_{i_0,\stackrel{\alpha}{\check{\dots}}\,\stackrel{\beta}{\check{\dots}},i_k}
(\lambda_0,\stackrel{\alpha}{\check{\dots}}\,\stackrel{\beta}{\check{\dots}},\lambda_k)
} \\
\displaystyle{
\,\,\,\,\,\,\,\,\,\,\,\,\,\,\,\,\,\,\,\,\,\,\,\,\,\,\,\,\,\,\,\,\,\,\,\,\,\,\,\,\,\,\,\,\,\,\,\,\,\,\,\,\,\,\,\,\,\,
\times
\big((\lambda_\alpha+\partial)^m K_{i,i_\alpha}(\lambda_\alpha)\big)
\big((\lambda_\beta+\partial)^n K_{j,i_\beta}(\lambda_\beta)\big)
} \\
\displaystyle{
+
\sum_{0\leq\beta<\alpha\leq k} (-1)^{\alpha+\beta+1} \sum_{i,j\in I,m,n\in\mb Z_+}
\frac{\partial}{\partial u_j^{(n)}} \frac{\partial}{\partial u_i^{(m)}}
P_{i_0,\stackrel{\beta}{\check{\dots}}\,\stackrel{\alpha}{\check{\dots}},i_k}
(\lambda_0,\stackrel{\beta}{\check{\dots}}\,\stackrel{\alpha}{\check{\dots}},\lambda_k)
} \\
\displaystyle{
\,\,\,\,\,\,\,\,\,\,\,\,\,\,\,\,\,\,\,\,\,\,\,\,\,\,\,\,\,\,\,\,\,\,\,\,\,\,\,\,\,\,\,\,\,\,\,\,\,\,\,
\times
\big((\lambda_\alpha+\partial)^m K_{i,i_\alpha}(\lambda_\alpha)\big)
\big((\lambda_\beta+\partial)^n K_{j,i_\beta}(\lambda_\beta)\big)
=0\,,
}
\end{array}
$$
which completes the proof of (a).

Next, we prove part (b). We have, by the definition \eqref{100518:eq4} of the $\mb F[\partial]$-module structure on $\tilde\Omega^k(\mc V)$,
$$
\begin{array}{l}
\displaystyle{
(\delta_K\partial P)_{i_0,\dots,i_k}(\lambda_0,\dots,\lambda_k)
=
\sum_{\alpha=0}^k (-1)^{\alpha}
} \\
\displaystyle{
\sum_{j\in I,n\in\mb Z_+}
\frac{\partial}{\partial u_j^{(n)}}
\Big(
(\lambda_0+\stackrel{\alpha}{\check{\dots}}+\lambda_k+\partial)
P_{i_0,\stackrel{\alpha}{\check{\dots}},i_k}(\lambda_0,\stackrel{\alpha}{\check{\dots}},\lambda_k)
\Big)
(\lambda_\alpha+\partial)^n
K_{j,i_\alpha}(\lambda_\alpha)
} \\
\displaystyle{
=
(\lambda_0+\dots+\lambda_k+\partial)
\sum_{\alpha=0}^k (-1)^{\alpha} \sum_{j\in I,n\in\mb Z_+}
} \\
\displaystyle{
\frac{\partial}{\partial u_j^{(n)}}
\Big(\!
P_{i_0,\stackrel{\alpha}{\check{\dots}},i_k}(\lambda_0,\stackrel{\alpha}{\check{\dots}},\lambda_k)
\!\Big)
(\lambda_\alpha+\partial)^n
K_{j,i_\alpha}(\lambda_\alpha)
=(\partial\delta_K P)_{i_0,\dots,i_k}\!(\lambda_0,\dots,\lambda_k).
}
\end{array}
$$
In the second equality we used \eqref{eq:comm_frac_b}.
The last assertion in (b) is obvious.
\end{proof}


\section{Computation of the variational Poisson cohomology}
\label{sec:11}

Throughout this section we assume that $\mc V$ is a normal algebra of differential functions
in finitely many differential variables $u_i,\,i\in I=\{1,\dots,\ell\}$,
i.e. $\frac{\partial}{\partial u_i^{(m)}}\mc V_{m,i}=\mc V_{m,i}$
for every $i\in I$ and $m\in\mb Z_+$,
where $\mc V_{m,i}$ is given by \eqref{eq:july21_1}.
We also assume that the space of quasiconstants $\mc F\subset\mc V$ is a field
(hence, the subalgebra of constants $\mc C\subset\mc F$ is a subfield).

We shall compute the cohomology of the complexes $(\tilde\Omega^\bullet(\mc V),\delta_K)$
and $(\Omega^\bullet(\mc V),\delta_K)$,
for any $\ell\times\ell$ matrix differential operator
$K(\partial)=\big(K_{ij}(\partial)\big)_{i,j\in I}$ of order $N$
with quasiconstant coefficients and invertible leading coefficient.
We use here the same approach as in \cite{BDSK}, where the case $K=\id$ was considered.

\subsection{Formality of the generalized de Rham complex}
\label{sec:11.1}

Fix a non-negative integer $N$.
We extend the filtration \eqref{eq:july21_1} of the algebra of differential functions
$\mc V=\tilde\Omega^0(\mc V)$,
to a filtration, depending on $N$,
of $\tilde\Omega^\bullet(\mc V)$.
For $m\in\mb Z_+$ and $i\in I$, we let
$\tilde\Omega^\bullet_{m,i}(\mc V)=\bigoplus_{k\in\mb Z_+}\tilde\Omega^k_{m,i}(\mc V)$,
where $\tilde\Omega^k_{m,i}(\mc V)$ consists of arrays
$P=\big(P_{i_1,\dots,i_k}(\lambda_1,\dots,\lambda_k)\big)_{i_1,\dots,i_k\in I}\in\tilde\Omega^k(\mc V)$
such that
\begin{equation}\label{eq:july24_1}
\begin{array}{l}
\frac{\partial}{\partial u_j^{(n)}}P_{i_1,\dots,i_k}(\lambda_1,\dots,\lambda_k) = 0
\,,\,\,\text{ if } (n,j)>(m,i)\,, \\
\partial_{\lambda_\alpha}^{n+N}P_{i_1,\dots,i_k}(\lambda_1,\dots,\lambda_k) = 0
\,,\,\,\text{ if } (n,i_\alpha)>(m,i)\,,
\end{array}
\end{equation}
for all $i_1,\dots,i_k\in I$,
where the inequalities are understood in the lexicographic order.
In other words, the coefficients of all the polynomials $P_{i_1,\dots,i_k}(\lambda_1,\dots,\lambda_k)$
lie in $\mc V_{m,i}$,
and, moreover, $P_{i_1,\dots,i_k}(\lambda_1,\dots,\lambda_k)$
has degree at most $m+N$ (resp. $m-1+N$)
in each variable $\lambda_\alpha$ with $i_\alpha\leq i$ (resp. $i_\alpha>i$).
We also let $\tilde\Omega^k_{n,0}(\mc V)=\tilde\Omega^k_{n-1,\ell}(\mc V)$ for $n\geq1$.

Finally, we let
$\tilde\Omega^\bullet_{0,0}=\bigoplus_{k\in\mb Z_+}\tilde\Omega^k_{0,0}$,
where $\tilde\Omega^k_{0,0}$ is a subspace of $\tilde\Omega^k(\mc V)$
consisting of arrays
$P=\big(P_{i_1,\dots,i_k}(\lambda_1,\dots,\lambda_k)\big)_{i_1,\dots,i_k\in I}$
(skewsymmetric with respect to simultaneous permutations of indexes and variables),
whose entries $P_{i_1,\dots,i_k}(\lambda_1,\dots,\lambda_k)$
are polynomials of degree at most $N-1$ in each variable $\lambda_1,\dots,\lambda_k$,
with quasiconstant coefficients.
In particular, $\tilde\Omega^0_{0,0}=\mc F$.
By the skewsymmetry condition,
\begin{equation}\label{100520:eq1}
\tilde\Omega^k_{0,0}=0 \,\,\text{ if }\,\,k>N\ell\,.
\end{equation}

Note that, if $K(\partial)$
is an $\ell\times\ell$ matrix differential operator with quasiconstant coefficients,
then the differential $\delta_K$ defined in \eqref{100519b:eq1}
is zero on $\tilde\Omega^\bullet_{0,0}$,
so that $(\tilde\Omega^\bullet_{0,0},0)$ is a subcomplex
of the generalized de Rham complex $(\tilde\Omega^\bullet(\mc V),\delta_K)$
Note also that $\tilde\Omega^\bullet(\mc V)$
is naturally a vector space over $\mc F$,
and the differential $\delta_K$ is $\mc F$-linear.
Also, all the $\tilde\Omega^k_{m,i}(\mc V)$ are $\mc F$-linear subspaces.
\begin{remark}\label{20120305:rem}
Due to formula \eqref{100430:eq2} (cf. Proposition \ref{100422x:prop1}) 
the restriction of the $\lambda$-bracket
from $\tilde W^{\var}(\Pi\mc V)=\tilde\Omega^\bullet(\mc V)$
to the subspace $\tilde\Omega^\bullet_{0,0}$ is zero.
Hence, $\tilde\Omega^\bullet_{0,0}$ is a subalgebra of the Lie conformal superalgebra
$\tilde\Omega^\bullet(\mc V)$.
\end{remark}

%
%

In this section we prove the following generalization of Theorem \ref{100430:th}(a):
\begin{theorem}\label{100520:th}
Let $\mc V$ be a normal algebra of differential functions
and assume that the subalgebra of quasiconstants $\mc F\subset\mc V$ is a field.
Let $K(\partial)=\big(K_{ij}(\partial)\big)_{i,j\in I}$
be an $\ell\times\ell$ matrix differential operator of order $N$ with quasiconstant coefficients,
and invertible leading coefficient $K_N\in\Mat_{\ell\times\ell}(\mc F)$.
Then:
\begin{enumerate}[(a)]
\item
The inclusion
$(\tilde\Omega^\bullet_{0,0},0)
\subset
(\tilde\Omega^\bullet(\mc V),\delta_K)$,
is a quasiisomorphism of complexes,
i.e. it induces a canonical Lie conformal superalgebra isomorphism of cohomology:
$$
H^k(\tilde\Omega^\bullet(\mc V),\delta_K)\simeq\tilde\Omega^k_{0,0}\,.
$$
\item
For $k\geq0$,
$H^k(\tilde\Omega^\bullet(\mc V),\delta_K)$
is a vector space over $\mc F$ of dimension $\binom{N\ell}{k}$.
\end{enumerate}
\end{theorem}

The proof of Theorem \ref{100520:th} consists of several steps.
First, we prove three lemmas which will be used in its proof.

In analogy with \eqref{100505:eq1},
for $S=\big(S_{ij}\big)_{i,j\in I}\in\Mat_{\ell\times\ell}(\mc F)$,
we define the map
$\Phi_S:\,\tilde\Omega^\bullet(\mc V)\to\tilde\Omega^\bullet(\mc V),\,P\mapsto \Phi_S(P)$,
given by the following equation
\begin{equation}\label{100520:eq2}
(\Phi_SP)_{i_1,\dots,i_k}(\lambda_1,\dots,\lambda_k)
=
\!\!\!
\sum_{j_1,\dots,j_k\in I}
\!\!\!
P_{j_1,\dots,j_k}(\lambda_1+\partial_1,\dots,\lambda_k+\partial_k)S_{j_1i_1}\dots S_{j_ki_k}
\,,
\end{equation}
where, as usual, $\partial_\alpha$ denotes $\partial$ acting on $S_{j_\alpha i_\alpha}$.
\begin{lemma}\label{100520:lem1}
\begin{enumerate}[(a)]
\item
For every $S\in\Mat_{\ell\times\ell}(\mc F)$, we have
$\Phi_S(\tilde\Omega^k(\mc V))\subset \tilde\Omega^k(\mc V)$
and $\Phi_S(\tilde\Omega^k_{0,0})\subset \tilde\Omega^k_{0,0}$.
\item
If $K(\partial)$ is an
$\ell\times\ell$ matrix differential operator with quasiconstant coefficients,
then
$$
\Phi_S(\delta_KP)=\delta_{K\circ S}\Phi_S(P)\,,
$$
where $(K\circ S)(\partial)=\big(\sum_{r\in I}K_{ir}(\partial)\circ S_{rj}\big)_{i,j\in I}$.
In other words, $\Phi_S$ is a homomorphism of complexes:
$(\tilde\Omega^\bullet(\mc V),\delta_K)\to(\tilde\Omega^\bullet(\mc V),\delta_{K\circ S})$.
\item
For $S,T\in\Mat_{\ell\times\ell}(\mc F)$, we have
$$
\Phi_S\circ \Phi_T=\Phi_{TS}\,.
$$
\item
If $S\in\Mat_{\ell\times\ell}(\mc F)$ is an invertible matrix,
then
$$
\Phi_S:\,(\tilde\Omega^\bullet(\mc V),\delta_K)
\stackrel{\sim}{\longrightarrow}
(\tilde\Omega^\bullet(\mc V),\delta_{K\circ S})\,,
$$
is an isomorphism of complexes,
which restricts to an automorphism of the subcomplex $(\tilde\Omega^\bullet_{0,0},0)$.
\end{enumerate}
\end{lemma}
\begin{proof}
Part (a) is clear.
Part (b) is proved by a straightforward computation using the definitions \eqref{100519b:eq1}
and \eqref{100520:eq2} of the differential $\delta_K$ and the map $\Phi_S$.
Part (c) is again straightforward, using the definition \eqref{100520:eq2} of $\Phi_S$.
Finally, part (d) immediately follows from (a), (b) and (c).
\end{proof}

\begin{lemma}\label{100520:lem2}
Let $K(\partial)=\big(K_{ij}(\partial)\big)_{i,j\in I}$
be an $\ell\times\ell$ matrix differential operator of order $N$ with quasiconstant coefficients,
and assume that its leading coefficient $K_N\in\Mat_{\ell\times\ell}(\mc F)$ is diagonal.
Then,
$$
\delta_K\big(\tilde\Omega^k_{m,i}(\mc V)\big)\,\subset\,\tilde\Omega^{k+1}_{m,i}(\mc V)\,,
$$
for every $k\in\mb Z_+,\,m\in\mb Z_+,\,i\in I$.
\end{lemma}
\begin{proof}
Let $P=\big(P_{i_1,\dots,i_k}(\lambda_1,\dots,\lambda_k)\big)_{i_1,\dots,i_k\in I}\in\tilde\Omega^k_{m,i}(\mc V)$.
By the definition \eqref{100519b:eq1} of $\delta_K$, we have,
using the assumption that $K(\partial)$ has quasiconstant coefficients,
$$
\begin{array}{l}
\displaystyle{
\frac{\partial}{\partial u_j^{(n)}}(\delta_KP)_{i_0,\dots,i_k}(\lambda_0,\dots,\lambda_k)
} \\
\displaystyle{
=
\sum_{\alpha=0}^k (-1)^{\alpha} \sum_{r\in I,p\in\mb Z_+}
\frac{\partial^2
P_{i_0,\stackrel{\alpha}{\check{\dots}},i_k}(\lambda_0,\stackrel{\alpha}{\check{\dots}},\lambda_k)
}{\partial u_r^{(p)}\partial u_j^{(n)}}
(\lambda_\alpha+\partial)^p
K_{r,i_\alpha}(\lambda_\alpha)\,,
}
\end{array}
$$
which is zero if $(n,j)>(m,i)$ by the assumption on $P$.
Next, we have
\begin{equation}\label{100520:eq3}
\begin{array}{l}
\displaystyle{
\partial_{\lambda_\alpha}^{n+N}
(\delta_KP)_{i_0,\dots,i_k}(\lambda_0,\dots,\lambda_k)
} \\
\displaystyle{
=
\sum_{\beta\neq\alpha}
(-1)^{\beta} \sum_{r\in I,p\in\mb Z_+}
\frac{\partial
}{\partial u_r^{(p)}}
\Big(
\partial_{\lambda_\alpha}^{n+N}
P_{i_0,\stackrel{\beta}{\check{\dots}},i_k}(\lambda_0,\stackrel{\beta}{\check{\dots}},\lambda_k)
\Big)
(\lambda_\beta+\partial)^p
K_{r,i_\alpha}(\lambda_\alpha)
} \\
\displaystyle{
+
(-1)^{\alpha} \sum_{r\in I}\sum_{p=0}^m
\frac{\partial
P_{i_0,\stackrel{\alpha}{\check{\dots}},i_k}(\lambda_0,\stackrel{\alpha}{\check{\dots}},\lambda_k)
}{\partial u_r^{(p)}}
\partial_{\lambda_\alpha}^{n+N}
(\lambda_\alpha+\partial)^p
K_{r,i_\alpha}(\lambda_\alpha)
\,.
}
\end{array}
\end{equation}
The first sum in the RHS of \eqref{100520:eq3} is zero for $(n,j)>(m,i)$
by the assumption on $P$.
Moreover, the second sum in the RHS of \eqref{100520:eq3} is zero for $n>m$
since, by assumption, $K(\lambda)$ has degree $N$.
Let then $n=m$ in this sum.
We have
$$
\partial_{\lambda_\alpha}^{m+N}(\lambda_\alpha+\partial)^pK_{r,i_\alpha}(\lambda_\alpha)
=(m+N)!\delta_{p,m}(K_N)_{r,i_\alpha}\,.
$$
Since, by assumption, $K_N$ is a diagonal matrix,
and since $\frac{\partial P_{i_1,\dots,i_k}(\lambda_1,\dots,\lambda_k)}{\partial u_r^{(m)}}$
is zero for $r>i$,
we conclude that the second sum in the RHS of \eqref{100520:eq3} is zero for $i_\alpha>i$,
as required.
\end{proof}

Now we are going to use the assumption that that $\mc V$ is normal and $\mc F\subset\mc V$ is a field.
Given $i\in I,m\in\mb Z_+$, choose an $\mc F$-subspace
$\mc U_{m,i}\subset\mc V_{m,i}$ complementary to the kernel of
the map $\frac{\partial}{\partial u_i^{(m)}}:\,\mc V_{m,i}\to\mc V_{m,i}$.
By normality of $\mc V$,
$\frac\partial{\partial u_i^{(m)}}$ restricts to an $\mc F$-linear isomorphism
$\frac{\partial}{\partial u_i^{(m)}}:\,\mc U_{m,i}\stackrel{\sim}{\longrightarrow}\mc V_{m,i}$,
and we denote by
$\tint du_i^{(m)}\cdot:\,\mc V_{m,i}\stackrel{\sim}{\longrightarrow}\mc U_{m,i}\subset\mc V_{m,i}$
the inverse $\mc F$-linear map,
so that
$\frac{\partial}{\partial u_i^{(m)}}\tint du_i^{(m)}f=f$ for every $f\in\mc V_{m,i}$.
Clearly, if we change the choice of the complementary subspace $\mc U_{m,i}$,
the integral $\tint du_i^{(m)} f\in\mc V_{m,i}$ of $f\in\mc V_{m,i}$
changes by adding an element of $\mc V_{m,i-1}$.

We extend the antiderivative $\tint du_i^{(m)}\cdot$ to the space of polynomials
in $\lambda_1,\dots,\lambda_k$ with coefficients in $\mc V_{m,i}$
by applying it to coefficients.
Clearly, the operators $\partial_{\lambda_\alpha}$ and $\tint du_i^{(m)}\cdot$,
acting on $\mb F[\lambda_1,\dots,\lambda_k]\otimes\mc V_{m,i}$, commute.

We define the \emph{local homotopy operators}
$h_{m,i}:\,\tilde\Omega^{k+1}_{m,i}(\mc V)\to\tilde\Omega^{k}_{m,i}(\mc V),\,k\geq0$,
by the following formula
\begin{equation}\label{eq:july24_4}
(h_{m,i}P)_{i_1,\dots,i_k}(\lambda_1,\dots,\lambda_k)
=
\tint du_i^{(m)}
\frac{\partial_{\mu}^{m+N}}{(m\!+\!N)!}
P_{i,i_1,\dots,i_k}(\mu,\lambda_1,\dots,\lambda_k)
\,.
\end{equation}
\begin{lemma}\label{100519:lem}
Let $P=\big(P_{i_0,\dots,i_k}(\lambda_0,\dots,\lambda_k)\big)_{i_0,\dots,i_k\in I}
\in\tilde\Omega^{k+1}_{m,i}(\mc V)$. Then:
\begin{enumerate}[(a)]
\item
$h_{m,i}P\in\tilde\Omega^k_{m,i}(\mc V)$.
\item
If $P\in\tilde\Omega^{k+1}_{m,i-1} \big( \subset\tilde\Omega^{k+1}_{m,i}(\mc V)\big)$, then $h_{m,i}P=0$.
\item
If $K(\partial)$ is an $\ell\times\ell$ matrix differential operator of order $N$
with quasiconstant coefficients and leading coefficient $\id$,
the operator $h_{m,i}$ satisfies the following homotopy condition:
\begin{equation}\label{100519:eq2}
h_{m,i}(\delta_KP)+\delta_K (h_{m,i}P)-P\,\in\,\tilde\Omega^k_{m,i-1}(\mc V)\,.
\end{equation}
\end{enumerate}
\end{lemma}
\begin{proof}
Clearly, $(h_{m,i}P)_{i_1,\dots,i_k}(\lambda_1,\dots,\lambda_k)$
is skewsymmetric with respect to simultaneous permutations of the variables $\lambda_1,\dots,\lambda_k$
and of the indexes $i_1,\dots,i_k$.
Moreover, by assumption on $P$ and by the definition of $\tint du_i^{(m)}\cdot$,
the coefficients of all the polynomials $(h_{m,i}P)_{i_1,\dots,i_k}(\lambda_1,\dots,\lambda_k)$
lie in $\mc V_{m,i}$.
Furthermore, if $(n,i_\alpha)>(m,i)$ with $\alpha\in\{1,\dots,k\}$, we have,
$$
\begin{array}{l}
\frac{\partial_{\lambda_\alpha}^{n+N}}{(n+N)!}(h_{m,i}P)_{i_1,\dots,i_k}(\lambda_1,\dots,\lambda_k) \\
\displaystyle{
= \tint du_i^{(m)}
\frac{\partial_{\mu}^{m+N}}{(m\!+\!N)!}
\frac{\partial_{\lambda_\alpha}^{n+N}}{(n+N)!}
P_{i,i_1,\dots,i_k}(\mu,\lambda_1,\dots,\lambda_k)=0
\,,
}
\end{array}
$$
by the assumption on $P$. Hence, $h_{m,i}P\in\tilde\Omega^k_{m,i}$, proving (a).
Part (b) is clear since, by definition, $P\in\tilde\Omega^{k+1}_{m,i-1}$
are such that $P_{i,i_1,\dots,i_k}(\mu,\lambda_1,\dots,\lambda_k)$
is a polynomial of degree at most $m+N-1$ in the variable $\mu$.
We are left to prove part (c).
By \eqref{100519b:eq1} and \eqref{eq:july24_4}, we have, for $P\in\tilde\Omega^{k+1}_{m,i}(\mc V)$,
\begin{equation}\label{100520:eq7}
\begin{array}{l}
\displaystyle{
(\delta_K h_{m,i}P)_{i_0,\dots,i_k}(\lambda_0,\dots,\lambda_k)
=
\sum_{\beta=0}^k
\sum_{j\in I,n\in\mb Z_+}
\frac{\partial}{\partial u_j^{(n)}}
} \\
\displaystyle{
\tint du_i^{(m)}
\big((\lambda_\beta+\partial)^n K_{ji_\beta}(\lambda_\beta)\big)
\frac{\partial_{\lambda_\beta}^{m+N}}{(m\!+\!N)!}
P_{i_0,\dots,\stackrel{\beta}{\check{i}},\dots,i_k}(\lambda_0,\dots,\lambda_k)
}
\end{array}
\end{equation}
(here we used the fact that $\tint du_i^{(m)}$ is $\mc F$-linear
and $K_{ij}(\lambda)$ has coefficients in $\mc F$),
and
\begin{equation}\label{100520:eq8}
\begin{array}{c}
\displaystyle{
(h_{m,i}\delta_K P)_{i_0,\dots,i_k}(\lambda_0,\dots,\lambda_k)
=
- \sum_{\beta=0}^k
\sum_{j\in I,n\in\mb Z_+}
\tint du_i^{(m)}
\frac{\partial}{\partial u_j^{(n)}}
} \\
\displaystyle{
\big((\lambda_\beta+\partial)^n K_{ji_\beta}(\lambda_\beta)\big)
\frac{\partial_{\lambda_\beta}^{m+N}}{(m\!+\!N)!}
P_{i_0,\dots,\stackrel{\beta}{\check{i}},\dots,i_k}(\lambda_0,\dots,\lambda_k)
} \\
\displaystyle{
+ \tint du_i^{(m)}
\frac{\partial}{\partial u_i^{(m)}}\
P_{i_0,\dots,i_k}(\lambda_0,\dots,\lambda_k) \,.
}
\end{array}
\end{equation}
By Lemma \ref{100520:lem2} and parts (a) and (b),
we know that $\delta_K(h_{m,i}P),\,h_{m,i}(\delta_K P)$ and $P$
all lie in $\tilde\Omega^{k+1}_{m,i}$.
Hence, in order to prove equation \eqref{100519:eq2} we only need to prove the following
two identities:
\begin{equation}\label{100520:eq6}
\begin{array}{l}
\displaystyle{
\frac{\partial}{\partial u_i^{(m)}}
\Big(
(h_{m,i}\delta_KP)_{i_0,\dots,i_k}(\lambda_0,\dots,\lambda_k)+
(\delta_K h_{m,i}P)_{i_0,\dots,i_k}(\lambda_0,\dots,\lambda_k)
\Big)
}\\
\displaystyle{
= \frac{\partial}{\partial u_i^{(m)}} P_{i_0,\dots,i_k}(\lambda_0,\dots,\lambda_k)\,;
}\\
\displaystyle{
\frac{\partial_{\lambda_\alpha}^{m+N}}{(m\!+\!N)!}
\Big(
(h_{m,i}\delta_KP)_{i_0,\dots,i_k}(\lambda_0,\dots,\lambda_k)+
(\delta_K h_{m,i}P)_{i_0,\dots,i_k}(\lambda_0,\dots,\lambda_k)
\Big)
}\\
\displaystyle{
= \frac{\partial_{\lambda_\alpha}^{m+N}}{(m\!+\!N)!}
P_{i_0,\dots,i_k}(\lambda_0,\dots,\lambda_k)
\,\,,\,\,\,\,\text{ if } i_\alpha=i \,.
}
\end{array}
\end{equation}
The first identity of \eqref{100520:eq6}
follows immediately from equations \eqref{100520:eq7} and \eqref{100520:eq8},
using that $\frac{\partial}{\partial u_i^{(m)}}\circ\tint du_i^{(m)}f=f$ for every $f\in\mc V_{m,i}$.
The second identity in \eqref{100520:eq6} follows by a straightforward computation
using the following two facts.
Since, by assumption, the leading coefficient of $K(\partial)$ is $\id$,
we have, for $(n,j)\leq(m,i)$,
$$
\frac{\partial_{\lambda_\alpha}^{m+N}}{(m\!+\!N)!}
\big((\lambda_\alpha+\partial)^n K_{ji}(\lambda_\alpha)\big)
=
\delta_{n,m}\delta_{j,i}\,.
$$
Moreover, by the skewsymmetry condition on $P$, we have,
if $i_\alpha=i_\beta=i$ for $\beta\neq\alpha$,
$$
\partial_{\lambda_\alpha}^{m+N}\partial_{\lambda_\beta}^{m+N}
P_{i_0,\dots,,i_k}(\lambda_0,\dots,\lambda_k)=0\,.
$$
\end{proof}

\begin{proof}[Proof of Theorem \ref{100520:th}]
By Lemma \ref{100520:lem1}(d), we have isomorphism of complexes
$\Phi_{K_N^{-1}}:\,(\tilde\Omega^\bullet,\delta_K)\to(\tilde\Omega^\bullet,\delta_{K\circ K_N^{-1}})$,
which induces an automorphism of the subcomplex $(\tilde\Omega^\bullet_{0,0},0)$.
Hence,
replacing $K(\partial)$ by $(K\circ K_N^{-1})(\partial)$,
it suffices to prove (a) for $K(\partial)$ with leading coefficient $\id$.

Let $P\in\tilde\Omega^k(\mc V)$ be such that $\delta_KP=0$.
For some $i\in I,m\in\mb Z_+$ we have $P\in\tilde\Omega^k_{m,i}(\mc V)$ and,
by Lemma \ref{100519:lem}(c), we have
$P=\delta_K(h_{m,i}P)+P_1$,
for some $P_1\in\tilde\Omega^k_{m,i-1}(\mc V)$ such that $\delta_KP_1=0$.
Repeating the same argument finitely many times, we get
that $P=\delta_KQ+R$, for some $Q\in\tilde\Omega^{k-1}_{m,i}(\mc V)$
and $R\in\tilde\Omega^k_{0,0}$.
Hence,
$$
\ker\big(\delta_K:\,\tilde\Omega^k(\mc V)\to\tilde\Omega^{k+1}(\mc V)\big)
=\delta_k\big(\tilde\Omega^{k-1}(\mc V)\big)+\tilde\Omega^k_{0,0}\,.
$$
To prove part (a) it remains to show that
$$
\delta_k\big(\tilde\Omega^{k-1}(\mc V)\big)\cap\tilde\Omega^k_{0,0}=0\,.
$$
Let $P=\delta_KQ\in\tilde\Omega^k_{0,0}$,
for some $Q\in\tilde\Omega^{k-1}_{m,i}$.
By Lemma \ref{100519:lem}(c), we have
$Q=\delta_K(h_{m,i}Q)+h_{m,i}(\delta_KQ)+Q_1$, for some $Q_1\in\tilde\Omega^{k-1}_{m,i-1}$,
and, by Lemma \ref{100519:lem}(b), we have
$h_{m,i}(\delta_KQ)=h_{m,i}P=0$.
Hence, $P=\delta_KQ_1\in\delta_K\tilde\Omega^{k-1}_{m,i-1}(\mc V)$.
Repeating the same argument finitely many times, we then get
$P\in\delta_K\tilde\Omega^{k-1}_{0,0}=0$.

Next, we prove part (b).
An element $P\in\tilde\Omega^k_{0,0}$
is uniquely determined by the collection of polynomials
$$
P_{\underbrace{1,..\,,1}_{n_1},\dots,\underbrace{\ell,..\,,\ell}_{n_\ell}}
(\lambda_1,\dots,\lambda_k)
\,\in\mb F[\lambda_1,\dots,\lambda_k]\otimes\mc F
\,,
$$
where $n_1,\dots,n_\ell\geq0$
are such that $n_1+\dots+n_\ell=k$,
which have degree at most $N-1$ in each variable $\lambda_\alpha,\,\alpha=1,\dots,k$,
and, for every $i=1,\dots,\ell$,
are skewsymmetric in the variables
$\lambda_{n_1+\dots+n_{i-1}+1},\dots,\lambda_{n_1+\dots+n_i}$.
Hence, $\tilde\Omega^k_{0,0}$ is a vector space over $\mc F$
of dimension $R_k$, given by
\begin{equation}\label{100521:eq1}
R_k=\sum_{\substack{n_1,\dots,n_\ell\in\mb Z_+ \\ n_1+\dots+n_\ell=k}}
\prod_{i=1}^\ell C(N,n_i)\,,
\end{equation}
where $C(N,n)$ is the dimension of the space of skewsymmetric polynomials
in $n$ variables of degree at most $N-1$ in each variable, i.e.
$C(N,n)=\binom{N}{n}$.
Taking the generating series of both sides of equation \eqref{100521:eq1},
we then get
$$
\sum_{k=0}^\infty R_kz^k=
\Big(
\sum_{n=0}^\infty \binom{N}{n} z^n\Big)^\ell
=(1+z)^{N\ell}\,,
$$
which implies $R_k=\binom{N\ell}{k}$, as required.
\end{proof}

\subsection{Cohomology of the generalized variational complex}
\label{sec:11.2}

Recall from Proposition \ref{100519b:prop}(b)
that we have a short exact sequence of complexes
$$
0\to(\partial\tilde\Omega^\bullet(\mc V),\delta_K)\stackrel{\alpha}{\to}(\tilde\Omega^\bullet(\mc V),\delta_K)
\stackrel{\beta}{\to}(\Omega^\bullet(\mc V),\delta_K)\to 0\,,
$$
where $\alpha$ is the inclusion map, and $\beta$ is the canonical quotient map
$\tilde\Omega^\bullet(\mc V)
\to\tilde\Omega^\bullet(\mc V)/\partial\tilde\Omega^\bullet(\mc V)=\Omega^\bullet(\mc V)$.
It induces a long exact sequence in cohomology:
\begin{equation}\label{100602:eq1}
\begin{array}{l}
\vphantom{\Bigg(}
0\to H^0(\partial\tilde\Omega^\bullet,\delta_K)
\stackrel{\alpha_0}{\to}
H^0(\tilde\Omega^\bullet,\delta_K)
\stackrel{\beta_0}{\to}
H^0(\Omega^\bullet,\delta_K)
\stackrel{\gamma_0}{\to}
H^{1}(\partial\tilde\Omega^\bullet,\delta_K)
\stackrel{\alpha_{1}}{\to}
\dots \\
\vphantom{\Bigg(}
\dots
\!\!\stackrel{\gamma_{k-1}}{\to}\!
H^k(\partial\tilde\Omega^\bullet,\delta_K)
\!\stackrel{\alpha_k}{\to}\!
H^k(\tilde\Omega^\bullet,\delta_K)
\!\stackrel{\beta_k}{\to}\!
H^k(\Omega^\bullet,\delta_K)
\!\stackrel{\gamma_k}{\to}\!
H^{k+1}(\partial\tilde\Omega^\bullet,\delta_K)
\!\stackrel{\alpha_{k+1}}{\to}\!\!
\dots
\end{array}
\end{equation}
Recall that, by Theorem \ref{100520:th},
for every $k\in\mb Z_+$, we have a canonical identification
$H^k(\tilde\Omega^\bullet(\mc V),\delta_K)=\tilde\Omega^k_{0,0}$.
Next, we want to describe $H^k(\partial\tilde\Omega^\bullet(\mc V),\delta_K)$
and the map
$\alpha_k:\,H^k(\partial\tilde\Omega^\bullet(\mc V),\delta_K)\to H^k(\tilde\Omega^\bullet(\mc V),\delta_K)$.
This is given by the following
\begin{lemma}\label{100602:lem1}
\begin{enumerate}[(a)]
\item
The inclusion
$(\partial\tilde\Omega^\bullet_{0,0},0)
\subset
(\partial\tilde\Omega^\bullet(\mc V),\delta_K)$,
is a quasiisomorphism of complexes,
i.e. it induces canonical isomorphisms:
$$
H^k(\partial\tilde\Omega^\bullet(\mc V),\delta_K)
\simeq \partial\tilde\Omega^k_{0,0}
\simeq
\left\{\begin{array}{ll}
\mc F/\mc C & \text{ for } k=0 \\
\tilde\Omega^k_{0,0} & \text{ for } k\geq1
\end{array}\right.
\,.
$$
\item
Under the identifications
$H^0(\tilde\Omega^\bullet(\mc V),\delta_K)=\mc F$
in Theorem \ref{100520:th} and
$H^0(\partial\tilde\Omega^\bullet(\mc V),\delta_K)=\mc F/\mc C$
in part (a),
the map $\alpha_0$ induces the map $\alpha_0:\,\mc F/\mc C\to\mc F$
given by $a+\mc C\mapsto\partial a$.
\item
For $k\geq1$, identifying
$H^k(\tilde\Omega^\bullet(\mc V),\delta_K)
=\tilde\Omega^k_{0,0}
=H^k(\partial\tilde\Omega^\bullet(\mc V),\delta_K)$
as in Theorem \ref{100520:th} and in part (a),
the map $\alpha_k$ induces the endomorphism $\alpha_k\in\End\big(\tilde\Omega^k_{0,0}\big)$
defined as follows.
For $P\in\tilde\Omega^k_{0,0}$,
there exist $Q\in\tilde\Omega^k(\mc V)$ and (a unique)
$R\in\tilde\Omega^k_{0,0}$
such that $\partial P=\delta_K Q+R$.
Then,
\begin{equation}\label{100602:eq2}
\alpha_k(P)=R\,.
\end{equation}
\item
Assuming that the leading coefficient of $K(\partial)$ is $\id$,
we can write $\alpha_k$ explicitly using the local homotopy operators \eqref{eq:july24_4}:
\begin{equation}\label{100602:eq3}
\alpha_k(P)
=
(1-\delta_K\circ h_{0,1})(1-\delta_K\circ h_{0,2})\dots(1-\delta_K\circ h_{0,\ell})\partial P\,.
\end{equation}
\end{enumerate}
\end{lemma}
\begin{proof}
The map $\partial:\,\Omega^k(\mc V)\to\Omega^k(\mc V)$ is injective for $k\geq1$,
while, for $k=0$, we have $\tilde\Omega^0(\mc V)=\mc V$ and $\ker(\partial|_{\mc V})=\mc C$,
the algebra of constants.
Since $\partial$ and $\delta_K$ commute, it follows that
$\ker\big(\delta_K\big|_{\partial\tilde\Omega^k(\mc V)}\big)
=\partial\ker\big(\delta_K\big|_{\tilde\Omega^k(\mc V)}\big)\subset\tilde\Omega^k(\mc V)$
for all $k\geq0$, and
$\delta_K\big(\partial\tilde\Omega^{k-1}(\mc V)\big)
=\partial\delta_K\big(\tilde\Omega^{k-1}(\mc V)\big)\subset\tilde\Omega^k(\mc V)$ for $k\geq1$.
Hence, we get
$$
H^0(\partial\tilde\Omega^\bullet(\mc V),\delta_K)
=\ker\big(\delta_K\big|_{\partial\tilde\Omega^0(\mc V)}\big)
=\partial\ker\big(\delta_K\big|_{\tilde\Omega^0(\mc V)}\big)
=\partial\mc F
\simeq\mc F/\mc C\,.
$$
In the second last equality we used the fact that 
$\ker\big(\delta_K\big|_{\tilde\Omega^0(\mc V)}\big)=\mc F$,
by the definition \eqref{100519b:eq1} of $\delta_K$ and that, by assumption, $K$
has invertible leading coefficient.
Moreover, the last isomorphism above is induced by
the surjective map $\partial:\,\mc F\to\partial\mc F$.
Similarly, for $k\geq1$, we have
$$
\begin{array}{l}
H^k(\partial\tilde\Omega^\bullet(\mc V),\delta_K)
=\ker\big(\delta_K\big|_{\partial\tilde\Omega^k(\mc V)}\big)
\big/ \delta_K\big(\partial\tilde\Omega^{k-1}(\mc V)\big) \\
\,\,\,\,\,\,\,\,\,\,\,\,\,\,\,\,\,\,
=\partial\ker\big(\delta_K\big|_{\tilde\Omega^k(\mc V)}\big)
\big/ \partial\delta_K\big(\tilde\Omega^{k-1}(\mc V)\big) \\
\,\,\,\,\,\,\,\,\,\,\,\,\,\,\,\,\,\,
\simeq
\ker\big(\delta_K\big|_{\tilde\Omega^k(\mc V)}\big)
\big/ \delta_K\big(\tilde\Omega^{k-1}(\mc V)\big)
= H^k(\tilde\Omega^\bullet(\mc V),\delta_K)
\simeq\tilde\Omega^k_{0,0}
\,.
\end{array}
$$
In the third identity we used the injectiveness of $\partial$. This proves part (a).

The map $\alpha_0:\,H^0(\partial\tilde\Omega(\mc V),\delta_K)\to H^0(\tilde\Omega(\mc V),\delta_K)$
is induced by the inclusion map
$\partial\tilde\Omega^0(\mc V)=\partial\mc V\subset\mc V=\tilde\Omega^0(\mc V)$.
Since $H^0(\tilde\Omega(\mc V),\delta_K)=\mc F$
and $H^0(\partial\tilde\Omega(\mc V),\delta_K)=\partial\mc F$,
the map $\alpha_0$ coincides with the inclusion map $\partial\mc F\subset\mc F$.
Part (b) follows from the identification
$\mc F/\mc C\simeq\partial\mc F$ via the map $a+\mc C\mapsto\partial a$.

For part (c) we use a similar argument.
The isomorphism $H^k(\tilde\Omega(\mc V),\delta_K)\simeq\tilde\Omega^k_{0,0}$
given by Theorem \ref{100520:th}(a)
maps $P+\delta_K(\tilde\Omega^{k-1}(\mc V))\in H^k(\tilde\Omega(\mc V),\delta_K)$
to the unique element $R\in\tilde\Omega^k_{0,0}$
such that $P-R\in\delta_K(\tilde\Omega^{k-1}(\mc V))$,
and the inverse map sends $P\in\tilde\Omega^k_{0,0}$
to $P+\delta_K(\tilde\Omega^{k-1}(\mc V))\in H^k(\tilde\Omega(\mc V),\delta_K)$.
Similarly, we have the canonical isomorphism,
$H^k(\partial\tilde\Omega(\mc V),\delta_K)\simeq\tilde\Omega^k_{0,0}$,
which maps
$\partial P+\delta_K(\partial\tilde\Omega^{k-1}(\mc V))\in H^k(\partial\tilde\Omega(\mc V),\delta_K)$
to the unique element $R\in\tilde\Omega^k_{0,0}$
such that $P-R\in\delta_K(\tilde\Omega^{k-1}(\mc V))$,
and the inverse map sends $P\in\tilde\Omega^k_{0,0}$
to $\partial P+\delta_K(\partial\tilde\Omega^{k-1}(\mc V))\in H^k(\partial\tilde\Omega(\mc V),\delta_K)$.
Equation \eqref{100602:eq2} follows
from the fact that
the map $\alpha_k:\,H^k(\partial\tilde\Omega(\mc V),\delta_K)\to H^k(\tilde\Omega(\mc V),\delta_K)$
is induced by the inclusion map
$\partial\tilde\Omega^k(\mc V)\subset\tilde\Omega^k(\mc V)$,
i.e. it sends
$\partial P+\delta_K(\partial\tilde\Omega^{k-1}(\mc V))\in H^k(\partial\tilde\Omega(\mc V),\delta_K)$
to
$\partial P+\delta_K(\tilde\Omega^{k-1}(\mc V))\in H^k(\tilde\Omega(\mc V),\delta_K)$.

We are left to prove part (d).
Given $P=\big(P_{i_1,\dots,i_k}(\lambda_1,\dots,\lambda_k)\big)_{i_1,\dots,i_k\in I}\in\tilde\Omega^k_{0,0}$,
the entries of the array $\partial P\in\tilde\Omega^k(\mc V)$
are the polynomials
$(\partial+\lambda_1+\dots+\lambda_k)P_{i_1,\dots,i_k}(\lambda_1,\dots,\lambda_k)$,
which have quasiconstant coefficients and have degree at most $N$ in each $\lambda_i$.
Hence, $\partial P\in\tilde\Omega^k_{0,\ell}(\mc F)\subset\tilde\Omega^k_{0,\ell}(\mc V)$.
It follows by Lemma \ref{100519:lem}(c) that
$(\id-\delta_K\circ h_{0,1})(\id-\delta_K\circ h_{0,2})\dots(\id-\delta_K\circ h_{0,\ell})\partial P$
lies in $\tilde\Omega^k_{0,0}$.
Since, obviously, this element differs from $\partial P$ by an exact element,
we conclude, by part (c), that it coincides with $\alpha_k(P)$.
\end{proof}

Using Theorem \ref{100520:th} and Lemma \ref{100602:lem1},
the long exact sequence \eqref{100602:eq1} becomes
\begin{equation}\label{100521:eq2}
\begin{array}{l}
\vphantom{\Bigg(}
0\to \mc F/\mc C
\stackrel{\partial}{\to}
\mc F
\stackrel{\beta_0}{\to}
H^0(\Omega^\bullet(\mc V),\delta_K)
\stackrel{\gamma_0}{\to}
\tilde\Omega^1_{0,0}
\stackrel{\alpha_{1}}{\to}
\tilde\Omega^1_{0,0}
\stackrel{\beta_{1}}{\to}
\dots \\
\vphantom{\Bigg(}
\dots
\stackrel{\gamma_{k-1}}{\to}
\tilde\Omega^k_{0,0}
\stackrel{\alpha_k}{\to}
\tilde\Omega^k_{0,0}
\stackrel{\beta_k}{\to}
H^k(\Omega^\bullet(\mc V),\delta_K)
\stackrel{\gamma_k}{\to}
\tilde\Omega^{k+1}_{0,0}
\stackrel{\alpha_{k+1}}{\to}
\tilde\Omega^{k+1}_{0,0}
\stackrel{\beta_{k+1}}{\to}
\dots
\end{array}
\end{equation}

Next, we study the maps $\beta_k,\,k\in\mb Z_+$.
First, it is clear that
$\beta_0:\,\tilde\Omega^0_{0,0}=\mc F\to H^0(\Omega^\bullet(\mc V),\delta_K)\subset\mc V/\partial\mc V$,
is given by $f\mapsto\tint f$.
In particular, $\beta_0=0$ if and only if $\partial\mc F=\mc F$.

For $k\in\mb Z_+$, let us consider the map
$\beta_{k+1}:\,\tilde\Omega^{k+1}_{0,0}\to H^{k+1}(\Omega^\bullet(\mc V),\delta_K)$.
Let $P\in\tilde\Omega^{k+1}_{0,0}$,
i.e. $P=\big(P_{i_0,\dots,i_k}(\lambda_0,\dots,\lambda_k)\big)_{i_0,\dots,i_k\in I}$
is a skewsymmetric array with respect to simultaneous permutations
of the indices $i_0,\dots,i_k$ and the variables $\lambda_0,\dots,\lambda_k$ and,
for each $k$-tuple $(i_0,\dots,i_k)$,
$P_{i_0,\dots,i_k}(\lambda_0,\dots,\lambda_k)\in\mb F[\lambda_0,\dots,\lambda_k]\otimes\mc F$
is a polynomial of degree at most $N-1$ in each variable $\lambda_i$.
Then, by definition, $\beta_{k+1}(P)\in H^k(\Omega^\bullet(\mc V),\delta_K)$ is
\begin{equation}\label{100627:eq1}
\beta_{k+1}(P)=\big(P_{i_0,\dots,i_k}(\lambda_0,\dots,\lambda_k)\big)_{i_0,\dots,i_k\in I}
+\delta_K\big(\Omega^{k}(\mc V)\big)\,,
\end{equation}
where $P_{i_0,\dots,i_k}(\lambda_0,\dots,\lambda_k)$
should now be viewed as an element of the space $\mb F_-[\lambda_0,\dots,\lambda_k]\otimes_{\mb F[\partial]}\mc F$.

Note that the space of exact elements $\delta_K\big(\Omega^{k}(\mc V)\big)$
contains all arrays
\begin{equation}\label{100627:eq2}
\Big(
\sum_{\alpha=0}^k (-1)^{\alpha} \sum_{j\in I}
K^*_{i_\alpha,j}(\lambda_0+\stackrel{\alpha}{\check{\dots}}+\lambda_k+\partial)
Q^j_{i_0,\stackrel{\alpha}{\check{\dots}},i_k}(\lambda_0,\stackrel{\alpha}{\check{\dots}},\lambda_k)
\Big)_{i_0,\dots,i_k\in I}\,,
\end{equation}
where
$Q^j_{i_1,\dots,i_{k}}(\lambda_1,\dots,\lambda_{k})\in\mb F[\lambda_1,\dots,\lambda_k]\otimes\mc F$
are polynomials with quasiconstant coefficients,
and they are skewsymmetric with respect to simultaneous permutations
of $i_1,\dots,i_{k}$ and $\lambda_1,\dots,\lambda_{k}$.
Indeed, recalling the definition \eqref{100519:eq1} of $\delta_K$,
we have that $\delta_K$ applied to the array
\begin{equation}\label{100627:eq3}
\Big(\sum_{j\in I}Q^j_{i_1,\dots,i_{k}}(\lambda_1,\dots,\lambda_{k})u_j\Big)_{i_1,\dots,i_{k}\in I}
\in\Omega^{k}(\mc V)\,,
\end{equation}
gives \eqref{100627:eq2}.

For example, for $k=0$, given $P=\big(P_i(\lambda)\big)_{i\in I}\in\tilde\Omega^1_{0,0}$,
we have
$\beta_1(P)=\big(p_i\big)_{i\in I}+\delta_K\big(\Omega^0(\mc V)\big)$,
where $p_i=P_i^*(0)\in\mc F$.
On the other hand, for $Q^j=f_j\in\mc F$, the array \eqref{100627:eq2} becomes
$\big(\sum_{j\in I}K^*_{i,j}(\partial)f_j\big)_{i\in I}$.
Hence, $\beta_1=0$ provided that the map
$K^*(\partial):\,\mc F^\ell\to\mc F^\ell$ is surjective.

In general, for $k\geq0$, let
$P=\big(P_{i_0,\dots,i_k}(\lambda_0,\dots,\lambda_k)\big)_{i_0,\dots,i_k\in I}
\in \tilde\Omega^{k+1}_{0,0}$,
i.e. $P$ is a skewsymmetric array with respect to simultaneous permutations
of the indices $i_0,\dots,i_k$ and the variables $\lambda_0,\dots,\lambda_k$ and
such that, for each $k$-tuple $(i_0,\dots,i_k)$,
$P_{i_0,\dots,i_k}(\lambda_0,\dots,\lambda_k)\in\mb F[\lambda_0,\dots,\lambda_k]\otimes\mc F$
is a polynomial of degree at most $N-1$ in each variable $\lambda_i$.
By equation \eqref{100627:eq1} and formulas \eqref{100627:eq2} and \eqref{100627:eq3},
we have that $\beta_{k+1}(P)=0$ provided that
there exists a collection of $\ell$ skewsymmetric arrays in $k$ variables
$$
\Big(Q^j_{i_1,\dots,i_{k}}(\lambda_1,\dots,\lambda_{k})\Big)_{i_1,\dots,i_{k}\in I}\,,
$$
indexed by $j=1,\dots,\ell$, where $Q^j_{i_1,\dots,i_{k}}(\lambda_1,\dots,\lambda_{k})$
are polynomials with quasiconstant coefficients,
such that
\begin{equation}\label{110213:eq1}
\begin{array}{l}
\displaystyle{
\sum_{\alpha=0}^k (-1)^{\alpha} \sum_{j\in I}
K^*_{i_\alpha,j}(\lambda_0+\stackrel{\alpha}{\check{\dots}}+\lambda_k+\partial)
Q^j_{i_0,\stackrel{\alpha}{\check{\dots}},i_k}(\lambda_0,\stackrel{\alpha}{\check{\dots}},\lambda_k)
}\\
\displaystyle{
-
P_{i_0,i_1,\dots,i_{k}}(\lambda_0,\lambda_1,\dots,\lambda_{k})
\in(\lambda_0+\dots+\lambda_k+\partial)\mb F[\lambda_0,\dots,\lambda_k]\otimes\mc F\,.
}
\end{array}
\end{equation}

Recall Definition \ref{def:linclosed} from Appendix \ref{app:3}
of a linearly closed differential field.
\begin{theorem}\label{100602:lem2}
Let $\mc V$ be a normal algebra of differential functions,
and assume that the algebra of quasiconstants $\mc F\subset\mc V$
is a linearly closed differential field.
Let $K(\partial)$ be an $\ell\times\ell$ matrix differential operator with quasiconstant coefficients
and invertible leading coefficient.
Then $\beta_k=0$ for every $k\geq0$ and every $\ell\geq1$.
\end{theorem}
\begin{proof}
The facts that $\beta_0=0$ and $\beta_1=0$ were pointed out above.
Using notation \eqref{110120:eq6} and \eqref{110126:eq4} in Appendix \ref{app:5},
we have that the array \eqref{100627:eq2} is equal to $\langle K^*\circ Q\rangle^-$.
Hence, condition \eqref{110213:eq1},
after replacing $\lambda_0$ by $-\lambda_1-\dots-\lambda_k-\partial$,
becomes $P=\langle K^*\circ Q\rangle^-$.
Hence, the assertion that $\beta_{k+1}(P)=0$ follows from
Theorem \ref{110127:conj1} in the Appendix.
\end{proof}
\begin{example}
In the case $k=1$ and $\ell=1$, the problem of solving
equation \eqref{110213:eq1} becomes:
for $P(\lambda)\in\mc F[\lambda]$ skewadjoint,
i.e. $P^*(\lambda):=P(-\lambda-\partial)=-P(\lambda)$,
we want to find $Q(\lambda)\in\mc F[\lambda]$,
such that
$P(\lambda)
=
Q^*(\lambda+\partial)K(\lambda)
-K^*(\lambda+\partial)Q(\lambda)$.
Solutions for certain choices of $K(\lambda)$ are the following:
\begin{enumerate}
\item
$K(\lambda)=1$: take $Q(\lambda)=\frac12 P(\lambda)$,
\item
$K(\lambda)=\lambda$: take $Q(\lambda)$ such that $\partial Q(\lambda)=P(\lambda)$,
\item
$K(\lambda)=\lambda^2$: take $Q(\lambda)=Q^*(\lambda)$ such that
$(\partial+2\lambda)\partial Q(\lambda)=P(\lambda)$,
\item
$K(\lambda)=\lambda^3$: take $Q(\lambda)=(\lambda-\partial)\alpha+R(\lambda)$,
with $R(\lambda)=R^*(\lambda)$, such that
$(\partial+2\lambda)\big(-\partial^3\alpha+2(\lambda^2+\lambda\partial+\partial^2)R(\lambda)\big)=P(\lambda)$.
\end{enumerate}
\end{example}

By the exact sequence \eqref{100521:eq2},  $\beta_k=0$ implies that
$\gamma_k:\,H^k(\Omega^\bullet(\mc V),\delta_K)\to\tilde\Omega^{k+1}_{0,0}$
is an embedding,
and its image coincides with the kernel of the endomorphism
$\alpha_{k+1}$ of the $\mc C$-vector space $\tilde\Omega^{k+1}_{0,0}$.
Hence, in order to compute the variational Poisson cohomology,
we need to study the maps $\alpha_{k+1}$ and $\gamma_k$.
In particular,
we will use the results of Appendix \ref{app:5}
to compute the dimension over $\mc C$ of $\ker(\alpha_{k+1})$,
which by the above observations coincides with the dimension
of $H^k(\Omega^\bullet(\mc V),\delta_K)$,
and, for each element $C\in\ker(\alpha_{k+1})$, we will find a representative
of $\gamma^{-1}(C)\in H^k(\Omega^\bullet(\mc V),\delta_K)$ in $\tilde\Omega^k(\mc V)$.
To start with, we need the following:
\begin{lemma}\label{110213:lem}
Suppose that the algebra of differential functions $\mc V$ is an extension of the algebra of differential
polynomials $\mc F\big[u_i^{(n)}\,\big|\,i\in I,n\in\mb Z_+\big]$,
for a differential field $\mc F$.
Then, there exists a direct sum (over $\mc F$) decomposition
\begin{equation}\label{110213:eq2}
\mc V
=
\mc F\oplus\Big(\bigoplus_{i\in I,n\in\mb Z_+}\mc Fu_i^{(n)}\Big)\oplus\mc U\,,
\end{equation}
where $\mc U\subset\mc V$ is an $\mc F$-linear subspace of $\mc V$ such that
\begin{equation}\label{110213:eq4}
\frac{\partial}{\partial u_j^{(m)}} \mc U\subset
\Big(\bigoplus_{i\in I,n\in\mb Z_+}\mc F u_i^{(n)}\Big)\oplus\mc U
\,\,,\,\,\,\,
\text{ for all } j\in I,\,m\in\mb Z_+\,.
\end{equation}
\end{lemma}
\begin{proof}
Consider the map: $\mc F\big[u_i^{(n)}\,\big|\,i\in I,n\in\mb Z_+\big]\twoheadrightarrow\mc F$
given by evaluating at $u_i^{(n)}=0,\,\forall i\in I,n\in\mb Z_+$,
and extend it to a linear over $\mc F$ map
$\varepsilon:\,\mc V\twoheadrightarrow\mc F$.
Let then
$$
\mc U:=\Big\{g\in\mc V\,\Big|\, \varepsilon(g)=0,\,
\varepsilon\Big(\frac{\partial g}{\partial u_i^{(n)}}\Big)=0\,\,\forall i\in I,n\in\mb Z_+\Big\}\,.
$$
Clearly, for every $f\in\mc V$ we have
$$
f-\varepsilon(f)-\sum_{i\in I,n\in\mb Z_+}
\varepsilon\Big(\frac{\partial f}{\partial u_i^{(n)}}\Big) u_i^{(n)}
\in\mc U\,,
$$
so that
$\mc V=\mc F+\Big(\bigoplus_{i\in I,n\in\mb Z_+}\mc Fu_i^{(n)}\Big)+\mc U$.
Moreover, if
$$
f=\alpha+\sum_{i\in I,n\in\mb Z_+}\beta_{i,n}u_i^{(n)}+g=0\,,
$$
with $\alpha,\,\beta_{i,n}\in\mc F$ and $g\in\mc U$,
then $\alpha=\varepsilon(f)=0$,
and $\beta_{i,n}=\varepsilon\big(\frac{\partial f}{\partial u_i^{(n)}}\big)=0$,
so that $\mc V$ admits the direct sum decomposition \eqref{110213:eq2}.

Let then $f\in\mc V$ and consider its decomposition given by \eqref{110213:eq2}:
$f=\alpha+\sum_{i\in I,n\in\mb Z_+}\beta_{i,n}u_i^{(n)}+g$,
where $\alpha,\,\beta_{i,n}\in\mc F$ and $g\in\mc U$.
We have $\varepsilon(f)=\alpha$,
proving that
\begin{equation}\label{110213:eq3}
\ker(\varepsilon)=\Big(\bigoplus_{i\in I,n\in\mb Z_+}\mc Fu_i^{(n)}\Big)\oplus\mc U\subset\mc V\,.
\end{equation}
To conclude, we note that, by the definition of $\mc U$,
if $g\in\mc U$ then $\frac{\partial g}{\partial u_i^{(n)}}\in\ker(\varepsilon)$ for every $i\in I,n\in\mb Z_+$,
which, together with \eqref{110213:eq3}, gives \eqref{110213:eq4}.
\end{proof}
Recall from Appendix \ref{app:5.1} that a $k$-differential operator on $\mc F^\ell$
is an array
$P=\big(P_{i_0,i_1,\dots,i_k}(\lambda_1,\dots,\lambda_k)\big)_{i_0,i_1,\dots,i_k\in I}$,
whose entries are polynomials in $\lambda_1,\dots,\lambda_k$ with coefficients in $\mc F$,
and it is said to be skewsymmetric if the entries
$P_{i_0,i_1,\dots,i_k}(\lambda_1,\dots,\lambda_k)$ are skewsymmetric
with respect to simultaneous permutations of the indices $i_1,\dots,i_k$
and the variables $\lambda_1,\dots,\lambda_k$.
Given an $\ell\times\ell$ matrix differential operator $K(\partial)$,
we denote by $\Sigma_k(K)$ the space of skewsymmetric $k$-differential operators on $\mc F^\ell$
whose entries are polynomials of degree at most $N-1$ in each variable $\lambda_1,\dots,\lambda_k$,
solving equation \eqref{110213:eq5}.
For example $\Sigma_0(K)$ consists of elements $P\in\mc F^\ell$ solving
$K(\partial)P=0$.
By Theorem \ref{110127:conj2}, if $\mc F$ is a linearly closed differential field,
then $\Sigma_k(K)$ is a vector space over $\mc C$
of dimension $\binom{N\ell}{k+1}$.
\begin{theorem}\label{110213:thm}
Let $k\in\mb Z_+$.
Let $\mc V$ be a normal algebra of differential functions,
and assume that the algebra of quasiconstants $\mc F\subset\mc V$
is a linearly closed differential field.
Let $K(\partial)$ be an $\ell\times\ell$ matrix differential operator
of order $N$ with quasiconstant coefficients
and invertible leading coefficient $K_N\in\Mat_{\ell\times\ell}(\mc F)$.
\begin{enumerate}[(a)]
\item
There is a canonical isomorphism of $\mc C$-vector spaces $\phi_k:\,\Sigma_k(K^*)\to\ker(\alpha_{k+1})$,
defined as follows: given $P\in\Sigma_k(K^*)$,
we let $\phi_k(P)=C\in\ker(\alpha_{k+1})$, where
\begin{equation}\label{110226:eq1}
\begin{array}{l}
\displaystyle{
\sum_{\alpha=0}^k (-1)^\alpha
\sum_{j\in I}
P_{j,i_0,\stackrel{\alpha}{\check{\dots}},i_k}
(\lambda_0,\stackrel{\alpha}{\check{\dots}},\lambda_k)
K_{j,i_\alpha}(\lambda_\alpha)
} \\
\displaystyle{
=
(\lambda_0+\lambda_1+\dots+\lambda_k+\partial)C_{i_0,i_1,\dots,i_k}(\lambda_0,\lambda_1,\dots,\lambda_k)\,,
}
\end{array}
\end{equation}
for all indices $i_0,\dots,i_k\in I$ (equality in $\mb F[\lambda_0,\dots,\lambda_k]\otimes\mc F$).
\item
There is a canonical isomorphism $\chi_k:\,\Sigma_k(K^*)\simeq H^k(\Omega^\bullet(\mc V),\delta_K)$
defined as follows:
given $P\in\Sigma_k(K^*)$,
we let $\chi_k(P)\in H^k(\Omega^\bullet(\mc V),\delta_K)$ be the cohomology class
with representative
\begin{equation}\label{110226:eq2}
\Big(\sum_{j\in I}P_{j,i_1,\dots,i_k}(\lambda_1,\dots,\lambda_k)u_j\Big)_{i_1,\dots,i_k\in I}\,\in\tilde\Omega^k(\mc V)\,.
\end{equation}
In particular,
\begin{equation}\label{110226:eq4}
\dim_{\mc C}(H^k(\Omega^\bullet(\mc V),\delta_K))=\binom{N\ell}{k+1}\,.
\end{equation}
\item
The maps $\gamma_k:\,H^k(\Omega^\bullet(\mc V),\delta_K)\to\ker(\alpha_{k+1})$
in the exact sequence \eqref{100521:eq2}
and $\phi_k:\,\Sigma_k(K^*)\to\ker(\alpha_{k+1})$ are compatible in the sense that
\begin{equation}\label{110226:eq3}
\phi_k=\gamma_k\circ\chi_k\,.
\end{equation}
For $C\in\ker(\alpha_{k+1})$,
let
$\phi_k^{-1}(C)=\big(P_{i_0,i_1,\dots,i_k}(\lambda_1,\dots,\lambda_k)\big)_{i_0,i_1,\dots,i_k\in I}\in\Sigma_k(K^*)$.
Then, the array \eqref{110226:eq2} in $\tilde\Omega^k(\mc V)$
is a representative of the cohomology class $\gamma_k^{-1}(C)\in H^k(\Omega^\bullet(\mc V),\delta_K)$.
\end{enumerate}
\end{theorem}
\begin{proof}
First, we prove that the map $\phi_k$ given by \eqref{110226:eq1} is well defined.
Let $P\in\Sigma_k(K^*)$, so that $\langle K^*\circ P\rangle^-=0$.
By equation \eqref{110222:eq4} from the Appendix,
we can rewrite this condition by saying that
\begin{equation}\label{110222:eq1}
\sum_{\alpha=0}^k (-1)^\alpha
\sum_{j\in I}
P_{j,i_0,\stackrel{\alpha}{\check{\dots}},i_k}
(\lambda_0,\stackrel{\alpha}{\check{\dots}},\lambda_k)
K_{j,i_\alpha}(\lambda_\alpha)
\end{equation}
becomes zero if we replace $\lambda_0$ by $-\lambda_1-\dots-\lambda_k-\partial$,
with $\partial$ acting from the left.
In other words, \eqref{110222:eq1}, as an element of $\mb F[\lambda_0,\lambda_1,\dots,\lambda_k]\otimes\mc F$,
is equal to
\begin{equation}\label{110222:eq2}
(\lambda_0+\lambda_1+\dots+\lambda_k+\partial)C_{i_0,i_1,\dots,i_k}(\lambda_0,\lambda_1,\dots,\lambda_k)\,,
\end{equation}
where $C=\big(C_{i_0,i_1,\dots,i_k}(\lambda_0,\lambda_1,\dots,\lambda_k)\big)_{i_0,i_1,\dots,i_k\in I}$
is a skewsymmetric array whose entries are polynomials with quasiconstant coefficients
of degree less than or equal to $N-1$, i.e. $C\in\tilde\Omega^{k+1}_{0,0}$.
Furthermore, we claim that $C$ lies in $\ker(\alpha_{k+1})$.
Indeed, taking $Q\in\tilde\Omega^k(\mc V)$ be the array \eqref{110226:eq2},
we have, by \eqref{100519b:eq1}, that $(\delta_KQ)_{i_0,\dots,i_k}(\lambda_0,\dots,\lambda_k)$
is equal to \eqref{110222:eq1}.
Hence, the equality of \eqref{110222:eq1} and \eqref{110222:eq2} implies that
$\delta_KQ=\partial C$.
Therefore, by Lemma \ref{100602:lem1}(c), we conclude that $\alpha_{k+1}(C)=0$,
proving that $\phi_k$ is well defined.

We next prove that the map $\phi_k:\,\Sigma_k(K^*)\to\ker(\alpha_{k+1})$ is injective.
Since, by assumption, $P_{i_0,\dots,i_k}(\lambda_1,\dots,\lambda_k)$
has degree less than or equal to $N-1$ in each variable,
the coefficient of $\lambda_0^N$ in \eqref{110222:eq1} is
\begin{equation}\label{110222:eq6}
\sum_{j\in I}
P_{j,i_1,\dots,i_k}
(\lambda_1,\dots,\lambda_k)
(K_N)_{j,i_0}\,.
\end{equation}
To say that $\phi_k(P)=0$ is equivalent to say that \eqref{110222:eq1},
viewed as an element of $\mb F[\lambda_0,\dots,\lambda_k]\otimes\mc F$,
is identically zero for all indices $i_0,\dots,i_k\in I$.
In particular \eqref{110222:eq6} is zero.
Since, by assumption, $K_N$ is an invertible matrix, it follows that
$P_{i_0,i_1,\dots,i_k}(\lambda_1,\dots,\lambda_k)=0$, for all indices $i_0,\dots,i_k$.
Hence, $\phi_k$ is injective.

To complete the proof of part (a) we are left with showing that the map $\phi_k$
is surjective.
Let $C=\big(C_{i_0,\dots,i_k}(\lambda_0,\dots,\lambda_k)\big)_{i_0,\dots,i_k\in I}$
be an element of $\ker(\alpha_{k+1})$.
By Lemma \ref{100602:lem1}(c), there exists an element
$Q\in\tilde\Omega^k(\mc V)$
such that
\begin{equation}\label{110222:eq7}
\partial C=\delta_K Q\,.
\end{equation}
By Lemma \ref{110213:lem}, we can assume that the coefficients of $Q$ are linear in the variables $u_j^{(n)}$, i.e.
\begin{equation}\label{110222:eq8}
Q_{i_1,\dots,i_k}(\lambda_1,\dots,\lambda_k)
=
\sum_{n=0}^M
\sum_{j\in I}
P^n_{j,i_1,\dots,i_k}(\lambda_1,\dots,\lambda_k) u_j^{(n)}
\end{equation}
with $M\in\mb Z_+$ and $P^n_{j,i_1,\dots,i_k}(\lambda_1,\dots,\lambda_k)\in\mb F[\lambda_1,\dots,\lambda_k]\otimes\mc F$.
Indeed, the quasiconstant part of $Q$ is killed by the differential $\delta_K$,
while $\delta_K$ applied to the $\mc U$-part of $Q$ has zero quasiconstant part.
Note that, if $Q$ is as in \eqref{110222:eq8}, then equation \eqref{110222:eq7} becomes
\begin{equation}\label{110223:eq1}
\begin{array}{l}
\displaystyle{
\sum_{\alpha=0}^k(-1)^\alpha \sum_{j\in I} \sum_{n=0}^M
P^n_{j,i_0,\stackrel{\alpha}{\check{\dots}},i_k}
(\lambda_0,\stackrel{\alpha}{\check{\dots}},\lambda_k)
(\lambda_\alpha+\partial)^n K_{j,i_\alpha}(\lambda_\alpha)
} \\
\displaystyle{
=(\lambda_0+\dots+\lambda_k+\partial)C_{i_0,\dots,i_k}(\lambda_0,\dots,\lambda_k)
\,\in\mb F[\lambda_0,\dots,\lambda_k]\otimes\mc F\,,
}
\end{array}
\end{equation}
for all choices of indices $i_0,\dots,i_k\in I$.

In order to prove surjectivity of $\phi_k$,
we will show that we can choose $Q$ as in \eqref{110222:eq8} with $M=0$
and $P^0_{j,i_1,\dots,i_k}(\lambda_1,\dots,\lambda_k)$ of degree
at most $N-1$ in each variable $\lambda_1,\dots,\lambda_k$,
such that equation \eqref{110223:eq1} holds with the given $C\in\ker(\alpha_{k+1})$.
In this case, by the definition \eqref{110226:eq1} of the map $\phi_k$,
$$
P=\big(P^0_{i_0,i_1,\dots,i_k}(\lambda_1,\dots,\lambda_k)\big)_{i_0,i_1,\dots i_k\in I}
$$
is an element of $\Sigma_k(K^*)$ such that $\phi_k(P)=C$.

We will achieve the desired form of $Q$ in three steps:
first we reduce to the case when
all polynomials $P^n_{j,i_1,\dots,i_k}(\lambda_1,\dots,\lambda_k)$
have degree less than or equal to $M+N$
in each variable $\lambda_\alpha$;
then we reduce to the case when $M=0$;
finally we reduce to the case when the polynomials
$P^0_{j,i_1,\dots,i_k}(\lambda_1,\dots,\lambda_k)$
have degree at most $N-1$ in each variable.
Note that in the case $k=0$ we only need to do the second step.

Let $k\geq1$ and let
$d$ be the maximal degree in one of the variables $\lambda_1,\dots,\lambda_k$
of all the polynomials $P^n_{j,i_1,\dots,i_k}(\lambda_1,\dots,\lambda_k)$
for $n=0,\dots,M$ and $j,i_1,\dots,i_k\in I$,
and assume that $d>M+N$.
By taking separately all terms in which some of the variables $\lambda_\alpha$ are raised to the power $d$,
we can write
\begin{equation}\label{110223:eq2}
\begin{array}{l}
\displaystyle{
P^n_{j,i_1,\dots,i_k}(\lambda_1,\dots,\lambda_k)
=
R^{n}_{j,i_1,\dots,i_k}\lambda_1^d\dots\lambda_k^d
+ \sum_{1\leq\beta\leq k} R^{n,\beta}_{j,i_1,\dots,i_k}(\lambda_\beta)\lambda_1^d\stackrel{\beta}{\check{\dots}}\lambda_k^d
}\\
\displaystyle{
+ \sum_{1\leq\beta<\gamma\leq k} R^{n,\beta,\gamma}_{j,i_1,\dots,i_k}
(\lambda_\beta,\lambda_\gamma)
\lambda_1^d\stackrel{\beta}{\check{\dots}}\,\stackrel{\gamma}{\check{\dots}}\lambda_k^d
+\dots
+R^{n}_{j,i_1,\dots,i_k}(\lambda_1,\dots,\lambda_k)
}\\
\displaystyle{
=\sum_{q=0}^k \sum_{1\leq\beta_1<\dots<\beta_q\leq k}
R^{n,\beta_1,\dots,\beta_q}_{j,i_1,\dots,i_k}(\lambda_{\beta_1},\dots,\lambda_{\beta_q})
\lambda_1^d\stackrel{\beta_1,\dots,\beta_q}{\check{\dots\dots}}\lambda_k^d
\,,}
\end{array}
\end{equation}
where $R^{n,\beta_1,\dots,\beta_q}_{j,i_1,\dots,i_k}(\lambda_{\beta_1},\dots,\lambda_{\beta_q})$
are polynomials with quasiconstant coefficients of degree strictly less than $d$ in each variable.
Then equation \eqref{110223:eq1} becomes
$$
\begin{array}{l}
\displaystyle{
\sum_{\alpha=0}^k(-1)^\alpha \sum_{j\in I} \sum_{n=0}^M
\sum_{q=0}^k
\sum_{\substack{0\leq\beta_1<\dots<\beta_q\leq k \\ (\beta_h<\alpha<\beta_{h+1})}}
\lambda_0^d\stackrel{\alpha,\beta_1,\dots,\beta_q}{\check{\dots\dots}}\lambda_k^d
} \\
\displaystyle{
\times
R^{n,\beta_1+1,\dots,\beta_h+1,\beta_{h+1},\dots\beta_q}_{j,i_0,\stackrel{\alpha}{\check{\dots}},i_k}(\lambda_{\beta_1},\dots,\lambda_{\beta_q})
(\lambda_\alpha+\partial)^n K_{j,i_\alpha}(\lambda_\alpha)
} \\
\displaystyle{
=(\lambda_0+\dots+\lambda_k+\partial)C_{i_0,\dots,i_k}(\lambda_0,\dots,\lambda_k)
\,.
}
\end{array}
$$

We can rewrite the above equation in the following equivalent form
\begin{equation}\label{110223:eq3}
\begin{array}{l}
\displaystyle{
\sum_{q=0}^k
\sum_{0\leq\beta_0<\dots<\beta_q\leq k }
\sum_{r=0}^q
(-1)^{\beta_r}
\sum_{j\in I} \sum_{n=0}^M
\lambda_0^d\stackrel{\beta_0,\dots,\beta_q}{\check{\dots\dots}}\lambda_k^d
} \\
\displaystyle{
\times R^{n,\beta_0+1,\dots,\beta_{r-1}+1,\beta_{r+1},\dots,\beta_q}_{j,i_0,\stackrel{\beta_r}{\check{\dots}},i_k}
(\lambda_{\beta_0},\stackrel{r}{\check{\dots}},\lambda_{\beta_q})
\times(\lambda_{\beta_r}+\partial)^n K_{j,i_{\beta_r}}(\lambda_{\beta_r})
} \\
\displaystyle{
=(\lambda_0+\dots+\lambda_k+\partial)C_{i_0,\dots,i_k}(\lambda_0,\dots,\lambda_k)
\,.
}
\end{array}
\end{equation}
Note that the RHS above has degree at most $N$ in each variable $\lambda_1,\dots,\lambda_k$.
Hence, by looking at the coefficient of $\lambda_1^d\dots\lambda_k^d$ in both sides of equation \eqref{110223:eq3},
we get
$$
\sum_{j\in I} \sum_{n=0}^M
R^{n}_{j,i_1,\dots,i_k}
(\lambda_0+\partial)^n K_{j,i_0}(\lambda_0)
=0\,.
$$
Since, by assumption, the leading coefficient $K_N$ of the differential operator $K(\partial)$
is an invertible matrix, one easily gets that $R^{n}_{j,i_1,\dots,i_k}=0$ for all $n=0,\dots,M$ and $j,i_1,\dots,i_k\in I$.
Hence, in the LHS of \eqref{110223:eq3} the term with $q=0$ vanishes.
Next, for $0\leq\beta_0<\beta_1\leq k$, by looking at the coefficient of
$\lambda_0^d\stackrel{\beta_0,\beta_1}{\check{\dots\dots}}\lambda_k^d$
in both sides of equation \eqref{110223:eq3}, we get
$$
\begin{array}{l}
\displaystyle{
T_{\beta_0,\beta_1}(\lambda_{\beta_0},\lambda_{\beta_1})
:=
(-1)^{\beta_0}
\sum_{j\in I} \sum_{n=0}^M
R^{n,\beta_1}_{j,i_0,\stackrel{\beta_0}{\check{\dots}},i_k}
(\lambda_{\beta_1})
(\lambda_{\beta_0}+\partial)^n K_{j,i_{\beta_0}}(\lambda_{\beta_0})
} \\
\displaystyle{
+
(-1)^{\beta_1}
\sum_{j\in I} \sum_{n=0}^M
R^{n,\beta_0+1}_{j,i_0,\stackrel{\beta_1}{\check{\dots}},i_k}
(\lambda_{\beta_0})
(\lambda_{\beta_1}+\partial)^n K_{j,i_{\beta_1}}(\lambda_{\beta_1})
=0
\,.
}
\end{array}
$$
On the other hand, the term with $q=1$ in the LHS of \eqref{110223:eq3}
is exactly
$$
\sum_{0\leq\beta_0<\beta_1\leq k}T_{\beta_0,\beta_1}(\lambda_{\beta_0},\lambda_{\beta_1})
\lambda_0^d\stackrel{\beta_0,\beta_1}{\check{\dots\dots}}\lambda_k^d\,,
$$
hence, it vanishes, and the sum over $q$ in the LHS of \eqref{110223:eq3} starts with $q=2$.
Repeating the same argument several times, we conclude that all the terms with $q\leq k-1$
in the LHS of \eqref{110223:eq3} vanish, hence the equation becomes
\begin{equation}\label{110223:eq4}
\begin{array}{c}
\displaystyle{
\sum_{\alpha=0}^k
(-1)^{\alpha}
\sum_{j\in I} \sum_{n=0}^M
R^{n,1,\dots,k}_{j,i_0,\stackrel{\alpha}{\check{\dots}},i_k}
(\lambda_0,\stackrel{\alpha}{\check{\dots}},\lambda_k)
(\lambda_\alpha+\partial)^n K_{j,i_{\alpha}}(\lambda_{\alpha})
} \\
\displaystyle{
=(\lambda_0+\dots+\lambda_k+\partial)C_{i_0,\dots,i_k}(\lambda_0,\dots,\lambda_k)
\,.
}
\end{array}
\end{equation}
Comparing equations \eqref{110223:eq1} and \eqref{110223:eq4},
we can replace the polynomials $P^n_{j,i_1,\dots,i_k}(\lambda_1,\dots,\lambda_k)$
by the polynomials $R^{n,1,\dots,k}_{j,i_1,\dots,i_k}(\lambda_1,\dots,\lambda_k)$,
which have degree strictly less than $d$.
Hence, repeating this argument several times,
we may assume that the degree of all polynomials
$P^n_{j,i_1,\dots,i_k}(\lambda_1,\dots,\lambda_k)$
is less than or equal to $M+N$, concluding the first step.

In the second step we want to reduce to the case when $M=0$.
For this, assuming $M\geq1$, we will reduce to the case when
$M$ is replaced by $M-1$.
We find an expansion of $P$ similar to the one discussed in equation \eqref{110223:eq2}.
Using the fact that $K(\partial)$ has order $N$
and its leading coefficient $K_N$ is an invertible matrix,
we can write
\begin{equation}\label{110224:eq1}
\begin{array}{c}
\displaystyle{
P^n_{j,i_1,\dots,i_k}(\lambda_1,\dots,\lambda_k)
= \sum_{q=0}^k \sum_{j_1,\dots,j_k\in I} \sum_{1\leq\beta_1<\dots<\beta_q\leq k}
Q^{n,\beta_1,\dots,\beta_q}_{j,j_1,\dots,j_k}(\lambda_{\beta_1},\dots,\lambda_{\beta_q})
}\\
\displaystyle{
\times
\delta_{j_{\beta_1},i_{\beta_1}}\dots\delta_{j_{\beta_q},i_{\beta_q}}
\big((\lambda_1+\partial)^MK_{j_1,i_1}(\lambda_1)\big)
\stackrel{\beta_1\dots\beta_q}{\check{\dots\dots}}
\big((\lambda_k+\partial)^MK_{j_k,i_k}(\lambda_k)\big)
\,,}
\end{array}
\end{equation}
where $Q^{n,\beta_1,\dots,\beta_q}_{j,j_1,\dots,j_k}(\lambda_{\beta_1},\dots,\lambda_{\beta_q})$
are polynomials with quasiconstant coefficients of degree strictly less than $M+N$
in each variable $\lambda_{\beta_1},\dots,\lambda_{\beta_q}$.
Then, equation \eqref{110223:eq1} becomes
$$
\begin{array}{l}
\displaystyle{
\sum_{\alpha=0}^k \sum_{j\in I} \sum_{n=0}^M
\sum_{q=0}^k
\!\!
\sum_{j_0,\stackrel{\alpha}{\check{\dots}},j_k\in I}
\!
\sum_{\substack{
0\leq\beta_1<\dots<\beta_q\leq k \\
(\beta_h<\alpha<\beta_{h+1})}}
\!\!\!\!\!\!\!\!\!
(-1)^\alpha
Q^{n,\beta_1+1,\dots,\beta_h+1,\beta_{h+1},\dots,\beta_q}_{j,j_0,\stackrel{\alpha}{\check{\dots}},j_k}
\!(\lambda_{\beta_1},\dots,\lambda_{\beta_q})
} \\
\displaystyle{
\vphantom{\Bigg(}
\times
\delta_{j_{\beta_1},i_{\beta_1}}\dots\delta_{j_{\beta_q},i_{\beta_q}}
\big((\lambda_0+\partial)^MK_{j_0,i_0}(\lambda_0)\big)
\stackrel{\alpha,\beta_1\dots\beta_q}{\check{\dots\dots}}
\big((\lambda_k+\partial)^MK_{j_k,i_k}(\lambda_k)\big)
} \\
\displaystyle{
\times
\big((\lambda_\alpha+\partial)^n K_{j,i_\alpha}(\lambda_\alpha)\big)
=(\lambda_0+\dots+\lambda_k+\partial)C_{i_0,\dots,i_k}(\lambda_0,\dots,\lambda_k)
\,,
}
\end{array}
$$
or, rearranging terms appropriately,
we can rewrite it in the following equivalent form
\begin{equation}\label{110224:eq2}
\begin{array}{l}
\displaystyle{
\sum_{n=0}^M
\sum_{q=0}^k \sum_{r=0}^q
\sum_{0\leq\beta_0<\dots<\beta_q\leq k}
\!\!\!\!
(-1)^{\beta_r}
\!\!\!\!\!\!
\sum_{j_0,\dots,j_k\in I}
Q^{n,\beta_0+1,\dots,\beta_{r-1}+1,\beta_{r+1},\dots,\beta_q}_{j_{\beta_r},j_0,\stackrel{\beta_r}{\check{\dots}},j_k}
(\lambda_{\beta_0},\stackrel{r}{\check{\dots}},\lambda_{\beta_q})
} \\
\displaystyle{
\vphantom{\Bigg(}
\times
\delta_{j_{\beta_0},i_{\beta_0}}\stackrel{r}{\check{\dots}}\delta_{j_{\beta_q},i_{\beta_q}}
\big((\lambda_0+\partial)^MK_{j_0,i_0}(\lambda_0)\big)
\stackrel{\beta_0\dots\beta_q}{\check{\dots\dots}}
\big((\lambda_k+\partial)^MK_{j_k,i_k}(\lambda_k)\big)
} \\
\displaystyle{
\times
\big((\lambda_{\beta_r}+\partial)^n K_{j_{\beta_r},i_{\beta_r}}(\lambda_{\beta_r})\big)
=(\lambda_0+\dots+\lambda_k+\partial)C_{i_0,\dots,i_k}(\lambda_0,\dots,\lambda_k)
\,.
}
\end{array}
\end{equation}
Note that the RHS has degree at most $N$ in each variable $\lambda_1,\dots,\lambda_k$.

By looking at the coefficient of $\lambda_0^{M+N}\dots\lambda_k^{M+N}$
in both sides of equation \eqref{110224:eq2}, we get, since $M\geq1$,
$$
\sum_{\alpha=0}^k
(-1)^{\alpha}
\sum_{j_0,\dots,j_k\in I}
Q^{M}_{j_{\alpha},j_0,\stackrel{\alpha}{\check{\dots}},j_k}
(K_N)_{j_0,i_0}
\dots
(K_N)_{j_k,i_k}
=0
\,.
$$
Since $K_N$ is invertible, we deduce that
$$
T_{i_0,\dots,i_k}
:=
\sum_{\alpha=0}^k
(-1)^{\alpha}
Q^{M}_{i_{\alpha},i_0,\stackrel{\alpha}{\check{\dots}},i_k}
=0
\,.
$$
On the other hand, the term with $q=0$ and $n=M$ in the LHS of \eqref{110224:eq2}
is equal to
$$
\sum_{j_0,\dots,j_k\in I}
T_{j_0,\dots,j_k}
\big((\lambda_0+\partial)^MK_{j_0,i_0}(\lambda_0)\big)
\dots
\big((\lambda_k+\partial)^MK_{j_k,i_k}(\lambda_k)\big)
\,,
$$
hence it vanishes.

Next, for $k\geq1$ fix $\alpha\in\{1,\dots,k\}$
and consider the coefficient of $\lambda_0^{M+N}\stackrel{\alpha}{\check{\dots}}\lambda_k^{M+N}$
in both sides of equation \eqref{110224:eq2}. In the RHS we get $0$ since $M\geq1$,
while in the LHS there are only two contributions, one coming from $q=0$ and $n\leq M-1$,
and the other coming from $q=1$ and $n=M$.
We thus get
$$
\begin{array}{l}
\displaystyle{
\sum_{n=0}^{M-1}
(-1)^{\alpha}
\sum_{j_0,\dots,j_k\in I}
Q^{n}_{j_\alpha,j_0,\stackrel{\alpha}{\check{\dots}},j_k}
(K_N)_{j_0,i_0}\stackrel{\alpha}{\check{\dots}} (K_N)_{j_0,i_0}
\big((\lambda_\alpha+\partial)^n K_{j_\alpha,i_\alpha}(\lambda_\alpha)\big)
} \\
\displaystyle{
+ \sum_{\substack{\beta=0 \\ (\beta<\alpha)}}^k
(-1)^{\beta}
\sum_{j_0,\dots,j_k\in I}
Q^{M,\alpha}_{j_{\beta},j_0,\stackrel{\beta}{\check{\dots}},j_k}(\lambda_{\alpha})
\delta_{j_{\alpha},i_{\alpha}}
(K_N)_{j_0,i_0}\stackrel{\alpha}{\check{\dots}} (K_N)_{j_k,i_k}
} \\
\displaystyle{
+ \sum_{\substack{\beta=0 \\ (\beta>\alpha)}}^k
(-1)^{\beta}
\sum_{j_0,\dots,j_k\in I}
Q^{M,\alpha+1}_{j_{\beta},j_0,\stackrel{\beta}{\check{\dots}},j_k}(\lambda_{\alpha})
\delta_{j_{\alpha},i_{\alpha}}
(K_N)_{j_0,i_0}\stackrel{\alpha}{\check{\dots}} (K_N)_{j_k,i_k}
=
0
\,.
}
\end{array}
$$
Again, since $K_N$ is invertible, we deduce that
$$
\begin{array}{l}
\displaystyle{
T^{\alpha}_{i_0,\dots,i_k}(\lambda_\alpha)
:=
\sum_{\substack{\beta=0 \\ (\beta<\alpha)}}^k
(-1)^{\beta}
Q^{M,\alpha}_{i_{\beta},i_0,\stackrel{\beta}{\check{\dots}},i_k}(\lambda_{\alpha})
+\sum_{\substack{\beta=0 \\ (\beta>\alpha)}}^k
(-1)^{\beta}
Q^{M,\alpha+1}_{i_{\beta},i_0,\stackrel{\beta}{\check{\dots}},i_k}(\lambda_{\alpha})
} \\
\displaystyle{
+
\sum_{n=0}^{M-1}
\sum_{j\in I}
(-1)^{\alpha}
Q^{n}_{j,i_0,\stackrel{\alpha}{\check{\dots}},i_k}
\big((\lambda_\alpha+\partial)^n K_{j,i_\alpha}(\lambda_\alpha)\big)
=
0
\,.
}
\end{array}
$$
On the other hand, the sum of the term with $q=1$ and $n=M$ together with the terms with $q=0$ and $n\leq M-1$
in the LHS of \eqref{110224:eq2} is equal to
$$
\begin{array}{l}
\displaystyle{
\sum_{j_0,\dots,j_k\in I}
\sum_{\alpha=0}^k
T^{\alpha}_{j_0,\dots,j_k}(\lambda_\alpha)
\delta_{j_\alpha,i_\alpha}
} \\
\displaystyle{
\times \big((\lambda_0+\partial)^MK_{j_0,i_0}(\lambda_0)\big)
\stackrel{\alpha}{\check{\dots}}
\big((\lambda_k+\partial)^MK_{j_k,i_k}(\lambda_k)\big)
\,,
}
\end{array}
$$
hence it vanishes.
Repeating the same argument several times,
at each step we prove that the sum of the term
with $q=q_0+1$ and $n=M$ together with the terms with $q=q_0$ and $n\leq M-1$
vanishes.
As a result, in the LHS of equation \eqref{110224:eq2}
only the terms with $q=k$ and $n<M$ survive.
Hence, equation \eqref{110224:eq2} becomes
$$
\begin{array}{c}
\displaystyle{
\sum_{n=0}^{M-1}
\sum_{\alpha=0}^k
(-1)^{\alpha}
\sum_{j\in I}
Q^{n,1,\dots,k}_{j,i_0,\stackrel{\alpha}{\check{\dots}},i_k}
(\lambda_0,\stackrel{\alpha}{\check{\dots}},\lambda_k)
(\lambda_\alpha+\partial)^n K_{j,i_\alpha}(\lambda_\alpha)
} \\
\displaystyle{
\vphantom{\Bigg(}
=(\lambda_0+\dots+\lambda_k+\partial)C_{i_0,\dots,i_k}(\lambda_0,\dots,\lambda_k)
\,.
}
\end{array}
$$
This is the same as \eqref{110223:eq1} with the polynomials
$P^n_{j,i_1,\dots,i_k}(\lambda_1,\dots,\lambda_k)$
replaced by $0$ for $n=M$, and by the polynomials
$Q^{n,1,\dots,k}_{j,i_1,\dots,i_k}(\lambda_1,\dots,\lambda_k)$ for $n<M$.
This completes the second step.

So far, we showed that we can choose $Q$ in \eqref{110222:eq8} of the form
\begin{equation}\label{110225:eq1}
Q_{i_1,\dots,i_k}(\lambda_1,\dots,\lambda_k)
=
\sum_{j\in I}
P_{j,i_1,\dots,i_k}(\lambda_1,\dots,\lambda_k) u_j\,,
\end{equation}
where $P_{j,i_1,\dots,i_k}(\lambda_1,\dots,\lambda_k)$ are polynomials
with quasiconstant coefficients of degree at most $N$ in each variable.
In this case equation \eqref{110223:eq1} reads
\begin{equation}\label{110225:eq2}
\begin{array}{l}
\displaystyle{
\sum_{\alpha=0}^k(-1)^\alpha \sum_{j\in I}
P_{j,i_0,\stackrel{\alpha}{\check{\dots}},i_k}
(\lambda_0,\stackrel{\alpha}{\check{\dots}},\lambda_k)
K_{j,i_\alpha}(\lambda_\alpha)
} \\
\displaystyle{
=(\lambda_0+\dots+\lambda_k+\partial)C_{i_0,\dots,i_k}(\lambda_0,\dots,\lambda_k)
\,\in\mb F[\lambda_0,\dots,\lambda_k]\otimes\mc F\,.
}
\end{array}
\end{equation}
To complete the proof of part (a), we are left with showing that we can choose the polynomials
$P_{j,i_1,\dots,i_k}(\lambda_1,\dots,\lambda_k)$ to be of degree at most $N-1$ in each
variable $\lambda_\alpha$,
such that equation \eqref{110225:eq2} still holds.
As before, we expand the polynomials $P_{j,i_1,\dots,i_k}(\lambda_1,\dots,\lambda_k)$ as in \eqref{110224:eq1}
\begin{equation}\label{110225:eq3}
\begin{array}{c}
\displaystyle{
P_{j,i_1,\dots,i_k}(\lambda_1,\dots,\lambda_k)
= \sum_{q=0}^k \sum_{j_1,\dots,j_k\in I} \sum_{1\leq\beta_1<\dots<\beta_q\leq k}
Q^{\beta_1,\dots,\beta_q}_{j,j_1,\dots,j_k}(\lambda_{\beta_1},\dots,\lambda_{\beta_q})
}\\
\displaystyle{
\times
\delta_{j_{\beta_1},i_{\beta_1}}\dots\delta_{j_{\beta_q},i_{\beta_q}}
K_{j_1,i_1}(\lambda_1)
\stackrel{\beta_1\dots\beta_q}{\check{\dots\dots}}
K_{j_k,i_k}(\lambda_k)
\,,}
\end{array}
\end{equation}
where the polynomials $Q^{\beta_1,\dots,\beta_q}_{j,j_1,\dots,j_k}(\lambda_{\beta_1},\dots,\lambda_{\beta_q})$
have degree strictly less than $N$.
Then, equation \eqref{110225:eq2} reads
\begin{equation}\label{110225:eq4}
\begin{array}{l}
\displaystyle{
\sum_{q=0}^k \sum_{j_0,\dots,j_k\in I}
\sum_{\alpha=0}^k
\sum_{\substack{
0\leq\beta_1<\dots<\beta_q\leq k \\
(\beta_h<\alpha<\beta_{h+1})}}
(-1)^\alpha
Q^{\beta_1+1,\dots,\beta_h+1,\beta_{h+1},\dots,\beta_q}_{j_\alpha,j_0,\stackrel{\alpha}{\check{\dots}},j_k}
(\lambda_{\beta_1},\dots,\lambda_{\beta_q})
} \\
\displaystyle{
\vphantom{\Bigg(}
\times
\delta_{j_{\beta_1},i_{\beta_1}}\dots\delta_{j_{\beta_q},i_{\beta_q}}
K_{j_0,i_0}(\lambda_0)
\stackrel{\beta_1\dots\beta_q}{\check{\dots\dots}}
K_{j_k,i_k}(\lambda_k)
} \\
\displaystyle{
\vphantom{\Bigg(}
=(\lambda_0+\dots+\lambda_k+\partial)C_{i_0,\dots,i_k}(\lambda_0,\dots,\lambda_k)
\,.
}
\end{array}
\end{equation}
We then proceed as in step two.
Note that, since by assumption the polynomials
$C_{i_0,\dots,i_k}(\lambda_0,\dots,\lambda_k)$ have degree at most $N-1$ in each variable,
in the right hand side of \eqref{110225:eq4} in each monomial
at most one variable $\lambda_\alpha$ appears in degree $N$.
Therefore, for $k\geq1$, comparing the coefficient of $\lambda_0^N\dots\lambda_k^N$ in both sides of \eqref{110225:eq4},
and using the fact that $K_N$ is invertible, we get
$$
T_{i_0,\dots,i_k}:=
\sum_{\alpha=0}^k
(-1)^\alpha
Q_{i_\alpha,i_0,\stackrel{\alpha}{\check{\dots}},i_k}=0\,,
$$
for every choice of indices $i_0,\dots,i_k\in I$.
But the term with $q=0$ in the LHS of \eqref{110225:eq4} is
$$
\sum_{j_0,\dots,j_k\in I}
T_{j_0,\dots,j_k} K_{j_0,i_0}(\lambda_0)\dots K_{j_k,i_k}(\lambda_k)\,,
$$
hence it vanishes.
Similarly, for $k\geq2$, given $\beta\in\{0,\dots,k\}$ and comparing the coefficient
of $\lambda_0^N\stackrel{\beta}{\check{\dots}}\lambda_k^N$ in both sides of \eqref{110225:eq4},
we get, again using the fact that $K_N$ is invertible,
$$
T^\beta_{i_0,\dots,i_k}(\lambda_\beta):=
\sum_{\alpha<\beta}
(-1)^\alpha
Q^{\beta}_{i_\alpha,i_0,\stackrel{\alpha}{\check{\dots}},i_k}(\lambda_\beta)
+\sum_{\alpha>\beta}
(-1)^\alpha
Q^{\beta+1}_{i_\alpha,i_0,\stackrel{\alpha}{\check{\dots}},i_k}(\lambda_\beta)
=0\,,
$$
for every choices of indices $i_0,\dots,i_k\in I$.
and the term with $q=1$ in the LHS of \eqref{110225:eq4} is exactly
$$
\sum_{\beta=0}^k \sum_{j_0,\dots,j_k\in I}
T^\beta_{j_0,\dots,j_k} \delta_{j_\beta,i_\beta} K_{j_0,i_0}(\lambda_0)\stackrel{\beta}{\check{\dots}} K_{j_k,i_k}(\lambda_k)\,,
$$
hence it vanishes.
Repeating the same argument several times, we prove that all the terms in the LHS of \eqref{110225:eq4}
with $q\leq k-1$ vanish.
Note that the same argument always works for $q\leq k-1$ since
the monomial $\lambda_0^N\stackrel{\beta_1\dots\beta_q}{\check{\dots\dots}}\lambda_k^N$
contains at least two variables raised to the power $N$.
In conclusion, equation \eqref{110225:eq4} is equivalent to
the same equation where in the LHS we only keep the term with $q=k$:
$$
\begin{array}{l}
\displaystyle{
\sum_{j\in I} \sum_{\alpha=0}^k
(-1)^\alpha
Q^{1,\dots,k}_{j,i_0,\stackrel{\alpha}{\check{\dots}},i_k}
(\lambda_0,\stackrel{\alpha}{\check{\dots}},\lambda_k)
K_{j,i_\alpha}(\lambda_\alpha)
} \\
\displaystyle{
\vphantom{\Bigg(}
=(\lambda_0+\dots+\lambda_k+\partial)C_{i_0,\dots,i_k}(\lambda_0,\dots,\lambda_k)
\,.
}
\end{array}
$$
In other words, we can replace the polynomials $P_{j,i_1,\dots,i_k}(\lambda_1,\dots,\lambda_k)$
defined in \eqref{110225:eq1}
by the polynomials $Q^{1,\dots,k}_{j,i_1,\dots,i_k}(\lambda_1,\dots,\lambda_k)$,
which have degree at most $N-1$ in each variable,
without changing the RHS of equation \eqref{110225:eq2}.
This completes the proof of part (a).

The dimension formula \eqref{110226:eq4}
follows from the first assertion in part (b)
and Theorem \ref{110127:conj2}.

Note that, if $P\in\Sigma_k(K^*)$, then $\delta_K$ applied to the array \eqref{110226:eq2}
is in the image of $\partial$, hence the array \eqref{110226:eq2} defines a cohomology class
in $\Omega^k(\mc V)$.
Therefore, $\chi_k$ is a well-defined map: $\Sigma_k(K^*)\to H^k(\Omega^\bullet(\mc V),\delta_K)$.

In order to complete the proof of both parts (b) and (c), we only need to check
that the map $\chi_k:\,\Sigma_k(K^*)\to H^k(\Omega^\bullet(\mc V),\delta_K)$
given by \eqref{110226:eq2} satisfies equation \eqref{110226:eq3}.
Indeed, the map $\phi_k:\,\Sigma_k(K^*)\to\ker(\alpha_{k+1})$ is an isomorphism by part (a),
and the map $\gamma_k:\,H^k(\Omega^\bullet(\mc V),\delta_K)\to\ker(\alpha_{k+1})$
is an isomorphism by Theorem \ref{100602:lem2} and the long exact sequence \eqref{100521:eq2}.
Hence, equation \eqref{110226:eq3} implies that the map $\chi_k$ must be an isomorphism as well.
Also, the last assertion in part (c) is clear since, by \eqref{110226:eq3},
we have $\gamma_k^{-1}=\chi_k\circ\phi_k^{-1}$.

Before proving equation \eqref{110226:eq3}, let us recall the usual
homological algebra definition of the boundary map $\gamma_k$
in the long exact sequence \eqref{100602:eq1}.
For $[\omega]\in H^k(\Omega^\bullet(\mc V),\delta_K)$,
we have
$$
\gamma_k([\omega])=\big[\alpha^{-1}\delta_K\beta^{-1}(\omega)\big]
\,\in\,H^{k+1}(\partial \tilde\Omega^\bullet(\mc V),\delta_K)\,.
$$
In other words, let $[\omega]\in H^k(\Omega^\bullet(\mc V),\delta_K)$
be the class of a closed element $\omega\in\Omega^k(\mc V)$.
Since $\beta:\,\tilde\Omega^k(\mc V)\to\Omega^k(\mc V)$ is surjective,
there exists $\eta\in\tilde\Omega^k(\mc V)$ such that $\beta(\eta)=\omega$.
Since, by assumption, $\delta_K\omega=0$, we have that $\delta_K(\eta)\in\ker(\beta)=\im(\alpha)$.
Hence, there exists $\zeta\in\partial \tilde\Omega^{k+1}(\mc V)$ such that
$\delta_K(\eta)=\alpha(\zeta)$.
Since $\alpha$ is injective, $\delta_K(\zeta)=0$,
and we let $\gamma_k([\omega])=[\zeta]\in H^{k+1}(\partial\tilde\Omega^\bullet(\mc V),\delta_K)$.
Using the identification of the exact sequence \eqref{100602:eq1}
with \eqref{100521:eq2},
the construction of the map $\gamma_k$ can be described as follows.
Consider a skewsymmetric array
$P=\big(P_{i_1,\dots,i_k}(\lambda_1,\dots,\lambda_k)\big)_{i_1,\dots,i_k\in I}\in\tilde\Omega^k(\mc V)$
such that, when viewed as an element in $\Omega^k(\mc V)$ (i.e. when we view its entries
in $\mb F_-[\lambda_1,\dots,\lambda_k]\otimes_{\mb F[\partial]}\mc V$), it is closed: $\delta_KP=0$
in $\Omega^{k+1}(\mc V)$.
By Theorem \ref{100520:th}(a) there exists a unique skewsymmetric array
$C=\big(C_{i_0,\dots,i_k}(\lambda_0,\dots,\lambda_k)\big)_{i_0,\dots,i_k\in I}\in\tilde\Omega^{k+1}_{0,0}$,
where $C_{i_0,\dots,i_k}(\lambda_0,\dots,\lambda_k)$ are polynomials with quasiconstant coefficients
of degree at most $N-1$ in each variable $\lambda_i$,
such that
$\delta_KP=\partial(C+\delta_KQ)$ for some $Q\in\tilde\Omega^k(\mc V)$.
Then $\gamma_k([P])=C$.

Given $P\in\Sigma_k(K^*)$, the array $\phi_k(P)=C\in\ker(\alpha_{k+1})$
is defined by equation \eqref{110226:eq1}.
On the other hand, a representative of the cohomology class $\chi_k(P)\in H^k(\Omega^\bullet(\mc V),\delta_K)$
is the array $Q\in\tilde\Omega^k(\mc V)$ given by \eqref{110226:eq2},
and, by the above observations, the array $\gamma_k(\chi_k(P))=C_1\in\ker(\alpha_{k+1})$
is defined by the equation $\delta_K Q=\partial C_1$.
Since this equation for $C_1$ coincides with equation \eqref{110226:eq1} for $C$,
we conclude that $C_1=C$, proving formula \eqref{110226:eq3}.
\end{proof}

\subsection{Explicit description of $H^0(\Omega^\bullet(\mc V),\delta_K)$
and $H^1(\Omega^\bullet(\mc V),\delta_K)$.}

Assume that $\mc V$ is a normal algebra of differential functions and $\mc F\subset\mc V$
is a linearly closed differential field, so that Theorem \ref{110213:thm} holds.
It is easy to see from the definition \eqref{100519:eq1} of the action
of $\delta_K$ on $\Omega^0=\mc V/\partial\mc V$ that
$$
H^0(\Omega^\bullet(\mc V),\delta_K)
=
\Big\{\tint f\in\mc V/\partial\mc V\,\Big|\,K^*(\partial)\frac{\delta f}{\delta u}=0\Big\}\,.
$$
The space $\Sigma_0(K^*)$ described before Theorem \ref{110213:thm} is
$$
\Sigma_0(K^*)
=
\Big\{P\in\mc F^\ell\,\Big|\,K^*(\partial)P=0\Big\}\,,
$$
and the isomorphism $\chi_0:\,\Sigma_0(K^*)\to H^0(\Omega^\bullet(\mc V),\delta_K)$,
defined in Theorem \ref{110213:thm}(b), is given by
$$
\chi_0(P)=\int\sum_{j\in I}P_ju_j\,.
$$
It is immediate to check that the variational derivative of $\tint\sum_j P_ju_j$ is $P$,
hence, if $P$ lies in $\Sigma_0(K^*)$, then $\chi_0(P)$ lies in $H^0(\Omega^\bullet(\mc V),\delta_K)$.

Recall from Section \ref{sec:8.2} that $\Omega^1(\mc V)$ is naturally identified with
$\mc V^{\oplus\ell}$.
Under this identification, the space of exact elements in $\Omega^1(\mc V)$ is
$$
B^1(\Omega^\bullet(\mc V),\delta_K)
=
\Big\{K^*(\partial)\frac{\delta f}{\delta u}
\,\Big|\,\tint f\in\mc V/\partial\mc V\Big\}\,.
$$
Moreover, it is not hard to check using the definition \eqref{100519:eq1}
of $\delta_K$, that the space of closed elements in $\Omega^1(\mc V)$ is
$$
Z^1(\Omega^\bullet(\mc V),\delta_K)
=
\Big\{F\in\mc V^\ell
\,\Big|\,
D_F(\partial)\circ K(\partial)=K^*(\partial)\circ D_F^*(\partial)
\Big\}\,,
$$
where $D_F(\partial)$ is the Frechet derivative \eqref{frechet}
and $D^*_F(\partial)$ its adjoint matrix differential operator.
On the other hand, it is easy to see that the space $\Sigma_1(K^*)$
consists of matrix differential operators
$P=\big(P_{ij}(\partial)\big)_{i,j\in I}$
with quasiconstant coefficients and of order at most $N-1$,
solving the equation
\begin{equation}\label{110226:eq5}
K^*(\partial)\circ P(\partial)=P^*(\partial)\circ K(\partial)\,.
\end{equation}
The isomorphism $\chi_1:\,\Sigma_1(K^*)\to H^1(\Omega^\bullet(\mc V),\delta_K)$
defined in Theorem \ref{110213:thm}(b) is given by $\chi_1(P)=F+\delta_K(\Omega^0(\mc V))$, where
\begin{equation}\label{110226:eq6}
F=
\big(\sum_{j\in I}P^*_{ij}(\partial)u_j\big)_{i\in I}\in\mc V^\ell\,.
\end{equation}
It is not hard to check that, if $F$ is as in \eqref{110226:eq6},
then its Frechet derivative is
$$
D_F(\partial)=P^*(\partial)\,,
$$
hence, if $P$ satisfies equation \eqref{110226:eq5}, then $F$ lies in $Z^1(\Omega^1(\mc V),\delta_K)$.

\begin{remark}
Recall that, if $H$ and $K$ are compatible Hamiltonian operators,
the Lenard scheme is the following recurrent relation:
$$
H(\partial)\frac{\delta h_{n}}{\delta u}
=
K(\partial)\frac{\delta h_{n+1}}{\delta u}
\,,
$$
or, equivalently,
$$
[H,\tint h_n]=[K,\tint h_{n+1}]\,,
$$
in the Lie superalgebra $W^{\var}(\Pi\mc V)\simeq\Omega^\bullet(\mc V)$.
The Hamiltonian functions $\tint h_n$ are constructed by induction on $n\in\mb Z_+$.
In fact, as explained in the introduction (see equation \eqref{110320:eq1}),
assuming that we have constructed $\tint h_j,\,j=0,\dots,n-1$
satisfying the Lenard recurrence formula,
then $[H,\tint h_{n-1}]$ is a closed element of $(\Omega^1(\mc V),\delta_K)$.
Hence, by equation \eqref{110226:eq6},
there exist a Hamiltonian function $\tint h_n\in\mc V/\partial\mc V$
and a unique $P\in\Sigma_1(K^*)$, i.e. a matrix differential operator $P=\big(P_{ij}(\partial)\big)_{i,j\in I}$
of order at most $N-1$ with quasiconstant coefficients
solving \eqref{110226:eq5}, such that the following equation holds in $\mc V^\ell$:
\begin{equation}\label{110526:eq1}
[H,\tint h_{n-1}]=[K,\tint h_n]+\big(\sum_{j\in I}P^*_{ij}(\partial)u_j\big)_{i\in I}\,.
\end{equation}
In order to complete the $n$-th step of the Lenard scheme, we have to show that $P=0$.
For this, the following observations may be used.

First note that, since $[H,K]=0$, $\ad H$ induces a well defined linear map
$H^k(\Omega^\bullet(\mc V),\delta_K)\to H^{k+1}(\Omega^\bullet(\mc V),\delta_K)$,
hence, thanks to the isomorphism $\chi_k:\Sigma_k(K^*)\to H^k(\Omega^\bullet(\mc V),\delta_K)$
defined in Theorem \ref{110213:thm}, we get an induced linear map
$\alpha_k^H:\,\Sigma_k(K^*)\to\Sigma_{k+1}(K^*)$.
On the other hand, applying $\ad H$ to both sides of equation \eqref{110526:eq1},
we get that
$(\ad H)\big(\sum_{j\in I}P^*_{ij}(\partial)u_j\big)_{i\in I}$ is an exact element of $(\Omega^2(\mc V),\delta_K)$,
or, equivalently, $\alpha_1^H(P)=0$.
Thus, in order to apply the Lenard scheme at the $n$-th step, it suffices to show that
$\ker(\alpha_1^H)=0$.

Recalling formula \eqref{100503:eq9} for the action of $\ad H$ on $\mc V^\ell=\Omega^1(\mc V)$,
it is not hard to show that the condition that $\alpha_1^H$ is injective translates to the condition
that, if
$$
X_{P^*(\partial)u}(H)-H(\partial)\circ P(\partial)-P^*(\partial)\circ H(\partial)
=
-K(\partial)\circ D_F^*(\partial)-D_F(\partial)\circ K(\partial)\,,
$$
for some $P\in\Sigma_1(K^*)$ and $F\in\mc V^\ell$, then $P=0$.
\end{remark}


\appendix
\numberwithin{equation}{subsection}
\numberwithin{theorem}{subsection}

\section{Systems of linear differential equations and (poly)differential operators}

In this Appendix we prove some facts about matrix differential and polydifferential operators needed
in the computation of the variational Poisson cohomology (cf. Section \ref{sec:11.2}).
In order to establish these facts, we use the theory of systems of linear differential equations
in several unknowns.
This theory has been developed by a number of authors, see \cite{Ler}, \cite{Vol}, \cite{Huf},
\cite{SK}, \cite{Miy}.
Our exposition (which we developed before becoming aware of the above references)
is given in the spirit of differential algebra, as the rest of the paper.

\subsection{Lemmas on differential operators}\label{app:1}

\noindent
Let $\mc M$ be a unital associative (not necessarily commutative) algebra, with a derivation $\partial$.
Consider the algebra of differential operators $\mc M[\partial]$.
Its elements are expressions of the form
\begin{equation}\label{101218:eq1}
P(\partial)=\sum_{n=0}^Na_n\partial^n\quad,\qquad a_n\in\mc M\,,
\end{equation}
which are multiplied according to the rule $\partial\circ a=a\partial+a'$.
If $a_N\neq0$, then we say that $P(\partial)$ has \emph{order} $\ord(P)=N$
and we call $a_N\in\mc M$ its \emph{leading coefficient}.
\begin{lemma}\label{101218:lem2}
If the differential operator
\begin{equation}\label{101218:eq2}
\sum_{m=0}^M \partial^{m+1}\circ a_m\partial^m
+\sum_{n=0}^N \partial^{n}\circ b_n\partial^n \in\mc M[\partial]
\end{equation}
is zero, then all the elements $a_m$ and $b_n$ are zero.
Hence, in a differential operator of the form \eqref{101218:eq2}, the elements $a_m$ and $b_n$
are uniquely determined.
\end{lemma}
\begin{proof}
In the contrary case, two things can happen: either
$a_M\neq0$ and $2M+1>2N$, or $b_N\neq0$ and $2N>2M+1$.
In the first case, the operator \eqref{101218:eq2} has order $2M+1$,
and the leading coefficient is $a_M$, a contradiction.
Similarly in the second case.
\end{proof}
\begin{lemma}\label{101218:lem1}
Let $p,q\in\mb Z_+$ and $a\in\mc M$. Then
\begin{enumerate}[(a)]
\item
for $p > q\in\mb Z_+$
$$
\begin{array}{c}
\displaystyle{
\partial^p\circ a\partial^q
=
\sum_{m=q}^{[(p+q-1)/2]} \binom{p-m-1}{m-q} \partial^{m+1}\circ a^{(p+q-2m-1)}\partial^m
} \\
\displaystyle{
+\sum_{m=q+1}^{[(p+q)/2]} \binom{p-m-1}{m-q-1} \partial^{m}\circ a^{(p+q-2m)}\partial^m\,;
}
\end{array}
$$
\item
for $p < q\in\mb Z_+$
$$
\begin{array}{c}
\displaystyle{
\partial^p\circ a\partial^q
=
\sum_{m=p}^{[(p+q-1)/2]}\gamma^{p,q}_m \partial^{m+1}\circ a^{(p+q-2m-1)}\partial^m
} \\
\displaystyle{
+\sum_{m=p}^{[(p+q)/2]}\delta^{p,q}_m \partial^{m}\circ a^{(p+q-2m)}\partial^m\,,
}
\end{array}
$$
where $\gamma^{p,q}_m$ and $\delta^{p,q}_m$ are integers.
\end{enumerate}
\end{lemma}
\begin{proof}
(a). By induction on $p-q$. For $p-q=1$, the statement is immediate to check.
For $p-q=2$, we have:
$\partial^{q+2}\circ a\partial^q
=\partial^{q+1}\circ a'\partial^q+\partial^{q+1}\circ a\partial^{q+1}$,
which agrees with our claim.
For $p-q\geq3$, we have, by induction,
$$
\begin{array}{l}
\displaystyle{
\partial^p\circ a\partial^q
=\partial^{p-1}\circ a'\partial^q+\partial^{p-1}\circ a\partial^{q+1}
} \\
\displaystyle{
=\sum_{m=q}^{[(p+q-2)/2]} \binom{p-m-2}{m-q} \partial^{m+1}\circ a^{(p+q-2m-1)}\partial^m
} \\
\displaystyle{
+\sum_{m=q+1}^{[(p+q-1)/2]} \binom{p-m-2}{m-q-1} \partial^{m}\circ a^{(p+q-2m)}\partial^m
} \\
\displaystyle{
+\sum_{m=q+1}^{[(p+q-1)/2]} \binom{p-m-2}{m-q-1} \partial^{m+1}\circ a^{(p+q-2m-1)}\partial^m
} \\
\displaystyle{
+\sum_{m=q+2}^{[(p+q)/2]} \binom{p-m-2}{m-q-2} \partial^{m}\circ a^{(p+q-2m)}\partial^m
} \\
\displaystyle{
=\sum_{m=q}^{[(p+q-1)/2]}
\binom{p-m-1}{m-q}
\partial^{m+1}\circ a^{(p+q-2m-1)}\partial^m
} \\
\displaystyle{
+\sum_{m=q+1}^{[(p+q-1)/2]}
\binom{p-m-1}{m-q-1}
\partial^{m}\circ a^{(p+q-2m)}\partial^m
} \,.
\end{array}
$$
In the last identity we used the Tartaglia-Pascal triangle.

(b). It follows from (a), since, by the binomial formula,
$$
\partial^p\circ a\partial^q=
\sum_{h=0}^{q-p}\binom{q-p}{h}(-1)^h
\partial^{q-h}\circ a^{(h)}\partial^p\,.
$$
\end{proof}
\begin{corollary}\label{101218:cor1}
Any differential operator $P(\partial)\in\mc M[\partial]$ of order less than or equal to $N$
can be written, in a unique way, in any of these three forms:
$$
\begin{array}{l}
\displaystyle{
P(\partial) =
\sum_{n=0}^Na_n\partial^n
= \sum_{n=0}^N \partial^n\circ b_n
}\\
\displaystyle{
= \sum_{m=0}^{[(N-1)/2]} \partial^{m+1}\circ c_n\partial^m
+ \sum_{n=0}^{[N/2]} \partial^n\circ d_n\partial^n\,.
}
\end{array}
$$
\end{corollary}
\begin{proof}
Existence is clear. Uniqueness of the first two forms is clear, and the third one
is Lemma \ref{101218:lem2}.
\end{proof}

Suppose next that $\mc M$ has an anti-involution $a\mapsto a^*$, commuting with $\partial$.
\begin{example}
If $\mc M=Mat_{\ell\times\ell}(\mc F)$, where $\mc F$ is a commutative differential algebra,
we let $a^*=a^T$, the transpose matrix.
\end{example}
\noindent
We extend $*$ to an anti-involution of $\mc M[\partial]$ by letting
$$
\big(\sum_{n=0}^Na_n\partial^n\big)^*=
\sum_{n=0}^N(-\partial)^n\circ a_n^*\,.
$$
We say that $P(\partial)$ is selfadjoint (respectively skewadjoint)
if $P^*(\partial)=P(\partial)$ (resp. $P^*(\partial)=-P(\partial)$).
\begin{lemma}\label{101218:lem3}
\begin{enumerate}[(a)]
\item
If $S(\partial)$ is a skewadjoint operator of order less than or equal to $N$,
then it can be written, in a unique way, in the form
\begin{equation}\label{110108:eq4}
S(\partial)=
\sum_{m=0}^{[(N-1)/2]}
\partial^m\circ
\Big(\partial\circ a_m+a_m\partial \Big)\partial^m
+\sum_{m=0}^{[N/2]}
\partial^m\circ b_m\partial^m
\end{equation}
where $a_m=a_m^*$ and $b_m=-b_m^*$.
\item
If $S(\partial)$ is a selfadjoint operator of order less than or equal to $N$,
then it can be written, in a unique way, in the form
$$
S(\partial)=
\sum_{m=0}^{[(N-1)/2]}
\partial^m\circ
\Big(\partial\circ a_m+a_m\partial \Big)\partial^m
+\sum_{m=0}^{[N/2]}
\partial^m\circ b_m\partial^m
$$
where $a_m=-a_m^*$ and $b_m=b_m^*$.
\end{enumerate}
\end{lemma}
\begin{proof}
Use  the third form of Corollary \ref{101218:cor1} and compute $S-S^*$ (resp. $S+S^*$).
\end{proof}

\subsection{Linear algebra over a differential field}\label{app:2}

Let $\mc F$ be a differential field, i.e. a field with a derivation $\partial$,
and let $\mc C=\{c\in\mc F\,|\,\partial c=0\}\subset\mc F$ be the subfield of constants.

{\bf Notation:} $a,b,c,\dots\in\mc F$, $\alpha,\beta,\gamma,\dots\in\mc C$, $u,v,w$ variables,
$m,n,p,q\in\mb Z_+$.

A system of $m$ linear differential equations in the variables $u_i,\,i=1,\dots,\ell$,
has the form
\begin{equation}\label{eq:star}
M(\partial)u=b\,,
\end{equation}
where
$$
u=
\left(\begin{array}{l}
u_1 \\ \vdots \\ u_\ell
\end{array}\right)
\,\,,\,\,\,\,
b=
\left(\begin{array}{l}
b_1 \\ \vdots \\ b_\ell
\end{array}\right)
\,\,,\,\,\,\,
M(\partial)=
\left(\begin{array}{lll}
L_{11}(\partial) & \dots & L_{1\ell}(\partial) \\
 & \dots & \\
L_{m1}(\partial) & \dots & L_{m\ell}(\partial)
\end{array}\right)
\,,
$$
with $b_i\in\mc F$ and $L_{ij}(\partial)\in\mc F[\partial]$.


In order to study this system of linear differential equations we will use the following simple result:
\begin{lemma}\label{101221:lem1}
Let $(n_{ij})$ be an $m\times\ell$ matrix with entries in $\mb Z$,
and let
\begin{equation}\label{b}
N_j=\max_i\{n_{ij}\}
\quad,\qquad
h_i=\min_j\{N_j-n_{ij}\}\,.
\end{equation}
Then
\begin{equation}\label{a}
n_{ij}\leq N_j-h_i\,\,,\,\,\,\,\forall i=1,\dots,m,\,j=1,\dots,\ell\,.
\end{equation}
Any other choice $N'_j,\,h'_i$ satisfying \eqref{a} is such that
$N'_j\geq N_j$ for all $j$,
and, if $N'_j=N_j$ for all $j$, then
$h'_i\leq h_i$ for all $i$.
\end{lemma}
\begin{proof}
Clearly, \eqref{a} holds.
Given $j$, there exists $i$ such that $N_j=n_{ij}$.
But, by assumption, $n_{ij}\leq N'_j-h'_i$. Hence $N_j\leq N'_j$.
Suppose now that $N_j=N'_j$ for all $j$.
Given $i$ there exists $j$ such that $h_i=N_j-n_{ij}=N'_j-n_{ij}$.
But $n_{ij}\leq N'_j-h'_i$, hence $h_i\geq h'_i$.
\end{proof}
\begin{definition}\label{101223:def}
The collection of integers $\{N_j,\,j=1,\dots,\ell;\,h_i,\,i=1,\dots\!$ $\dots,m\}$ 
satisfying \eqref{a} is called a \emph{majorant} of the matrix $(n_{ij})$.
\end{definition}
%

Consider the system of equations \eqref{eq:star}.
A \emph{majorant} $\{N_j;h_i\}$
of the $m\times\ell$ matrix differential operator $M(\partial)$
is defined as a majorant of its matrix of orders $(n_{ij})$.
Given an arbitrary majorant $\{N_j;h_i\}$ of the matrix differential operator $M(\partial)$,
we can write the $i,j$ entry of $M(\partial)$ in the form
$$
L_{ij}(\partial)
=
\sum_{n=0}^{N_j-h_i} a_{ij;n}\partial^n
\,\,,\,\,\,\,
a_{ij;n}\in\mc F\,.
$$
We define the corresponding \emph{leading matrix} as
the following $m\times\ell$ matrix whose entries are monomials in an indeterminant $\xi$
with coefficients in $\mc F$:
\begin{equation}\label{101221:eq1b}
\bar M(\xi)
=
\left(\begin{array}{lll}
a_{11;N_1-h_1}\xi^{N_1-h_1} & \dots & a_{1\ell;N_\ell-h_1} \xi^{N_\ell-h_1} \\
 & \dots & \\
a_{m1;N_1-h_m}\xi^{N_1-h_m} & \dots & a_{m\ell;N_\ell-h_m}\xi^{N_\ell-h_m}
\end{array}\right)
\,.
\end{equation}
Clearly, if $m=\ell$, we have
\begin{equation}\label{110404:eq1}
\det(\bar M(\xi))
=
\det(\bar M(1))\,\xi^d\,,
\end{equation}
where
\begin{equation}\label{110404:eq2}
d=\sum_{j=1}^\ell(N_j-h_j)\,.
\end{equation}
Note that this matrix depends on the choice of the majorant $\{N_j;\,h_i\}$,

%

Permuting the equations in the system \eqref{eq:star} if necessary,
we can (and will) assume that $h_1\geq\dots\geq h_\ell$.

As in linear algebra,
the set of solutions of the system \eqref{eq:star} does not change
if we exchange two equations or if we add to the $i$-th equation
the $j$-th equation, with $j\neq i$, to which we apply a differential operator $P(\partial)$.
Since we want to preserve the fact that $\ord(L_{ij})\leq N_j-h_i$,
we give the following:
\begin{definition}\label{101222:def}
An \emph{elementary row operation} of the matrix differential operator $M(\partial)$
is either a permutation of two rows of it,
or the operation $\mc T(i,j;P)$, where $1\leq i\neq j\leq m$
and $P(\partial)$ is a differential operator,
which replaces the $j$-th row by itself minus $i$-th row
multiplied on the left by $P(\partial)$.
Assuming that $h_1\geq\dots\geq h_\ell$,
we say that the elementary row operation $\mc T(i,j;P)$
is \emph{majorant preserving}
if $i<j$ and $P(\partial)$ has order less than or equal to $h_i-h_j$.
\end{definition}

\begin{remark}\label{101222:rem}
After a majorant preserving row operation $\mc T(i,j;P)$,
the leading matrix $\bar M(\xi)$ in \eqref{101221:eq1b}
is unchanged unless $P(\partial)$ has order equal to $h_i-h_j$,
and in this case it changes by an elementary row operation over $\mc F$,
namely we add to the $j$-th row of $\bar M(\xi)$ the $i$-th row multiplied by the
leading coefficient of $P(\partial)$.
\end{remark}

Using the usual Gauss elimination, we can get the (well known) analogues
of standard linear algebra theorems for matrix differential operators.
In particular, we have the following
\begin{lemma}\label{110406:lem}
Any $m\times\ell$ matrix differential operator $M(\partial)$
can be brought by elementary row operations to a row echelon form.
\end{lemma}
\begin{proof}
Let $j_1$ be the first non zero column of $M(\partial)$.
Among all matrices obtained from $M(\partial)$ by elementary row operations,
chose one for which the first entry of column $j_1$,
$L_{1j_1}(\partial)$, is non zero of minimal possible order (minimal among the orders
of the $(1,j_1)$ entry in all these matrices).
Clearly, all the other entries in column $j_1$ must be divisible (on the left) by $L_{1j_1}(\partial)$,
and using elementary row operations we can make them zero.
Then, we proceed by induction on the submatrix with first row deleted.
\end{proof}
We next discuss majorant preserving Gauss elimination for a matrix differential operator.
\begin{lemma}\label{101222:lem}
Consider the $m\times\ell$ matrix differential operator
$M(\partial)=\big(L_{ij}(\partial)\big)$ with $m\leq\ell$,
and let $\{N_1,\dots,N_\ell;\,h_1\geq\dots\geq h_m\}$
be a majorant of $M(\partial)$.
Suppose, moreover, that
$\ord(L_{jj})=N_j-h_j$ for $1\leq j\leq m-1$,
and $\ord(L_{ij})<N_j-h_j$ for $1\leq j<i\leq m-1$.
Then, we can perform majorant preserving elementary row operations
on the $m$-th row of $M(\partial)$ so that its new
$m$-th row $\tilde L_{mj}(\partial),\,j=1,\dots,\ell$,
satisfies:
$$
\ord(\tilde L_{mj})<N_j-h_j
\,\, \text{ for } j<m
\,\,,\,\,\,\,
\ord(\tilde L_{mj})\leq N_j-h_m
\,\, \text{ for } j\geq m\,.
$$
\end{lemma}
\begin{proof}
By assumption, the $m$-th row of the starting matrix $M(\partial)$ satisfies
$$
\ord(L_{mj})\leq N_j-h_m
\,\text{ for } 1\leq j\leq\ell\,.
$$
Applying to $M(\partial)$ the elementary row operation $\mc T(1,m;c\partial^{h_1-h_m})$,
where $c=a_{m1;N_1-h_m}/a_{11;N_1-h_1}$,
we get a new matrix satisfying
$$
\ord(\tilde L_{m1})\leq N_1-h_m-1
\,\,,\,\,\,\,
\ord(\tilde L_{mj})\leq N_j-h_m
\,\text{ for } 2\leq j\leq\ell\,.
$$
Next,
applying the elementary row operation $\mc T(2,m;c\partial^{h_2-h_m})$,
for suitable $c\in\mc F$,
we get a new matrix satisfying
$$
\begin{array}{l}
\ord(\tilde L_{m1})\leq N_1-h_m-1
\,\,,\,\,\,\,
\ord(\tilde L_{m2})\leq N_2-h_m-1
\,,\\
\ord(\tilde L_{mj})\leq N_j-h_m
\,\text{ for } 3\leq j\leq\ell\,.
\end{array}
$$
Proceeding in the same way, we get after $m-1$ steps,
a new matrix satisfying
$$
\begin{array}{l}
\ord(\tilde L_{mj})\leq N_j-h_m-1
\,\text{ for } 1\leq j\leq m-1
\,,\\
\ord(\tilde L_{mj})\leq N_j-h_m
\,\text{ for } m\leq j\leq\ell\,.
\end{array}
$$
If $h_1=h_2=\dots=h_m$, we are done.
Otherwise, $h_1>h_m$ and
we apply to the last matrix
the elementary row operation $\mc T(1,m;c\partial^{h_1-h_m-1})$,
where $c=\tilde a_{m1;N_1-h_m-1}/a_{11;N_1-h_1}$.
We thus get a new matrix satisfying
$$
\begin{array}{l}
\ord(\tilde L_{m1})\leq N_1-h_m-2
\,\,,\,\,\,\,
\ord(\tilde L_{mj})\leq N_j-h_m-1
\,\text{ for } 2\leq j\leq m-1
\,,\\
\ord(\tilde L_{mj})\leq N_j-h_m
\,\text{ for } m\leq j\leq\ell\,.
\end{array}
$$
Next, if $h_2>h_m$,
we apply the row operation $\mc T(2,m;c\partial^{h_2-h_m-1})$,
for suitable $c\in\mc F$,
and we get a new matrix satisfying
$$
\begin{array}{l}
\ord(\tilde L_{m1})\leq N_1-h_m-2
\,\,,\,\,\,\,
\ord(\tilde L_{m2})\leq N_2-h_m-2
\,,\\
\ord(\tilde L_{mj})\leq N_j-h_m-1
\,\text{ for } 3\leq j\leq m-1
\,,\\
\ord(\tilde L_{mj})\leq N_j-h_m
\,\text{ for } m\leq j\leq\ell\,.
\end{array}
$$
Let $r\leq m-1$ be such that $h_r>h_{r+1}=h_m$.
We proceed in the same way and we get,
after $r$ steps, a new matrix satisfying
$$
\begin{array}{l}
\ord(\tilde L_{mj})\leq N_j-h_m-2
\,\text{ for } 1\leq j\leq r
\,,\\
\ord(\tilde L_{mj})\leq N_j-h_j-1
\,\text{ for } r+1\leq j\leq m-1
\,,\\
\ord(\tilde L_{mj})\leq N_j-h_m
\,\text{ for } m\leq j\leq\ell\,.
\end{array}
$$
If $h_r-h_m\geq 2$,
we again apply consecutively elementary row operations
$\mc T(1,m;c_1\partial^{h_1-h_m-2}),\,
\mc T(2,m;c_2\partial^{h_2-h_m-2}),
\dots,\,
\mc T(r,m;c_r\partial^{h_r-h_m-2})$,
for appropriate $c_1,\dots,c_r\in\mc F$.
As a result we get a new matrix satisfying
$$
\begin{array}{l}
\ord(\tilde L_{mj})\leq N_j-h_m-3
\,\text{ for } 1\leq j\leq r
\,,\\
\ord(\tilde L_{mj})\leq N_j-h_j-1
\,\text{ for } r+1\leq j\leq m-1
\,,\\
\ord(\tilde L_{mj})\leq N_j-h_m
\,\text{ for } m\leq j\leq\ell\,,
\end{array}
$$
and proceeding as before $h_r-h_m$ times,
we get a matrix satisfying
$$
\begin{array}{l}
\ord(\tilde L_{mj})\leq N_j-h_r-1
\,\text{ for } 1\leq j\leq r
\,,\\
\ord(\tilde L_{mj})\leq N_j-h_j-1
\,\text{ for } r+1\leq j\leq m-1
\,,\\
\ord(\tilde L_{mj})\leq N_j-h_m
\,\text{ for } m\leq j\leq\ell\,.
\end{array}
$$
If $h_1=\dots=h_r$, we are done.
Otherwise,
let $s\leq r-1$ be such that $h_s>h_{s+1}=\dots=h_r$.
We proceed in the same way as before to get,
after a finite number of steps, a new matrix satisfying
$$
\begin{array}{l}
\ord(\tilde L_{mj})\leq N_j-h_s-1
\,\text{ for } 1\leq j\leq s
\,,\\
\ord(\tilde L_{mj})\leq N_j-h_j-1
\,\text{ for } s+1\leq j\leq m-1
\,,\\
\ord(\tilde L_{mj})\leq N_j-h_m
\,\text{ for } m\leq j\leq\ell\,.
\end{array}
$$
Continuing along these lines, one gets the desired result.
\end{proof}

\begin{proposition}\label{101222:prop}
Let $M(\partial)=\big(L_{ij}(\partial)\big)$
be an $\ell\times\ell$ matrix differential operator
with majorant $\{N_1,\dots,N_\ell;\,h_1\geq\dots\geq h_\ell\}$.
Assume that the leading matrix $\bar M(\xi)$
associated to this majorant, defined in \eqref{101221:eq1b}, is non degenerate.
Then, after possibly permuting the columns of $M(\partial)$,
and after applying majorant preserving elementary row operations,
we get a matrix of the form
$\tilde M(\partial)=\big(\tilde L_{ij}(\partial)\big)$, where
$$
\begin{array}{l}
\ord(\tilde L_{jj})=N_j-h_j
\,\text{ for } 1\leq j\leq\ell
\,,\\
\ord(\tilde L_{ij})\leq N_j-h_i
\,\text{ for } 1\leq i<j\leq\ell
\,,\\
\ord(\tilde L_{ij})<N_j-h_j
\,\text{ for } 1\leq j<i\leq\ell
\,.
\end{array}
$$
\end{proposition}
\begin{proof}
Since the first row of the leading matrix $\bar M(\xi)$ is non zero,
after possibly exchanging the first column of $M(\partial)$
with its $j$-th column, $j>1$,
we can assume that $L_{11}(\partial)$ has order $N_1-h_1$.
Applying Lemma \ref{101222:lem} to the first two rows of the matrix $M(\partial)$,
we get, after elementary row operations on the second row,
a new matrix $\tilde M(\partial)$
with $\tilde L_{21}(\partial)$ of order strictly less than $N_1-h_1$,
and $\tilde L_{2j}(\partial)$ of order less than or equal to $N_j-h_2$ for $j\geq 2$.
By Remark \ref{101222:rem}, the leading matrix $\bar{\tilde M}(\xi)$ of the new matrix $\tilde M(\partial)$
is again non degenerate,
and it has zero in position $(2,1)$.
In particular, the first two rows of $\bar{\tilde M}(\xi)$ are linearly independent,
and, after possibly exchanging the second column
with the $j$-th column with $j>2$,
we can assume that $\tilde a_{22;N_2-h_2}\neq0$,
i.e. $\ord(\tilde L_{22})=N_2-h_2$.
Applying Lemma \ref{101222:lem} to the first three rows of the matrix $\tilde M(\partial)$,
we get, after elementary row operations on the third row,
a new matrix with
$\ord(\tilde L_{31})<N_1-h_1,\,\ord(\tilde L_{32})<N_2-h_2,\,
\ord(\tilde L_{3j})\leq N_j-h_3,\,j\geq 3$.
Repeating the same procedure for each subsequent row, we get the desired result.
\end{proof}
\begin{proposition}\label{101223:prop}
Let $M(\partial)=\big(L_{ij}(\partial)\big)$ be an $\ell\times\ell$
matrix differential operator as in the conclusion of Proposition \ref{101222:prop},
i.e.
$\ord(L_{jj})=N_j-h_j$ for $1\leq j\leq\ell$,
$\ord(L_{ij})\leq N_j-h_i$ for $1\leq i<j\leq\ell$,
and $\ord(L_{ij})<N_j-h_j$ for $1\leq j<i\leq\ell$,
where $N_1,\dots,N_\ell$ and $h_i\geq\dots\geq h_\ell$ are non negative integers.
Then it has the following majorant $\{N'_j;h'_i\}$:
$$
N'_1=N_1-h_1,\dots,N'_\ell=N_\ell-h_\ell
\,\,;\,\,\,\,
h'_1=\dots=h'_\ell=0\,.
$$
The leading matrix $\bar M(\xi)$ associated
to this majorant of $M(\partial)$ is upper triangular with non zero diagonal entries,
and the $(ij)$ entry with $i<j$ is zero unless $h_i=h_j$.
%
\end{proposition}
\begin{proof}
Obvious.
\end{proof}


The ring $\mc F[\partial]$ of scalar differential operators over $\mc F$
is naturally embedded in the ring of pseudodifferential operators $\mc F[\partial][[\partial^{-1}]]$,
which is a skewfield.

All the above definitions and statements have an obvious generalization
to pseudodifferential operators.
In particular, we define a majorant $\{N_j;h_i\}$
of an $m\times\ell$ matrix pseudodifferential operator $M(\partial)$
in the same way as before, and the corresponding leading matrix $\bar M(\xi)$
by the same equation \eqref{101221:eq1b}
(except that here we allow negative powers of the variable $\xi$).
We also define elementary row operations and majorant preserving elementary row operations
as in Definition \ref{101222:def},
except that we allow $P(\partial)$ to be a pseudodifferential operator.
Finally, in the case of pseudodifferential operators
Propositions \ref{101222:prop} and \ref{101223:prop} still hold,
and have the following stronger analogue:
\begin{proposition}\label{110404:propa}
Let $M(\partial)=\big(L_{ij}(\partial)\big)$
be an $\ell\times\ell$ matrix pseudodifferential operator
with majorant $\{N_1,\dots,N_\ell;\,h_1\geq\dots\geq h_\ell\}$.
Assume that the leading matrix $\bar M(\xi)$
associated to this majorant is non degenerate.
Then, after possibly permuting the columns of $M(\partial)$,
and after applying majorant preserving elementary row operations,
we get an upper triangular matrix
$\tilde M(\partial)=\big(\tilde L_{ij}(\partial)\big)$, with
$\ord(\tilde L_{jj})=N_j-h_j$ for all $j=1,\dots,\ell$.
The resulting matrix $\tilde M(\partial)$
has the following majorant $\{\tilde N_j;\tilde h_i\}$:
$$
\tilde N_1=N_1-h_1,\dots,\tilde N_\ell=N_\ell-h_\ell
\,\,;\,\,\,\,
\tilde h_1=\dots=\tilde h_\ell=0\,.
$$
\end{proposition}
\begin{proof}
First, following the proof of Proposition \ref{101222:prop},
we can apply majorant preserving elementary row operations, and possibly permutations of columns,
to reduce $M(\partial)$ to a matrix pseudodifferential operator satisfying the following conditions:
$$
\ord(\tilde L_{ij})<\ord(L_{jj})=N_j-h_j
\,\text{ for all } 1\leq j<i\leq\ell
\,.
$$
Let then $i>j$, and recall that, by assumption, $h_j\geq h_i$,
and, by the above condition,
$P(\partial)=L_{jj}(\partial)L_{ij}(\partial)^{-1}$
has negative order.
Hence, the elementary row operation $\mc T(j,i;P)$
is majorant preserving.
Applying such elementary row operations a finite number of times, we get the
desired upper triangular matrix.
The last statement is obvious.
\end{proof}


Recall that any $\ell\times\ell$ matrix pseudodifferential operator $M(\partial)$
has the Dieudonn\'e determinant
of the form $\det(M(\partial))=c\xi^d$, where $c\in\mc F$, $\xi$ is an indeterminate, and $d\in\mb Z$.
In fact, the Dieudonn\'e determinant is defined for square matrices over an arbitrary
skewfield $\mc K$, and it takes values in $\mc K^\times/(\mc K^\times,\mc K^\times)\cup\{0\}$,
\cite{Die}, \cite{Art}.
By definition,
$\det(M(\partial))$ changes sign if we permute two rows or two columns of $M(\partial)$,
and it is unchanged under any elementary row operation $\mc T(i,j;P)$
in Definition \ref{101222:def}, for arbitrary $i\neq j$ and a pseudodifferential operator $P(\partial)$.
Also, if $M(\partial)$ is upper triangular,
with diagonal entries $L_{ii}(\partial)$ of order $n_i$ and leading coefficient $a_i$,
then $\det(M(\partial))=\big(\prod_i a_i\big) \xi^{\sum_in_i}$.
It is proved in the above references that the Dieudonn\'e determinant is well defined
and $\det(A(\partial)B(\partial))=\det(A(\partial))\det(B(\partial))$
for every $\ell\times\ell$ matrix pseudodifferential operators $A(\partial)$
and $B(\partial)$.
Moreover, we have the following proposition (cf. \cite{Huf}, \cite{SK}, \cite{Miy}):
\begin{proposition}\label{110404:propb}
If $M(\partial)$ is an $\ell\times\ell$ matrix pseudodifferential operator
with non degenerate leading matrix $\bar M(\xi)$ (for a certain majorant $\{N_j;h_i\}$
of $M(\partial)$), then
$$
\det(M(\partial))=\det(\bar M(\xi))\,.
$$
In particular,
$\deg_\xi\det(M(\partial))=\sum_{j=1}^\ell(N_j-h_j)$.
\end{proposition}
\begin{proof}
It follows Proposition \ref{110404:propa} since,
by Remark \ref{101222:rem}, the determinant of the leading matrix $\bar M(\xi)$
is unchanged by majorant preserving elementary row operations.
\end{proof}
\begin{example}
If $\bar M(\xi)$ is degenerate, we can still have $\det(M(\partial))\neq0$.
For example, the matrix differential operator
$$
M(\partial)=\left(\begin{array}{cc} 1 & a \\ \partial & a\partial\end{array}\right)\,,
$$
has the majorant $N_1=N_2=1,\,h_1=1,h_2=0$,
and the corresponding leading matrix
$$
\bar M(\xi)=\left(\begin{array}{cc}1&a \\ \xi & a\xi\end{array}\right)
$$
is degenerate.
However, $M(\partial)$ can be brought, by elementary row operations, to
the matrix
$$
\left(\begin{array}{cc} 1 & a \\ 0 & -a'\end{array}\right)\,,
$$
which shows that $\det(M(\partial))=-a'$.
\end{example}

\subsection{Linearly closed differential fields}\label{app:3}

\begin{definition}\label{def:linclosed}
A differential field $\mc F$ is called \emph{linearly closed}
if any linear differential equation,
$$
a_nu^{(n)}+\dots+a_1u'+a_0u=b\,,
$$
with $n\geq0$, $a_0,\dots,a_n\in\mc F,\, a_n\neq0$,
has a solution in $\mc F$ for every $b\in\mc F$,
and it has a non zero solution for $b=0$,
provided that $n\geq1$.
\end{definition}
\begin{remark}
For a linearly closed differential field $\mc F$
and a non zero differential operator  $L(\partial)\in\mc F[\partial]$,
the map $L(\partial):\,\mc F\to\mc F$ given by $a\mapsto L(\partial)a$
is surjective.
Indeed, by definition, the differential equation $L(\partial)u=b$ has a solution
in $\mc F$ for every $b\in\mc F$.
\end{remark}
\begin{remark}
Any differential field $\mc F$ can be embedded in a linearly closed one.
Note also that a differentially closed field is automatically linearly closed.
\end{remark}
\begin{remark}\label{110406:rem}
Let $\mc F$ be a linearly closed differential field.
Letting $x\in\mc F$ be a solution of $\partial x=1$
we get that $\mc F$ contains the field of rational functions over $\mc C$ in $x$.
In particular, $\mc F$ is infinite dimensional over $\mc C$.
\end{remark}
\begin{theorem}\label{101218:thm1}
Let $\mc F$ be a differential field.
Consider a linear differential equation of order $N$ over $\mc F$ in the variable $u$:
$L(\partial)u=0$, where
$$
L(\partial)=a_N \partial^N+a_{N-1}\partial^{N-1}+\dots+a_1\partial+a_0\,\in\mc F[\partial]
\,\,,\,\,\,\,
a_N\neq0\,.
$$
\begin{enumerate}[(a)]
\item
The space of solutions of this equation is a vector space over $\mc C$ of dimension at most $N$.
\item
If $\mc F$ is linearly closed, then the space of solution has dimension equal to $N$.
\end{enumerate}
\end{theorem}
\begin{proof}
We prove (a) by induction on $N$. For $N=0$, it is clear.
For $N\geq1$, if there are no non zero solutions, we are done
(note that this does not happen if $\mc F$ is linearly closed).
If $a\in\mc F$ is a non zero solution of $L(\partial)u=0$,
we divide $L(\partial)$ by $\partial-a'/a$ with remainder, to get
$L(\partial)=L_1(\partial)(\partial-a'/a)+R$, where $L_1$ has order $N-1$ in $\partial$
and $R\in\mc F$.
Since $L(\partial)a=0$, it follows that $R=0$.
By inductive assumption, the space of solutions of $L_1(\partial)u=0$ has
dimension at most $N-1$ over $\mc C$.
Consider the linear map over $\mc C$, $b\mapsto (\partial-a'/a)b,\,b\in\mc F$.
It is immediate to check that it maps surjectively the space of solutions for $L(\partial)$
onto the space of solutions of $L_1(\partial)$, and its kernel is $\mc Ca$.
The statement (a) follows.
For part (b) we use the same argument.
\end{proof}
\begin{theorem}\label{110404:thm}
\begin{enumerate}[(a)]
\item
Let $M(\partial)$ be an $\ell\times\ell$ matrix differential operator
over a differential field $\mc F$.
\begin{enumerate}[i)]
\item
If $\det(M(\partial))\neq0$, then $\dim_{\mc C}(\ker M(\partial))\leq \deg_\xi\det(M(\partial))$.
\item
If $\im M(\partial)\subset\mc F^\ell$ has finite codimension over $\mc C$,
then $\det(M(\partial))\neq0$, provided that $\mc C\neq\mc F$.
\end{enumerate}
\item
Assuming that the differential field $\mc F$ is linearly closed,
the following statements are equivalent
for an $\ell\times\ell$ matrix differential operator $M(\partial)$:
\begin{enumerate}[i)]
\item
$\det(M(\partial))\neq0$,
\item
$\dim_{\mc C}(\ker M(\partial))<\infty$,
\item
$\det(M(\partial))\neq0$ and $\dim_{\mc C}(\ker M(\partial))=\deg_\xi\det(M(\partial))$,
\item
$\codim_{\mc C}\im M(\partial)<\infty$,
\item
$M(\partial):\,\mc F^\ell\to\mc F^\ell$ is surjective.
\end{enumerate}
\item
Let $M(\partial)$ be an $m\times\ell$ matrix differential operator
over a linearly closed differential field $\mc F$,
such that $\ker(M(\partial))$ has finite dimension over $\mc C$
and $\im(M(\partial))$ has finite codimension over $\mc C$.
Then necessarily $m=\ell$ and $\det(M(\partial))\neq0$.
\end{enumerate}
\end{theorem}
\begin{proof}
Since the dimension (over $\mc C$) of $\ker(M(\partial))$
and the codimension of $\im(M(\partial))$ are unchanged by elementary row operations on $M(\partial)$,
we may assume, by Lemma \ref{110406:lem}, that $M(\partial)$ is in row echelon form.
Assume first that $M(\partial)$ is an $\ell\times\ell$ matrix.
If $\det(M(\partial))\neq0$, it means that its diagonal entries $L_{ii}(\partial)$ are all non zero,
say of order $n_i$.
Hence the corresponding homogeneous system $M(\partial)u=0$
is upper triangular and, by Theorem \ref{101218:thm1},
its space of solutions has dimension less than or equal to $\sum_i n_i=\deg_\xi\det(M(\partial))$
(and equal to it, provided that $\mc F$ is linearly closed).
This proves part (a)(i) and, in (b), (i) implies (iii) (and hence it is equivalent to it).
Similarly, if $\det(M(\partial))=0$,
then the last row of $M(\partial)$ is zero,
so that $\im M(\partial)$ is of infinite codimension over $\mc C$.
Here we are using the fact that $\mc F$ is infinite dimensional over $\mc C$,
since any $f\in\mc F$, such that $f'\neq0$, is not algebraic over $\mc C$.
(Indeed, if $f\in\mc F$ is algebraic and $P(f)=0$ is its monic minimal polynomial over $\mc C$,
then $0=\partial P(f)=P'(f)f'$,
so that $P'(f)=0$ if $f'\neq0$, a contradiction.)
This proves part (a)(ii) and, in (b), that (iv) implies (i).

Since, in statement (b), condition (iii) obviously implies (ii),
and (v) obviously implies (iv),
in order to prove part (b) we only need to prove that (ii) implies (v).
Assume that $\mc F$ is linearly closed (hence infinite dimensional over $\mc C$,
by Remark \ref{110406:rem}),
and that, by condition (ii), $M(\partial)$ has finite dimensional kernel over $\mc C$.
Then its last diagonal entry is non zero,
otherwise, by Theorem \ref{101218:thm1}, there is a solution of the homogeneous system $M(\partial)u=0$
for every choice of $u_\ell$.
Therefore $M(\partial)$ is upper triangular with non zero diagonal entries,
and then, again by Theorem \ref{101218:thm1}, the inhomogeneous system $M(\partial)u=b$
has a solution for every $b\in\mc F^\ell$, i.e. $M(\partial)$ is surjective.
This completes the proof of part (b).

Finally, we prove part (c).
Assume, as before, that $M(\partial)$ is in row echelon form.
The homogeneous system $M(\partial)u=0$
admits a solution
for every choice of a coordinate $u_i$ which does not correspond to a pivot of $M(\partial)$.
Hence, since $\ker(M(\partial))$ is finite dimensional over $\mc C$,
we must have $m\geq\ell$.
If $m>\ell$, then the last $m-\ell$ rows of $M(\partial)$ are zero,
and then the image of $M(\partial)$ has infinite codimension.
\end{proof}
\begin{corollary}\label{101218:thm2}(cf. \cite{Huf,Miy})
Let $M(\partial)$ be an $\ell\times\ell$ matrix differential operator
with coefficients in a differential field $\mc F$.
Suppose that the leading matrix $\bar M(\xi)$,
associated to a majorant $\{N_j;h_i\}$ of $M(\partial)$, is non degenerate.
Then the space of solutions for the homogeneous system $M(\partial)u=0$
has dimension over $\mc C$ less than or equal to
$$
d = \sum_{j=1}^\ell (N_j-h_j)\,\,\, \Big(= \deg_\xi\det(M(\partial))\Big)\,.
$$
Moreover, if $\mc F$ is a linearly closed differential field
then the inhomogeneous system $M(\partial) u=b$ has a solution for every $b\in\mc F^\ell$,
and the space of solutions for the homogeneous  system $M(\partial)u=0$
has dimension equal to $d$.
\end{corollary}
\begin{proof}
By Proposition \ref{110404:propb}, if $\bar M(\xi)$ is non degenerate,
then $\det(M(\partial))=\det(\bar M(\xi))$,
and $\deg_\xi\det(\bar M(\partial))=\sum_j(N_j-h_j)=d$.
Hence, by Theorem \ref{110404:thm}, $\dim_{\mc C}(\ker M(\partial))\leq d$ and,
if $\mc F$ is linearly closed, $M(\partial)$ is surjective and $\dim_{\mc C}(\ker M(\partial))=d$.

\end{proof}

%

%
%

In Section \ref{app:5} we will need the following slight generalization of Corollary \ref{101218:thm2}.
\begin{corollary}\label{110328:cor}
Let $\mc F$ be a linearly closed differential field with subfield of constants $\mc C$.
Let $A(\partial)$ be an $\ell\times\ell$ matrix differential operator
such that $\det(A(\partial))\neq0$.
Let $M(\partial)=(L_{ij}(\partial))$ be an $\ell\times\ell$ matrix pseudodifferential operator
with non degenerate leading matrix $\bar M(\xi)$ associated to a majorant $\{N_j,\,j=1,\dots\ell;h_i,\,i=1,\dots\ell\}$.
Assume, moreover, that $A(\partial)M(\partial)$ is a matrix differential operator.
Then the inhomogeneous system of differential equations $A(\partial)M(\partial)u=b$
has a solution for every $b\in\mc F^\ell$,
and the space of solutions of the corresponding homogeneous system $A(\partial)M(\partial)u=0$
has dimension over $\mc C$ equal to
$$
d=\dim_{\mc C}(\ker A(\partial))+\sum_{j=1}^\ell (N_j-h_j)\,.
$$
\end{corollary}
\begin{proof}
We have $\det(A(\partial)M(\partial))=\det(A(\partial))\det(M(\partial))\neq0$.
Moreover,
by Proposition \ref{110404:propb}, we have
$\deg_\xi\det(M(\partial))=\sum_j(N_j-h_j)$,
while, by Theorem \ref{110404:thm} (b)(iii), we have
$\deg_\xi\det(A(\partial))=\dim_{\mc C}(\ker A(\partial))$.
Therefore
$$
\deg_\xi\det(A(\partial)M(\partial))
=
\dim_{\mc C}(\ker A(\partial))+\sum_j(N_j-h_j)\,.
$$
The statement follows from Theorem \ref{110404:thm}(b) applied to the matrix differential operator $A(\partial)M(\partial)$.
\end{proof}

\subsection{Main results}\label{app:4}

\subsubsection{The scalar case}

\begin{theorem}\label{110105:thm1}
Let $\mc F$ be a linearly closed differential field,
and let $K(\partial)\in\mc F[\partial]$ be a non zero scalar differential operator.
For every skewadjoint differential operator $S(\partial)$,
there exists a differential operator $P(\partial)$ such that
\begin{equation}\label{110105:eq1}
K(\partial)\circ P(\partial)-P^*(\partial)\circ K^*(\partial)=S(\partial)\,.
\end{equation}
\end{theorem}
\begin{proof}
Let $K(\partial)$ be of order $N$ with leading coefficient $k_N\neq0$.
Note that, replacing $P(\partial)$ by $k_N P(\partial)$,
we can reduce to the case when $k_N=1$.

If $S(\partial)$ has order $n\geq N$ (clearly $n$ must be odd)
with leading coefficient $a\in\mc F$,
letting $P(\partial)=\frac{a}{2}\partial^{n-N}$, we have that
$S(\partial)-K(\partial)\circ P(\partial)+P^*(\partial)\circ K^*(\partial)$
is a skewadjoint differential operator of order
strictly less than $n$.
Hence, repeating the same argument a finite number of times,
we reduce to the case when $S(\partial)$ has order $n\leq N-1$.
In particular, for $N=1$ there is nothing to prove since $S$ is skewadjoint,
hence zero.

In fact, we will consider the case when $\ord(S)=n\leq 2N-3$,
which, for $N\geq2$, covers all possibilities.
We will prove that in this case we can find $P(\partial)$
solving \eqref{110105:eq1} of order less than or equal to $N-2$.
By Corollary \ref{101218:cor1} and Lemma \ref{101218:lem3},
the operators $K(\partial),\,P(\partial)$ and $S(\partial)$ can be written, uniquely,
in the forms
\begin{equation}\label{110108:eq1}
\begin{array}{l}
\displaystyle{
K(\partial)=\sum_{n=0}^N\partial^n\circ k_n
\,\,,\,\,\,\,
P(\partial)=\sum_{n=0}^{N-2}u_n\partial^n
\,,}\\
\displaystyle{
S(\partial)=S_+(\partial)-S_-(\partial)
\,\,,\,\,\,
S_+(\partial)=
\sum_{m=0}^{N-2}\partial^{m+1}\circ s_m\partial^m=S_-^*(\partial)\,.
}
\end{array}
\end{equation}
Clearly, equation \eqref{110105:eq1}
is equivalent to say that $K(\partial)\circ P(\partial)$ and $S_+(\partial)$ differ by a selfadjoint operator.
By Lemma \ref{101218:lem1} $K(\partial)\circ P(\partial)$ is, up to adding a selfadjoint operator,
equal to
$$
\sum_{p=0}^N\sum_{q=0}^{N-2}
\sum_{m=\min(p,q)}^{\left[\frac{p+q-1}{2}\right]}
\gamma^{p,q}_m\partial^{m+1}\circ (k_pu_q)^{(p+q-2m-1)}\partial^m\,,
$$
where $\gamma^{p,q}_m$ are integers and $\gamma^{p,q}_m=\binom{p-m-1}{m-q}$
for $p>q$.
Exchanging the order of summation,
the above expression can be written in the form
\begin{equation}\label{110105:eq3}
\sum_{m=0}^{N-2}\sum_{p,q\in\mc D_{N,m}}
\gamma^{p,q}_m\partial^{m+1}\circ (k_pu_q)^{(p+q-2m-1)}\partial^m\,,
\end{equation}
where
\begin{equation}\label{110105:eq2}
\mc D_{N,m}
=\Big\{
p,q\in\mb Z_+\,\Big|\,
p\leq N,\,
q\leq N-2,\,
\min(p,q)\leq m,\,
p+q\geq 2m+1\Big\}\,.
\end{equation}
Comparing \eqref{110105:eq3} with the expression \eqref{110108:eq1} for $S_+(\partial)$,
we conclude that equation \eqref{110105:eq1}
is equivalent to the following system of $N-1$ linear differential equations
in the $N-1$ variables $u_i,\,i=0,\dots,N-2$:
\begin{equation}\label{110105:eq4}
\sum_{p,q\in\mc D_{N,m}}
\gamma^{p,q}_m (k_pu_q)^{(p+q-2m-1)}
=s_m\,,
\end{equation}
for $m=0,\dots,N-2$.

The system \eqref{110105:eq4}
is of the form $M(\partial)u=s$,
where
$u=(u_q)_{q=0}^{N-2},\,s=(s_m)_{m=0}^{N-2}$
and $M=(L_{mq}(\partial))$ is the matrix differential operator with entries
$$
L_{mq}(\partial)=\sum_{p\,:\,(p,q)\in\mc D_{N,m}}
\gamma^{p,q}_m\, \partial^{p+q-2m-1}\circ k_p
\,\,,\,\,\,\,
0\leq m,q\leq N-2
\,.
$$
Note that $L_{mq}(\partial)$ has order less than or equal to $N_q-h_m$, where
$N_q=N+q-1$ and $h_m=2m$.
The leading matrix associated to this majorant, defined by \eqref{101221:eq1b}, has
$\bar M(\xi)=\big(\bar M_{mq}\xi^{N+q-2m-1}\big)_{m,q=0}^{N-2}$, where
$$
\bar M_{mq}=
\left\{\begin{array}{ll}
0 & \text{ if } 0\leq m<q\leq N-2 \\
\binom{N-m-1}{m-q}
& \text{ if } 0\leq q\leq m\leq N-2
\end{array}\right.\,.
$$
In particular $\bar M(1)$ is upper triangular with 1's on the diagonal.
Hence, by Corollary \ref{101218:thm2} we conclude that
the system \eqref{110105:eq4} has solutions.
\end{proof}
\begin{theorem}\label{110108:thm}
Let $K(\partial)\in\mc F[\partial]$ be a scalar differential operator
of order $N$ over a differential field $\mc F$.
Then the set of differential operators $P(\partial)$ of order at most $N-1$
such that $K(\partial)\circ P(\partial)$ is selfadjoint
is a vector space over $\mc C$ of dimension
less than or equal to $\binom{N}{2}$
and, if $\mc F$ is linearly closed, it has dimension equal to $\binom{N}{2}$.
\end{theorem}
\begin{proof}
First we note that,
since $K(\partial)\circ P(\partial)$ is selfadjoint of order less than or equal to $2N-1$,
it must have order at most $2N-2$, i.e. $P(\partial)$ has order at most $N-2$.
The condition on $P(\partial)$ means that $P(\partial)$ is a solution of \eqref{110105:eq1}
with $S(\partial)=0$.
Hence, if we expand $K(\partial)$ and $P(\partial)$ as in \eqref{110108:eq1},
the condition on $P(\partial)$ reduces to the system of differential equations
\eqref{110105:eq4} with $s_m=0$.
As observed in the proof of Theorem \ref{110105:thm1},
the matrix $M(\partial)$ of coefficients has majorant
$\{N_j=N+j-1;h_i=2i\}$
and the corresponding leading matrix $\bar M(\xi)$ is non degenerate.
Hence, by Corollary \ref{101218:thm2},
the space of solutions has dimension at most
$$
d=\sum_{i=0}^{N-2}(N_i-h_i)
=\sum_{i=0}^{N-2}(N-1-i)=\binom N2\,,
$$
and equal to $d$ if $\mc F$ is linearly closed.
\end{proof}

\subsubsection{The matrix case}

\begin{theorem}\label{110105:thm1b}
Let $K(\partial)\in\Mat_{\ell\times\ell}(\mc F[\partial])$
be an $\ell\times\ell$ matrix differential operator
of order $N$ with invertible leading coefficient, over a differential field $\mc F$.
\begin{enumerate}[(a)]
\item
If $\mc F$ is linearly closed,
then for every skewadjoint $\ell\times\ell$ matrix differential operator $S(\partial)$,
there exists an $\ell\times\ell$ matrix differential operator $P(\partial)$ such that
\begin{equation}\label{110105:eq1b}
K(\partial)\circ P(\partial)-P^*(\partial)\circ K^*(\partial)=S(\partial)\,.
\end{equation}
\item
The set of differential operators $P(\partial)$ of order at most $N-1$
such that $K(\partial)\circ P(\partial)$ is selfadjoint
is a vector space over $\mc C$ of dimension
less than or equal to $d=\binom{N\ell}{2}$,
and equal to $d$ provided that $\mc F$ is linearly closed.
\end{enumerate}
\end{theorem}
\begin{proof}
We follow the same steps as in the proof of Theorems \ref{110105:thm1} and \ref{110108:thm}.
Let $K_N\in\Mat_{\ell\times\ell}(\mc F)$ be the leading coefficient of $K(\partial)$.
Replacing $K(\partial)$ by $K(\partial)\circ K_N^{-1}$
and $P(\partial)$ by $K_N P(\partial)$,
we can reduce to the case when $K(\partial)$ has leading coefficient $K_N=\id$.

Let $S(\partial)$ be of order $n$, with leading coefficient $S_n\in\Mat_{\ell\times\ell}(\mc F)$.
Since, by assumption, $S(\partial)$ is skewadjoint,
we have $S_n^T=(-1)^{n+1}S_n$.
If $n\geq N$,
letting $P(\partial)=\frac{1}{2}S_n\partial^{n-N}+P_1(\partial)$,
the equation for $P(\partial)$ becomes
$$
\begin{array}{l}
\displaystyle{
K(\partial)\circ P_1(\partial)-P_1^*(\partial)\circ K^*(\partial)
} \\
\displaystyle{
= S(\partial)-\frac{1}{2} K(\partial)\circ S_n\partial^{n-N}
+(-1)^{N}\frac{1}{2} \partial^{n-N}\circ S_n K^*(\partial)\,.
}
\end{array}
$$
Note that the RHS of the above equation
is a skewadjoint $\ell\times\ell$ matrix differential operator
of order strictly less than $n$.
Hence, repeating the same argument a finite number of times,
we reduce to the case when $S(\partial)$ has order $n\leq N-1$.

In fact, we will consider the more general case when $\ord(S)=n\leq 2N-1$.
We will prove that, in this case, we can find $P(\partial)$
solving \eqref{110105:eq1b} of order less than or equal to $N-1$.
By Corollary \ref{101218:cor1} and Lemma \ref{101218:lem3},
the operators $K(\partial),\,P(\partial)$ and $S(\partial)$ can be written, uniquely,
in the forms
\begin{equation}\label{110108:eq1b}
\begin{array}{l}
\displaystyle{
K(\partial)=\sum_{n=0}^N\partial^n\circ K_n
\,\,,\,\,\,\,
P(\partial)=\sum_{n=0}^{N-1}U_n\partial^n
\,,}\\
\displaystyle{
S(\partial)=
\sum_{m=0}^{N-1}\partial^m\circ
\big(\partial\circ A_m+A_m\partial+B_m\big)
\partial^m
\,,
}
\end{array}
\end{equation}
where $K_n,\,U_n,\,A_n,\,B_n\,\in\Mat_{\ell\times\ell}(\mc F)$
and $A_m^T=A_m,\,B_m^T=-B_m$.
By Lemma \ref{101218:lem1}  we have
\begin{equation}\label{110108:eq2}
\begin{array}{l}
\displaystyle{
K(\partial)\circ P(\partial)
=
\sum_{p=0}^N\sum_{q=0}^{N-1}
\partial^p\circ(K_pU_q)\partial^q
}\\
\displaystyle{
= \sum_{p=0}^N\sum_{q=0}^{N-1}
\sum_{m=0}^{N-1}
\gamma^{p,q}_m\partial^{m+1}\circ (K_pU_q)^{(p+q-2m-1)}\partial^m
}\\
\displaystyle{
+ \sum_{p=0}^N\sum_{q=0}^{N-1}
\sum_{m=0}^{N-1}
\delta^{p,q}_m\partial^m\circ (K_pU_q)^{(p+q-2m)}\partial^m
\,,
}
\end{array}
\end{equation}
where
$\gamma^{p,q}_m$ and $\delta^{p,q}_m$ are integers and
$$
\begin{array}{ll}
\gamma^{p,q}_m=0
& \text{ unless } 0\leq p\leq N,\,0\leq q\leq N-1\\
& \text{ and } \min(p,q)\leq m\leq \left[\frac{p+q-1}{2}\right]\,, \\
\delta^{p,q}_m=0
& \text{ unless } 0\leq p\leq N,\,0\leq q\leq N-1 \\
& \text{ and } \min(p,q+1)\leq m\leq \left[\frac{p+q}{2}\right]\,, \\
\gamma^{p,q}_m=\binom{p-m-1}{m-q} & \text{ if }
0\leq q<p\leq N,\, q\leq m\leq \left[\frac{p+q-1}{2}\right]\,,  \\
\delta^{p,q}_m=\binom{p-m-1}{m-q-1} & \text{ if }
0\leq q<p\leq N,\, q+1\leq m\leq \left[\frac{p+q}{2}\right]\,.
\end{array}
$$
We thus get from \eqref{110108:eq2}
\begin{equation}\label{110108:eq3}
\begin{array}{l}
\displaystyle{
K(\partial)\circ P(\partial)
-P^*(\partial)\circ K^*(\partial)
= \sum_{p=0}^N\sum_{q=0}^{N-1}
\sum_{m=0}^{N-1}
} \\
\displaystyle{
\partial^m\circ\Big(
\gamma^{p,q}_m\partial\circ (K_pU_q)^{(p+q-2m-1)}
+\gamma^{p,q}_m(U_q^T K_p^T)^{(p+q-2m-1)} \partial
}\\
\displaystyle{
+\delta^{p,q}_m (K_pU_q)^{(p+q-2m)}
-\delta^{p,q}_m (U_q^TK_p^T)^{(p+q-2m)}
\Big)
\partial^m
\,.
}
\end{array}
\end{equation}
Comparing \eqref{110108:eq3} and \eqref{110108:eq1b}
we get, from the uniqueness of the decomposition  \eqref{110108:eq4},
the following system of equations ($m=0,\dots,N-1$):
\begin{equation}\label{110108:eq5}
\begin{array}{l}
\displaystyle{
\frac12 \sum_{p=0}^N\sum_{q=0}^{N-1}
\gamma^{p,q}_m (K_pU_q+U_q^TK_p^T)^{(p+q-2m-1)} = A_m
}\\
\displaystyle{
\frac12\sum_{p=0}^N\sum_{q=0}^{N-1}
\big(\gamma^{p,q}_m+2\delta^{p,q}_m\big) (K_pU_q-U_q^TK_p^T)^{(p+q-2m)} = B_m\,
}
\end{array}
\end{equation}
This can be viewed as a system of linear differential equations
in the $\ell^2N$ entries $u_{qij}$ of the $\ell\times\ell$ matrices $U_q,\,q=0,\dots,N-1$.
The number of independent equations is also $\ell^2N$,
since both sides of the first equation are manifestly symmetric
and both sides of the second equation are manifestly skewsymmetric.

We make a change of variables $X_q=\frac12(U_q+U_q^T)=\big(x_{qij}\big)_{i,j=1}^\ell$
and $Y_q=\frac12(U_q-U_q^T)=\big(y_{qij}\big)_{i,j=1}^\ell$.
In these variables, the system \eqref{110108:eq5}  has the form
$$
\begin{array}{l}
\displaystyle{
\sum_{q,i',j'}L_{mij+;qi'j'+}(\partial)x_{qi'j'}
+\sum_{q,i',j'}L_{mij+;qi'j'-}(\partial)y_{qi'j'}=a_{mij}
}\\
\displaystyle{
\sum_{q,i',j'}L_{mij-;qi'j'+}(\partial)x_{qi'j'}
+\sum_{q,i',j'}L_{mij-;qi'j'-}(\partial)y_{qi'j'}=b_{mij}
}
\end{array}
$$
where $A_m=\big(a_{mij}\big)_{i,j=1}^\ell$ and $B_m=\big(b_{mij}\big)_{i,j=1}^\ell$,
and, for $m,q=0,\dots,N-1$ and $i,j,i',j'=1,\dots\ell$,
$$
\begin{array}{l}
\displaystyle{
L_{mij+,qi'j'+}(\partial)
=\frac12 \sum_{p=0}^N \gamma^{p,q}_m \partial^{p+q-2m-1}
\big(\delta_{j,j'}k_{pii'}+\delta_{i,i'}k_{pjj'}\big)\,,
} \\
\displaystyle{
L_{mij+,qi'j'-}(\partial)
=\frac12 \sum_{p=0}^N \gamma^{p,q}_m \partial^{p+q-2m-1}
\big(\delta_{j,j'}k_{pii'}-\delta_{i,i'}k_{pjj'}\big)\,,
} \\
\displaystyle{
L_{mij-,qi'j'+}(\partial)
=\frac12 \sum_{p=0}^N (\gamma^{p,q}_m+2\delta^{p,q}_m) \partial^{p+q-2m}
\big(\delta_{j,j'}k_{pii'}-\delta_{i,i'}k_{pjj'}\big)\,,
} \\
\displaystyle{
L_{mij-,qi'j'-}(\partial)
=\frac12 \sum_{p=0}^N (\gamma^{p,q}_m+2\delta^{p,q}_m) \partial^{p+q-2m}
\big(\delta_{j,j'}k_{pii'}+\delta_{i,i'}k_{pjj'}\big)\,,
}
\end{array}
$$
where $k_{pij}$ denotes the $(i,j)$ entry of the matrix $K_p$.
Hence,  the system \eqref{110108:eq5} is associated to the $\ell^2N\times\ell^2N$
matrix differential operator
$M(\partial)=\big(L_{mij\varepsilon;qi'j'\varepsilon'}(\partial)\big)$
with rows and columns indexed by quadruples $(m,i,j,\varepsilon)$,
where $m=0,\dots,N-1,\,\varepsilon=\pm,\,i,j=1,\dots,\ell$ with $i\leq j$ if $\varepsilon=+$
and with $i<j$ if $\varepsilon=-$.
In particular, $\ord(L_{mij\varepsilon;qi'j'\varepsilon'}(\partial))\leq N+q-2m-\frac{1+\varepsilon1}{2}
= N_{qi'j'\varepsilon'}-h_{mij\varepsilon}$,
where
\begin{equation}\label{110112:eq1}
N_{qij\varepsilon}=N+q
\,\,,\,\,\,\,
h_{mij\varepsilon}=2m+\frac{1+\varepsilon1}{2}\,.
\end{equation}

In order to apply Corollary \ref{101218:thm2},
we need to check that the leading matrix $\bar M(\xi)$ (defined by \eqref{101221:eq1b})
associated to the majorant \eqref{110112:eq1}
is non degenerate
(or equivalently, by equation \eqref{110404:eq1}, we need to check that $\bar M(1)$ is non degenerate).
Recalling that $K_N=\id$, we have the following formulas
for the entries $l_{mij\varepsilon,qi'j'\varepsilon'}$ of the matrix $\bar M(1)$:
$$
l_{mij\varepsilon,qi'j'\varepsilon'}
=\left\{\begin{array}{ll}
\gamma^{N,q}_m\delta_{i,i'}\delta_{j,j'} & \text{ for } \varepsilon=\varepsilon'=+ \,,\\
\gamma^{N,q}_m+2\delta^{N,q}_m \delta_{i,i'}\delta_{j,j'} & \text{ for } \varepsilon=\varepsilon'=- \,,\\
0 & \text{ for } \varepsilon\neq\varepsilon' \,.
\end{array}\right.
$$
In particular, $\bar M(1)$ is a block diagonal matrix, with upper triangular blocks,
having 1's on the diagonal, hence it is non degenerate.
By Corollary \ref{101218:thm2} we conclude that
the system \eqref{110108:eq5} always has solutions if $\mc F$ is linearly closed,
proving part (a).
Moreover, by Corollary \ref{101218:thm2},
the space of solutions of the corresponding homogeneous system
has dimension over $\mc C$ less than or equal to (and equal if $\mc F$ is linearly closed)
$$
\begin{array}{l}
\displaystyle{
d=
\sum_{m=0}^{N-1}\sum_{1\leq i\leq j\leq\ell}^\ell (N_{mij+}-h_{mij+})
+\sum_{m=0}^{N-1}\sum_{1\leq i< j\leq\ell}^\ell (N_{mij-}-h_{mij-})
}\\
\displaystyle{
=\frac{\ell(\ell+1)}{2}\sum_{m=0}^{N-1}(N-m-1)
+\frac{\ell(\ell-1)}{2}\sum_{m=0}^{N-1}(N-m)
=\binom{N\ell}{2}\,,
}
\end{array}
$$
proving part (b).
\end{proof}

\subsection{Generalization to polydifferential operators}\label{app:5}

\subsubsection{Preliminaries on polydifferential operators on $\mc F^\ell$}\label{app:5.1}

The goal of this technical section is to provide lemmas which will be used in the proof
of Corollary \ref{110314:cor2} and Theorem \ref{110127:conj1} in the next section.

For $k\in\mb Z_+$, a $k$-\emph{differential operator on} $\mc F^\ell$ is, by definition,
an array
\begin{equation}\label{110126:eq0}
\big(P_{i_0,i_1,\dots,i_k}(\lambda_1,\dots,\lambda_k)\big)_{i_0,i_1,\dots,i_k\in\{1,\dots,\ell\}}\,,
\end{equation}
consisting of $k$-differential operators on $\mc F$, i.e.
\begin{equation}\label{110126:eq2}
P_{i_0,i_1,\dots,i_k}(\lambda_1,\dots,\lambda_k)
=\sum_{n_1,\dots,n_k=0}^N p^{n_1,\dots,n_k}_{i_0,i_1,\dots,i_k}
\lambda_1^{n_1}\dots\lambda_k^{n_k}\,.
\end{equation}
The reason for this name is that to an array as in \eqref{110126:eq0}
we can associate a $k$-linear (over $\mc C$) map
$P:\,(\mc F^\ell)^{\otimes k}
=\mc F^\ell\otimes{_\mc C}\dots\otimes_{\mc C}\mc F^\ell
\to\mc F^\ell$
given by
$$
P(F^1\otimes\dots\otimes F^k)_{i_0}
=\sum_{\substack{i_1,\dots,i_k\in I \\ n_1,\dots,n_k\in\mb Z_+}}
p^{n_1,\dots,n_k}_{i_0,i_1,\dots,i_k}
(\partial^{n_1}(F^1)_{i_1})\dots(\partial^{n_k}(F^k)_{i_k})\,.
$$

The symmetric group $S_k=Perm(1,2,\dots,k)$ acts on such arrays by simultaneous
permutations of the indices $i_1,\dots,i_k$ and variables $\lambda_1,\dots,\lambda_k$:
\begin{equation}\label{110126:eq1}
\begin{array}{l}
P^\sigma_{i_0,i_1,\dots,i_k}(\lambda_1,\dots,\lambda_k)
=
P_{i_0,i_{\sigma^{-1}(1)},\dots,i_{\sigma^{-1}(k)}}
(\lambda_{\sigma^{-1}(1)},\dots,\lambda_{\sigma^{-1}(k)}) \\
=
\displaystyle{
\sum_{n_1,\dots,n_k=0}^N
p^{n_{\sigma^{-1}(1)},\dots,n_{\sigma^{-1}(k)}}_{i_0,i_{\sigma^{-1}(1)},\dots,i_{\sigma^{-1}(k)}}
\lambda_1^{n_1}\dots\lambda_k^{n_k}
\,\,,\,\,\,\,
\sigma\in S_k
}\,.
\end{array}
\end{equation}
We extend this to an action of the group $S_{k+1}=Perm(0,1,\dots,k)$,
denoted $P\mapsto P^\sigma,\,\sigma\in S_{k+1}$,
as follows. If $\tau_\alpha$ is the transposition of $0$ and $\alpha\in\{1,\dots,k\}$ and if $P$ is as in
\eqref{110126:eq2}, we let
\begin{equation}\label{110126:eq3}
P^{\tau_\alpha}_{i_0,i_1,\dots,i_k}(\lambda_1,\dots,\lambda_k)
=\!\!\!
\sum_{n_1,\dots,n_k=0}^N \!\!\!
(-\lambda_1-\dots-\lambda_k-\partial)^{n_\alpha}
\lambda_1^{n_1}\stackrel{\alpha}{\check{\dots}}\lambda_k^{n_k}
p^{n_1,\dots,n_k}_{i_\alpha,i_1,\dots ,\stackrel{\alpha}{\check{i_0}},\dots,i_k}\,.
\end{equation}
The reason for this definition is that we have the identity
\begin{equation}\label{110120:eq2}
\tint F^0\cdot P(F^1,\dots,F^k)
=
\tint F^{\sigma(0)}\cdot P^\sigma(F^{\sigma(1)},\dots,F^{\sigma(k)})\,,
\end{equation}
for every $\sigma\in S _{k+1}$,
and every $F^0,\dots,F^k\in\mc F^\ell$,
where as usual $\tint$ denotes the canonical map $\mc F\to\mc F/\partial\mc F$.
The pairing $\tint F\cdot G,\,F,G\in\mc F^\ell$,
may be degenerate,
but if we replace $\mc F$ by the algebra of differential polynomials $\mc F[u,u',\dots]$,
this pairing is always non degenerate.
Extending the map $P:\,(\mc F^\ell)^{\otimes k}\to\mc F^\ell$
to a map $(\mc F[u,u',\dots]^\ell)^{\otimes k}\to\mc F[u,u',\dots]^\ell$ in the obvious way,
we get that formula \eqref{110120:eq2} uniquely determines the action of $S_{k+1}$
on the space of $k$-differential operators on $\mc F^\ell$.

A $k$-differential operator $P$ on $\mc F^\ell$
is called \emph{skewsymmetric} (respectively \emph{totally
skewsymmetric}) if $P^\sigma=\sign(\sigma)P$ for every $\sigma\in S_k$
(resp. $\sigma\in S_{k+1}$).
Given an array $P$ as in \eqref{110126:eq0},
we define its \emph{total skewsymmetrization},
denoted $\langle P\rangle^-$,
by the following formula:
\begin{equation}\label{110120:eq5}
\langle P\rangle^-
=
\frac{1}{(k+1)!}
\sum_{\sigma\in S_{k+1}} \sign(\sigma)P^\sigma\,.
\end{equation}
Clearly, $\langle P\rangle^-=P$ if and only if $P$ is totally skewsymmetric.
Note that if $P$ is already skewsymmetric, then it total skewsymmetrization is
\begin{equation}\label{110120:eq6}
\langle P\rangle^-
=
\frac{1}{k+1}
\Big(P-
\sum_{\alpha=1}^k P^{\tau_\alpha}\Big)\,.
\end{equation}

We define a structure of a $\Mat_{\ell\times\ell}(\mc F[\partial])$-module
on the space of $k$-differential operators on $\mc F^\ell$ as follows.
For an $\ell\times\ell$ matrix differential operator
$K(\partial)=\big(K_{ij}(\partial)\big)_{i,j\in\{1,\dots,\ell\}}$
and an array $P$ as in \eqref{110126:eq0},
we let
\begin{equation}\label{110126:eq4}
(K\circ P)_{i_0,i_1,\dots,i_k}(\lambda_1,\dots,\lambda_k)
=
\sum_{j=1}^\ell
K_{i_0,j}(\lambda_1+\dots+\lambda_k+\partial)P_{j,i_1,\dots,i_k}(\lambda_1,\dots,\lambda_k)\,.
\end{equation}
Clearly, if $P$ is skewsymmetric so is $K\circ P$,
but in general it is not totally skewsymmetric, even if $P$ is totally skewsymmetric.

Note that 0-differential operators on $\mc F^\ell$ are elements of $\mc F^\ell$,
and in this case we have $\langle P\rangle^-=P$ and $K\circ P=K(\partial)P$
for every $P\in\mc F^\ell$.
The space of $1$-differential operators on $\mc F^\ell$
is identified with $\Mat_{\ell\times\ell}(\mc F[\partial])$ by letting $\lambda=\partial$.
In this case the action of $\tau=(0,1)\in S_2$, defined in \eqref{110126:eq3},
coincides with taking the adjoint matrix differential operator: $P^\tau=P^*$,
and the $\Mat_{\ell\times\ell}(\mc F[\partial])$-module structure defined in \eqref{110126:eq4}
coincides with the left multiplication in $\Mat_{\ell\times\ell}(\mc F[\partial])$.

In order to simplify notation, we let
\begin{equation}\label{110313:eq2}
\lambda_0
=
-\lambda_1-\dots-\lambda_k-\partial\,.
\end{equation}
With this notation, the action of the symmetric group $S_{k+1}$ given by \eqref{110126:eq1}
and \eqref{110126:eq3} reads
$$
P^\sigma_{i_0,i_1,\dots,i_k}(\lambda_1,\dots,\lambda_k)
=
P_{i_{\sigma^{-1}(0)},i_{\sigma^{-1}(1)},\dots,i_{\sigma^{-1}(k)}}
(\lambda_{\sigma^{-1}(1)},\dots,\lambda_{\sigma^{-1}(k)})
\,,
$$
where $\lambda_0$, when it appears, acts from the left on the coefficients of $P$.
Moreover, the $\Mat_{\ell\times\ell}(\mb F[\partial])$-module structure \eqref{110126:eq4}
becomes
$$
(K\circ P)_{i_0,i_1,\dots,i_k}(\lambda_1,\dots,\lambda_k)
=
\sum_{j=1}^\ell
P_{j,i_1,\dots,i_k}(\lambda_1,\dots,\lambda_k) K^*_{j,i_0}(\lambda_0) \,,
$$
where again $\lambda_0$ is assumed to be moved to the left.
If $P$ is a skewsymmetric $k$-differential operator on $\mc F^\ell$,
we then have, by \eqref{110120:eq6},
\begin{equation}\label{110222:eq4}
\begin{array}{l}
\langle K\circ P\rangle^-_{i_0,i_1,\dots,i_k}(\lambda_1,\dots,\lambda_k) \\
\displaystyle{
=\frac1{k+1}
\sum_{\alpha=0}^k (-1)^\alpha
\sum_{j=1}^\ell
P_{j,i_0,\stackrel{\alpha}{\check{\dots}},i_k}
(\lambda_0,\stackrel{\alpha}{\check{\dots}},\lambda_k)
K^*_{j,i_\alpha}(\lambda_\alpha)\,.
}
\end{array}
\end{equation}

In this section we will generalize the results of Sections \ref{app:1} and \ref{app:4}
to the case of $k$-differential operators on $\mc F^\ell$ for arbitrary $k\geq0$.
First, we introduce some more notation.
Given a $(k+1)$-tuple in $\mb Z_+^{k+1}$, written with Latin letters,
we use the corresponding Greek letters to denote its non decreasing reordering;
for example,  $\mu_0\geq\mu_1\geq\dots\geq\mu_k$ will denote the reordering of
$(m_0,m_1,\dots,m_k)\in\mb Z_+^{k+1}$,
while $\nu_0\geq\nu_1\geq\dots\geq\nu_k$ will denote
the reordering of $(n_0,n_1,\dots,n_k)\in\mb Z_+^{k+1}$.
\begin{lemma}\label{110311:lem0}
Let $V$ be a vector space over a field $\mc C$ and let $\partial$ be an endomorphism of $V$.
Let $v_{m_0,\dots,m_k}$ be elements of $V$, labeled by the indices
$m_0,\dots,m_k\in\{0,\dots,N\}$ satisfying $\mu_0-\mu_1=1$.
Then they satisfy
the following equation in $\mc C[\lambda_1,\dots,\lambda_k]\otimes V$,
where $\lambda_0$ is as in \eqref{110313:eq2},
\begin{equation}\label{110312:eq1}
\sum_{\substack{m_0,\dots,m_k=0 \\ (\mu_0-\mu_1=1)}}^N
\lambda_0^{m_0}\lambda_1^{m_1}\dots\lambda_k^{m_k}
v_{m_0,\dots,m_k}
=0
\,,
\end{equation}
if and only if the following two conditions hold:
\begin{enumerate}[(i)]
\item
$v_{m_0,\dots,m_k}=v_{m'_0,\dots,m'_k}$
if the $(k+1)$-tuple $(m'_0,\dots,m'_k)$ is obtained from
$(m_0,\dots,m_k)$ by a transposition of entries $m_i=\mu_0$ and $m_j=\mu_1$;
\item
given a $(k+1)$-tuple $(m_0,\dots,m_k)\in\{0,\dots,N\}^{k+1}$ satisfying $\mu_0-\mu_1=1$,
subtracting 1 from the maximal entry $m_i=\mu_0$,
we get a $(k+1)$-tuple $(n_0,\dots,n_k)\in\{0,\dots,N-1\}^{k+1}$ satisfying $\nu_0=\nu_1$,
and we denote $w_{n_0,\dots,n_k}=v_{m_0,\dots,m_k}$
(note that, by condition (i), $w_{n_0,\dots,n_k}$ is well defined);
these elements should satisfy the following system of equations:
\begin{equation}\label{110523:eq1}
\partial w_{n_0,\dots,n_k}
+\sum_{h=0}^k w_{n_0,\dots,n_h-1,\dots,n_k}=0\,,
\end{equation}
where the summand $w_{n_0,\dots,n_h-1,\dots,n_k}$ in the LHS is considered to be zero
unless $n_h\geq1$ and the maximal two of the indices
$n_0,\dots,n_h-1,\dots,n_k$ are equal.
\end{enumerate}
In particular, if $\partial$ is injective, then equation \eqref{110312:eq1}
has only the zero solution.
\end{lemma}
\begin{proof}
First, we prove that if elements $v_{m_0,\dots,m_k}\!\in V$, labeled by
$m_0,\dots,m_k\!\in$ $\{0,\dots,N\}$ such that $\mu_0-\mu_1=1$,
satisfy conditions (i) and (ii) of the lemma,
then equation \eqref{110312:eq1} holds.
By condition (i), the LHS of \eqref{110312:eq1} is equal to
\begin{equation}\label{110524:eq1}
\sum_{\substack{n_0,\dots,n_k=0 \\ (\nu_0=\nu_1)}}^{N-1}
\sum_{i\,|\,n_i=\nu_0}\lambda_0^{n_0}\dots\lambda_i^{n_i+1}\dots\lambda_k^{n_k}
w_{n_0,\dots,n_k}\,.
\end{equation}
Since $\lambda_0+\dots+\lambda_k+\partial=0$, \eqref{110524:eq1}
can be rewritten as
$$
\begin{array}{c}
\displaystyle{
-\sum_{\substack{n_0,\dots,n_k=0 \\ (\nu_0=\nu_1)}}^{N-1}
\sum_{j\,|\,n_j<\nu_0}\lambda_0^{n_0}\dots\lambda_j^{n_j+1}\dots\lambda_k^{n_k}
w_{n_0,\dots,n_k}
} \\
\displaystyle{
-\sum_{\substack{n_0,\dots,n_k=0 \\ (\nu_0=\nu_1)}}^{N-1}
\lambda_0^{n_0}\dots\lambda_k^{n_k}
\partial w_{n_0,\dots,n_k}\,.
}
\end{array}
$$
Renaming $n_j+1$ by $n_j$, the above expression becomes
\begin{equation}\label{110524:eq2}
-\sum_{\substack{n_0,\dots,n_k=0 \\ (\nu_0=\nu_1)}}^{N-1}
\lambda_0^{n_0}\dots\lambda_k^{n_k}
\Big(\sum_{j=0}^k w_{n_0,\dots,n_j-1,\dots,n_k}+\partial w_{n_0,\dots,n_k}\Big)\,,
\end{equation}
which is zero by equation \eqref{110523:eq1}.

Next, assuming that equation \eqref{110312:eq1} holds, we prove, by induction on $N$,
that conditions (i) and (ii) necessarily hold.
Recalling \eqref{110313:eq2}, the coefficient of $\lambda_1^{2N-1}$
in the LHS of \eqref{110312:eq1} is
$$
(-1)^N\sum_{m_2,\dots,m_k=0}^{N-1}\lambda_2^{m_2}\dots\lambda_k^{m_k}
\big(v_{N,N-1,m_2,\dots,m_k}-v_{N-1,N,m_2,\dots,m_k}\big)\,,
$$
so that, necessarily, $v_{N,N-1,m_2,\dots,m_k}=v_{N-1,N,m_2,\dots,m_k}$
for all $m_2,\dots,m_k\in\{0,\dots,N-1\}$.
More generally, if we replace $\lambda_i$
by $-\lambda_0-\stackrel{i}{\check{\dots}}-\lambda_k-\partial$
and we consider the LHS of \eqref{110312:eq1} as a polynomial in the variables
$\lambda_0,\stackrel{i}{\check{\dots}},\lambda_k$,
we conclude, by looking at the coefficient of $\lambda_j^{2N-1}$ in equation \eqref{110312:eq1},
that
$$
v_{m_0,\dots,\stackrel{i}{\check{N}},\dots,\stackrel{j}{\check{N-1}},\dots,m_k}
=
v_{m_0,\dots,\stackrel{i}{\check{N-1}},\dots,\stackrel{j}{\check{N}},\dots,m_k}\,,
$$
i.e. condition (i) holds for all the elements $v_{m_0,\dots,m_k}$ with $\mu_0=N$,
and we can introduce the notation $w_{n_0,\dots,n_k}$
for indices $n_0,\dots,n_k\in\mb Z_+$ satisfying $\nu_0=\nu_1=N-1$.
Using this fact, by the same computation that we did in the first part of the proof
to derive formula \eqref{110524:eq2},
we get that equation \eqref{110312:eq1} is equivalent to
\begin{equation}\label{110524:eq3}
-\sum_{\substack{n_0,\dots,n_k\in\mb Z_+ \\ (\nu_0=\nu_1=N-1)}}
\lambda_0^{n_0}\dots\lambda_k^{n_k}
x_{n_0,\dots,n_k}
+ \sum_{\substack{m_0,\dots,m_k=0 \\ (\mu_0-\mu_1=1)}}^{N-1}
\lambda_0^{m_0}\lambda_1^{m_1}\dots\lambda_k^{m_k}
v_{m_0,\dots,m_k}
=0 \,,
\end{equation}
where
$$
x_{n_0,\dots,n_k}
=
\sum_{j=0}^k w_{n_0,\dots,n_j-1,\dots,n_k}+\partial w_{n_0,\dots,n_k}\,.
$$
The coefficient of $\lambda_1^{2N-2}\lambda_2^{n_2}\!\dots\lambda_k^{n_k}$
in the LHS of \eqref{110524:eq3} is $x_{N-1,N-1,n_2,\dots,n_k}$,
and, in general,
if we replace $\lambda_i$
by $-\lambda_0-\stackrel{i}{\check{\dots}}-\lambda_k-\partial$,
the coefficient of $\lambda_j^{2N-2}\lambda_0^{n_0}\stackrel{i,j}{\check{\dots}}\lambda_k^{n_k}$
in LHS of \eqref{110524:eq3}
is $x_{n_0,\dots,\stackrel{i}{\check{N-1}},\dots,\stackrel{j}{\check{N-1}},\dots,n_k}$.
Hence, equation \eqref{110524:eq3}
implies that $x_{n_0,\dots,n_k}=0$ for all indices $n_0,\dots,n_k$.
In other words, the elements $v_{m_0,\dots,m_k}$ satisfy condition (ii) when $\mu_0=N$.
Finally, equation \eqref{110524:eq3} allows us to replace $N$ by $N-1$
in equation \eqref{110312:eq1},
so that the claim follows by the inductive assumption.

For the last statement of the lemma,
if $\partial$ is injective equation \eqref{110523:eq1} implies that $w_{n_0,\dots,n_k}=0$,
by induction on $\sum_in_i$.
\end{proof}
\begin{remark}
The system of equations \eqref{110523:eq1} has the form
\begin{equation}\label{110524:eq4}
\partial Y=AY
\end{equation}
where $Y\in V^d$, $\partial\in\End(V)$ and $A\in\Mat_{d\times d}(\mc C)$ is a nilpotent matrix.
In this case, the dimension of the space of solutions of \eqref{110524:eq4}
is equal to
$\sum_i\dim(\ker \partial^{d_i+1})$, where $d_i$ are the sizes
of the Jordan blocks of $A$.
Indeed, if $A$ is a Jordan block of size $d$,
it is immediate to check that the space of solutions
of \eqref{110524:eq4} has dimension equal to $\dim(\ker \partial^{d+1})$,
and the general case reduces to this one.
\end{remark}
\begin{lemma}\label{110316:lem1}
For $n_0,\dots,n_k\in\mb Z_+$, and $m_0,\dots,m_k\in\mb Z_+$ such that $\mu_0-\mu_1=0$ or $1$,
we define the integers $b^{n_0,\dots,n_k}_{m_0,\dots,m_k}$ recursively by
the following formulas:
\begin{enumerate}[(i)]
\item
if $\nu_0-\nu_1=0$ or $1$, we let $b^{n_0,\dots,n_k}_{m_0,\dots,m_k}=\delta_{m_0,n_0}\dots\delta_{m_k,n_k}$,
\item
if $n_\alpha=\nu_0\geq\nu_1+2$, we let
\begin{equation}\label{110316:eq2}
b^{n_0,\dots,n_k}_{m_0,\dots,m_k}
=
-\sum_{\beta\neq\alpha}
b^{n_0,\dots n_\beta+1,\dots,n_\alpha-1,\dots,n_k}_{m_0,\dots,m_k}
-b^{n_0,\dots,n_\alpha-1,\dots,n_k}_{m_0,\dots,m_k}\,.
\end{equation}
\end{enumerate}
Then the following identity holds
in $\mb F[\lambda_1,\dots,\lambda_k,\partial]$, for every $n_0,\dots,n_k\in\mb Z_+$:
\begin{equation}\label{110316:eq1}
\lambda_0^{n_0}\lambda_1^{n_1}\dots\lambda_k^{n_k}
=
\sum_{\substack{m_0,\dots,m_k\in\mb Z_+ \\ (\mu_0-\mu_1=0 \text{ or } 1)}}
b^{n_0,\dots,n_k}_{m_0,\dots,m_k}
\lambda_0^{m_0}\lambda_1^{m_1}\dots\lambda_k^{m_k}
\partial^{\sum_i(n_i-m_i)}\,,
\end{equation}
Furthermore, if $n_\alpha=\nu_0$, then the coefficient $b^{n_0,\dots,n_k}_{m_0,\dots,m_k}$ is zero
unless $m_\alpha=\mu_0\leq\nu_0$,
$m_\beta\geq n_\beta$ for every $\beta\neq\alpha$,
and $\sum_i(n_i-m_i)\geq0$.
In particular, if $\mu_0=\nu_0$, then $b^{n_0,\dots,n_k}_{m_0,\dots,m_k}$ is zero
unless $m_i=n_i$ for every $i=0,\dots,k$.
\end{lemma}
\begin{proof}
If $\nu_0-\nu_1=0$ or $1$, then obviously \eqref{110316:eq1} holds
for $b^{n_0,\dots,n_k}_{m_0,\dots,m_k}=\delta_{m_0,n_0}\delta_{m_1,n_1}\dots\delta_{m_k,n_k}$.
If $n_\alpha=\nu_0\geq\nu_1+2$,
substituting
$\lambda_\alpha=-\lambda_0-\stackrel{\alpha}{\check{\dots}}-\lambda_k-\partial$,
we get,
\begin{equation}\label{110316:eq3}
\lambda_0^{n_0}\lambda_1^{n_1}\dots\lambda_k^{n_k}
=
-\sum_{\beta\neq\alpha}
\lambda_0^{n_0}\!\!\!\dots\!\lambda_\beta^{n_\beta+1}\!\!\!\dots\!\lambda_\alpha^{n_\alpha-1}\!\!\!\dots\!\lambda_k^{n_k}
-\lambda_0^{n_0}\!\!\!\dots\!\lambda_\alpha^{n_\alpha-1}\!\!\!\dots\!\lambda_k^{n_k}\partial \,.
\end{equation}
By induction on $\nu_0-\nu_1$, the RHS of \eqref{110316:eq3} is
$$
\begin{array}{l}
\displaystyle{
\sum_{\substack{m_0,\dots,m_k\in\mb Z_+ \\ (\mu_0-\mu_1=0 \text{ or } 1)}}
\Big(
-\sum_{\beta\neq\alpha}
b^{n_0,\dots,n_\beta+1,\dots,n_\alpha-1,\dots,n_k}_{m_0,\dots,m_k}
-b^{n_0,\dots,n_\alpha-1,\dots,n_k}_{m_0,\dots,m_k}
\Big)
} \\
\displaystyle{
\times \lambda_0^{m_0}\dots\lambda_k^{m_k}
\partial^{\sum_i(n_i-m_i)}
=\!\!\!
\sum_{\substack{m_0,\dots,m_k\in\mb Z_+ \\ (\mu_0-\mu_1=0 \text{ or } 1)}}
\!\!\!
b^{n_0,\dots,n_k}_{m_0,\dots,m_k}
\lambda_0^{m_0}\dots\lambda_k^{m_k}
\partial^{\sum_i(n_i-m_i)}\,,
}
\end{array}
$$
proving \eqref{110316:eq1}.
The last statement of the lemma is obvious for $\nu_0-\nu_1=0$ or $1$,
while, for $n_\alpha=\nu_0\geq\nu_1+2$, it follows
by the recursive formula \eqref{110316:eq2} and an easy induction on $\nu_0-\nu_1$.
\end{proof}
\begin{lemma}\label{110311:lem1}
\begin{enumerate}[(a)]
\item
There exist unique numbers $c^{n_0,\dots,n_k}_{m_0,\dots,m_k}$,
for $n_0,\dots,n_k\in\mb Z_+$, and $m_0,\dots,m_k\in\mb Z_+$ satisfying $\mu_0-\mu_1=1$,
such that the following identity holds
in $\mb F[\lambda_1,\dots,\lambda_k,\partial^{\pm1}]$:
\begin{equation}\label{110311:eq1}
\lambda_0^{n_0}\lambda_1^{n_1}\dots\lambda_k^{n_k}
=
\sum_{\substack{m_0,\dots,m_k\in\mb Z_+ \\ (\mu_0-\mu_1=1)}}
c^{n_0,\dots,n_k}_{m_0,\dots,m_k}
\lambda_0^{m_0}\lambda_1^{m_1}\dots\lambda_k^{m_k}
\partial^{\sum_i(n_i-m_i)}\,,
\end{equation}
for every $n_0,\dots,n_k\in\mb Z_+$,
where $\lambda_0$ is as in \eqref{110313:eq2}.
\item
The coefficients $c^{n_0,\dots,n_k}_{m_0,\dots,m_k}$ in part (a)
are integers satisfying the following properties:
\begin{enumerate}[(i)]
\item
if $\nu_0=\nu_1+1$, then
$c^{n_0,\dots,n_k}_{m_0,\dots,m_k}=\delta_{m_0,n_0}\delta_{m_1,n_1}\dots\delta_{m_k,n_k}$;
\item
$c^{n_{\sigma(0)},\dots,n_{\sigma(k)}}_{m_{\sigma(0)},\dots,m_{\sigma(k)}}
=
c^{n_0,\dots,n_k}_{m_0,\dots,m_k}$
for all $\sigma\in Perm(0,1,\dots,k)$;
\item
for every $n_0,\dots,n_k,m_0,\dots,m_k\in\mb Z_+$
such that $\mu_0-\mu_1=1$, we have the following recurrent formulas
\begin{equation}\label{110313:eq5}
c^{n_0,\dots,n_k}_{m_0,\dots,m_k}
=
-\sum_{\alpha=0}^k
c^{n_0,\dots,n_\alpha+1,\dots,n_k}_{m_0,\dots,m_k}
\,,
\end{equation}
and, for every $\alpha=0,\dots,k$,
\begin{equation}\label{110313:eq6}
c^{n_0,\dots,n_k}_{m_0,\dots,m_k}
=
-\sum_{\beta\neq\alpha}
c^{n_0,\dots,n_\beta+1,\dots,n_\alpha-1,\dots,n_k}_{m_0,\dots,m_k}
-c^{n_0,\dots,n_\alpha-1,\dots,n_k}_{m_0,\dots,m_k}\,.
\end{equation}
\item
$c^{n_0,\dots,n_k}_{m_0,\dots,m_k}$ is zero unless $\mu_0\,\big(=\mu_1+1\big)\,\leq\nu_0+1$;
\item
if $\nu_0>\nu_1$,
then $c^{n_0,\dots,n_k}_{m_0,\dots,m_k}$ is zero unless $\mu_0\leq\nu_0$;
\item
if $n_\alpha=\nu_0$,
then $c^{n_0,\dots,n_k}_{m_0,\dots,m_k}$ is zero unless $m_\alpha\geq \max(\mu_1,\nu_1)$;
\item
if $n_\beta\leq\nu_1$,
then $c^{n_0,\dots,n_k}_{m_0,\dots,m_k}$ is zero unless $m_\beta\geq n_\beta$.
\end{enumerate}
\end{enumerate}
\end{lemma}
\begin{proof}
The uniqueness of the decomposition \eqref{110311:eq1} follows from Lemma \ref{110311:lem0}
in the case $V=\mb F[\partial^{\pm1}]$, with $\partial$ acting by left multiplication,
and we want to prove the existence.
By Lemma \ref{110316:lem1} the monomial
$\lambda_0^{n_0}\dots\lambda_k^{n_k}$ is a linear combination
over $\mb F[\partial]$ of monomials $\lambda_0^{m_0}\dots\lambda_k^{m_k}$
with $\mu_0-\mu_1=0$ or $1$.
Hence, we are left to consider the monomials with $\nu_0=\nu_1$.
In this case,
multiplying the monomial $\lambda_0^{n_0}\lambda_1^{n_1}\dots\lambda_k^{n_k}$ by
$$
1=-\lambda_0\partial^{-1}-\lambda_1\partial^{-1}-\dots-\lambda_k\partial^{-1}\,,
$$
we get
\begin{equation}\label{110313:eq3}
\lambda_0^{n_0}\lambda_1^{n_1}\dots\lambda_k^{n_k}
=
-\sum_{\alpha=0}^k
\lambda_0^{n_0}\dots\lambda_\alpha^{n_\alpha+1}\dots\lambda_k^{n_k} \partial^{-1}
\,.
\end{equation}
All the monomials which appear in the RHS have the difference between maximal and second maximal
upper indices equal to 0 or 1.
Hence, \eqref{110311:eq1} holds by induction on $\sum_i(\nu_0-n_i)$.

We next prove part (b). Property (i) is clear.
Given a permutation $\sigma\in S_{k+1}$ and $n_0,\dots,n_k\in\mb Z_+$,
we have by part (a) (after changing the indices of summation),
$$
\lambda_0^{n_{\sigma(0)}}\dots \lambda_k^{n_{\sigma(k)}}
=
\sum_{\substack{m_0,\dots,m_k\in\mb Z_+ \\ (\mu_0-\mu_1=1)}}
c^{n_{\sigma(0)},\dots,n_{\sigma(k)}}_{m_{\sigma(0)},\dots,m_{\sigma(k)}}
\lambda_0^{m_{\sigma(0)}}\dots\lambda_k^{m_{\sigma(k)}}
\partial^{\sum_i(n_i-m_i)}\,.
$$
On the other hand, by permuting the variables $\lambda_0,\dots,\lambda_k$,
since the condition $\lambda_0+\dots+\lambda_k+\partial=0$
is invariant, part (a) says that we have a unique decomposition
$$
\lambda_{\sigma^{-1}(0)}^{n_0}\dots \lambda_{\sigma^{-1}(k)}^{n_k}
=
\sum_{\substack{m_0,\dots,m_k\in\mb Z_+ \\ (\mu_0-\mu_1=1)}}
c^{n_0,\dots,n_k}_{m_0,\dots,m_k}
\lambda_{\sigma^{-1}(0)}^{m_0}\dots\lambda_{\sigma^{-1}(k)}^{m_k}
\partial^{\sum_i(n_i-m_i)}\,.
$$
Comparing the above two identities we conclude, by the uniqueness of the decomposition \eqref{110311:eq1},
that
$c^{n_{\sigma(0)},\dots,n_{\sigma(k)}}_{m_{\sigma(0)},\dots,m_{\sigma(k)}}=c^{n_0,\dots,n_k}_{m_0,\dots,m_k}$,
proving property (ii).

The two identities \eqref{110313:eq5} and \eqref{110313:eq6} follow immediately by part (a)
and the equations \eqref{110313:eq3} and \eqref{110316:eq3} respectively.

Next, we prove properties (iv)--(vii).
Let $n_0,\dots,n_k\in\mb Z_+$.
If $\nu_0=\nu_1+1$,
then by (i) we have $c^{n_0,\dots,n_k}_{m_0,\dots,m_k}=\delta_{m_0,n_0}\dots\delta_{m_k,n_k}$,
hence this coefficient is zero unless $m_0=n_0,\dots m_k=n_k$.
Properties (iv)--(vii), in this case, trivially hold.

Suppose next that $n_0,\dots,n_k\in\mb Z_+$ are such that $\nu_0=\nu_1$.
We prove, by induction on $\sum_i(\nu_0-n_i)$,
that properties (iv)---(vii) hold, i.e.
$c^{n_0,\dots,n_k}_{m_0,\dots,m_k}$ is zero unless, respectively:
\begin{enumerate}
\item[(iv)] $\mu_0\leq\nu_0+1$,
\item[(vi)] $m_\alpha\geq\mu_1$ and $\nu_1$, for $\alpha$ such that $n_\alpha=\nu_0$,
\item[(vii)] $m_\beta\geq n_\beta$ for all $\beta=0,\dots,k$.
\end{enumerate}
By equation \eqref{110313:eq5} we have
\begin{equation}\label{110313:eq7}
c^{n_0,\dots,n_k}_{m_0,\dots,m_k}
=
-\sum_{\beta\,|\,n_\beta=\nu_0}
\delta_{m_0,n_0}\dots\delta_{m_\beta,n_\beta+1}\dots\delta_{m_k,n_k}
-\!\!\!
\sum_{\gamma\,|\,n_\gamma\leq\nu_0-1}
c^{n_0,\dots,n_\gamma+1,\dots,n_k}_{m_0,\dots,m_k}\,.
\end{equation}
Given $\beta$ such that $n_\beta=\nu_0$, the corresponding term in the first sum of the RHS
is zero unless $m_0=n_0,\dots,m_\beta=n_\beta+1,\dots,m_k=n_k$.
In particular one easily checks that it is zero unless
all conditions (iv), (vi) and (vii) above hold.
Next, let $\gamma$ be such that $n_\gamma\leq\nu_0-1$,
and consider the corresponding term in the second summand of the RHS of \eqref{110313:eq7}.
It has maximal and second maximal upper indices both equal to $\nu_0$.
Hence, we can apply the inductive assumption to deduce that
it is zero unless all conditions (iv), (vi), and (vii) hold.

Finally, suppose that $n_\alpha=\nu_0\geq\nu_1+2$.
We prove by induction on $\nu_0-\nu_1$
that properties (iv)--(vii) hold,
i.e. $c^{n_0,\dots,n_k}_{m_0,\dots,m_k}$ is zero unless
\begin{enumerate}
\item[(v)] $\mu_0\leq\nu_0$,
\item[(vi)] $m_\alpha\geq\mu_1$ and $\nu_1$,
\item[(vii)] $m_\beta\geq n_\beta$ for all $\beta\neq\alpha$,
\end{enumerate}
Consider equation \eqref{110313:eq6}.
In all terms in the RHS the maximal upper index is $n_\alpha-1=\nu_0-1$,
while the second maximal upper index is either $\nu_1$ or $\nu_1+1$.
Hence, we can use the results in the previous case and the inductive assumption
to deduce that all terms in the RHS of \eqref{110313:eq6} are zero unless
all conditions (v), (vi) and (vii) above are met.
\end{proof}
\begin{remark}
In the special case $k=1$, Lemma \ref{110311:lem1} follows from Lemma \ref{101218:lem1}.
In fact, in this case we can use Lemma \ref{101218:lem1} to compute explicitly
the coefficients $c^{p,q}_{m+1,m}$ and $c^{p,q}_{m,m+1}$ defined in Lemma \ref{110311:lem1}.
For $p=q$ we have $c^{p,p}_{m+1,m}=c^{p,p}_{m,m+1}=-\delta_{m,p}$.
For $p>q$ we have
$c^{p,q}_{m+1,m}=(-1)^{m+p+1}\binom{p-m}{m-q}$ when $q\leq m\leq\big[(p+q)/2\big]$ and zero otherwise,
and $c^{p,q}_{m,m+1}=(-1)^{m+p+1}\binom{p-m-1}{m-q-1}$ when $q+1\leq m\leq\big[(p+q)/2\big]$ and zero otherwise.
For $p<q$ they are obtained using the symmetry condition (ii) in Lemma \ref{110311:lem1}.
\end{remark}
\begin{remark}
The polynomials
$(1+x_1+\dots+x_k)^{m_0}x_1^{m_1}\dots x_k^{m_k}$,
with $m_0,m_1,\dots,m_k\in\mb Z_+$ such that $\mu_0-\mu_1=1$,
form a basis (over $\mb Z$) of the ring $\mb Z[x_1,\dots,x_k]$.
Indeed, the linear independence of these elements follows from
Lemma \ref{110311:lem0} in the special case when $V=\mb Q$ and $\partial=1$,
while the fact that they span the whole polynomial ring follows
from Lemma \ref{110311:lem1}.
\end{remark}
The following is a simple combinatorial result that will be used in the proof of the
subsequent Lemma \ref{110318:cor1}.
\begin{lemma}\label{110319:lem}
Let $m_1,\dots,m_k,n_1,\dots,n_k\in\!\mb Z_+$ be such that
$(\mu_1,\dots,\mu_k)<(\nu_1,\dots,\nu_k)$
in the lexicographic order.
Then:
\begin{enumerate}[(a)]
\item
$m_\alpha<n_\alpha$ for some $\alpha\in\{1,\dots,k\}$;
\item
suppose, moreover, that $m_\alpha<n_\alpha$ for exactly one index $\alpha\in\{1,\dots,k\}$,
and $m_\beta\geq n_\beta$ for every other $\beta\neq\alpha$,
and let $n_\alpha=\nu_i$,
then
$\mu_1=\nu_1,\dots,\mu_{i-1}=\nu_{i-1},\mu_i\leq\nu_i$.
\end{enumerate}
\end{lemma}
\begin{proof}
Let $\{\alpha_1,\dots,\alpha_k\}$ and $\{\beta_1,\dots,\beta_k\}$
be permutations of $\{1,\dots,k\}$ such that
$m_{\alpha_1}\geq m_{\alpha_2}\geq\dots\geq m_{\alpha_k}$
and $n_{\beta_1}\geq n_{\beta_2}\geq\dots\geq n_{\beta_k}$.
In particular, $\mu_i=m_{\alpha_i}$ and $\nu_i=n_{\beta_i}$
for every $i$.

For part (a), suppose, by contradiction, that $m_\alpha\geq n_\alpha$ for every $\alpha=1,\dots,k$.
Then, clearly, $\mu_1\geq m_{\beta_1}\geq n_{\beta_1}=\nu_1$.
Since, by assumption, $\mu_1\leq\nu_1$,
we conclude that $\mu_1=\nu_1$.
Suppose, by induction, that $\mu_1=\nu_1,\dots,\mu_{i-1}=\nu_{i-1}$
for $i\geq2$.
If $\{\alpha_1,\dots,\alpha_{i-1}\}=\{\beta_1,\dots,\beta_{i-1}\}$,
we have $\beta_i\not\in\{\alpha_1,\dots,\alpha_{i-1}\}$,
and therefore $\mu_i\geq m_{\beta_i}\geq n_{\beta_i}=\nu_i$.
Similarly, if $\{\alpha_1,\dots,\alpha_{i-1}\}\neq\{\beta_1,\dots,\beta_{i-1}\}$,
we have $\beta_j\not\in\{\alpha_1,\dots,\alpha_{i-1}\}$ for some $j\leq i-1$,
and therefore $\mu_i\geq m_{\beta_j}\geq n_{\beta_j}=\nu_j\geq\nu_i$.
In both cases we have $\mu_i\geq\nu_i$ and,
since, by assumption, $\mu_i\leq\nu_i$,
we conclude that $\mu_i=\nu_i$.
It follows that $(\mu_1,\dots,\mu_k)=(\nu_1,\dots,\nu_k)$,
a contradiction.

For part (b) we use a similar argument.
In our notation, $\alpha=\beta_i$ for some $i\geq 1$.
If $i=1$, there is nothing to prove.
If $i\geq2$,
we have $\mu_1\geq m_{\beta_1}\geq n_{\beta_1}=\nu_1$,
and therefore $\mu_1=\nu_1$.
Hence, the claim is proved for $i=2$.
Suppose next that $i\geq 3$ and assume, by induction,
that $\mu_1=\nu_1,\dots,\mu_{j}=\nu_{j}$, where $j\leq i-2$.
If $\{\alpha_1,\dots,\alpha_j\}=\{\beta_1,\dots,\beta_j\}$,
we have $\beta_{j+1}\not\in\{\alpha_1,\dots,\alpha_j\}$,
and therefore $\mu_{j+1}\geq m_{\beta_{j+1}}\geq n_{\beta_{j+1}}=\nu_{j+1}$.
Similarly, if $\{\alpha_1,\dots,\alpha_j\}\neq\{\beta_1,\dots,\beta_j\}$,
we have $\beta_h\not\in\{\alpha_1,\dots,\alpha_j\}$ for some $h\leq j$,
and therefore $\mu_{j+1}\geq m_{\beta_h}\geq n_{\beta_h}=\nu_h\geq\nu_{j+1}$.
In both cases we have $\mu_{j+1}\geq\nu_{j+1}$ and, therefore, $\mu_{j+1}=\nu_{j+1}$.
In conclusion, $\mu_1=\nu_1,\dots,\mu_{i-1}=\nu_{i-1}$,
proving the claim.
\end{proof}
\begin{lemma}\label{110318:cor1}
Let $m_1,\dots,m_k,n_1,\dots,n_k\in\{0,\dots, N-1\}$ be such that
$(\mu_1,\dots,\mu_k)<(\nu_1,\dots,\nu_k)$
in the lexicographic order.
Then:
\begin{enumerate}[(a)]
\item
$c^{N,n_1,\dots,n_k}_{\mu_1+1,m_1,\dots,m_k}=0$;
\item
$c^{n_\alpha,n_1,\dots,\stackrel{\alpha}{\check{N}},\dots,n_k}_{\mu_1+1,m_1,\dots,m_k}=0$
for every $\alpha=1,\dots,k$;
\item
$c^{N,n_1,\dots,n_k}_{\nu_1+1,n_1,\dots,n_k}=(-1)^{N-\nu_1-1}$;
\item
$c^{n_\alpha,n_1,\dots,\stackrel{\alpha}{\check{N}},\dots,n_k}_{\nu_1+1,n_1,\dots,n_k}=0$
for every $\alpha=1,\dots,k$.
\end{enumerate}
\end{lemma}
\begin{proof}
By Lemma \ref{110319:lem}(a),
$m_\gamma<n_\gamma$ for some $\gamma\neq0$.
Then, by property (vii) of Lemma \ref{110311:lem1}(b),
we have $c^{N,n_1,\dots,n_k}_{\mu_1+1,m_1,\dots,m_k}=0$,
proving part (a).

For part (b), we first observe that,
by property (vi) of Lemma \ref{110311:lem1}(b),
if $\alpha\neq0$ then
$c^{n_\alpha,n_1,\dots,\stackrel{\alpha}{\check{N}},\dots,n_k}_{\mu_1+1,m_1,\dots,m_k}=0$
unless $m_\alpha=\mu_1$.
Moreover, if $m_\gamma<n_\gamma$ for $\gamma\neq\alpha$,
again by property (vii) of Lemma \ref{110311:lem1}(b),
we have $c^{n_\alpha,n_1,\dots,\stackrel{\alpha}{\check{N}},\dots,n_k}_{\mu_1+1,m_1,\dots,m_k}=0$.
Hence, by Lemma \ref{110319:lem}(a),
we only have to prove (b)
in the case when
$\mu_1=m_\alpha<n_\alpha$,
and $m_\beta\geq n_\beta$ for every $\beta\neq\alpha$.
Suppose, in this case,
that $n_\alpha=\nu_i$, for $i\geq2$.
Then, by Lemma \ref{110319:lem}(b),
we have $\mu_1<m_\alpha<n_\alpha=\nu_i\leq \nu_1=\mu_1$,
which is impossible.
Hence, we are left to consider the case when
$m_\alpha=\mu_1<n_\alpha=\nu_1$ and $m_\beta\geq n_\beta$ for every $\beta\neq\alpha$.
By property (vii) of Lemma \ref{110311:lem1}(b),
we have that $c^{n_\alpha,n_1,\dots,\stackrel{\alpha}{\check{N}},\dots,n_k}_{\mu_1+1,m_1,\dots,m_k}=0$
unless $\mu_1+1\geq n_\alpha$.
Hence, we only need to consider the case when $\mu_1+1=\nu_1$.
But in this case $m_\alpha=\mu_1<\mu_1+1=\nu_1$,
and hence
$c^{n_\alpha,n_1,\dots,\stackrel{\alpha}{\check{N}},\dots,n_k}_{\mu_1+1,m_1,\dots,m_k}=0$
by property (vi) of Lemma \ref{110311:lem1}(b).
This completes the proof of part (b).

For $N=n_1+1$, we have $c^{N,n_1,\dots,n_k}_{\nu_1+1,n_1,\dots,n_k}=1$
by property (i) of Lemma \ref{110311:lem1}(b).
For $N\geq n_1+2$, we have, by the recursive formula \eqref{110313:eq6},
$$
c^{N,n_1,\dots,n_k}_{\nu_1+1,n_1,\dots,n_k}
=
-\sum_{\beta\neq0}
c^{N-1,n_1,\dots n_\beta+1,\dots ,n_k}_{\nu_1+1,n_1,\dots,n_k}
-c^{N-1,n_1,\dots ,n_k}_{\nu_1+1,n_1,\dots,n_k}\,.
$$
Since $N-1$ is the maximal upper index in all terms of the RHS,
the first term in the RHS is zero by property (vii) of Lemma \ref{110311:lem1}(b).
Hence, we get $c^{N,n_1,\dots,n_k}_{\nu_1+1,n_1,\dots,n_k}=-c^{N-1,n_1,\dots ,n_k}_{\nu_1+1,n_1,\dots,n_k}$,
which, by induction, implies part (c).

We are left to prove part (d).
If $N=\nu_1+1$,
by property (i) of Lemma \ref{110311:lem1}(b),
we have that
$c^{n_\alpha,n_1,\dots,\stackrel{\alpha}{\check{N}},\dots,n_k}_{\nu_1+1,n_1,\dots,n_k}=0$
for every $\alpha\neq0$, since $n_\alpha\neq N$.
For $N\geq n_1+2$, we have, by the recursive formula \eqref{110313:eq6},
$$
\begin{array}{c}
\displaystyle{
c^{n_\alpha,n_1,\dots,\stackrel{\alpha}{\check{N}},\dots,n_k}_{\nu_1+1,n_1,\dots,n_k}
=
-\sum_{\beta\neq0,\alpha}
c^{n_\alpha,n_1,\dots, n_\beta+1,\dots,\stackrel{\alpha}{\check{N-1}},\dots,n_k}_{\nu_1+1,n_1,\dots,n_k}
} \\
\displaystyle{
-c^{n_\alpha+1,n_1,\dots,\stackrel{\alpha}{\check{N-1}},\dots,n_k}_{\nu_1+1,n_1,\dots,n_k}
-c^{n_\alpha,n_1,\dots,\stackrel{\alpha}{\check{N-1}},\dots,n_k}_{\nu_1+1,n_1,\dots,n_k}\,.
}
\end{array}
$$
Note that $N-1$ is the maximal upper index in all terms of the RHS.
Hence,
the first term in the RHS is zero by property (vii) of Lemma \ref{110311:lem1}(b),
since $n_\beta<n_\beta+1$.
Moreover, the second term in the RHS is zero by property (vi) of Lemma \ref{110311:lem1}(b)
since $n_\alpha<n_\alpha+1\leq\nu_1$.
Hence, we get
$c^{n_\alpha,n_1,\dots,\stackrel{\alpha}{\check{N}},\dots,n_k}_{\nu_1+1,n_1,\dots,n_k}
=
-c^{n_\alpha,n_1,\dots,\stackrel{\alpha}{\check{N-1}},\dots,n_k}_{\nu_1+1,n_1,\dots,n_k}$,
which, by induction, implies that
$c^{n_\alpha,n_1,\dots,\stackrel{\alpha}{\check{N}},\dots,n_k}_{\nu_1+1,n_1,\dots,n_k}=0$,
proving (d).
\end{proof}
\begin{corollary}\label{110314:cor1}
If $\partial:\mc F\to\mc F$ is surjective,
then any polynomial $S\in\mb F[\lambda_1,\dots,\lambda_k]\otimes\mc F$
admits a decomposition of the following form:
$$
S(\lambda_1,\dots,\lambda_k)
=
\sum_{\substack{m_0,\dots,m_k=0 \\ (\mu_0-\mu_1=1)}}^N
\lambda_0^{m_0}\dots\lambda_k^{m_k}s_{m_0,\dots,m_k}\,,
$$
where $\lambda_0$ is as in \eqref{110313:eq2} (and it acts on the coefficients $s_{m_0,\dots,m_k}$).
\end{corollary}
\begin{proof}
Expanding $S$ and substituting each monomial $\lambda_1^{n_1}\dots\lambda_k^{n_k}$
with the RHS of \eqref{110311:eq1} (for $n_0=0$),
we get the desired expansion, using the assumption that $\partial$ is surjective.
\end{proof}

\subsubsection{The main results for polydifferential operators}\label{app:5.2}

\begin{corollary}\label{110314:cor2}
\begin{enumerate}[(a)]
\item
Let $S=\big(S_{i_0,\dots,i_k}(\lambda_1,\dots,\lambda_k)\big)_{i_0,\dots,i_k\in I}$
be a totally skewsymmetric $k$-differential operator on $\mc F^\ell$.
Assuming that $\partial:\,\mc F\to\mc F$ is surjective, $S$ admits a decomposition
\begin{equation}\label{110314:eq1}
S_{i_0,\dots,i_k}(\lambda_1,\dots,\lambda_k)
=
\sum_{\substack{m_0,\dots,m_k=0 \\ (\mu_0-\mu_1=1)}}^M
\lambda_0^{m_0}\dots\lambda_k^{m_k}s^{m_0,\dots,m_k}_{i_0,\dots,i_k}\,,
\end{equation}
with $\lambda_0$ is as in \eqref{110313:eq2},
where the coefficients $s^{m_0,\dots,m_k}_{i_0,\dots,i_k}\!\in\mc F$ are skewsymmetric
with respect to simultaneous permutations of upper and lower indices:
\begin{equation}\label{110316:eq5}
s^{m_{\sigma(0)},\dots,m_{\sigma(k)}}_{i_{\sigma(0)},\dots,i_{\sigma(k)}}
=
\sign(\sigma) s^{m_0,\dots,m_k}_{i_0,\dots,i_k}
\,\,\,\,\forall\sigma\in S_{k+1}\,.
\end{equation}
\item
Assuming that $\mc F$ is linearly closed,
the space of vectors $\{s^{m_0,\dots,m_k}_{i_0,\dots,i_k}\in\mc F\}$,
labeled by $i_0,\dots,i_k\in\{1,\dots,\ell\}$ and $m_0,\dots,m_k\in\{0,\dots,N\}$
such that $\mu_0-\mu_1=1$,
which are skewsymmetric with respect to simultaneous permutations of upper and lower indices,
and which solve
\begin{equation}\label{110406:eq7}
\sum_{\substack{m_0,\dots,m_k=0 \\ (\mu_0-\mu_1=1)}}^N
\lambda_0^{m_0}\dots\lambda_k^{m_k}s^{m_0,\dots,m_k}_{i_0,\dots,i_k}=0\,,
\end{equation}
has dimension over $\mc C$ equal to
\begin{equation}\label{110406:eq8}
D=
\sum_{n=0}^{N-1}\bigg(\binom{(n+1)\ell}{k+1}-\binom{n\ell}{k+1}-\ell\binom{n\ell}{k}\bigg)
\end{equation}
\end{enumerate}
\end{corollary}
\begin{proof}
By Corollary \eqref{110314:cor1}, each polynomial
$S_{i_0,\dots,i_k}(\lambda_1,\dots,\lambda_k)$
admits a decomposition as in \eqref{110314:eq1},
for some $M\in\mb Z_+$ and some coefficients $s^{m_0,\dots,m_k}_{i_0,\dots,i_k}\in\mc F$.
Applying total skewsymmetrization to both sides of \eqref{110314:eq1},
we can replace the coefficients $s^{m_0,\dots,m_k}_{i_0,\dots,i_k}$
by totally skewsymmetric ones, proving part (a).

All the solutions of equation \eqref{110406:eq7}
are given by elements $s^{m_0,\dots,m_k}_{i_0,\dots,i_k}\in\mc F$
satisfying conditions (i) and (ii) of Lemma \ref{110311:lem0}.
Therefore, the space of arrays $\{s^{m_0,\dots,m_k}_{i_0,\dots,i_k}\in\mc F\}$
as in part (b) is in bijective correspondence with the space of
arrays $T=\{t^{n_0,\dots,n_k}_{i_0,\dots,i_k}\in\mc F\}$,
labeled by $i_0,\dots,i_k\in\{1,\dots,\ell\}$ and $n_0,\dots,n_k\in\{0,\dots,N-1\}$
such that $\nu_0=\nu_1$,
which are skewsymmetric with respect to simultaneous permutations of upper and lower indices,
and which solve the system of equations
$$
\partial t^{n_0,\dots,n_k}_{i_0,\dots,i_k}
+\sum_{h=0}^k t^{n_0,\dots,n_h-1,\dots,n_k}_{i_0,\dots,i_k}=0\,.
$$
This is a system of linear differential equations of the form $\partial T+AT=0$,
hence, since by assumption $\mc F$ is linearly closed,
the space of solutions has dimension equal to the number of unknowns.

The functions $t^{n_0,\dots,n_k}_{i_0,\dots,i_k}$
are labeled by the index set
$$
\tilde{\mc C}
=
\{1,\dots,\ell\}^{k+1}\times
\Big\{
(n_0,\dots,n_k)\in\{0,\dots,N-1\}^{k+1}
\,\Big|\,
\nu_0=\nu_1
\Big\}
\,,
$$
and since they are skewsymmetric with respect to simultaneous permutations
of indices $n_0,\dots,n_k$ and $i_0,\dots,i_k$,
we can say that the entries of the array $T$ are labeled
by the $S_{k+1}$-orbits in $\tilde{\mc C}$ with trivial stabilizer.
Therefore $D=\#(\mc C)$, where
$$
\mc C=\Big\{\omega\in\tilde{\mc C}/S_{k+1}\,\Big|\,\text{Stab}(\omega)=\{1\}\Big\}\,.
$$
We can decompose the index set $\tilde{\mc C}$ as disjoint union of the subsets
$\tilde{\mc C}_{s,n},\,s=2,\dots,k+1,n=0,\dots,N-1$,
defined by
$$
\tilde{\mc C}_{s,n}
\!=
\{1,\dots,\ell\}^{k+1}\times
\Big\{\!
(n_0,\dots,n_k)\in\{0,\dots,N-1\}^{k+1}
\,\Big|\,
n=\nu_0=\nu_{s-1}>\nu_s
\!\Big\}
$$
and the action of the permutation group $S_{k+1}$ preserves each of these subsets.
Hence,
$D=\sum_{n=0}^{N-1}\sum_{s=2}^{k+1}\#(\mc C_{s,n})$, where
$$
\mc C_{s,n}=\Big\{\omega\in\tilde{\mc C}_{s,n}/S_{k+1}\,\Big|\,\text{Stab}(\omega)=\{1\}\Big\}\,.
$$
It is easy to see, by putting all maximal elements in the first positions,
that the set $\mc C_{s,n}$ is in bijection with the cartesian product
of the set of $S_s$-orbits with trivial stabilizer
in $\{1,\dots,\ell\}^{k+1}\times\{0,\dots,N-1\}^{k+1}$
of the form $((i_0,\dots,i_s),(n,\dots,n))$,
and the set of $S_{k+1-s}$-orbits with trivial stabilizer
in $\{1,\dots,\ell\}^{k+1}\times\{0,\dots,N-1\}^{k+1}$
of the form $((i_s,\dots,i_k),(n_s,\dots,n_k))$ with $n_s,\dots,n_k<n$.
Clearly, the cardinality of the first set is $\binom{\ell}{s}$,
and the cardinality of the second set is $\binom{n\ell}{k+1-s}$.
Hence,
$$
D
=\sum_{n=0}^{N-1}\sum_{s=2}^{k+1}\binom{\ell}{s}\binom{n\ell}{k+1-s}\,.
$$
This formula implies equation \eqref{110406:eq8},
since $\binom{(n+1)\ell}{k+1}=\sum_{s=0}^{k+1}\binom{\ell}{s}\binom{n\ell}{k+1-s}$.
\end{proof}

%

\begin{theorem}\label{110127:conj1}
Let $k\in\mb Z_+$, and
let $K(\partial)\in\Mat_{\ell\times\ell}(\mc F[\partial])$
be an $\ell\times\ell$ matrix differential operator
of order $N$ with invertible leading coefficient, over a
linearly closed differential field $\mc F$.
Then for every totally skewsymmetric $k$-differential operator $S$ on $\mc F^\ell$
there exists a skewsymmetric $k$-differential operator $P$ on $\mc F^\ell$
such that
\begin{equation}\label{110316:eq11}
\langle K\circ P\rangle^- = S\,.
\end{equation}
\end{theorem}
\begin{proof}
For $k=0$ we have $P\in\mc F^\ell$ and $\langle K\circ P\rangle^-=K(\partial)P$.
Hence, the claim follows from Corollary \ref{101218:thm2} by taking the majorant
$N_j=N\,\forall j,\,h_i=0\,\forall i$ of the matrix differential operator $K(\partial)$.

Next let $k\geq1$.
Let $K(\partial)=\sum_{n=0}^N(-\partial)^n\circ K_n$,
where $(-1)^NK_N\neq 0$ is the leading coefficient.
If we let $K_1(\partial)=K(\partial)\circ K_N^{-1}$
and $P_1(\lambda_1,\dots,\lambda_k)=K_NP_1(\lambda_1,\dots,\lambda_k)$,
we have, by \eqref{110126:eq4}, $K\circ P=K_1\circ P_1$.
Hence, we may assume that $K_N=\id$.

Let $S$ be a totally skewsymmetric $k$-differential operator on $\mc F^\ell$.
By Corollary \ref{110314:cor2}, $S$ admits a decomposition as in \eqref{110314:eq1}.
The first part of the proof will consist in reducing to the case when $M=N$.

Let $M_0\geq M_1\geq\dots\geq M_k$ be the maximal, in the lexicographic order,
among all non increasing $(k+1)$-tuples $\mu_0\geq \mu_1\geq\dots\geq \mu_k$
such that $s^{\mu_0,\dots,\mu_k}_{i_0,\dots,i_k}\neq0$ for some $i_0,\dots,i_k\in I$.
Clearly, by the skewsymmetry condition \eqref{110316:eq5},
for $m_0,\dots,m_k\in\mb Z_+$
we have that $s^{m_0,\dots,m_k}_{i_0,\dots,i_k}$ is zero unless $(\mu_0,\dots,\mu_k)\leq(M_0,\dots,M_k)$.
Hence, the decomposition \eqref{110314:eq1} of $S$ can be rewritten as follows
\begin{equation}\label{110316:eq4}
S_{i_0,\dots,i_k}(\lambda_1,\dots,\lambda_k)
=
\sum_{\substack{m_0,\dots,m_k\in\mb Z_+ \\
\mu_0-\mu_1=0 \text{ or } 1 \\
(\mu_0,\dots,\mu_k)\leq(M_0,\dots,M_k)
}}
\lambda_0^{m_0}\dots\lambda_k^{m_k}s^{m_0,\dots,m_k}_{i_0,\dots,i_k}
\,.
\end{equation}
Notice that in \eqref{110316:eq4}, for reasons that will become clear later,
we allow terms with $\mu_0-\mu_1$ equal 0 or 1 (even though, by Corollary \ref{110314:cor2},
we could restrict to the terms with $\mu_0-\mu_1=1$).
If $M_0\geq N$, let $P^0$ be the following $k$-differential operator on $\mc F^\ell$:
$$
P^0_{i_0,\dots,i_k}
=
C
\lambda_0^{M_0-N}\lambda_1^{M_1}\dots\lambda_k^{M_k}
s^{M_0,\dots,M_k}_{i_0,\dots,i_k}
\,.
$$
where $C$ denotes the cardinality of the orbit of $(M_0,\dots,M_k)$
under the action of the symmetric group $S_{k+1}$.
By \eqref{110126:eq4} and the assumption that $K_N=\id$, we have
\begin{equation}\label{110316:eq6}
\begin{array}{l}
\displaystyle{
(K\circ P^0)_{i_0,\dots,i_k}
= C
\lambda_0^{M_0}\lambda_1^{M_1}\dots\lambda_k^{M_k}
s^{M_0,\dots,M_k}_{i_0,\dots,i_k}
} \\
\displaystyle{
+ C
\sum_{j=1}^\ell\!\! \sum_{n=0}^{N-1} \sum_{h=0}^{M_0-N} \!\! \binom{M_0\!-\!N}{h}
\lambda_0^{M_0-N+n-h}
\lambda_1^{M_1}\!\!\dots\lambda_k^{M_k}
(\partial^h K_n)_{i_0,j}s^{M_0,M_1,\dots,M_k}_{j,i_1,\dots,i_k}
\,.
}
\end{array}
\end{equation}
Let $P$ be the skewsymmetrization (over $S_k$) of $P^0$.
Clearly, the skewsymmetrization (over $S_k$) of $K\circ P^0$
is equal to $K\circ P$,
and therefore $\langle K\circ P\rangle^-=\langle K\circ P^0\rangle^-$.
Hence, taking the total skewsymmetrization of both sides of \eqref{110316:eq6}, we get
\begin{equation}\label{110316:eq7}
\begin{array}{l}
\displaystyle{
\langle K\circ P\rangle^-_{i_0,\dots,i_k}(\lambda_1,\dots,\lambda_k)
} \\
\displaystyle{
=
\frac{C}{(k+1)!}
\sum_{\sigma\in S_{k+1}}\sign(\sigma)
\lambda_{\sigma^{-1}(0)}^{M_0}\lambda_{\sigma^{-1}(1)}^{M_1}\dots\lambda_{\sigma^{-1}(k)}^{M_k}
s^{M_0,\dots,M_k}_{i_{\sigma^{-1}(0)},\dots,i_{\sigma^{-1}(k)}}
} \\
\displaystyle{
+
\sum_{j=1}^\ell \sum_{n=0}^{N-1} \sum_{h=0}^{M_0-N} \binom{M_0\!-\!N}{h}
\frac{C}{(k+1)!}
\sum_{\sigma\in S_{k+1}}\sign(\sigma)
} \\
\displaystyle{
\vphantom{\Bigg(}
\times \lambda_{\sigma^{-1}(0)}^{M_0-N+n-h}
\lambda_{\sigma^{-1}(1)}^{M_1}\dots\lambda_{\sigma^{-1}(k)}^{M_k}
(\partial^h K_n)_{i_{\sigma^{-1}(0)},j}s^{M_0,M_1,\dots,M_k}_{j,i_{\sigma^{-1}(1)},\dots,i_{\sigma^{-1}(k)}}
\,.
}
\end{array}
\end{equation}
By the skewsymmetry condition \eqref{110316:eq5} on the coefficients $s^{m_0,\dots,m_k}_{i_0,\dots,i_k}$,
and since $\frac{(k+1)!}C$ is the cardinality of the stabilizer of $(M_0,\dots,M_k)$
under the action of $S_{k+1}$,
the first term in the RHS of \eqref{110316:eq7} is equal to
$$
\sum_{\substack{m_0,\dots,m_k\in\mb Z_+ \\
(\mu_0,\dots,\mu_k)=(M_0,\dots,M_k)
}}
\lambda_0^{m_0}\dots\lambda_k^{m_k}s^{m_0,\dots,m_k}_{i_0,\dots,i_k}
\,.
$$
Moreover, each monomial
$\lambda_{\sigma^{-1}(0)}^{M_0-N+n-h}
\lambda_{\sigma^{-1}(1)}^{M_1}\dots\lambda_{\sigma^{-1}(k)}^{M_k}$
which enters in the second term of the RHS of \eqref{110316:eq7}
can be expanded, using Lemma \ref{110316:lem1}, as
$$
\begin{array}{l}
\displaystyle{
\lambda_{\sigma^{-1}(0)}^{M_0-N+n-h}
\lambda_{\sigma^{-1}(1)}^{M_1}\dots\lambda_{\sigma^{-1}(k)}^{M_k}
} \\
\displaystyle{
=
\!\!\!\!
\sum_{\substack{m_0,\dots,m_k\in\mb Z_+ \\ (\mu_0-\mu_1=0 \text{ or } 1)}}
\!\!\!\!\!
b^{M_{\sigma(0)},\dots,\stackrel{\sigma^{-1}(0)}{\check{M_0-N+n-h}},\dots,M_{\sigma(K)}}_{m_0,\dots,m_k}
\lambda_0^{m_0}\lambda_1^{m_1}\dots\lambda_k^{m_k}
\partial^{\sum_i(M_i-m_i)-N+n-h}
}
\end{array}
$$
and, by the last statement of Lemma \ref{110316:lem1},
$b^{M_{\sigma(0)},\dots,\stackrel{\sigma^{-1}(0)}{\check{M_0-N+n-h}},\dots,M_{\sigma(K)}}_{m_0,\dots,m_k}$
is zero unless $\mu_0\leq N_0$ and,
for $\mu_0=N_0$, it is zero unless $m_\alpha=M_{\sigma(\alpha)}-\delta_{\alpha,\sigma^{-1}(0)}(N-n+h)$,
for every $\alpha=0,\dots,k$.
In particular, since $n\leq N-1$, this coefficient is zero
unless $(\mu_0,\dots,\mu_k)<(M_0,\dots,M_k)$.
Putting together the above observations, we can write
$\langle K\circ P\rangle^-$ as
\begin{equation}\label{110316:eq8}
\begin{array}{l}
\displaystyle{
\langle K\circ P\rangle^-_{i_0,\dots,i_k}(\lambda_1,\dots,\lambda_k)
} \\
\displaystyle{
=
\sum_{\substack{m_0,\dots,m_k\in\mb Z_+ \\
(\mu_0,\dots,\mu_k)=(M_0,\dots,M_k)
}}
\lambda_0^{m_0}\dots\lambda_k^{m_k}s^{m_0,\dots,m_k}_{i_0,\dots,i_k}
} \\
\displaystyle{
+
\sum_{\substack{m_0,\dots,m_k\in\mb Z_+ \\
\mu_0-\mu_1=0 \text{ or } 1 \\
(\mu_0,\dots,\mu_k)<(M_0,\dots,M_k)
}}
\lambda_0^{m_0}\dots\lambda_k^{m_k} t^{m_0,\dots,m_k}_{i_0,\dots,i_k}\,,
}
\end{array}
\end{equation}
for some coefficients $t^{m_0,\dots,m_k}_{i_0,\dots,i_k}\in\mc F$.
It follows that $S-\langle K\circ P\rangle^-$ has a decomposition as in \eqref{110316:eq4}
where only terms with $(\mu_0,\dots,\mu_k)<(M_0,\dots,M_k)$ appear.
Hence, by induction,
we can reduce to the case when $S$ has a decomposition as in \eqref{110316:eq4}
with $M_0\leq N-1$.

To conclude, we further reduce the monomials
$\lambda_0^{m_0}\dots\lambda_k^{m_k}$ in the expansion \eqref{110316:eq4} of $S$
with $\mu_0=\mu_1$ using Lemma \ref{110311:lem1}.
By property (iv) of Lemma \ref{110311:lem1}(b),
the only monomials which enter in the obtained decomposition are such that $\mu_0\leq M_0+1\leq N$.
Therefore, this $S$ admits a decomposition as in \eqref{110314:eq1} with $M=N$,
completing the first step of the proof.

Let then $S$ have the following decomposition:
\begin{equation}\label{110316:eq9}
S_{i_0,\dots,i_k}(\lambda_1,\dots,\lambda_k)
=
\sum_{\substack{m_0,\dots,m_k=0 \\ (\mu_0-\mu_1=1)}}^N
\lambda_0^{m_0}\dots\lambda_k^{m_k}s^{m_0,\dots,m_k}_{i_0,\dots,i_k}\,,
\end{equation}
with coefficients $s^{m_0,\dots,m_k}_{i_0,\dots,i_k}$ skewsymmetric
with respect to simultaneous permutations of upper and lower indices.
We want to prove that there exists a skewsymmetric $k$-differential operator $P$
of degree at most $N-1$ in each variable,
\begin{equation}\label{110316:eq10}
P_{i_0,\dots,i_k}(\lambda_1,\dots,\lambda_k)
=
\sum_{n_1,\dots,n_k=0}^{N-1}
\lambda_1^{n_1}\dots\lambda_k^{n_k}p^{n_1,\dots,n_k}_{i_0,i_1,\dots,i_k}\,,
\end{equation}
with coefficients $p^{m_1,\dots,m_k}_{i_0,\dots,i_k}\in\mc F$ skewsymmetric
with respect to the action of $S_k=Perm(1,\dots,k)$ on upper and lower indices simultaneously,
satisfying equation \eqref{110316:eq11}.
By \eqref{110126:eq4} we have
$$
(K\circ P)_{i_0,\dots,i_k}(\lambda_1,\dots,\lambda_k)
=
\!\!
\sum_{n_0=0}^N
\sum_{n_1,\dots,n_k=0}^{N-1}
\sum_{j=1}^\ell
\!
\lambda_0^{n_0}\lambda_1^{n_1}\dots\lambda_k^{n_k} (K_{n_0})_{i_0,j} p^{n_1,\dots,n_k}_{j,i_1,\dots,i_k}\,,
$$
and, by Lemma \ref{110311:lem1}, we get
\begin{equation}\label{110316:eq12}
\begin{array}{c}
\displaystyle{
(K\circ P)_{i_0,\dots,i_k}(\lambda_1,\dots,\lambda_k)
=
\sum_{\substack{m_0,\dots,m_k=0 \\ (\mu_0-\mu_1=1)}}^N
\sum_{n_0=0}^N\sum_{n_1,\dots,n_k=0}^{N-1}
\sum_{j=1}^\ell
c^{n_0,\dots,n_k}_{m_0,\dots,m_k}
} \\
\displaystyle{
\vphantom{\Bigg(}
\times \lambda_0^{m_0}\lambda_1^{m_1}\dots\lambda_k^{m_k}
\partial^{\sum_i(n_i-m_i)}
(K_{n_0})_{i_0,j} p^{n_1,\dots,n_k}_{j,i_1,\dots,i_k}\,.
}
\end{array}
\end{equation}
In the RHS above we can take the sum over $m_0,\dots,m_k\leq N$
for the following reason.
If $n_0=N \,(\,>n_1,\dots,n_k)$, by property (v) of Lemma \ref{110311:lem1}(b),
$c^{n_0,\dots,n_k}_{m_0,\dots,m_k}$ is zero unless $\mu_0\leq \nu_0=N$,
while, if $n_0\leq N-1$ we have, by property (iv) of Lemma \ref{110311:lem1}(b),
that $c^{n_0,\dots,n_k}_{m_0,\dots,m_k}$ is zero unless $\mu_0\leq \nu_0+1\leq N-1+1=N$.
Taking the skewsymmetrization of both sides of equation \eqref{110316:eq12}, we have,
by \eqref{110120:eq6} and by the symmetry property (ii) of Lemma \ref{110311:lem1}(b),
\begin{equation}\label{110316:eq14}
\begin{array}{c}
\displaystyle{
\langle K\circ P\rangle^-_{i_0,\dots,i_k}(\lambda_1,\dots,\lambda_k)
=
\frac1{k+1}
\sum_{\substack{m_0,\dots,m_k=0 \\ (\mu_0-\mu_1=1)}}^N
\sum_{n_0=0}^N\sum_{n_1,\dots,n_k=0}^{N-1}
\sum_{j=1}^\ell
} \\
\displaystyle{
\lambda_0^{m_0}\dots\lambda_k^{m_k}
\partial^{\sum_i(n_i-m_i)}
\Big(
c^{n_0,\dots,n_k}_{m_0,\dots,m_k}
(K_{n_0})_{i_0,j} p^{n_1,\dots,n_k}_{j,i_1,\dots,i_k}
} \\
\displaystyle{
\vphantom{\Bigg(}
-\sum_{\alpha=1}^k
c^{n_\alpha,n_1,\dots,\stackrel{\alpha}{\check{n_0}},\dots,n_k}_{m_0,\dots,m_k}
(K_{n_0})_{i_\alpha,j} p^{n_1,\dots,n_k}_{j,i_1,\dots,\stackrel{\alpha}{\check{i_0}},\dots,i_k}
\Big)
\,.
}
\end{array}
\end{equation}
Comparing equations \eqref{110316:eq9} and \eqref{110316:eq14}, we get the following equation
in $\mb F[\lambda_1,\dots,\lambda_k]\otimes\mc F$:
\begin{equation}\label{110317:eq1}
\begin{array}{c}
\displaystyle{
\sum_{\substack{m_0,\dots,m_k=0 \\ (\mu_0-\mu_1=1)}}^N
\!\!\!\!\!\!
\lambda_0^{m_0}\!\dots\lambda_k^{m_k}
\Bigg(
\frac1{k+1}
\sum_{n_0=0}^N\sum_{n_1,\dots,n_k=0}^{N-1}
\sum_{j=1}^\ell
\partial^{\sum_i(n_i-m_i)}
\Big(
c^{n_0,\dots,n_k}_{m_0,\dots,m_k}
(K_{n_0})_{i_0,j}
} \\
\displaystyle{
\vphantom{\Bigg(}
\times p^{n_1,\dots,n_k}_{j,i_1,\dots,i_k}
-\sum_{\alpha=1}^k
c^{n_\alpha,n_1,\dots,\stackrel{\alpha}{\check{n_0}},\dots,n_k}_{m_0,\dots,m_k}
(K_{n_0})_{i_\alpha,j} p^{n_1,\dots,n_k}_{j,i_1,\dots,\stackrel{\alpha}{\check{i_0}},\dots,i_k}
\Big)
- s^{m_0,\dots,m_k}_{i_0,\dots,i_k}
\Bigg)=0\,.
}
\end{array}
\end{equation}
The above equation should be read as an equation in
the unknown variables
\begin{equation}\label{110317:eq4}
X=
\Big(
p^{n_1,\dots,n_k}_{i_0,i_1,\dots,i_k}\in\mc F
\Big)_{\substack{1\leq i_0,\dots,i_k\leq\ell \\ 0\leq n_1,\dots,n_k\leq N-1}}\,,
\end{equation}
such that $p^{n_{\sigma(1)},\dots,n_{\sigma(k)}}_{i_0,i_{\sigma(k)},\dots,i_{\sigma(k)}}
=\sign(\sigma)p^{n_1,\dots,n_k}_{i_0,i_1,\dots,i_k}$ for every $\sigma\in S_k$,
and the element
\begin{equation}\label{110317:eq5}
B=
\Big(
s^{m_0,\dots,m_k}_{i_0,\dots,i_k}\in\mc F
\Big)_{\substack{1\leq i_0,\dots,i_k\leq\ell \\ 0\leq m_0,m_1,\dots,m_k\leq N \\ (\mu_0-\mu_1=1) }}\,,
\end{equation}
skewsymmetric with respect to the action of $S_{k+1}$,
is given.
To complete the proof of the theorem, we only need to show that this equation
admits a solution.

Note that the coefficient of $\lambda_0^{m_0}\dots\lambda_k^{m_k}$ in the LHS of \eqref{110317:eq1}
can be rewritten, up to the summand $-s^{m_0,\dots,m_k}_{i_0,\dots,i_k}$, as
\begin{equation}\label{110406:eq6}
\frac1{k+1}
\sum_{n_0,\dots,n_k=0}^{N}
\sum_{j=1}^\ell
\partial^{\sum_i(n_i-m_i)}
\sum_{\alpha=0}^k
(-1)^\alpha
c^{n_0,\dots,n_k}_{m_0,\dots,m_k}
(K_{n_\alpha})_{i_\alpha,j}
p^{n_0,n_1,\stackrel{\alpha}{\check{\dots}},n_k}_{j,i_0,\stackrel{\alpha}{\check{\dots}},i_k}\,,
\end{equation}
and, in this form, it is manifestly skewsymmetric with respect to simultaneous permutations
of $i_0,\dots,i_k$ and $m_0,\dots,m_k$.

The variables $p^{n_1,\dots,n_k}_{j_0,j_1,\dots,j_k}$ are labeled by the index set
$$
\tilde{\mc A}
=
\{1,\dots,\ell\}^{k+1}\times\{0,\dots,N-1\}^{k}\,.
$$
On the other hand, since, by assumption,
they are skewsymmetric with respect to simultaneous permutations of indices $n_1,\dots,n_k$
and $j_1,\dots,j_k$,
we can say that the entries of the variable $X$ are labeled
by the $S_k$-orbits in $\tilde{\mc A}$
with trivial stabilizer,
where $S_k=Perm(1,\dots,k)$ acts on the element $((j_0,j_1,\dots,j_k),(n_1,\dots,n_k))$
by fixing $j_0$ and permuting, simultaneously, the other entries.
We therefore write $X=\big(p_a\big)_{a\in\mc A}$, where
\begin{equation}\label{110317:eq2}
\mc A=\Big\{\omega\in\tilde{\mc A}/S_k\,\Big|\,\text{Stab}(\omega)=\{1\}\Big\}\,,
\end{equation}
and, for
\begin{equation}\label{110406:eq2}
a=S_k\cdot((j_0,j_1,\dots,j_k),(n_1,\dots,n_k))\in\mc A\,,
\end{equation}
we let $p_a=\pm p^{n_1,\dots,n_k}_{j_0,j_1,\dots,j_k}$.
The sign of $p_a$ is fixed by taking + for the unique representative of $a$ with $n_1\geq n_2\geq\dots\geq n_k$
and $j_s>j_{s+1}$ if $n_s=n_{s+1}$.
Similarly,
the functions $s^{m_0,\dots,m_k}_{i_0,\dots,i_k}$
are labeled by the index set
$$
\tilde{\mc B}
=
\{1,\dots,\ell\}^{k+1}\times
\Big\{
(m_0,\dots,m_k)\in\{0,\dots,N\}^{k+1}
\,\Big|\,
\mu_0-\mu_1=1
\Big\}
\,,
$$
and since, by assumption,
they are skewsymmetric with respect to simultaneous permutations of indices $m_0,\dots,m_k$
and $i_0,\dots,i_k$,
we can say that the entries of the given array $B$ are labeled
by the $S_{k+1}$-orbits in $\tilde{\mc B}$
with trivial stabilizer,
where $S_{k+1}=Perm(0,\dots,k)$ acts diagonally on the
element $((i_0,\dots,i_k),(m_0,\dots,m_k))$.
We therefore write $B=\big(s_b\big)_{b\in\mc B}$, where
\begin{equation}\label{110317:eq3}
\mc B=\Big\{\omega\in\tilde{\mc B}/S_{k+1}\,\Big|\,\text{Stab}(\omega)=\{1\}\Big\}\,,
\end{equation}
and, for
\begin{equation}\label{110406:eq3}
b=S_{k+1}\cdot((i_0,\dots,i_k),(m_0,\dots,m_k))\in\mc B\,,
\end{equation}
we let $s_b=\pm s^{m_0,\dots,m_k}_{i_0,\dots,i_k}$.
As before, the sign of $s_b$ is fixed by taking + for the unique representative of $b$ with $m_0=m_1+1>m_1\geq\dots\geq m_k$
and $i_s>i_{s+1}$ if $m_s=m_{s+1}$.

For $a$ as in \eqref{110406:eq2},
we let
\begin{equation}\label{110407:eq4}
\varphi(a)=S_{k+1}\cdot((j_0,j_1,\dots,j_k),(\max(n_1,\dots,n_k)+1,n_1,\dots,n_k))\,.
\end{equation}
It is not hard to check that $\varphi$ is a well-defined bijective map $\mc A\stackrel{\sim}{\rightarrow}\mc B$.
In particular, the vectors $X$ and $B$ have the same number of entries.
In fact,
\begin{equation}\label{110406:eq1}
\#(\mc A)=\#(\mc B)=
\ell\binom{N\ell}{k}\,.
\end{equation}

Equation \eqref{110317:eq1} is equivalent to the following equation
\begin{equation}\label{110406:eq5}
A(\partial)(M(\partial)X-B)=0\,,
\end{equation}
where
$X$ and $B$ are as in \eqref{110317:eq4} and \eqref{110317:eq5} respectively,
and $M(\partial)$ and $A(\partial)$ are defined as follows.
First, $M(\partial)=\big(L_{b,a}(\partial)\big)_{b\in\mc B,a\in\mc A}$,\
given by \eqref{110406:eq6},
is the square matrix pseudodifferential operator with entries
\begin{equation}\label{110317:eq8}
\begin{array}{c}
\displaystyle{
L_{b,a}(\partial)
=
\frac1{k+1}
\sum_{n_0=0}^N
\sum_{j=1}^\ell
\partial^{\sum_i(n_i-m_i)}
\Big(
c^{n_0,\dots,n_k}_{m_0,\dots,m_k}
(K_{n_0})_{i_0,j}
\delta_{j,j_0}\delta_{j_1,i_1}\dots\delta_{j_k,i_k}
} \\
\displaystyle{
-\sum_{\alpha=1}^k
c^{n_\alpha,n_1,\dots,\stackrel{\alpha}{\check{n_0}},\dots,n_k}_{m_0,\dots,m_k}
(K_{n_0})_{i_\alpha,j}
\delta_{j,j_0}\delta_{j_\alpha,i_0}
\delta_{j_1,i_1}\stackrel{\alpha}{\check{\dots}}\delta_{j_k,i_k}
\Big)
\,,
}
\end{array}
\end{equation}
for $a$ and $b$ as in \eqref{110406:eq2} and \eqref{110406:eq3} respectively.
Note that in order to say that the entries of the matrix $M(\partial)$ can be labeled by the set $\mc B\times\mc A$
(and not by the set $\tilde{\mc B}\times\mc A$)
we are using the fact that $M(\partial)X$, given by \eqref{110406:eq6},
is manifestly skewsymmetric with respect to the action of $S_{k+1}$.
To define $A(\partial)$, consider first the map
$\tilde A$ from $\mc F^{\tilde{\mc B}}$ to the space $\mc D$ of $k$-differential operators on $\mc F^\ell$
of degree at most $2N-1$ in each variable, given by
$\tilde B=\{s^{m_0,\dots,m_k}_{i_0,\dots,i_k}\}
\mapsto \tilde A(\tilde B)$, where
$$
\tilde A(\tilde B)_{i_1,\dots,i_k}(\lambda_1,\dots,\lambda_k)
=\sum_{\tilde{\mc B}}(-\lambda_1-\dots-\lambda_k-\partial)^{m_0}\lambda_1^{m_1}\dots\lambda_k^{m_k}s^{m_0,\dots,m_k}_{i_0,\dots,i_k}\,.
$$
Note that $\tilde A$ is a $\mc C$-linear, not an $\mc F$-linear map,
but the space $\mc D$ is in fact a finite dimensional vector space over $\mc F$, say of dimension $d$,
and, for any choice of basis of it, $\tilde A$ becomes a $d\times\#(\tilde{\mc B})$ matrix differential operator.
We next consider $\mc F^{\mc B}$ as the subspace of $\mc F^{\tilde{\mc B}}$
consisting of the elements $\tilde B=\{s^{m_0,\dots,m_k}_{i_0,\dots,i_k}\}$ skewsymmetric with respect
to the action of $S_{k+1}$.
The restriction of $\tilde A$ to the subspace $\mc F^{\mc B}\subset\mc F^{\tilde{\mc B}}$
is then a $d\times\#(\mc B)$ matrix differential operator
from $\mc F^{\mc B}$ to $\mc D$, for any choice of basis of $\mc D$ over $\mc F$
(once we fix representatives, we can consider $\mc B$ as a subset of $\tilde{\mc B}$,
and the matrix of $\tilde A|_{\mc F^{\mc B}}$
consists of the rows of the matrix of $\tilde A$ corresponding to the indices in $\mc B$).
By Lemma \ref{110406:lem}, we can choose a basis of $\mc D$ such that the matrix for
$\tilde A|_{\mc F^{\mc B}}:\,\mc F^{\mc B}\to\mc D$
is in row echelon form, say with $k$ pivots and $d-k$ zero rows.
We then let $A(\partial):\,\mc F^{\mc B}\to\mc F^k$
the $k\times\#(\mc B)$ matrix differential operator given by the first $k$ rows of this matrix.
It is then clear from the above construction that equation \eqref{110317:eq1} is equivalent to equation \eqref{110406:eq5}.
Therefore, to prove the theorem,
we need to show that, for every $B\in\mc F^{\mc B}$
there exists $X\in\mc F^{\mc A}$ solving equation \eqref{110406:eq5}.

Notice that, even though in the expression of the LHS of \eqref{110317:eq1}
there appear negative powers of $\partial$,
we know that all such negative powers cancel out when we compute all the sums in \eqref{110317:eq1},
since this equation is the same as \eqref{110316:eq11}, which does not involve any negative power of $\partial$.
In terms of equation \eqref{110406:eq5} this means that,
even though $M(\partial)$ is a matrix pseudodifferential operator,
the product $A(\partial)M(\partial)$ is a matrix differential operator.

By construction, $A(\partial)$ is surjective, and its kernel consists of
elements $B=\{s^{m_0,\dots,m_k}_{i_0,\dots,i_k}\}$ solving equation \eqref{110406:eq7}.
Hence, by Corollary \ref{110314:cor2}(b), it is of finite dimension over $\mc C$.
Therefore, by Theorem \ref{110404:thm}(c), it is a square matrix differential operator with non zero determinant.

Next, note that $L_{b,a}(\partial)$ has order less than or equal to $N_a-h_b$, where,
for $a$ and $b$ as in \eqref{110406:eq2} and \eqref{110406:eq3} respectively,
\begin{equation}\label{110407:eq1}
N_a=N+\sum_{i=1}^k n_i
\,\,,\,\,\,\,
h_b=\sum_{i=0}^k m_i\,.
\end{equation}
The leading matrix associated to this majorant (see equation \eqref{101221:eq1b}) is
$\bar M(\xi)=\big(m_{b,a}\xi^{N_a-h_b}\big)_{b\in\mc B,a\in\mc A}$, where
$$
\begin{array}{l}
\displaystyle{
m_{b,a}
=
\frac1{k+1}
\Big(
c^{N,n_1,\dots,n_k}_{m_0,\dots,m_k}
\delta_{i_0,j_0}\delta_{j_1,i_1}\dots\delta_{j_k,i_k}
} \\
\displaystyle{
-\sum_{\alpha=1}^k
c^{n_\alpha,n_1,\dots,\stackrel{\alpha}{\check{N}},\dots,n_k}_{m_0,\dots,m_k}
\delta_{j_0,i_\alpha}\delta_{j_\alpha,i_0}
\delta_{j_1,i_1}\stackrel{\alpha}{\check{\dots}}\delta_{j_k,i_k}
\Big)\,.
}
\end{array}
$$
We want to prove that the leading matrix $\bar M(\xi)$ or, equivalently, $\bar M(1)$ (see \eqref{110404:eq1}),
is non degenerate.

In order to prove this, we fix a total ordering of the sets $\mc A\simeq\mc B$
(identified via $\varphi$),
and we prove that, with respect to this ordering, the matrix $\bar M(1)$ is lower triangular
with non zero diagonal entries.

Given elements
$b=S_{k+1}\cdot((i_0,\dots,i_k),(m_0,\dots,m_k))\in\mc B$
and $b'=S_{k+1}\cdot((i'_0,\dots,i'_k),(m'_0,\dots,m'_k))\,\in\mc B$,
we say that $b>b'$ if
$(\mu_0,\dots,\mu_k)>(\mu'_0,\dots,\mu'_k)$ in the lexicographic order,
or $(\mu_0,\dots,\mu_k)=(\mu'_0,\dots,\mu'_k)$ and $b>b'$ in some total ordering
of the remaining indices (which will play no role).
Therefore,
for $a\in\mc A$ and $b\in\mc B$ as in \eqref{110406:eq2} and \eqref{110406:eq3} respectively,
using the map $\varphi$ we have that $b>a$ if
$(\mu_0,\mu_1,\dots,\mu_k)>(\nu_1+1,\nu_1,\dots,\nu_k)$ in the lexicographic order,
or $(\mu_0,\mu_1,\dots,\mu_k)=(\nu_1+1,\nu_1,\dots,\nu_k)$ and $b>\varphi(a)$ in some total ordering
of the remaining indices.

We want to prove that, with respect to this ordering, for $a\in\mc A$ and $b\in\mc B$, we have
\begin{equation}\label{110318:eq1}
m_{b,a}=0 \text{ if } b<a
\,\,,\,\,\,\,
\text{ and }
m_{b,a}\neq 0 \text{ if } b=\varphi(a)\,.
\end{equation}
This follows from Lemma \ref{110318:cor1}.
Indeed, for $b\leq a$, we have in particular that
$(\mu_0,\mu_1,\dots,\mu_k)\leq(\nu_1+1,\nu_1,\dots,\nu_k)$.
Hence, by Lemma \ref{110318:cor1}(a) and (b)
we have $m_{b,a}=0$ unless $(\mu_0,\mu_1,\dots,\mu_k)=(\nu_1+1,\nu_1,\dots,\nu_k)$,
and, in this case,
by Lemma \ref{110318:cor1}(c) and (d),
$m_{b,a}=\frac{(-1)^{N-\nu_1-1}}{k+1}
\delta_{i_0,j_0}\delta_{j_1,i_1}\dots\delta_{j_k,i_k}$.

Summarizing the above results, $M(\partial)$ is a square matrix pseudodifferential operator
of size $\#(\mc A)=\#(\mc B)$, with non degenerate leading matrix associated 
to the majorant \eqref{110407:eq1},
$A(\partial)$ is a square matrix differential operator of the same size, with non zero determinant,
and $A(\partial)M(\partial)$ is a matrix differential operator.
By Corollary \ref{110328:cor}, it follows that equation \eqref{110406:eq5} 
has a solution for every $B$,
completing the proof of the theorem.
\end{proof}
\begin{theorem}\label{110127:conj2}
Let $\mc F$ be a linearly closed differential field with subfield of constants $\mc C\subset\mc F$.
Let $k\in\mb Z_+$, and
let $K(\partial)\in\Mat_{\ell\times\ell}(\mc F[\partial])$
be an $\ell\times\ell$ matrix differential operator
of order $N$ with invertible leading coefficient, over $\mc F$.
Then, the set of skewsymmetric $k$-differential operators $P$ on $\mc F^\ell$
of degree at most $N-1$ in each variable
such that
\begin{equation}\label{110213:eq5}
\langle K\circ P\rangle^-=0\,,
\end{equation}
is a vector space over $\mc C$ of dimension
$$
d=\binom{N\ell}{k+1}\,.
$$
\end{theorem}
\begin{proof}
As in the proof of Theorem \ref{110127:conj1},
for $k=0$ we have $P\in\mc F^\ell$ and $\langle K\circ P\rangle^-=K(\partial)P$.
Hence, the statement follows from Corollary \ref{101218:thm2} by taking the majorant
$N_j=N\,\forall j,\,h_i=0\,\forall i$ of the matrix differential operator $K(\partial)$.

Let then $k\geq1$. By the discussion in the proof of Theorem \ref{110127:conj1},
equation \eqref{110213:eq5} is the same as equation \eqref{110406:eq5} with $B=0$,
and, moreover, by Corollary \ref{110328:cor}, the space of solutions has dimension over $\mc C$
equal to
\begin{equation}\label{110407:eq2}
d=\dim_{\mc C}(\ker A(\partial))+\sum_{a\in\mc A}(N_a-h_{\varphi(a)})\,.
\end{equation}
By the construction of the matrix differential operator $A(\partial)$,
the equation $A(\partial)B=0$, for $B=\{s^{m_0,\dots,m_k}_{i_0,\dots,i_k}\}\in\mc F^{\mc B}$,
is equivalent to equation \eqref{110406:eq7}.
Hence, by Corollary \ref{110314:cor2}(b), $\ker(A(\partial))$
has dimension over $\mc C$ equal to \eqref{110406:eq8}.
Moreover, we have, recalling \eqref{110317:eq2}, \eqref{110407:eq4} and \eqref{110407:eq1}
(letting $a\in\mc A$ as in \eqref{110406:eq2}),
\begin{equation}\label{110407:eq3}
\begin{array}{l}
\displaystyle{
\sum_{a\in\mc A}(N_a-h_{\varphi(a)})
}\\
\displaystyle{
=
\sum_{a\in\mc A}\big(N-1-\max(n_1,\dots,n_k)\big)
=
\sum_{n=0}^{N-1}(N-n-1)\#\Big\{a\in\mc A\,\Big|\,\nu_1=n\Big\}
}\\
\displaystyle{
=
\sum_{n=0}^{N-1}(N-n-1)
\Big(
\#\Big\{a\in\mc A\,\Big|\,\nu_1\leq n\Big\}
-\#\Big\{a\in\mc A\,\Big|\,\nu_1\leq n-1\Big\}
\Big)
}\\
\displaystyle{
=
\ell\sum_{n=0}^{N-1}(N-1-n)\Big(\binom{(n+1)\ell}{k}-\binom{n\ell}{k}\Big)\,.
}
\end{array}
\end{equation}
In the last identity we used equation \eqref{110406:eq1}, with $N-1$ replaced by $n$ or $n-1$.
Putting together \eqref{110406:eq8} and \eqref{110407:eq3}, we get
$$
d=
\sum_{n=0}^{N-1}\bigg(\binom{(n+1)\ell}{k+1}-\binom{n\ell}{k+1}
+\ell(N-1-n)\binom{(n+1)\ell}{k}-\ell(N-n)\binom{n\ell}{k}\bigg)
\,,
$$
which is a telescopic sum equal to $\binom{N\ell}{k+1}$,
proving the theorem.
\end{proof}



\end{document}